\documentclass[11pt,letterpaper]{article}

\usepackage{amssymb,amsthm,amsmath,multirow,amsfonts,enumerate,setspace,pdflscape,lscape,mathtools,comment,upgreek,bbm}
\usepackage[T1]{fontenc}
\usepackage{lmodern}
\usepackage[margin=1in]{geometry}
\usepackage[compact]{titlesec}
\usepackage[normalem]{ulem}
\usepackage{float}
\usepackage[makeroom]{cancel}
\usepackage[round,authoryear]{natbib}
\usepackage[colorlinks=true,linkcolor=red,urlcolor=blue,citecolor=blue,anchorcolor=blue,
pdftex,breaklinks,pdfencoding=auto,psdextra,
bookmarksopenlevel=1,bookmarksopen=true]{hyperref}
\usepackage{xr}
\usepackage{booktabs}
\usepackage{graphicx,epstopdf,color,xcolor}
\graphicspath{{graphics/}}
\usepackage{subcaption} 
\makeatletter
\DeclareRobustCommand\citepos
{\begingroup\def\NAT@nmfmt##1{{\NAT@up##1's}}%
	\NAT@swafalse\let\NAT@ctype\z@\NAT@partrue
	\@ifstar{\NAT@fulltrue\NAT@citetp}{\NAT@fullfalse\NAT@citetp}}
\makeatother

\makeatletter
\renewcommand*{\eqref}[1]{\hyperref[{#1}]{\textup{\tagform@{\ref*{#1}}}}}
\makeatother

\definecolor{pink}{RGB}{219, 48, 122} 
\usepackage[inline,shortlabels]{enumitem}

\expandafter
\def \expandafter \normalsize \expandafter{\normalsize \setlength \abovedisplayskip{8pt plus 2pt minus 5pt}}
\expandafter
\def \expandafter \normalsize \expandafter{\normalsize \setlength \abovedisplayshortskip{0pt plus 2pt}}
\expandafter
\def \expandafter \normalsize \expandafter{\normalsize \setlength \belowdisplayskip{8pt plus 2pt minus 5pt}}
\expandafter
\def \expandafter \normalsize \expandafter{\normalsize \setlength \belowdisplayshortskip{5pt plus 2pt minus 3pt}}

\DeclareMathOperator{\ran}{ran}

\DeclareMathOperator{\spn}{span}

\DeclareMathOperator{\HS}{HS}

\DeclareMathOperator{\sgn}{sgn}
\DeclareMathOperator{\K}{{\mathrm{K}}}

\theoremstyle{plain}
\newtheorem{theorem}{Theorem}
\newtheorem{theo}{Theorem}

\newtheorem{lemma}{Lemma}
\newtheorem{corollary}{Corollary}
\theoremstyle{definition}
\newtheorem{assumption}{Assumption}
\theoremstyle{example}
\newtheorem{example}{Example}
\newtheorem{rmk}{Remark}
\newenvironment{remark}{\begin{rmk}}{\hfill\end{rmk}}
\newtheorem{assumpM}{Assumption}
 
\newtheorem{assumpMM}{Assumption}
 
\newtheorem{assumpA}{Assumption}
 
\newtheorem{assumpB}{Assumption}
 
\def\dto{\overset{d}\rightarrow}
\def\pto{\overset{p}\rightarrow}

\numberwithin{equation}{section}

\makeatletter

\usepackage{upgreek}
\newcommand{\abar}{\upalpha}

\newcommand{\bt}{\textcolor{black}}  
  
\newcommand{\op}{\mathrm{op}}  
\renewcommand{\hat}{\widehat}

\newcommand{\commRV}[1]{{\leavevmode\color{black}#1}}
\newcommand{\commRVV}[1]{{\leavevmode\color{black}#1}}
\begin{document} 
	
	\title{Functional instrumental variable regression with an application to estimating the impact of immigration on native wages\thanks{The authors express
deep appreciation to the co-editor and three anonymous referees for their invaluable and insightful suggestions. We are also thankful to  Morten {\O}.\ Nielsen  and seminar participants at the University of Sydney, University of Queensland and SETA 2022 for their helpful comments. Data and R code to replicate the empirical results in Table~\ref{tab.emp1} are available on the authors' websites. } }
	\author{Dakyung Seong\textsuperscript{a}\thanks{Corresponding author. Address: School of Economics, University of Sydney, Camperdown, 2006, NSW, Australia.\\ E-mail addresses: \texttt{dakyung.seong@sydney.edu.au (D.\ Seong), \texttt{won-ki.seo@sydney.edu.au (W.-K.\ Seo)}}}
		\qquad\qquad
		Won-Ki Seo\textsuperscript{a}   \\
		\small{\textsuperscript{a} School of Economics, University of Sydney, Australia}
	}
	\date{}
	\maketitle
	\vspace{-2em}
	\begin{abstract}
		Functional linear regression gets its popularity as a statistical tool to study the relationship between function-valued response and exogenous explanatory variables. However, in practice, it is hard to expect that the explanatory variables of interest are perfectly exogenous, due to, for example, the presence of omitted variables and measurement error. Despite its empirical relevance, it was not until recently that this issue of endogeneity was studied in the literature on functional regression, and the development in this direction does not seem to sufficiently meet practitioners' needs;  for example, this issue has been discussed with paying particular attention on consistent estimation and thus distributional properties of the proposed estimators still remain to be further explored. 	To fill this gap, this paper proposes new consistent FPCA-based instrumental variable estimators and  develops their asymptotic properties in detail. 
Simulation experiments under a wide range of settings show that the proposed estimators perform considerably well.  We apply our methodology to estimate the impact of immigration on native wages.	
		
		
	\end{abstract}
	
	\section{Introduction \label{sec:intro}}
	The recent developments in data collection and storage technologies ignite studies on how to use more complicated observations such as curves, probability density functions, or images. This area of study, commonly called functional data analysis, has become popular in statistics, and researchers in various fields, including economics, have been benefited from advances in this area. In particular, for practitioners who are interested in studying the relationship between two or more such variables, functional linear models are of central importance, and crucial contributions on this topic include \cite{Bosq2000}, \cite*{Yao2005}, \cite{Mas2007}, \cite{Hall2007}, \cite{Park2012}, \cite{crambes2013asymptotics}, \cite*{Benatia2017}, \cite{imaizumi2018}, \cite*{Sun2018}, and \cite*{Chen_et_al_2020} to name only a few. 
	
	The existing statistical approaches for estimating the functional linear model, including those proposed in the aforementioned literature, 
	are mostly established under the assumption that the explanatory variable of interest is exogenous, meaning that it is uncorrelated with the regression error. However, this assumption is not likely to hold in practice;  that is, explanatory variables are often {\it{endogenous}}. The issue of endogeneity is particularly relevant in the context of functional linear models because functional observations used in the analysis are typically constructed by smoothing their discrete, and often sparsely observed, realizations 
(see, e.g., \citealp{Yao2005}). If this being the case, the functional observations may inevitably involve small or large measurement errors, which leads to the violation of the exogeneity condition  at least to some degree (see, e.g., Sections \ref{model} and  \ref{sec:sim2}). This issue may hinder practitioners from applying the functional linear model. 
	
	If the volume of the literature on functional linear regression under the exogeneity condition provides any indication, the developments made so far to deal with endogeneity do not seem to sufficiently meet  practitioners' needs. Although a few papers, such as \cite{Benatia2017} and \cite{Chen_et_al_2020}, study the issue of endogeneity in the functional linear model, not much is known about asymptotic distributions of their estimators and how to implement statistical inference on the parameter of interest;  
		this may limit the practical applicability of the functional endogenous linear model. We will fill this gap to some extent by providing new estimators and inferential methods based on their asymptotic properties. This is a crucial point where the present paper is differentiated from the existing ones concerning the issue of endogeneity in the functional linear model.

Specifically, this paper provides new estimation results for the functional endogenous linear model  based on (i) the functional principal component analysis (FPCA) and (ii) the instrumental variable (IV) approach. The former has been widely adopted by researchers dealing with functional data (see e.g., \citealp{Ramsay2005}; \citealp{shang2014survey}), and the latter has also been widely adopted in order to address endogeneity not only in the conventional Euclidean space setting (see e.g., \citealp{Bekker1994}; \citealp{Chao2005}; \citealp{Newey2009a}), but also in the setting involving functional observations  (see, e.g., \citealp{Carrasco2012}; \citealp{Florence2015}; \citealp{Benatia2017}; \citealp{Chen_et_al_2020};  \citealp{Babii2021}). However, the application of the FPCA to the functional endogenous linear model has not been fully explored.


We consider the case where the response variable $y_t$, explanatory variable $x_t$, and instrumental variable $z_t$ are all function-valued; of course, with slight modifications, our results to be subsequently given can be adjusted for the case where $y_t$ is scalar- or vector-valued.  Unlike in the most of the papers mentioned above, we do not require the variables of interest to be independent and identically distributed (iid), but allow those to exhibit some weak dependence so that our methodology can be applied to various empirical examples. Given that many functional observations considered in the literature for applications in fields of energy, environmental and financial economics tend to involve time dependence, this extension may  be attractive to practitioners. 

 Among the aforementioned papers, the study by \cite{Chen_et_al_2020} is most closely related to the present paper in the sense that they consider FPCA-based consistent estimation of function-on-function regression models with endogeneity introduced by measurement errors. \cite{Benatia2017} earlier considered a similar model and proposed a consistent estimation method, but their theoretical results are obtained from a quite different theoretical methodology (ridge-type regularization). We complement these studies by providing new FPCA-based estimators and in-depth discussion on their asymptotic properties.

Technically, we view the function-valued variables of interest as random variables taking values in a Hilbert space of square-integrable functions, and then propose our FPCA-based functional IV estimator (FIVE).  As is well known in the literature, estimation of a model involving function-valued random variables is not straightforward because some important sample operators, such as the covariance of such a random variable, are not invertible over the entire Hilbert space(s). 	We circumvent this issue by employing a rank-regularized inverse of such an operator,  and this is the point where we make use of the FPCA.   The reason why we focus on this  regularzation scheme comes not from its theoretical superiority, but merely from its popularity in the literature. Other schemes such as ridge-type regularization (e.g., \citealp{Florence2015}; \citealp{Benatia2017}) may be alternatively adopted, and are expected to have their own merits (see Remark~\ref{rem:comb}).
\phantomsection\label{rvlabel01}\commRVV{It is worth summarizing some crucial differences between our estimators and the alternative estimator proposed by \cite{Benatia2017} based on ridge regularization. To the best of our knowledge, there has not been exploration of asymptotic inference on specific characteristics of the regression (coefficient) operator for their estimator; this seems to be because of a nontrivial challenge associated with a particular asymptotic bias (see Section 4 of \citealp{Benatia2017}). In contrast, in this paper, we tackle the issue by employing the FPCA augmented with a proper extension of the asymptotic approach introduced by \cite{Hall2007}. As a result, this paper provides mathematical conditions that support a valid asymptotic inference on the regression operator. Moreover, \citepos{Benatia2017} methodology does not take into account for the more general presence of weak dependence, although we believe their results can be extended to such a setting.}

This paper studies  in depth the asymptotic properties of the proposed estimators. It is first shown that, under some mild conditions on the data generating process of $\{y_t,x_t,z_t\}_{t\geq1}$,  the FIVE achieves the weak (convergence in probability) and strong (almost sure convergence) consistencies as long as the regularization parameter, which is introduced for a rank-regularized inverse of a certain sample operator used to construct the estimator, decays to zero at an appropriate rate. We then establish more detailed asymptotic properties of the FIVE under some nonrestrictive assumptions on the eigenstructure of the cross-covariance operator of the explanatory variable $x_t$ and the IV $z_t$. By doing so, we can see how the cross-covariance structure of $x_t$ and $z_t$ and  the choice of the regularization parameter affect the convergence of the FIVE toward its true counterpart. In addition to these results, we show that 
the FIVE is asymptotically normal in a pointwise sense if it is centered at a certain operator that is slightly biased from the true parameter of interest;  moreover, if certain additional conditions are satisfied, such a bias becomes asymptotically negligible and thus, in this case,  the FIVE centered at the true parameter becomes asymptotically normal.   The asymptotic normality results given in this paper are quite different from similar results given in a finite dimensional setting in the sense that the convergence rate is (i) possibly random and (ii) not uniformly given over the entire Hilbert space on which our estimator is defined. This result implies that the proposed estimator does not weakly converge to any elements in the usual operator topology, which \commRV{generalizes} what  \cite{Mas2007} earlier found in the context of functional autoregressive (AR) models of order 1.  
Based on our study of the FIVE, we also propose a different but closely related estimator, called the functional two-stage least square estimator (F2SLSE) and obtain its asymptotic properties in a similar manner.  We discuss how  our estimators and their asymptotic properties can be used to implement usual statistical inference on the parameter of interest. 

To see how the asymptotic properties of our estimators are revealed in finite samples, we implement Monte Carlo experiments under various simulation designs. The simulation results are quite satisfactory. Overall, it seems that our estimators can be good alternatives or sometimes complements to some existing estimators that are closely related to ours.

As an empirical illustration, we study the impact of immigration on native wages. Specifically, we employ a model that is similar to those considered by, e.g., \cite*{Dustmann2012} and \cite{Sharpe2020101902}. The previous literature in this area, including \cite{Ottaviano2011}, \cite{Card2009}, and the aforementioned articles, show that  an inflow of immigrants differently affects native wages depending on skill levels (captured by, e.g., years of education and experience) of both natives and immigrants.	We, in this paper, investigate such heterogeneous effects using our functional linear model, which is initiated by viewing both the labor supply and the native wage  as functions of a certain measure of workers' skill (will be detailed in Section~\ref{sec:emp}). This approach has a couple of advantages compared to that taken in the earlier literature. For example, in the previous literature, workers of various levels of skill are often classified into a few skill groups before analysis, which  is necessitated to reduce dimensionality of the considered model (see Example \ref{exam1} and Section \ref{sec:emp}). However, such a pre-classification, which may affect estimation results and their interpretation, is not required in our approach. Moreover, our methodology allows for studying if an inflow of immigrants in a particular skill group heterogenously affects workers equipped with different skill levels.  Using the  methodology developed in this paper, we find evidence supporting the presence of heterogeneous effects of immigration.

This paper is organized as follows. Section~\ref{model} introduces a functional endogenous linear model and provides  motivating examples. In Section~\ref{sec:estimators}, we define the FIVE and discuss its asymptotic properties. Section~\ref{sec:f2sls00} introduces the F2SLSE and discusses its asymptotic properties. 
Section~\ref{sec:sim} reports simulation results and details our empirical example. Section~\ref{sec:con} concludes. The mathematical proofs of the theoretical results can be found in the Supplementary Material.  
\section{Functional endogenous linear model}\label{model}
\subsection{Endogeneity and motivating examples}\label{model1}
We suppose that a stationary  sequence of random functions $\{y_t, x_t,u_t\}_{t\geq 1}$ satisfies the following:
\begin{equation} \label{eqmodel0}
y_t = c_y +  \mathcal Ax_t + u_t,  
\end{equation}
where $c_y$ is the intercept function and  $\mathcal A$ is a linear operator satisfying certain conditions to be clarified.  In \eqref{eqmodel0}, $y_t$, $x_t$ and $u_t$ will be technically understood as random variables taking values in separable Hilbert spaces. Appendix~\ref{appintro} briefly introduces the definitions of a Hilbert-valued random variable $X$, its expectation (denoted $\mathbb{E}[X]$), covariance operator (denoted $\mathcal C_{XX}\coloneqq\mathbb{E}[(X-\mathbb{E}[X]) \otimes (X-\mathbb{E}[X])]$),    and cross-covariance operator with another Hilbert-valued random variable $Y$ (denoted $\mathcal C_{XY} \coloneqq \mathbb{E}[(X-\mathbb{E}[X]) \otimes (Y-\mathbb{E}[Y])]$), where $\otimes$ signifies the tensor product defined \commRV{by $X\otimes Y(\cdot) = \langle X, \cdot \rangle Y $ for any random or nonrandom $X$ and $Y$ taking values in $\mathcal H$ (see \eqref{eqtensor})}. 


We say that the explanatory variable $x_t$ is endogenous if the cross-covariance of $x_t$ and $u_t$, given by the operator $\mathbb{E}[(x_t-\mathbb{E}[x_t])\otimes (u_t-\mathbb{E}[u_t])]$, is nonzero. The present paper focuses on estimation and inference of the functional linear model in the presence of endogeneity. Below we provide specific and practical examples that motivate this model of interest.  

\begin{example}[Effects of immigration on the native labor market]\normalfont \label{exam1} 
\commRV{
In Section \ref{sec:emp}, we will explore a functional version of a well-known linear regression model that examines the skill-dependent effects of immigrant inflows on native workers' wages. In this example, the dependent variable \(y_t\) and the explanatory variable \(x_t\) are functions representing skill-specific changes in wage and the share of immigrants, respectively, at time $t$. As will be detailed in Section \ref{sec:emp}, this model suitably extends existing approaches that typically require pre-classifying workers into only a few groups (e.g., low, mid, and high skilled groups) and/or may not allow spillover effects across different skill groups. As is known in the literature (see, e.g.,\ \citealp{llull2018}), estimating this model involves challenging issues, including the endogenous occupational adjustment of workers}.
\end{example}

\begin{example}[Functional AR model with measurement errors] \label{exam2}\normalfont 	
	The functional AR model has been used in many applications involving functional data. We in this example consider the functional AR model where each observation is contaminated by a measurement error; this may be understood as a special case of the model considered in \cite{Chen_et_al_2020}. An example can be found in the recent literature on forecasting of probability densities; see e.g., \cite*{kokoszka2019forecasting}. Since true probability density functions are not observable in practice, they need to be replaced by appropriate nonparametric estimates that inherently involve estimation errors. 
Beyond this specific case, it seems to be quite common in practice that the true functional realization $y_t^\circ$ cannot be observed and thus has to be replaced by an estimate $y_t$, obtained by smoothing its discrete realizations. In these cases, it may be natural to assume that $y_t$ contains a measurement error $e_t$, i.e.,  $y_t = y_t^\circ + e_t$. If $\{y_t^\circ\}_{t\geq 1}$ satisfies the stationary AR law of motion given by  $y_{t}^{\circ}  = \mathcal A y_{t-1} ^\circ + \epsilon_{t}$ for $t \geq 1$, with  $\mathbb{E}[\varepsilon_t] = 0$ and $\mathbb{E}[y_{t-1}^\circ \otimes \varepsilon_t] = 0$, we have 		\begin{equation*}  	y_{t}  = \mathcal A y_{t-1} + u_{t}, \quad\text{where}\quad u_t = e_{t} - \mathcal Ae_{t-1} + \epsilon_{t} . 
	\end{equation*}		In this case, $\mathbb{E}[y_{t-1}\otimes u_t] \neq 0$ in general, and hence $y_{t-1}$ is endogenous. 
\end{example}

It is expected from Example \ref{exam2} that endogeneity can arise in many practical applications of the functional linear model, where $x_t$ is incompletely observed. In such a case, the exogeneity condition is likely violated. 
\phantomsection\label{rvlabel03}\commRV{In particular, due to advancements in data collection techniques, it is now possible to construct density- or curve-valued economic variables from large datasets. As a result, the analysis of these variables has gained popularity, as evidenced by various empirical examples in the literature (e.g., \citealp{Benatia2017}; \citealp{Babii2021}; \citealp{Chen_et_al_2020}; \citealp{Nielsen2019}; \citealp{seo2020functional}). In economic functional data, observations are often incomplete, and the discrete and finite realizations used to construct functional observations may not be enough to fully capture the entire function. Therefore, practitioners may remain cautious about potential endogeneity, even if the considered regressor $x_t$ is presumed to be  exogenous. This limitation could hinder the practical use of the functional linear model.}  As well expected from the literature on the standard linear simultaneous equation model, endogeneity should be properly addressed for consistent estimation of the regression operator. A widely used strategy to do  this is the IV approach, which will be pursued in our Hilbert space setting.


\subsection{Model, assumptions, and notation}\label{model2}
	To facilitate the subsequent discussions, it may be helpful to introduce some additional notation. We first let  $\mathcal H$ denote the Hilbert space of square-integrable functions defined on the unit interval $[0,1]$, where the inner product $\langle \cdot, \cdot \rangle$ is defined by 
	$\langle \zeta_1,\zeta_2 \rangle = \int_{0}^1 \zeta_1(s) \zeta_2(s) ds$ for $\zeta_1,\zeta_2 \in \mathcal H$ and $\|\cdot\|= \langle \cdot,\cdot\rangle^{1/2}$ defines the norm of $\mathcal H$. $\mathcal{L}_{\mathcal H}$ denotes the space of bounded  linear operators acting on $\mathcal H$, equipped with the  operator norm $\|\mathcal T\|_{\op} = \sup_{ \|\zeta\|\leq 1} \|\mathcal T\zeta\|$. For any $\mathcal T \in \mathcal L_{\mathcal H}$,  we let $\mathcal T^\ast$, $\ran \mathcal T$, $\ker \mathcal T$, and $\|\mathcal T\|_{\HS}$ denote the adjoint, range, kernel, and the Hilbert-Schmidt norm of $\mathcal T$, which are briefly reviewed in Appendix \ref{appprelim}; in that section,  various properties of $\mathcal T \in \mathcal L_{\mathcal H}$ such as nonnegativity, positivity, self-adjointness, compactness and Hilbert-Schmidtness are also reviewed. For any nonnegative, self-adjoint and compact $\mathcal T$, we may write  $\mathcal T = \sum_{j=1}^\infty a_j\zeta_{j} \otimes \zeta_{j}$ for some nonnegative sequence $\{a_j\}_{j \geq 1}$ and some orthonormal basis $\{\zeta_{j}\}_{j \geq 1}$. Then $\mathcal T^{1/2}$ can be well defined by replacing $a_j$ with $\sqrt{a_j}$.

This paper concerns the case where the response variable $y_t$ and the endogenous explanatory variable $x_t$ are infinite-dimensional random variables taking values in separable Hilbert spaces. We hereafter conveniently assume that all of such variables 
take values in $\mathcal H$.
This setup in fact encompasses a more general scenario where $y_t$ and $x_t$ take values in different separable Hilbert spaces of infinite dimension, say $\mathcal H_y$ and $\mathcal H_x$. This is because that these spaces are all isomorphic to $\mathcal H$ (see e.g., Corollary 5.5 in \citeauthor{Conway1994}, \citeyear{Conway1994}), and thus there is no loss of generality by assuming that $\mathcal H_y=\mathcal H_x =\mathcal H$. 
We further assume for convenience that  $y_t$ and $x_t$ have zero means, i.e., $\mathbb{E}[y_t]=\mathbb{E}[x_t]=0$; this assumption naturally makes $c_y$ in \eqref{eqmodel0} be suppressed to zero. The extension to the case where the means are unknown and needed to be estimated is straightforward. 
After adopting all such simplifying assumptions, the functional endogenous linear model, which will  be subsequently considered, is given as follows:  for a linear operator $\mathcal A:\mathcal H\mapsto \mathcal H$,
\begin{equation}\label{eqmodel}
	y_t = \mathcal A x_t + u_t, \quad\text{where}\quad \mathbb{E}[x_t \otimes u_t] \neq 0\quad\text{and} \quad \mathbb{E}[u_t] = 0. 
\end{equation}
\phantomsection\label{rvlabel07}\commRV{We then let $z_t$ (to be called the IV) be another zero-mean $\mathcal H$-valued random variable satisfying $\mathbb{E}[z_t\otimes u_t]=0$.}
For notational convenience, we use $\mathcal C_{zz}$, $\mathcal C_{xz}$, $\mathcal C_{yz}$ and   $\mathcal C_{uz}$ to denote the following operators,  
\begin{equation*}
	\mathcal C_{zz}=\mathbb{E} [z_{t} \otimes z_{t}], \quad \mathcal  C_{xz} = \mathbb{E} [x_{t} \otimes z_{t}], \quad \mathcal C_{yz}=\mathbb{E} [y_{t} \otimes z_{t}], \quad\text{and}\quad \mathcal C_{uz}=\mathbb{E} [u_{t} \otimes z_{t}].
\end{equation*}
Similarly, let $\widehat{\mathcal C}_{zz}$,  $\widehat{\mathcal C}_{xz}$, $\widehat{\mathcal C}_{yz}$ and  $\widehat{\mathcal C}_{uz}$ denote their sample counterparts that are computed as follows:   
\begin{equation*}
	\widehat{\mathcal C}_{zz} = \frac{1}{T} \sum_{t=1}^T z_t \otimes z_{t}, \quad  \widehat{\mathcal C}_{xz} = \frac{1}{T} \sum_{t=1}^T x_t \otimes z_{t}, \quad \widehat{\mathcal C}_{yz}= \frac{1}{T} \sum_{t=1}^T y_t \otimes z_{t}, \quad \text{and}\quad\widehat{\mathcal C}_{uz}= \frac{1}{T} \sum_{t=1}^T u_t \otimes z_{t}.
\end{equation*}

 We will employ the following assumptions throughout the paper: below, $\mathfrak F_t$ denotes the filtration given by $\mathfrak F_t = \sigma\left( \{z_s\}_{s\leq t+1}, \{u_s\}_{s\leq t} \right)$, and $\widehat{\mathcal C}_{uu} = T^{-1} \sum_{t=1}^T u_t \otimes u_t$.
\begin{assumpM} \label{assum1} \begin{enumerate*}[(a)]
		\item\label{assum1.1} \eqref{eqmodel}  holds, 
		\item\label{assum1.2} $\{x_t,z_t\}_{t\geq 1}$ is  stationary and geometrically strongly mixing in $\mathcal H \times \mathcal H$, $\mathbb{E}[\|x_t\|^2] < \infty$, and $\mathbb{E}[\|z_t\|^2] < \infty$,
		\item\label{assum1.3} $\mathbb{E}[u_t|\mathfrak F_{t-1}]=0$, 	
		\item\label{assum1.4}  $\mathbb{E}[u_t\otimes u_t|\mathfrak F_{t-1}] = \mathcal C_{uu}$, and  $ \sup_{1\leq t \leq T} \mathbb{E}[\|u_t\|^{2+\delta}|\mathfrak F_{t-1}]<\infty$ for $\delta>0$, 
		\item\label{assum1.5} $\mathcal A$ is Hilbert-Schmidt,
		\item\label{assum1.7} $\Vert \widehat{\mathcal C}_{xz} - \mathcal C_{xz} \Vert_{\HS} $, $\Vert \widehat{\mathcal C}_{zz} - \mathcal C_{zz} \Vert_{\HS}$ and $\Vert \widehat{\mathcal C}_{uz} - \mathcal C_{uz} \Vert_{\HS}$ are $O_p (T^{-1/2})$,
		\item\label{assum1.7a} $\|\widehat{\mathcal C}_{uu} - {\mathcal C}_{uu} \|_{\HS} = o_p(1)$, 	\item\label{assum1.add} $\ker \mathcal C_{xz} = \{0\}.$
	\end{enumerate*}
\end{assumpM} 
By Assumption \ref{assum1}.\ref{assum1.2}, we allow $\{x_t,z_t\}_{t\geq 1}$ to be a weakly dependent sequence; this is because (i) we want to accommodate various empirical examples such as those given in \citet[Chapters 13-16]{HK2012} by not restricting our attention to the iid case and (ii) the variables to be considered in our empirical application (Section~\ref{sec:emp}) naturally exhibit time series dependence. In Assumptions \ref{assum1}.\ref{assum1.3} and \ref{assum1}.\ref{assum1.4}, the error term $u_t$ is assumed to be a homoscedastic martingale difference sequence. Assumption \ref{assum1}.\ref{assum1.3} states the conditions required for the IV $z_t$ in this setting (see Example \ref{exam2a} for a possible IV presented for the model in Example \ref{exam2}), and this condition implies that $u_t$ is uncorrelated with $z_t$. 
In Assumption \ref{assum1}.\ref{assum1.4}, we impose some requirements on the moments of $u_t$. 
We here note that, if $\{z_t,u_t\}_{t\geq 1}$ is an iid sequence, as often assumed in the literature, Assumptions \ref{assum1}.\ref{assum1.3} and \ref{assum1}.\ref{assum1.4} reduce to the following: 
$$\mathbb{E}[u_t|z_t] = 0, \quad  \mathbb{E}[u_t\otimes u_t|z_{t}] = \mathcal C_{uu},\quad \text{and}\quad     \mathbb{E}[\|u_t\|^{2+\delta}|z_t]<\infty \text{ for some } \delta>0.$$  
The Hilbert-Schmidt condition of $\mathcal A$ given in Assumption \ref{assum1}.\ref{assum1.5}  
would become redundant if we considered a finite dimensional Hilbert space, but in our setting it imposes a nontrivial mathematical condition on $\mathcal A$.
In Assumptions  \ref{assum1}.\ref{assum1.7}  and \ref{assum1}.\ref{assum1.7a}, high-level conditions on  limiting behaviors of some sample operators are given, and these are for mathematical convenience.   We  first note that $\{x_t\otimes z_t-\mathcal C_{xz}\}_{t\geq 1}$, $\{ z_t \otimes z_t-\mathcal C_{zz} \}_{t\geq1}$,  and $\{u_t \otimes z_t -\mathcal C_{uz}\}_{t\geq1}$ are sequences in the Hilbert space of Hilbert-Schmidt operators, denoted by $\mathcal S_{\mathcal H}$ (see Section \ref{sec:strongfive} of the Supplementary Material). If those sequences are iid (resp.\ geometrically strongly mixing), then Assumption  \ref{assum1}.\ref{assum1.7}  holds once $\mathbb E [(\Vert x_t \Vert \Vert z_t\Vert)^{\upsilon}]$, $\mathbb E [\Vert z_t \Vert ^{2\upsilon}]$, and $\mathbb E [(\Vert u_t \Vert \Vert z_t \Vert)^{\upsilon}]$ are finite for some $\upsilon \geq 2$ (resp.\ $\upsilon \geq 2 + \delta$ for some $\delta>0$) \citep[Theorems 2.7 and 2.17]{Bosq2000}; such primitive sufficient conditions can also be found for martingale differences \citep[Theorem 2.16]{Bosq2000} and weakly stationary sequences \citep[Theorems 2.18]{Bosq2000}. We also observe that $\{u_t \otimes u_t-\mathcal C_{uu}\}_{t\geq1}$ is a martingale difference sequence in $\mathcal S_{\mathcal H}$, and some primitive sufficient conditions for Assumption  \ref{assum1}.\ref{assum1.7a} can be found in e.g.,\ Theorems 2.11 and 2.14 of \cite{Bosq2000}. \commRV{Lastly, Assumption \ref{assum1}.\ref{assum1.add} enables us to identify the unique bounded linear operator $\mathcal A$ satisfying \eqref{eqmodel} using the IV $z_t$, which will be discussed in more detail in Section \ref{sec:estimators}}.



\begin{example}\label{exam2a} \normalfont Consider Example \ref{exam2} in Section \ref{model1}.
\commRV{Suppose that the sequence of measurement errors $\{e_t\}_{t\geq 1}$ satisfies that $\mathbb{E}[e_t|\mathcal G_{t-1}]=0$ where 
	$\mathcal G_{t-1} = \sigma\left( \{y_s\}_{s\leq t-1},   \{e_s\}_{s\leq t-1}, \{\varepsilon_s\}_{s\leq t}\right)$.} Then $y_{t-\ell}$ for $\ell > 1$ satisfies {the exogeneity condition implied by} Assumption \ref{assum1}.\ref{assum1.3}. Note that $y_{t-2}$ satisfies \begin{equation}
		\mathbb E [y_{t-2} \otimes y_{t}] = \mathcal A\mathbb E [ y_{t-2} \otimes y_{t-1}] + \mathbb E [y_{t-2}\otimes u_t] = \mathcal A\mathbb E [ y_{t-2} \otimes y_{t-1}].  \label{modified_yule}
	\end{equation} 
	This reveals a theoretical connection between our approach and the modified Yule-Walker method, which was introduced to deal with uncorrelated measurement errors in AR models. In a univariate AR(1) case, the modified Yule-Walker  estimator can be obtained by replacing the population moments in \eqref{modified_yule} by their sample counterparts; see \cite{walker1960some} and \citet[Section 4.4]{staudenmayer2005measurement}. As may be expected from this example, if $\{y_t,x_t\}_{t\geq 1}$ is a time series satisfying \eqref{eqmodel}, the lagged explanatory variables $\{x_{t-\ell}\}_{t\geq 1}$ for $\ell=1,\ldots, L$ may be suitable candidates for $z_t$. \cite{Chen_et_al_2020} noted this and proposed $\sum_{\ell=1}^{L} x_{t-\ell}$ as the IV to obtain a consistent estimator of $\mathcal A$. In this regard, our model is related to that of \cite{Chen_et_al_2020}, but the theory and methodology that are subsequently pursued in this paper move in an apparently different direction. 
\end{example}

\begin{remark}\label{remrvadd1}\normalfont 
\commRV{Consider the standard Euclidean space setting where \(\mathcal C_{xz}\) is invertible. From \eqref{eqmodel} and Assumption \ref{assum1}.\ref{assum1.3}, the standard IV estimator, \( \widehat{\mathcal C}_{xz}^{-1} \widehat{\mathcal C}_{yz}\) can be defined and it can be understood as a sample analogue of \(\mathcal A\). Alternatively, \eqref{eqmodel} and Assumption \ref{assum1}.\ref{assum1.3} imply that  \(\mathcal C_{xz}^\ast \mathcal C_{yz}=\mathcal C_{xz}^\ast \mathcal C_{xz}    \mathcal A^\ast\).  Thus, based on this equation, we may define another IV estimator, \((\widehat{\mathcal C}_{xz}^\ast \widehat{\mathcal C}_{xz})^{-1} \widehat{\mathcal C}_{xz}^\ast\widehat{\mathcal C}_{yz}\), which can be understood as a GMM estimator with an identity weight. Our FIVE is defined as in the latter case. This choice offers a theoretical advantage especially in our functional setup because we can utilize well-established mathematical results concerning the spectral properties of self-adjoint operators ($(\widehat{\mathcal C}_{xz}^\ast \widehat{\mathcal C}_{xz})^{-1}$) to study the asymptotic properties of the FIVE.}
\end{remark}

\section{Functional IV estimator  \label{sec:estimators}}
This section discusses estimation of the model \eqref{eqmodel} given observations $\{y_t,x_t,z_t\}_{t=1}^T$. We first propose the FIVE in detail and study its asymptotic properties.    
\subsection{The proposed estimator}\label{sec:estimators1}
We  find from \eqref{eqmodel}  that $\mathcal C_{yz}^\ast = \mathbb E[z_t \otimes y_t] = \mathcal A \mathbb E[z_t \otimes x_t] = \mathcal A \mathcal C_{xz}^\ast$ and hence, 
\begin{equation}
	\mathcal C_{yz}^\ast \mathcal C_{xz} = \mathcal A \mathcal C_{xz}^\ast \mathcal C_{xz}. \label{popmoment1}
\end{equation}
As discussed in \cite{Mas2007} and \cite{Benatia2017}, $\mathcal A$ is a uniquely identified bounded linear operator if and only if $ \ker \mathcal C_{xz}  = \{0\}$ (see Assumption \ref{assum1}.\ref{assum1.add}), and we note that all the eigenvalues of $\mathcal C_{xz}^\ast \mathcal C_{xz}$ are positive under the condition \citep[Remark 2.1]{Mas2007}. In the sequel,  we thus let $\{\lambda_j^2\}_{j\geq 1}$ denote the collection of the eigenvalues of $\mathcal C_{xz}^\ast \mathcal C_{xz}$  ordered from the largest to the smallest, and represent $\mathcal C_{xz}^\ast \mathcal C_{xz}$ as its spectral decomposition given by 
\begin{equation*}
	\mathcal C_{xz}^\ast \mathcal C_{xz} = \sum_{j=1}^\infty \lambda_j^2 f_j \otimes f_j,
\end{equation*}   
where $f_j$ is the eigenfunction corresponding to $\lambda_j^2$.
Given \eqref{popmoment1}, it may be natural to consider an estimator $\bar{\mathcal A}$ that satisfies the equation $\widehat{\mathcal C}_{yz}^\ast\widehat{\mathcal C}_{xz}=\bar{\mathcal A} \widehat{\mathcal C}_{xz}^\ast \widehat{\mathcal C}_{xz}$, obtained by replacing $\mathcal C_{yz}$ and $\mathcal C_{xz}$ with their sample counterparts.  However, it is  generally impossible to directly compute the estimator $\bar{\mathcal A}$ from this equation since $\widehat {\mathcal C}_{xz} ^\ast \widehat {\mathcal C}_{xz}$ is not invertible over the entire Hilbert space $\mathcal H$. We circumvent this issue by employing a regularized inverse of $\widehat {\mathcal C}_{xz} ^\ast \widehat {\mathcal C}_{xz}$ which may be understood as the well defined inverse on a strict subspace of $\mathcal H$.  

To this end, we first note that $\widehat{\mathcal C}_{xz}^\ast  \widehat{\mathcal C}_{xz}$ is nonnegative, self-adjoint and compact and hence allows the following  representation: 
\begin{equation*}
	\widehat{\mathcal C}_{xz}^\ast  \widehat{\mathcal C}_{xz} = \sum_{j=1}^\infty \hat{\lambda}_j ^2 \hat{f}_j \otimes \hat{f}_j,  
\end{equation*} 
where $\{ \hat{\lambda}_j^2 , \hat{f}_j  \}_{j\geq 1}$ are the pairs of eigenvalues and eigenfunctions, and $\hat{\lambda}_1^2 \geq \ldots \geq \hat{\lambda}_T^2 \geq 0 = \hat{\lambda}_{T+1}^2=\ldots$.
We then define $\K$  as the random integer determined by the threshold parameter $\abar >0$ such that
\begin{equation}
	\K = \#\{j  : \hat{\lambda}_j^2 > 1/\abar\} .  \label{eqdef2}
\end{equation}
Using the first $\K$ eigenfunctions of $\widehat{\mathcal C}_{xz}^\ast  \widehat{\mathcal C}_{xz}$, its rank-regularized inverse, denoted $(\widehat{\mathcal C}_{xz}^\ast  \widehat{\mathcal C}_{xz})_{\K} ^{-1}$, and the FIVE, denoted $\widehat { \mathcal A}$, are defined as follows: 
\begin{equation}
	\widehat{\mathcal A} =\widehat{\mathcal C}_{yz}^\ast  \widehat{\mathcal C}_{xz} (\widehat{\mathcal C}_{xz}^\ast  \widehat{\mathcal C}_{xz})_{\K}^{-1}, \quad \text{where}\quad	(\widehat{\mathcal C}_{xz}^\ast  \widehat{\mathcal C}_{xz})_{\K}^{-1} = \sum_{j=1}^{\K} \hat{\lambda}_j^{-2} \hat{f}_j \otimes \hat{f}_j.\label{eqdef3}
\end{equation} 
The largest eigenvalue of the regularized inverse $(\widehat{\mathcal C}_{xz}^\ast  \widehat{\mathcal C}_{xz})_{\K}^{-1}$ is bounded above by $\abar$ and thus the regularized inverse is a well defined bounded linear operator for every $\abar>0$. It is worth mentioning that the FIVE becomes equivalent to the estimator proposed by \cite{Park2012} in the case where  $z_t = x_t$ and $\K$ is deterministically chosen by researchers (see Remark~\ref{remini} below), so our estimator may be understood as an extension of their estimator. We also note that the FIVE $\widehat{\mathcal A}$ may be viewed as a sample-analogue of $\mathcal A$ satisfying \eqref{popmoment1} in the sense that $\widehat{\mathcal A}$ is the solution to $\widehat{\mathcal C}_{yz}^\ast\widehat{\mathcal C}_{xz}=\widehat{\mathcal A} \widehat{\mathcal C}_{xz}^\ast \widehat{\mathcal C}_{xz}$ on the restricted domain given by $\widehat{\mathcal H}_{\K} = \spn\{\hat{f_j}\}_{j=1}^{\K}$. {Section \ref{sec_comp} of the Supplementary Material discusses on how $\widehat{\mathcal A}$ can be computed from the data using the FPCA.}

\begin{remark} \label{remini} \normalfont
\commRV{It should be noted that \(\K\) in \((\widehat{\mathcal C}_{xz}^\ast  \widehat{\mathcal C}_{xz})_{\K}^{-1}\) is by construction a random variable associated with the choice of \(\abar\). In the literature where a similar regularized inverse is discussed,  \(\K\) is chosen by practitioners and hence regarded as deterministic (e.g., \citealp{Mas2007}; \citealp{Park2012}). However, even in this case, it is generally recommended to choose  \(\K\) taking the eigenvalues of  \(\widehat{\mathcal C}_{xz}^\ast  \widehat{\mathcal C}_{xz}\) into account, and thus treating $\K$ as a random variable appears to be natural. This alternative perspective on  \(\K\) helps practitioners more directly control the degree of instability of the regularized inverse, measured by its largest eigenvalue, by means of the  parameter  \(\abar\) that they choose. Moreover, this approach distinguishes our asymptotic approach from those in the aforementioned papers.}
\end{remark}

\subsection{General asymptotic properties} \label{sec:asym1}
As may be deduced from our construction of the FIVE, $\widehat{\mathcal A} = 0$ is imposed outside a subspace whose dimension increases as $\abar$ gets larger. Thus, for $\widehat{\mathcal A}$ to be a consistent estimator of $\mathcal A$ defined on the entire space $\mathcal H$,  the {regularization} parameter $\abar$ given in \eqref{eqdef2} needs to diverge to infinity. Taking this into consideration, we  investigate the asymptotic properties of the FIVE when $T \to \infty$ and $\abar\to \infty$ jointly. We will employ the following assumption:
\begin{assumpA}\label{assum1eigen}
	$\lambda_1^2 > \lambda_2^2> \ldots > 0$.
\end{assumpA}
That is, the eigenvalues of $\mathcal C_{xz}^\ast \mathcal C_{xz}$ are required to be distinct. This is employed to see asymptotic properties of the FIVE in detail and does not seem to be restrictive in practice; in fact, similar assumptions have been employed in the literature on functional linear models, see e.g., \citet[Section 8.3]{Bosq2000}, \cite{Mas2007}, \cite{Hall2007} and \cite{Park2012} to name only a few.

We now provide the asymptotic properties of the estimator $\widehat{\mathcal A}$ when $\abar$ and $T$ grow jointly without bound.  To this end, we consider the following decomposition of $\widehat {\mathcal A} - \mathcal A$: 
\begin{equation}\label{eqdecom}
	\widehat{\mathcal A}-\mathcal A = (	\widehat{\mathcal A}-\mathcal A\widehat{\Pi}_{\K}) - \mathcal A (\mathcal I-\widehat{\Pi}_{\K}),
\end{equation} 
where  $\widehat{\Pi}_{\K}$ denotes the orthogonal projection defined by $	\widehat{\Pi}_{\K} = \sum_{j=1}^{\K} \hat{f}_j \otimes \hat{f}_j$ and $\mathcal  I$ is the identity operator  acting on $\mathcal H$. Given that the FIVE is computed on the restricted domain $\ran \widehat{\Pi}_{\K}$ (note that $\widehat{\mathcal A}=0$ on $\ran (\mathcal I-\widehat{\Pi}_{\K})$ by construction)  the first term of \eqref{eqdecom} may be understood as the deviation of $\widehat{\mathcal A}$ from $\mathcal A$ on  $\ran \widehat{\Pi}_{\K}$. Thus, this term is hereafter called the deviation component on the restricted domain (the DR component). On the other hand, the second term $\mathcal A (\mathcal  I-\widehat{\Pi}_{\K})$ may be understood as the bias induced by the fact that $\widehat{\mathcal A}$ is enforced to zero on  $\ran (\mathcal  I-\widehat{\Pi}_{\K})$. We thus call this term  the regularization bias component (the RB component). 
Our first result below shows that both the DR and RB components are asymptotically negligible and thus $\widehat{\mathcal A}$ becomes weakly consistent once the regularization parameter $\abar$ diverges to infinity at an appropriate rate; in the next theorem,  we let $\tau(\abar)$ be a random function that increases without bound as $\abar\to \infty$, which is defined by $\tau(\abar) = \sum_{j=1}^{\K} \tau_j$, where $\tau_j =  2\sqrt{2} \max\{(\lambda_{j-1}^2-\lambda_{j}^2)^{-1}, (\lambda_{j}^2-\lambda_{j+1}^2)^{-1}\}$.
\begin{theorem} \label{thm2w1}
	Suppose that Assumptions \ref{assum1} and \ref{assum1eigen} are satisfied, 
	$T^{-1/2}\tau(\abar) \pto 0$ and  $T^{-1}\abar \to 0$ as $\abar \to \infty$ and $T\to \infty$. Then 
	\begin{equation*}
		\|\widehat{\mathcal A} - \mathcal A\widehat{\Pi}_{\K}\|_{\op}^2 =O_p(T^{-1}\abar)\quad \text{and} \quad 	\|\mathcal A (\mathcal I-\widehat{\Pi}_{\K})\|_{\op}^2 =o_p(1).	
	\end{equation*} 
\end{theorem}
The following is an immediate consequence of Theorem \ref{thm2w1} and  Assumption  \ref{assum1}.\ref{assum1.7a}.
\begin{corollary}\label{cor1} Suppose that the assumptions in Theorem \ref{thm2w1} are satisfied, and let $\hat{u}_t = y_t-\widehat{\mathcal A}x_t$. Then 
	$\|T^{-1}\sum_{t=1}^T \hat{u}_t \otimes \hat{u}_t  - \mathcal C_{uu} \|_{\op} \pto 0$.
\end{corollary} 
The condition imposed on $\tau(\abar)$ in Theorem \ref{thm2w1} does not place any essential restrictions on the eigenvalues of $\mathcal C_{xz}^\ast \mathcal C_{xz}$. Given that $\tau(\abar)$ increases as $\abar$ (and thus $\K$) gets larger, the condition, together with the requirement that $T^{-1}\abar \to 0$, merely tells us that $\abar$ needs to grow with a sufficiently slower rate than $T$ for the weak consistency of the FIVE. In fact, under some additional conditions, the strong (almost sure) consistency of the estimator can also be derived; we need more mathematical preliminaries to present this result, and thus leave the discussion to Section~\ref{sec:strongfive} of the Supplementary Material.

\begin{remark}\normalfont \label{rem:comb}
	{The result in Theorem \ref{thm2w1} is, at least to some extent, related to a similar consistency result given by  \cite{Benatia2017} for their functional IV estimator.} In order to obtain an estimator from \eqref{popmoment1}, the authors employ the ridge regularized inverse of $\widehat{\mathcal C}_{xz}^\ast \widehat{\mathcal C}_{xz}$, while we use a rank-regularized inverse of $\widehat{\mathcal C}_{xz}^\ast \widehat{\mathcal C}_{xz}$. This makes a significant difference in asymptotic approaches to establish consistency in the two papers. For example, our result is based on the FPCA and thus we require the eigenvalues of $\mathcal C_{xz} ^\ast \mathcal C_{xz}$ to be distinct, 
	which is not required in \cite{Benatia2017}. In addition, \citepos{Benatia2017} approach  restricts the range of $\mathcal A$ to a certain subspace of $\mathcal H$, called the $\beta$-regularity space, while we need to restrict the increasing rate of $\abar$ or $\K$ depending on $\tau(\abar)$. 
\end{remark}

Under stronger assumptions than what we require for the weak consistency of $\widehat {\mathcal A}$, we can further find that (i) the decaying rate of $	\widehat{\mathcal A}-\mathcal A$ is not uniform over the entire Hilbert space $\mathcal H$ and (ii) the choice of  $\abar$  can affect the decaying rates of the DR and RB components in different directions. These results are given as consequences of the following asymptotic normality result of the DR component; in the theorem below, $({\mathcal C}_{xz}^\ast{\mathcal C}_{xz})_{\K}^{-1}$ denotes the operator {given by $\sum_{j=1}^{\K}\lambda_j^{-2}f_j\otimes f_j$ and $N(0,\mathcal G)$ denotes Gaussian random element in $\mathcal H$ with covariance operator $\mathcal G$.}

\begin{theorem}\label{thm2w2}  
	Suppose that Assumptions \ref{assum1} and \ref{assum1eigen} are satisfied, $\abar^{1/2}T^{-1/2}\tau(\abar) \pto 0$ and $T^{-1}\abar \to 0$ as $\abar \to \infty$ and $T\to \infty$. Then the following hold for any $\zeta \in \mathcal H$.\begin{enumerate}[(i)] 
		\item\label{thm2w2a} 	$\sqrt{{{T}}/{{\theta_{\K}(\zeta)}}}(\widehat{\mathcal A}- \mathcal A \widehat{\Pi}_{\K})\zeta \dto  N(0, \mathcal  C_{uu})$, 	where $	{\theta_{\K}(\zeta)}= \langle \zeta,  ({\mathcal C}_{xz}^\ast{\mathcal C}_{xz})_{\K}^{-1}{\mathcal C}_{xz}^\ast \mathcal C_{zz}{\mathcal C}_{xz}({\mathcal C}_{xz}^\ast{\mathcal C}_{xz})_{\K}^{-1}\zeta \rangle$. 
		\item \label{thm2w2b} If $\widehat{\theta }_{\K}(\zeta) \coloneqq \langle \zeta,  (\widehat{\mathcal C}_{xz}^\ast\widehat{\mathcal C}_{xz})_{\K}^{-1}\widehat{\mathcal C}_{xz}^\ast \widehat{\mathcal C}_{zz}\widehat{\mathcal C}_{xz}(\widehat{\mathcal C}_{xz}^\ast\widehat{\mathcal C}_{xz})_{\K}^{-1} \zeta \rangle$, then  $ |\widehat{\theta }_{\K}(\zeta) -  \theta_ {\K}(\zeta) | \pto 0$.
	\end{enumerate} 
\end{theorem} 

Depending on the choice of $\zeta$, $\theta_{\K}(\zeta)$ may be convergent or divergent in probability (see Remark \ref{rem:conv} below), 
and thus the convergence rate of  the DR component $(\widehat{\mathcal A}- \mathcal A \widehat{\Pi}_{\K})\zeta$ depends on $\zeta$. This finding is not completely new; similar results were formerly observed by \cite{Mas2007} and  \cite{Hu2016} in the context of functional AR(1) models. If $\theta_{\K} (\zeta)$ is convergent {in probability}, then $(\widehat{\mathcal A}- \mathcal A \widehat{\Pi}_{\K})\zeta$ converges at $\sqrt{T}$-rate, otherwise it converges at a slower rate given by $\sqrt{{{T}}/{{\theta_{\K}(\zeta)}}}$ which is random (because of the randomness of $\K$). As noted by  \cite{Mas2007}, this discrepancy in convergence rates implies that (i) there exists no sequence of normalizing constants $c_T$ such that $c_T (\widehat{\mathcal A}- \mathcal A \widehat{\Pi}_{\K})\zeta$  weakly converges to a well defined limiting distribution uniformly in $\zeta \in \mathcal H$, and therefore it is impossible that $\widehat{\mathcal A}- \mathcal A \widehat{\Pi}_{\K}$ weakly converges to a well defined bounded linear operator in the topology of $\mathcal L_{\mathcal H}$, and (ii) this statement is also true if $\mathcal A \widehat{\Pi}_{\K}$ is replaced by $\mathcal A$ (see Theorem 3.1 of \citealp{Mas2007}).  \phantomsection\label{rvlabel06}\commRV{Moreover, it can be deduced from Theorem~\ref{thm2w2} that as $\K$ increases, the regularization parameter $\abar$ induces a trade-off between the decaying rates of the DR and RB components when $\theta_{\K}(\zeta)$ is not convergent.  If $\K$ relative to $T$ increases by a larger choice of $\abar$, then the operator norm of the RB component tends to shrink to zero at a faster rate. On the other hand,  this change results in a faster divergence of  ${\theta}_{\K}(\zeta)$, and  Theorem~\ref{thm2w2} shows that the DR component will decay at a slower rate in such a case.}

\begin{remark}\label{rem:conv} \normalfont
	As a simple way to see if $\theta_{\K}(\zeta)$ can be either convergent or divergent depending on the choice of $\zeta$, it is useful to assume that $x_t = z_t + v_t$, where $\{v_t\}_{t\geq 1}$ satisfies that $\mathbb{E}[z_t \otimes v_t] = 0$. For $j \geq 1$, let $\{\mu_j,g_j\}_{j\geq 1}$ be  the pairs of eigenvalues and eigenfunctions  of $\mathcal C_{zz}$. Then it can be shown that  $\text{plim}_{T\to \infty}\theta_{\K}(\zeta)$ is simply given by $\sum_{j=1}^{\infty} \mu_j^{-1} \langle g_j, \zeta \rangle^2$, and this quantity is bounded only when $\zeta$ belongs to a certain strict subspace of $\mathcal H$; see \citet*[Section 3.2]{carrasco2007linear}. 
\end{remark}

In applications involving economic or statistical time series, practitioners are often interested in the marginal effect of some additive and hypothetical perturbation, say $\zeta$, in $x_t$ on $y_t$. In the considered linear model, this marginal effect is simply given by $\mathcal A\zeta$, which can be consistently estimated by $\widehat{\mathcal A}\zeta$. Let $\{\zeta_{\K}\}$ be a sequence of random elements given by $\zeta_{\K} = \widehat{\Pi}_{\K}\zeta$. Then $\zeta_{\K}$ is given by  the orthogonal projection of the new perturbation $\zeta$ onto the subspace on which the sample cross-covariance of $x_t$ and $z_t$ is the most explained, in a certain sense (see Remark \ref{rem3}), among all the subspaces of  dimension $\K$; that is, $\zeta_{\K}$ is the best linear approximation of $\zeta$ based on the covariation of the explanatory and instrumental variables. Therefore, $\zeta_{\K}$ may be interpreted as a nice approximation showing how a hypothetical perturbation $\zeta$ can be revealed given the dataset, and we thus call $\zeta_{\K}$ a data-supporting approximation of $\zeta$. The following is an immediate consequence  of Theorem \ref{thm2w2}: under the assumptions employed in Theorem \ref{thm2w2}, 
\begin{equation*}
	|\hat{\theta}_{\K}(\zeta_{\K}) - \theta_{\K}(\zeta_{\K})|\pto 0 \quad \text{and}  \quad \sqrt{{{T}}/{{\theta_{\K}(\zeta_K)}}}(\widehat{\mathcal A}- \mathcal A)\zeta_{\K} \dto  N(0, \mathcal  C_{uu}).
\end{equation*} Based on this result, we may implement standard statistical inference on various characteristics of $\mathcal A \zeta_{\K}$, which may provide a practical and interpretable insight for practitioners.  We illustrate this by constructing a confidence interval for the random variable given by $\langle {\mathcal A}\zeta_{\K}, \psi\rangle$ for some $\psi\in \mathcal H$.
In fact, various characteristics of ${\mathcal A}\zeta_{\K}$ may be written in this form; for example, if $\psi(s)=1\{s_1\leq s\leq s_2\}$ then  $\langle {\mathcal A}\zeta_{\K}, \psi\rangle = \int_{s_1}^{s_2}  {\mathcal A}\zeta_{\K}(s)ds$ means the locally (if $s_1 \neq 0$ or $s_2 \neq 1$) or globally (if $s_1=0$ and $s_2 = 1$) aggregated marginal effect on $y_t$. 
We then consider the interval whose endpoints are given as follows: 
\begin{equation}
	\langle\widehat{\mathcal A}\zeta_{\K},\psi\rangle \pm \Phi^{-1}(1-{\varpi}/2) \sqrt{{\hat{\theta}_{\K}(\zeta_{\K})}\langle \widehat{\mathcal C}_{\hat{u}\hat{u}} \psi, \psi \rangle/T} , \label{eq:conf}
\end{equation}
where  $\Phi^{-1}(\cdot)$ is the quantile function of the standard normal distribution and $\widehat{\mathcal C}_{\hat{u}\hat{u}}=T^{-1}\sum_{t=1}^T \hat{u}_t \otimes \hat{u}_t$. Based on  Theorem~\ref{thm2w2} and Corollary~\ref{cor1}, the intervals that are repeatedly constructed as in \eqref{eq:conf} are expected to include $\langle \mathcal A  \zeta_{\K} , \psi \rangle$ with $100(1-{\varpi}$)$\%$ of probability for a large $T$. 	Of course, \eqref{eq:conf} may not be quite satisfactory for practitioners who want to consider a purely hypothetical perturbation $\zeta$ without data-supporting approximation;   however, given that (i) the discrepancy between $\zeta$ and $\zeta_{\K}$ caused by the noninvertibility of $\widehat{\mathcal C}_{xz}^\ast  \widehat{\mathcal C}_{xz}$, which inevitably arises in our functional setting, and (ii) the magnitude of discrepancy is expected to be small since it is, anyhow, asymptotically negligible, a small bias caused by the data-supporting approximation may be understood as a cost to implement standard inference based on asymptotic normality in our setting. Furthermore and more importantly, it will be shown in Section \ref{sec:asym2} that, if certain conditions, which are not that restrictive, are satisfied, then the convergence result given in Theorem \ref{thm2w2}.\ref{thm2w2a} holds even if $\mathcal A\widehat{\Pi}_{\K}$ is replaced by $\mathcal A$  (see Remark \ref{remnormality}); this, of course, implies that \eqref{eq:conf} can be understood as a confidence interval for $\langle \mathcal A\zeta,\psi \rangle$ with no data-supporting approximation.

		\begin{remark}\normalfont\label{rem3}
			Note that $\|\widehat{\mathcal C}_{xz}\|_{\HS}^2 = \sum_{j=1}^\infty \widehat{\lambda}_j^2= \sum_{j=1}^\infty \|\widehat{\mathcal C}_{xz}\zeta_j\|^2$ holds for any arbitrary orthonormal basis $\{\zeta_{j}\}_{j\geq 1}$ of $\mathcal H$. We then may define the proportion of the sample cross-covariance operator explained by the first $\K$ orthonormal vectors as 
			\begin{equation*}  
				\sum_{j=1}^{\K} \|\widehat{\mathcal C}_{xz} \zeta_j\|^2 \Big/ 	\sum_{j=1}^\infty \widehat{\lambda}_j^2  =   \sum_{k=1}^{\K} \sum_{j=1}^\infty \hat{\lambda}_j^2  \langle \hat{f}_j,\zeta_k \rangle^2\Big/ 	\sum_{j=1}^\infty \widehat{\lambda}_j^2 .
			\end{equation*}
			Provided that $\{\widehat{\lambda}_j^2,\widehat{f}_j\}_{j\geq 1}$ is the sequence of the eigenelements of $\widehat{\mathcal C}_{xz}^\ast\widehat{\mathcal C}_{xz}$, it is deduced from the results given in \citet[Theorem 3.2 and Section 3.2]{HK2012} that  the above quantity is bounded above by $\sum_{j=1}^{\K}  \widehat{\lambda}_j^2 \Big/ 	\sum_{j=1}^\infty \widehat{\lambda}_j^2$, and this upper bound is attained if and only if $\zeta_j = \pm \widehat{f}_j$ for $j=1,\ldots,\K$. This shows that,  among all the subspaces of dimension $\K$, $\ran \widehat{\Pi}_{\K}$ is the unique subspace that explains the most proportion of the squared Hilbert-Schmidt norm of $\widehat{\mathcal C}_{xz}$.
		\end{remark}
		
\begin{remark} \label{remsigtest}\normalfont
\commRV{In Section \ref{sec:hypo} of the Supplementary Material, we develop a significance test to examine whether various characteristics of \(y_t\), expressed as \(\langle y_t,\psi \rangle\) for some \(\psi \in \mathcal H\), depend on $x_t$. A crucial input to this test is a consistent estimator of $\mathcal A$, and the FIVE (and also the F2SLSE to be developed in Section \ref{sec:f2sls00}) can be used.}
\end{remark}
		
		\subsection{Refinements of the general asymptotic results} \label{sec:asym2}
		In Section \ref{sec:asym1}, we established some general asymptotic properties of the FIVE, which do not require any specific assumptions on the eigenstructure of the cross-covariance of $x_t$ and $z_t$ other than the assumption of distinct eigenvalues. The results given in the previous section tell us that the FIVE is a reasonable estimator in this functional setting. However,  what can be learned from Theorems  \ref{thm2w1} and \ref{thm2w2} is not rich enough; we only know that the FIVE is consistent (Theorem \ref{thm2w1}) and its DR component is asymptotically normal in a pointwise sense (Theorem \ref{thm2w2}) if $\abar$ diverges to infinity at a sufficiently slow rate. We in this section investigate the asymptotic behavior of the FIVE in more detail under a set of assumptions \commRV{which is stronger than Assumption \ref{assum1eigen} but not restrictive in practice}. By doing so, we will obtain useful refinements of Theorems \ref{thm2w1} and \ref{thm2w2}.  The specific assumptions that we need are given as follows: in the assumption below, {we note that $\mathcal C_{xz}\mathcal C_{xz}^\ast$ allows the spectral decomposition $\mathcal C_{xz} \mathcal C_{xz}^\ast=\sum_{j=1}^\infty \lambda_j^2 \xi_{j} \otimes \xi_j$ for some orthonormal basis $\{\xi_j\}_{j\geq 1}$, and let}   $\upsilon_{t}(j,\ell) = \langle x_t,f_j \rangle\langle z_t,\xi_{\ell}\rangle -\mathbb{E}[\langle x_t,f_j \rangle\langle z_t,\xi_{\ell}\rangle]$ for $j,\ell \geq 1$.
		\begin{assumpA}		\label{assumconvrate} There exist constants $c_{\circ}>0$, $\rho >2 $, $\varsigma > 1/2$, $\gamma >1/2$ and $m>1$  satisfying the following: 	\begin{enumerate*}[(a)] \item\label{assumconvrate.1} $\lambda_j^2 \leq c_{\circ} j^{-\rho}$, \item\label{assumconvrate.2} $\lambda_{j} ^2 -\lambda_{j+1} ^2 \geq c_{\circ}^{-1} j^{-\rho-1}$, \item\label{assumconvrate.3} $|\langle \mathcal A f_j , \xi_\ell \rangle| \leq c_{\circ} j^{-\varsigma} \ell ^{-\gamma} $, \item\label{assumconvrate.4} $\mathbb{E}[\upsilon_{t}(j,\ell)\upsilon_{t-s}(j,\ell)] \leq c_{\circ} s^{-m} \mathbb E[\upsilon_t ^2 (j,\ell)]$ for $s\geq 1$, and furthermore,  $\mathbb E[ \Vert \langle x_t , f_j\rangle z_t \Vert ^2 ] \leq c_{\circ} \lambda_j^2$ and $\mathbb E[ \Vert \langle z_t , \xi_j\rangle x_t \Vert ^2 ] \leq c_{\circ}  \lambda_j^2$. 
			\end{enumerate*}
		\end{assumpA}
		Assumptions~\ref{assumconvrate}.\ref{assumconvrate.1} and \ref{assumconvrate}.\ref{assumconvrate.2} restrict the eigenstructure of $\mathcal C_{xz}$ (or equivalently $\mathcal C_{xz}^\ast\mathcal C_{xz}$), and these are adapted from similar conditions in \cite{Hall2007} and \cite{imaizumi2018}.\footnote{In fact, Assumptions~\ref{assumconvrate}.\ref{assumconvrate.1} and ~\ref{assumconvrate}.\ref{assumconvrate.2} can be replaced by the following conditions: $|\lambda_j|\leq c_{\circ} j^{-\rho/2}$ and $|\lambda_j|-|\lambda_{j+1}|\geq c_{\circ} j^{-\rho/2-1}$ for $\rho>2$. These conditions are directly comparable with similar conditions given for the eigenvalues of $\mathbb{E}[x_t\otimes x_t]$ in \cite{Hall2007} and \cite{imaizumi2018}.} Assumption~\ref{assumconvrate}.\ref{assumconvrate.3} is a very natural condition given that $\langle \mathcal Af_j, \xi_{\ell} \rangle$ must be square-summable with respect to both $j$ and $\ell$; in this assumption, it is worth mentioning that $\varsigma$ is the parameter determining the smoothness of $\mathcal A$ on $\ran \mathcal C_{xz}^\ast \mathcal C_{xz} $.  As may be deduced from the definition of $\upsilon_{t}(j,\ell)$ and Assumption \ref{assum1}.\ref{assum1.2}, $\{\upsilon_{t}(j,\ell)\}_{t\geq 1}$ is a stationary sequence in $\mathbb{R}$ for each $j$ and $\ell$, and the former condition of Assumption~\ref{assumconvrate}.\ref{assumconvrate.4} states that its lag-$s$ autocovariance function decays at a sufficiently fast rate; this condition is satisfied for a wide class of stationary processes. Note that  both $\mathbb{E}[\|\langle x_t,f_j\rangle z_t\|^2]$ and $\mathbb{E}[\|\langle z_t,\xi_j\rangle x_t\|^2]$ naturally decrease as $j$ gets larger and its decaying rate is restricted by  Assumption~\ref{assumconvrate}.\ref{assumconvrate.4}. 
		{Specifically, we require the second moments of $\|\langle x_t,f_j\rangle z_t\|$ and $\|\langle z_t,\xi_j\rangle x_t\|$ as functions of $j$ have a constant multiple of $\lambda_j^2$ as their upper envelope}; a similar condition for the iid case can be found in e.g., \cite{Hall2007}.  
		
		The following theorem refines the result given in Theorem \ref{thm2w1} under Assumption \ref{assumconvrate}. 
		\begin{theorem}\label{thm:convrate}
			Suppose that Assumptions \ref{assum1} and \ref{assumconvrate} are satisfied and $\abar = o(T^{\rho/(2\rho+2)})$. Then, $\|\widehat{\mathcal A} - \mathcal A\widehat{\Pi}_{\K}\|_{\op}^2 =O_p(T^{-1}\abar)$ as in Theorem \ref{thm2w1}, and
			\begin{equation} \label{eqthmconvrate}
				\|\mathcal A (\mathcal I-\widehat{\Pi}_{\K})\|_{\op}^2 = O_p (T^{-1}\abar \max\{1, \abar^{(3-2\varsigma)/\rho}\} + \abar^{(1-2\varsigma)/\rho} ). 
			\end{equation} 
			Thus, $\|\widehat{\mathcal A} - \mathcal A\|_{\op} = o_p(1)$ for any $\rho > 2$ and $\varsigma > 1/2$. 
		\end{theorem} 
		Some comments on the requirement $\abar= o(T^{\rho/(2\rho+2)})$ are first in order. This condition is needed in our proof of Theorem \ref{thm:convrate} to deal with estimation errors associated with $\widehat{\lambda}_j$ (see Remark \ref{remalpha}). This may be replaced by a sufficient and more convenient condition given by $\abar = o(T^{1/3})$, which does not depend on the value of $\rho$ under Assumption \ref{assumconvrate} requiring $\rho>2$.

		Theorem~\ref{thm:convrate} not only gives us a more detailed consistency result than that given in Theorem \ref{thm2w1}, but also better clarifies how certain parameters appearing in Assumption \ref{assumconvrate} can affect the convergence rate of the FIVE. 
		Specifically, in the above theorem, the convergence rate is characterized by the regularization parameter $\abar$, the smoothness $\varsigma$ of $\mathcal A$ on $\ran \mathcal C_{xz}^\ast \mathcal C_{xz}$, and  the decaying rate $\rho$ of $\lambda_j^2$ (as a function of $j$). From \eqref{eqthmconvrate}, it is evident that if $\abar$ grows to infinity at a sufficiently slow rate, then the convergence rate of the RB term will be dominantly determined by the second term (appearing in \eqref{eqthmconvrate}), whose convergence rate is positively related to $\abar$. Therefore, in this case, we expect a slower convergence rate of the RB component; this is a quite natural property that can also be deduced from our earlier discussion following Theorem \ref{thm2w2} in Section \ref{sec:asym1}. Moreover, it can be shown that the convergence rate of the FIVE is generally positively (resp.\ negative) related to $\varsigma$ (resp.\ $\rho$); the former is immediately seen from  \eqref{eqthmconvrate}, and the latter is discussed in detail in  Remark \ref{fiverho}.

		\begin{remark}\normalfont
			In fact, the results given in Theorems \ref{thm2w1} and \ref{thm:convrate} hold even if $\|\cdot\|_{\op}$ is replaced by $\|\cdot\|_{\HS}$, which can be seen in our proofs of those theorems (see Section \ref{sec:pf} of the Supplementary Material). 
		\end{remark}
		
		\begin{remark}\normalfont\label{remalpha}
			The requirement $\abar = o(T^{\rho/(2\rho+2)})$ is used to ensure the existence of a small constant, say $\tilde{c}_{\circ}$, such that $\mathbb{P}(|\widehat{\lambda}_j-\lambda_{\ell}|\geq \tilde{c}_{\circ}|{\lambda}_j-\lambda_{\ell}|$ for all $1\leq j  \leq \K$ and $\ell \neq j) \to 1$; the detailed reason why this result can be obtained from the requirement on $\abar$ is given in our proof of Theorem \ref{thm:convrate}, which is quite similar to the discussion given by \cite{imaizumi2018} following their Theorem 1.
		\end{remark}
		
		\begin{remark}\normalfont\label{fiverho} If the eigenvalues of $\mathcal C_{xz}^\ast\mathcal C_{xz}$ decay to zero at a fast rate  (and thus the eigenvalues of $\widehat{\mathcal C}_{xz}^\ast\widehat{\mathcal C}_{xz}$ tend to do so), then the rank-regularized inverse $(\widehat{\mathcal C}_{xz}^\ast \widehat{\mathcal  C}_{xz})_{\K}^{-1}$ tends to be more unstable (unless $\K$ becomes smaller) than in the case with more slowly decaying eigenvalues. Thus, it is expected that the convergence rate of the FIVE generally becomes slower as $\rho$ increases. This can be seen from the  asymptotic results given in Theorem \ref{thm:convrate}. Let  $\rho$ change with the other parameters being fixed. In addition, it is assumed that $\abar$ satisfies the condition $\abar = o(T^{\rho/(2\rho+2)})$ both before and after any change in $\rho$ and thus no adjustment in $\abar$ is required; note that this can always be done by letting $\abar =o(T^{1/3})$ if necessary. From \eqref{eqthmconvrate} it can be shown that the RB component is (i) $O_p(T^{-1}\abar) +O_p(\abar ^{(1-2\varsigma)/\rho})$ if $\varsigma\geq3/2$ and (ii) $O_p(T^{-1} \abar ^{(\rho-2\varsigma+3)/\rho}) + O_p(\abar^{(1-2\varsigma)/\rho} )$ if $\varsigma \in(1/2,3/2)$. In case (i), the decaying rate of the second term is negatively related to $\rho$, and thus an increase in $\rho$ does not yield a faster convergence rate. In case (ii), the decaying rates of both terms are negatively related to $\rho$ since $(\rho-2\varsigma +3)/\rho$ is positive and strictly increasing in $\rho$. 
\end{remark}

We next refine our pointwise asymptotic normality result  under Assumption \ref{assumconvrate}. To this end, it is convenient to decompose  the RB component again as follows:
\begin{equation}\label{eqdecom1}
\mathcal A(\widehat{\Pi}_{\K} - \mathcal I) = \mathcal A(\widehat{\Pi}_{\K} - \Pi_{\K}) + \mathcal A({\Pi}_{\K} - \mathcal I),
\end{equation}
where $\Pi_{\K} = \sum_{j=1}^{\K} f_j \otimes f_j$, and this may be understood as the population counterpart of $\widehat{\Pi}_{\K}$. The next theorem refines the results  in Theorem \ref{thm2w2}, \commRV{but in order to simplify the subsequent discussion, we for now only consider the case where $\rho/2+2<\varsigma+\delta_{\zeta}$; the result without this condition is given in Appendix \ref{sec_thm_conti1}.} 
\begin{theorem}\label{thm:convrate2}
Suppose that Assumptions \ref{assum1} and \ref{assumconvrate} are satisfied,  $\zeta \in \mathcal H$ satisfies $\langle f_{j} , \zeta \rangle   \leq c_\zeta j^{-\delta_{\zeta}} $ for some $c_{\zeta} >0$ and $\delta_{\zeta} > 1/2$,  \commRV{$\rho/2+2<\varsigma+\delta_{\zeta}$} and 			\begin{equation}\label{eqrbalpha1}
\commRV{T^{-1} \max\left\{\abar^{(3\rho -2\delta_{\zeta}+1)/\rho}, \abar^{(\rho+1)/\rho} \right\} = o(1)}. 
\end{equation} 
Then,  Theorem \ref{thm2w2} holds and 
\begin{align}
\Vert \mathcal A(\widehat \Pi_{\K}-\Pi_{\K}) \zeta \Vert =  O_p(T^{-1/2})\quad \text{ and }\quad	\|\mathcal A(\Pi_{\K} - \mathcal I)\zeta\| = O_p(\abar^{(1/2-\varsigma- \delta_{\zeta})/\rho}). \nonumber 
\end{align}
\end{theorem}
Theorem \ref{thm:convrate2} refines  Theorem \ref{thm2w2} by providing a detailed asymptotic order of the RB component. Some remarks on the theorem are given in  Remarks \ref{remalpha2} and \ref{remnormality} below; particularly, in the latter remark, an improvement of the asymptotic normality result  in Theorem \ref{thm2w2} is discussed.  

\begin{remark}\normalfont \label{remalpha2}
The growing rate of $\abar$ required for Theorem \ref{thm:convrate2} depends on both $\rho$ and $\delta_{\zeta}$. If $\delta_{\zeta}$ is sufficiently large so that $2\delta_{\zeta} \geq \rho-1$, then \eqref{eqrbalpha1} can be simplified to $\abar = o(T^{1/3})$. Moreover, the condition $\delta_{\zeta} > 1/2$ is natural since $\langle f_j,\zeta\rangle$ must be square-summable with respect to $j$. 
\end{remark}
\begin{remark}[Pointwise asymptotic normality of the FIVE]\normalfont \label{remnormality}
A consequence of Theorem \ref{thm:convrate2} is that, if $\mathcal A$ is smooth enough and $\langle f_j,{\zeta} \rangle$ decays to zero at a sufficiently fast rate as $j$ increases, then $\sqrt{{{T}}/{{\theta_{\K}(\zeta)}}}\|\mathcal A(\widehat{\Pi}_{\K} - \mathcal I)\|_{\op} \pto 0$ and thus the result given in Theorem \ref{thm2w2}.\ref{thm2w2a} can be strengthened to the following: 
\begin{equation} \label{eqnormality}
\sqrt{{{T}}/{{\theta_{\K}(\zeta)}}}(\widehat{\mathcal A}- \mathcal A )\zeta \dto  N(0, \mathcal  C_{uu}).
\end{equation}
In this case,  of course, \eqref{eq:conf} is understood as a confidence interval for $\langle \mathcal A\zeta,\psi \rangle$.  
In particular, if (i) $\varsigma + \delta_{\zeta} > \rho/2 +2$ and (ii) $T\abar^{1-2\varsigma-2\delta_{\zeta}} = O(1)$ (note that both conditions are easier to hold if $\varsigma$ and $\delta_{\zeta}$ are large), we have 
\begin{equation*}
\sqrt{{{T}}/{{\theta_{\K}(\zeta)}}}\mathcal A(\widehat{\Pi}_{\K} - \mathcal I)\zeta= O_p(1/\sqrt{\theta_{\K}(\zeta)}).
\end{equation*} 
The above quantity converges to zero if $\theta_{\K}(\zeta) \pto \infty$, which is likely to happen in practice for many possible choices of $\zeta$; for example, if we assume that $\zeta$ is arbitrarily and randomly chosen from $\mathcal H$, $\mathbb{P}\{\theta_{\K}(\zeta) < c <\infty\} \to 0$ as $\K \to \infty$ since $\theta_{\K}$ is convergent only on a strict subspace of $\mathcal H$. Given that $\delta_{\zeta} > 1/2$, the aforementioned conditions for \eqref{eqnormality} are satisfied if (i)$'$ $\varsigma > \rho/2 + 3/2$ and (ii)$'$ $T\abar^{-2\varsigma} = O(1)$ regardless of the value of $\delta_{\zeta}$. The former condition (i)$'$ requires $\mathcal A$ to have a sufficient smoothness depending on the decaying rate of the eigenvalues of $\mathcal C_{xz}^\ast \mathcal C_{xz}$, and this seems not to be restrictive; in fact, in the literature on functional linear models, it is  common to impose such a smoothness condition on $\mathcal A$ depending on the eigenvalues of a certain covariance or cross-covariance operator (e.g. \citealp{Hall2007}; \citealp{imaizumi2018}; \citealp{Chen_et_al_2020}).
\end{remark}

\section{Functional two-stage least square estimator}\label{sec:f2sls00} 

\subsection{The proposed estimator}\label{sec:f2sls0}
\phantomsection\label{rvlabel08}\commRV{In the context of Euclidean space setting, our FIVE reduces to a certain IV estimator as in Remark~\ref{remrvadd1}. As may be expected from the existing results, 
this estimator generally exhibits a larger asymptotic variance compared to the two-stage least square estimator (2SLSE), which achieves asymptotic efficiency within some context.  In this section, we explore an extension of the conventional 2SLSE tailored to our functional setting.} 
To the best of the authors' knowledge, such an estimator has not been considered in the context of function-on-function regression, although a similar estimator was studied by \cite{Florence2015} in the case where the dependent variable is scalar-valued.
For this estimator, it is assumed that $x_t$ and $z_t$ satisfy the so-called first-stage relationship:
\begin{assumpMM} \label{assum1a} \begin{enumerate*}[(a)]\item\label{assum1a1}  Assumption~\ref{assum1} holds,  \item \label{assum1a2}  $x_t = \mathcal B z_t + v_t$, where $\mathcal B \in \mathcal L_{\mathcal H}$ and  $\mathbb{E}[v_t|\mathfrak F_{t-1}]= 0$ ($\mathfrak F_{t-1}$ is defined in Assumption~\ref{assum1}), \item\label{assum1a4}   $\widehat{\mathcal C}_{vz}=T^{-1} \sum_{t=1}^T v_t\otimes z_t$ satisfies that $\|\widehat{\mathcal C}_{vz}\|_{\HS} = O_p(T^{-1/2})$.
	\end{enumerate*}

\end{assumpMM} 


If we consider the case $\mathcal H = \mathbb{R}^n$, then the 2SLSE is defined as follows:
\begin{equation}
\widetilde{\mathcal A}^\circ = \widehat{\mathcal C}_{yz}^\ast \widehat{\mathcal C}_{zz}^{-1} \widehat{\mathcal C}_{xz} \left(\widehat{\mathcal C}_{xz}^\ast \widehat{\mathcal C}_{zz}^{-1} \widehat{\mathcal C}_{xz}\right)^{-1}, \nonumber
\end{equation} 
and it is widely known that $\widetilde{\mathcal A}^\circ$ has many desirable properties as an estimator of $\mathcal A$. Coming back to our functional setting, it is not difficult to see that the use of the standard 2SLSE is problematic since it involves $\widehat{\mathcal C}_{zz}^{-1}$ and  $(\widehat{\mathcal C}_{xz}^\ast \widehat{\mathcal C}_{zz}^{-1} \widehat{\mathcal C}_{xz})^{-1}$ which are not well defined as bounded linear operators.  

To have a well-behaved analogue of the 2SLSE in our setting, we regularize those inverses as we did in Section~\ref{sec:estimators} and  propose an alternative estimator. To this end, we hereafter let  $\mathcal T_{\K}^{-1}$ denote the regularized inverse of a compact operator $\mathcal T$ based on its first $\K$ eigenelements (this is defined in the same way as $(\widehat{\mathcal C}_{xz}^\ast  \widehat{\mathcal C}_{xz})_{\K}^{-1}$ given in \eqref{eqdef3}). Our proposed estimator is defined as follows:
\begin{equation*} \label{estf2slse}
\widetilde{\mathcal A} =  \widehat{\mathcal P} \widehat{ \mathcal Q}_{\K_2}^{-1}, \quad \text{where}\quad \widehat{\mathcal P}= \widehat{\mathcal C}_{xz}^\ast (\widehat{\mathcal C}_{zz})_{\K_1}^{-1} \widehat{\mathcal C}_{xz}  \quad \text{and}\quad  \widehat{ \mathcal Q} =  \widehat{\mathcal C}_{xz}^\ast (\widehat{\mathcal C}_{zz})_{\K_1}^{-1} \widehat{\mathcal C}_{xz},
\end{equation*} 
and, if we let $\{\widehat{\mu}_j\}_{j\geq 1}$ (resp.\ $\{\widehat{\nu}_j\}_{j\geq 1}$) be the ordered (from the largest to the smallest) eigenvalues of $\widehat{\mathcal C}_{zz}$ (resp.\ $\widehat{\mathcal Q}$),\footnote{The eigenvalues of $\widehat{\mathcal C}_{zz}$ and $\widehat{\mathcal Q}$ are almost surely positive since they are nonnegative self-adjoint by construction.} then $\K_1$ and $\K_2$ are defined as 
\begin{equation*} \label{2slsek1}
\K_1 = \# \{ j: \hat\mu_j^2 >1/\abar_1 \} \quad \text{and} \quad 		\K_2  = \#\{j  : \hat{\nu}_j^2 > 1/\abar_2\}.
\end{equation*}
Note that by definition, $\K_2 \leq \K_1$ holds almost surely because  $(\widehat{\mathcal C}_{zz})_{\K_1}^{-1}$ is of finite rank $\K_1$. We conveniently call $\widetilde{\mathcal A}$ the F2SLSE.

To investigate the asymptotic properties of the F2SLSE, it is necessary to establish some preliminary results and fix notation. First, we note that the operator $\mathcal Q$ defined by $\mathcal Q = \mathcal C_{xz}^\ast \mathcal C_{zz}^{-1} \mathcal C_{xz}$ can be understood as a well defined compact operator. To be more specific,  Lemma \ref{lem:ftsls} (and the following discussion given in Section \ref{sec:pfe} of the Supplementary Material) shows that $\mathcal C_{zz}^{-1/2} \mathcal C_{xz} = \mathcal R_{xz}\mathcal C_{xx}^{1/2}$ for some unique bounded linear operator $\mathcal R_{xz}$, which may be understood as the correlation operator of $x_t$ and $z_t$, and thus $\mathcal Q = \mathcal C_{xx}^{1/2}\mathcal R_{xz}^\ast\mathcal R_{xz}  \mathcal C_{xx}^{1/2}$. From similar arguments, it can be easily shown that the operator $\mathcal P = 	\mathcal C_{yz}^\ast 	\mathcal C_{zz}^{-1} \mathcal C_{xz}$ is also well defined. We then let $\{\mu_j,g_j\}_{j\geq 1}$ (resp.\ $\{\nu_j,h_j\}_{j\geq 1}$) be the eigenelements of $\mathcal C_{zz}$ (resp.\ $\mathcal Q$), i.e.,
\begin{equation*}
\mathcal C_{zz} = \sum_{j=1}^\infty \mu_j g_j \otimes g_j \quad \text{and} \quad  \mathcal Q = \sum_{j=1}^\infty \nu_j h_j\otimes h_j.
\end{equation*}   
Since $\mathcal C_{zz}$ and $\mathcal Q$ are self-adjoint and nonnegative, $\mu_j$ and $\nu_j$ are all nonnegative. We know from \eqref{eqmodel} that the population relationship $\mathcal P=\mathcal A \mathcal Q$ holds.  

\subsection{General asymptotic properties}\label{sec:asym1a}
This section discusses the asymptotic properties of the F2SLSE under the following assumption:
\begin{assumpB}\label{assum1eigen2}
 $\mu_1 >\mu_2 >\ldots > 0$ and $\nu_1 >\nu_2 >\ldots > 0$.
\end{assumpB}
It should be noted that the operator $\mathcal A$ satisfying $\mathcal P=\mathcal A\mathcal Q$ is uniquely identified if Assumptions \ref{assum1a} and \ref{assum1eigen2} are satisfied. To see this in detail,  note that $\mathcal A$ satisfying $\mathcal C_{yz}^\ast \mathcal C_{xz} = \mathcal A \mathcal C_{xz}^\ast \mathcal C_{xz}$ (resp.\ $\mathcal P=\mathcal A \mathcal Q$) is identified if and only if $\mathcal C_{xz} $ (resp.\ $\mathcal C_{zz} ^{-1/2} \mathcal C_{xz}$) is injective. If Assumption~\ref{assum1eigen2} is satisfied and hence $\mathcal C_{zz} $ is injective, then $\ker \mathcal C_{xz} (=\ker \mathcal C_{xz}^\ast \mathcal C_{xz})$  becomes identical to $ \ker \mathcal C_{zz} ^{-1/2} \mathcal C_{xz} (= \ker \mathcal Q)$.

As in Section \ref{sec:asym1}, we also consider the following decomposition:
\begin{align}\label{eqdecom2}
\widetilde{\mathcal A}-\mathcal A = (	\widetilde{\mathcal A}-\mathcal A\widetilde{\Pi}_{\K_2}) - \mathcal A (\mathcal I-\widetilde{\Pi}_{\K_2}),
\end{align}
where  $\widetilde{\Pi}_{\K_2}$ denotes the orthogonal projection defined by $	\widetilde{\Pi}_{\K_2} = \sum_{j=1}^{\K_2} \widehat{h}_j \otimes \widehat{h}_j$ and $\{\widehat{h}_j\}_{j=1}^{\K_2}$ is the collection of the eigenvectors of $\widehat{\mathcal Q}$ corresponding to the first $\K_2$ leading eigenvalues. The two terms in \eqref{eqdecom2} are similarly interpreted as in the case of the FIVE (see \eqref{eqdecom1}), and we thus call the first (resp.\ the second) term the DR (resp.\ RB) component.  

We first show that  both the DR and RB components are asymptotically negligible (and thus $\widehat{\mathcal A}$ is weakly consistent) if the regularization parameters $\abar_1$ and $\abar_2$ diverge to infinity at appropriate rates: in the theorem below, we let 
$\tau_{1,j} =  2\sqrt{2} \max\{(\mu_{j-1}-\mu_{j})^{-1},(\mu_{j}-\mu_{j+1})^{-1}\}$ and $\tau_{2,j} =  2\sqrt{2} \max\{(\nu_{j-1}-\nu_{j})^{-1}, (\nu_{j}-\nu_{j+1})^{-1}\}$.
\begin{theorem} \label{thm2w}
Suppose that Assumptions \ref{assum1a} and \ref{assum1eigen2} are satisfied, and $T^{-1/2} (\sum_{j=1}^{\K_1} \mu_j \tau_{1,j} ) (\sum_{j=1}^{\K_2}\tau_{2,j} )\pto0$,  $(\sum_{j=\K_1+1}^{\infty}\mu_j)(\sum_{j=1}^{\K_2}\tau_{2,j} ) \pto 0$,  $\abar_1^{-1} \abar_2 \to 0$,  
and $T^{-1}\abar_1 \to 0$ as $\abar_{1}\to \infty$, $\abar_{2}\to \infty$ and $T\to \infty$. Then 	
\begin{equation*}
\|\widetilde{\mathcal A} - \mathcal A\widetilde{\Pi}_{\K_2}\|_{\op}^2 =O_p(T^{-1} \abar_1^{1/2}\abar_2^{1/2})\quad \text{and} \quad 	\|\mathcal A (\mathcal I-\widetilde{\Pi}_{\K_2})\|_{\op}^2 =o_p(1).	
\end{equation*}
\end{theorem}
An immediate consequence of Theorem \ref{thm2w} is given as follows:
\begin{corollary}\label{cor2}  Suppose that the assumptions in Theorem \ref{thm2w2} are satisfied and let $\tilde{u}_t = y_t-\widetilde{\mathcal A}x_t$.  Then $
\|T^{-1}\sum_{t=1}^T \tilde{u}_t \otimes \tilde{u}_t  - \mathcal C_{uu} \|_{\op} \pto 0$. 
\end{corollary} 
As in the case of the FIVE, the conditions imposed on the quantities $(\sum_{j=1}^{\K_1} \mu_j \tau_{1,j})(\sum_{j=1}^{\K_2}\tau_{2,j})$ and $(\sum_{j=\K_1+1}^{\infty}\mu_j)(\sum_{j=1}^{\K_2}\tau_{2,j} )$ are understood not as special restrictions on the eigenvalues, but as requirements on the growing rates of $\abar_1$ and $\abar_2$ in our asymptotic theory. Specifically, the condition on the former quantity merely requires $\abar_1$ and $\abar_2$ to increase slowly so that  $\K_1$ and $\K_2$ tend to grow with sufficiently slower rates than $T$. Moreover, given that, for fixed $\abar_2$, $(\sum_{j=\K_1+1}^{\infty}\mu_j)(\sum_{j=1}^{\K_2}\tau_{2,j} )$ can be arbitrarily small by choosing $\abar_1$ large enough, the condition on the latter quantity merely tells us that the growing rate of $\abar_1$ needs to be sufficiently higher than that of $\abar_2$. In addition to the weak consistency given by Theorem \ref{thm2w},  the strong consistency of the F2SLSE can be derived under  additional conditions; this is discussed in Section \ref{sec:strong2sls} of the Supplementary Material. 

We also obtain  an asymptotic normality result similar to that given by Theorem \ref{thm2w2} for the FIVE: 
\begin{theorem}\label{thm2wa}
Suppose that the assumptions in Theorem \ref{thm2w} are satisfied, $ T^{-1/2} \abar_{1}^{1/2}\sum_{j=1}^{\K_1}\tau_{1,j}\pto 0$, $ T^{-1/2}\abar_2^{1/2} (\sum_{j=1}^{\K_1} \mu_j \tau_{1,j} ) (\sum_{j=1}^{\K_2}\tau_{2,j} )\pto0$,  $\abar_2^{1/2}(\sum_{j=\K_1+1}^{\infty}\mu_j)(\sum_{j=1}^{\K_2}\tau_{2,j} ) \pto 0$, $\abar_1^{-1} \abar_2 \to 0$, and $T^{-1}\abar_1\to 0$ as $\abar_1\to \infty$, $\abar_2\to \infty$ and $T \to \infty$.
Then the following hold for any $\zeta \in \mathcal H$. 
\begin{enumerate}[(i)]
\item\label{thm2waa} $\sqrt{{{T}}/{{\phi_{\K_2}(\zeta)}}}(\widetilde{\mathcal A}- \mathcal A \widetilde{\Pi}_{\K_2}) \zeta \dto N(0, \mathcal C_{uu})$, 	where $	{\phi_{\K_2}(\zeta)}= \langle \zeta,  \mathcal Q_{\K_2}^{-1}\zeta \rangle$.  
\item\label{thm2wab} If $\widehat{\phi}_{\K_2}(\zeta)  \coloneqq \langle \zeta,  \widehat{\mathcal Q}_{\K_2}^{-1} \zeta \rangle$, then  $ |\widehat{\phi}_{\K_2}(\zeta) - {\phi}_{\K_2}(\zeta)|\pto 0$.
\end{enumerate}
\end{theorem}
As shown, we need more stringent requirements on the growing rates of $\abar_1$ and $\abar_2$. This is mainly due to that the F2SLSE involves the doubly regularized inverse $\widehat{\mathcal Q}_{\K_2}^{-1}$, which may be ill-behaved if $\abar_1$ and $\abar_2$ do not grow at sufficiently slow rates. This implies that there is no reason why the F2SLSE is generally preferred to the FIVE in this functional setting, unlike what we can expect from the general preference for the 2SLSE by practitioners in the Euclidean space setting.\footnote{Moreover, the two estimators have different RB terms whose magnitudes depend on various parameters (e.g., the eigenvalues of $\mathcal C_{xz}^\ast\mathcal C_{xz}$ and $\mathcal Q$), and thus the RB component of the FIVE can have a smaller asymptotic order. } 

\begin{remark}\normalfont \label{remadd1}
\commRV{While the functional form of the F2SLSE closely resembles the conventional 2SLSE in the Euclidean space setting, a crucial distinction exists. In the conventional case, the 2SLSE gains its theoretical superiority by efficiently combining regressors using a larger number of instruments. However, in our setting, one endogenous functional regressor is instrumented by one another functional variable. This explains, at least to some degree, why the F2SLSE does not supersede the FIVE and why the standard properties of the 2SLSE do not naturally extend to our functional framework. For example, for any element $\zeta \in \mathcal H$, Theorems \ref{thm2w2}  and \ref{thm2wa} tell us that the DR components of the FIVE and F2SLSE, respectively, converge to the same Gaussian random element with $\sqrt{T/\theta_{\K}(\zeta)}$-rate and $\sqrt{T/\phi_{\K_2}(\zeta)}$-rate, which are possibly random quantities depending on $\zeta$. Given this, it is not generally possible to conclude that the F2SLSE is asymptotically better than the FIVE, although we disregard the RB components of the FIVE and F2SLSE.  This insight originates from a detailed discussion by an anonymous referee, to whom we are indebted.} 
\end{remark}

\subsection{Refinements of the general asymptotic results}\label{sec:asym2a}
We provide refinements of Theorems \ref{thm2w} and \ref{thm2wa} \commRV{under the following set of assumptions which is stronger than Assumption \ref{assum1eigen2}}: below, we let  $\tilde{\upsilon}_{t}(j,\ell) = \langle z_t,g_j \rangle\langle z_t,g_{\ell}\rangle -\mathbb{E}[\langle z_t,g_j \rangle\langle z_t,g_{\ell}\rangle]$ for $j,\ell \geq 1$.
\begin{assumpB}\label{assconvratetsls}  There exist constants $c_{\circ}>0$, $\rho_{\mu} > 2$, $\rho_{\nu} > 2$, $\varsigma_{\mu}>1/2$, $\varsigma_{\nu}>1/2$, $\gamma_{\mu}>1/2$, $\gamma_{\nu}>1/2$ and $m>1$  satisfying the following:
\begin{enumerate*}[(a)] 
\item \label{assconvratetsls.1}${\mu}_j^2 \leq  c_{\circ}j^{-\rho_{\mu}}$,
\item \label{assconvratetsls.2} $\mu_j^2 - \mu_{j+1}^2 \geq c_\circ ^{-1} j^{-\rho_{\mu}-1}$, 
\item \label{assconvratetsls.3} ${\nu}_j^2 \leq  c_{\circ}j^{-\rho_{\nu}}$,
\item \label{assconvratetsls.4} $\nu_j^2 - \nu_{j+1}^2 \geq c_\circ ^{-1} j^{-\rho_{\nu}-1}$,
\item  \label{assconvratetsls.70} $\langle   h_j ,\mathcal Ah_\ell \rangle \leq c_\circ j^{-\gamma_{\nu}} \ell ^{-\varsigma_{\nu}}$, 
\item  \label{assconvratetsls.7}$\langle   h_j ,\mathcal Bg_\ell \rangle \leq c_\circ j^{-\gamma_{\mu}} \ell ^{-\varsigma_{\mu}}$ and $\gamma_{\mu} {\leq \rho_{\nu}/4+1/2}$,
\item\label{assconvratetsls.5} $\mathbb{E}[\tilde{\upsilon}_{t}(j,\ell)\tilde{\upsilon}_{t-s}(j,\ell)] \leq c_{\circ} s^{-m} \mathbb E[\tilde{\upsilon}_t ^2 (j,\ell)]$ for $s\geq 1$, $\mathbb E [\Vert\langle z_t, g_j \rangle z_t \Vert ^2]\leq c_{\circ} \mu_j^2$, 
and  $\mathbb{E}[\Vert \langle x_t, h_j \rangle z_t \Vert ^2 ] \leq c_{\circ}\Vert \mathcal C_{xz} h_j \Vert ^2$.
\end{enumerate*}
\end{assumpB}
 
The conditions are somewhat similar to those in Assumption \ref{assumconvrate}, and thus we omit detailed comments except the following two points:  (i) from a technical point of view,  Assumption \ref{assconvratetsls}.\ref{assconvratetsls.5} is similar to Assumption \ref{assumconvrate}.\ref{assumconvrate.4} employed for our study of the FIVE and helps us obtain convergence rates of the eigenelements of $\widehat{\mathcal Q}$ (which are crucial inputs to our main results of the F2SLSE), and \commRV{(ii) we require a smoothness condition on \(\mathcal B\) which characterizes the linear relationship between \(x_t\) and \(z_t\) whereas such a condition is not necessary in the case of the FIVE. This  reveals that the data generating process (DGP) is more restricted for our asymptotic analysis of the F2SLSE.} 

Our next result refines Theorem \ref{thm2w} by providing a more detailed result on the RB component.
\begin{theorem}\label{thm3:convrate}
Suppose that Assumptions \ref{assum1a} and \ref{assconvratetsls} are satisfied, $\abar_1 = o(T^{\rho_{\mu}/(2\rho_{\mu}+2)})$ and $\abar_2 = o(\abar_1^{\rho_{\nu}/(2\rho_{\nu}+2)})$.  Then,  $\|\widetilde{\mathcal A} - \mathcal A\widetilde{\Pi}_{\K_2}\|_{\op}^2 = O_p(T^{-1}\abar_1^{1/2}\abar_2^{1/2})$ as in Theorem \ref{thm2w}, and 
\begin{equation} \label{eqthmconvrate2}
\|\mathcal A (\mathcal I-\widetilde{\Pi}_{\K_2})\|_{\op}^2  =  O_p(\abar_1^{-1} \abar_2 \max\{1, \abar_2^{(3-2\varsigma_{\nu})/\rho_{\nu}}\} + \abar_2^{(1-2\varsigma_{\nu})/\rho_{\nu}}). 
\end{equation}
Thus, $\|\widetilde{\mathcal A}-\mathcal A\|_{\op} = o_p(1)$ for any $\rho_{\mu} > 2$, $\rho_{\nu} > 2$, $\varsigma_{\mu} > 1/2$ and $\varsigma_{\nu} > 1/2$. 
\end{theorem}
The convergence rate of the RB component, described in the above theorem, depends not only on the regularization parameters, but also on smoothness of $\mathcal A$ as in the case of the FIVE. However, the convergence rate described in \eqref{eqthmconvrate2} is generally slower than that of the FIVE, and this is somewhat expected from the fact that the F2SLSE involves a doubly regularized (and thus less stable) inverse. Despite this  disadvantage of the F2SLSE over the FIVE, our simulation results support that the F2SLSE performs comparably well among a set of the competing estimators (including the FIVE), and thus this estimator can also be used in practice. 

Using Assumption \ref{assconvratetsls}, the next theorem refines Theorem \ref{thm2wa}, but as in Section \ref{sec:asym2}, we for now only focus on the case where $\rho_{\nu}/2+2<\varsigma_{\nu}+\delta_{\zeta}$. The  result without this condition is provided in the Supplementary Material (see Section \ref{sec_thm_conti2}). In the theorem below, we, as in \eqref{eqdecom1}, consider the decomposition of the RB component given by
\begin{equation*}\label{eqdecom3}
\mathcal A(\widetilde{\Pi}_{\K_2} - \mathcal I) = \mathcal A(\widetilde{\Pi}_{\K_2} - \Pi_{\K_2}) + \mathcal A({\Pi}_{\K_2} - \mathcal I),
\end{equation*}
where $\Pi_{\K_2} = \sum_{j=1}^{\K_2} h_j \otimes h_j$ is understood as the population counterpart of $\widetilde{\Pi}_{\K_2}$.

\begin{theorem}\label{thm4:convrate}
Suppose that Assumptions  \ref{assum1a} and \ref{assconvratetsls} are satisfied, $\zeta \in \mathcal H$ satisfies $\langle h_{j} , \zeta \rangle   \leq c_\zeta j^{-\delta_{\zeta}} $ for some $c_{\zeta} >0$ and $\delta_{\zeta} > 1/2$,  \commRV{$\rho_{\nu}/2+2<\varsigma_{\nu}+\delta_{\zeta}$}, and the following hold:
\begin{equation} \label{eq001thm}
\abar_1 = o(T^{\rho_{\mu}/(2\rho_{\mu}+2)}), \quad
\commRV{\abar_1^{-1} \max\left\{\abar_2^{{(3\rho_{\nu}-2\delta_{\zeta}+1)}/{\rho_{\nu}}},\abar_2^{(\rho_{\nu}+1)/\rho_{\nu}}\right\} = o(1).}
\end{equation}
Then Theorem \ref{thm2wa} holds, and furthermore,
\begin{equation*}
\Vert \mathcal A(\widetilde \Pi_{\K_2}-\Pi_{\K_2}) \zeta \Vert    = 
O_p(\abar_1^{-1/2})\quad\text{ and }\quad
	\|\mathcal A(\Pi_{\K_2} - \mathcal I)\zeta\|  =O_p(\abar_2^{(1/2-\varsigma_{\nu}-\delta_{\zeta})/\rho_{\nu}}).  \label{eqthmrb4a}
\end{equation*}
\end{theorem}
Obtaining a pointwise asymptotic normality result that is not dependent on the RB component as in Remark \ref{remnormality} requires more stringent conditions, which will be detailed in Remark \ref{remnormality2} below. 
\begin{remark}[Pointwise asymptotic normality of the F2SLSE]\normalfont \label{remnormality2}	
In order to strengthen the result given by Theorem~\ref{thm2wa}.\ref{thm2waa} using Theorem~\ref{thm4:convrate} as in the case of the FIVE, we need more stringent conditions. For example, suppose as in Remark \ref{remnormality} that $\mathcal A$ is smooth enough so that $\varsigma_{\nu} > \rho_{\nu}/2 + 3/2$ and $T \abar_2^{-(1-2\varsigma_{\nu}-2\delta_{\zeta})} = O(1)$. We know from Theorem \ref{thm4:convrate} that  $\sqrt{{{T}/{{\phi_{\K_2}(\zeta)}}}}\|\mathcal A(\widetilde{\Pi}_{\K_2} - \mathcal I)\zeta\| = O_p(\sqrt{T\abar_1^{-1}/\phi_{\K_2}(\zeta)})$, and find that $\sqrt{{{T}}/{{\phi_{\K_2}(\zeta)}}}(\mathcal A\widetilde{\Pi}_{\K_2} - \mathcal A)\zeta \dto  N(0, \mathcal  C_{uu})$  if $\phi_{\K_2}(\zeta)$ diverges at a faster rate than that of $T\abar_1^{-1}$. In the case of the FIVE and under an analogous smoothness condition, recall that only $\theta_{\K}(\zeta) \pto \infty$ is needed to obtain a similar result; see Remark \ref{remnormality}.
\end{remark}

\section{Numerical Studies \label{sec:sim}}
We first investigate the finite sample performance of our estimators via Monte Carlo studies. In Sections~\ref{subsub: exp1.1}--\ref{sec:sim2}, the number of replications is set to 1,000 and all the considered random variables are demeaned before  computing  the estimators of $\mathcal A$. Section \ref{sec:emp} provides an empirical application. 
\subsection{Experiment 1: Functional linear simultaneous equation model \label{subsub: exp1.1}}
We consider the following functional linear simultaneous model: for  $t\geq 1$,
\begin{equation} \label{eqsimdgp}y_t = \mathcal A x_t  + u_t,\quad\quad\quad x_t  = \vartheta \mathcal{B}   z_t + v_t,		
\end{equation}
where $u_t =0.8 v_t + 0.6 \varepsilon_t$, $ \{v_t\}_{t \geq 1}$  and $\{\varepsilon_t\}_{t \geq 1} $ are mutually independent iid sequences of standard Brownian bridges satisfying $\mathbb{E}[v_t\otimes \varepsilon_{\ell}] =0$ for all $t,\ell \geq 1$. The constant $\vartheta$ is chosen in such a way that the first-stage functional coefficient of determination (see, \citealp{Yao2005}), defined by $\mathbb E[\Vert \vartheta \mathcal B z_t \Vert ^2]/\mathbb E [\Vert x_t\Vert^2]$, has a specific value of $\mathtt{r}^2$. In this section, we will focus on empirical MSEs of a few estimators at various levels of $\mathtt{r}^2$, and in particular we consider $\mathtt{r}^2 \in \{ 0.1,0.2,\ldots,0.5 \}$.

The DGP here is specially designed to examine the performance of our estimators when all the employed assumptions (Assumptions \ref{assum1}, \ref{assum1eigen}, \ref{assumconvrate}, \ref{assum1a}, \ref{assum1eigen2} and \ref{assconvratetsls}) are satisfied (see Section \ref{sec_app_num1} of the Supplementary Material). Specifically, we let $\{z_t\}_{t \geq 1}$ be an iid sequence of standard Brownian bridges satisfying $\mathbb{E}[z_t\otimes v_t]=\mathbb{E}[z_t\otimes u_t]=0$. Then, we have $\mu_j = (j\pi)^{-2} $ and $g_j(s)   = \sqrt{2}\sin(j\pi s)$ for $s \in [0,1]$, see, e.g., \cite*{JaimezBonnet}. The operators $\mathcal A$ and $\mathcal B$ are defined as follows:  \begin{equation*}
		\mathcal A = \sum_{j=1} ^{\infty} a_j g_j \otimes g_j, \quad 	\mathcal B = \sum_{j=1}^\infty b_jg_j \otimes g_j,   \quad a_j=  j^{-n_a}, \quad b_j=j^{-n_b}, \quad n_a \in \{3,5\},\quad n_b \in \{0.75, 1.5\}.
\end{equation*}  
In this setup,  $f_j = g_j$. In view of the fact that function-valued random variables are only partially observed in practice, we assume that the discrete realizations of   $y_t$, $x_t$ and $z_t$ at 50 equally-spaced points of $[0,1]$ are available. Then, following the literature, e.g., \citet[Chapter 5]{Ramsay2005}, we represent functional variables $y_t$, $x_t$ and $z_t$ by using 31 Fourier basis functions.  

We will compare the performance of our estimators with the ridge regularized IV estimator (RIVE) of \citet[eqn.\ 34]{Benatia2017} with denoting their regularization parameter to $\abar^{-1}$ to keep notational consistency. To compute the FIVE and RIVE, we consider $ \delta_{\abar}{T^{-0.4}}\Vert  \widehat{\mathcal C}_{xz}\Vert_{\HS}^2 $  as candidates for the inverse of $\abar$. This candidate value is calculated at 20 equidistant points of $\delta_{\abar}$ ranging from $0.1$ to $T^{0.2}$. 
Among such candidates,  we choose the value that minimizes the empirical MSE of each estimator.  The F2SLSE needs two regularization parameters: $\abar_{1}$ and $\abar_{2}$. The parameter $\abar_1$ is chosen as the FIVE and RIVE with $\Vert  \widehat{\mathcal C}_{xz}\Vert_{\HS} ^2$ being replaced by $\Vert  \widehat{\mathcal C}_{zz} \Vert_{\HS}^2$. Once $\abar_1$ is chosen, we similarly choose the inverse of $\abar_2$ from $\delta_{\abar_2} (\abar_1^{-1}\Vert \mathcal C_{zz}\Vert_{\HS}^2)^{1/2}\Vert \widehat{\mathcal Q}_{\K_1}\Vert_{\HS} ^2 $ with $\delta_{\abar_2}$ being 20 equidistant points between $T^{0.05}$ and $ T^{0.2}$. This setup enforces $\abar_1$ to grow at a faster rate than that of $\abar_2$.


 To save space, we report estimation results only for the case with $T=500$; the results with a smaller sample size are qualitatively similar and are reported in Section \ref{sec:addtab} of the Supplementary Material.  Figure~\ref{fig:box} reports boxplots  (without outliers) of the empirical MSE estimated with the FIVE (red), the F2SLSE (blue) and the RIVE (green).  The first interesting observation in the figure is that our estimators tend to produce smaller MSEs  when the signal from $x_t$ to $y_t$ is more concentrated on the first few components, i.e., when $n_a =5$. This observation is consistent regardless of the values of $n_b$ and $\mathtt{r}^2$. This may not be surprising because when $n_a$ is large,  the first few $f_j$'s (=$g_j$'s) summarize the most significant information of $\mathcal A$. 


\begin{figure}[h!]
\centering	\caption{Boxplots of the empirical MSEs ($T=500$)}\label{fig:box} 
	\begin{subfigure}{.4\textwidth}\subcaption{$(n_{a}, n_b ) =(3,3/4)$ }
	\includegraphics[width=\textwidth ]{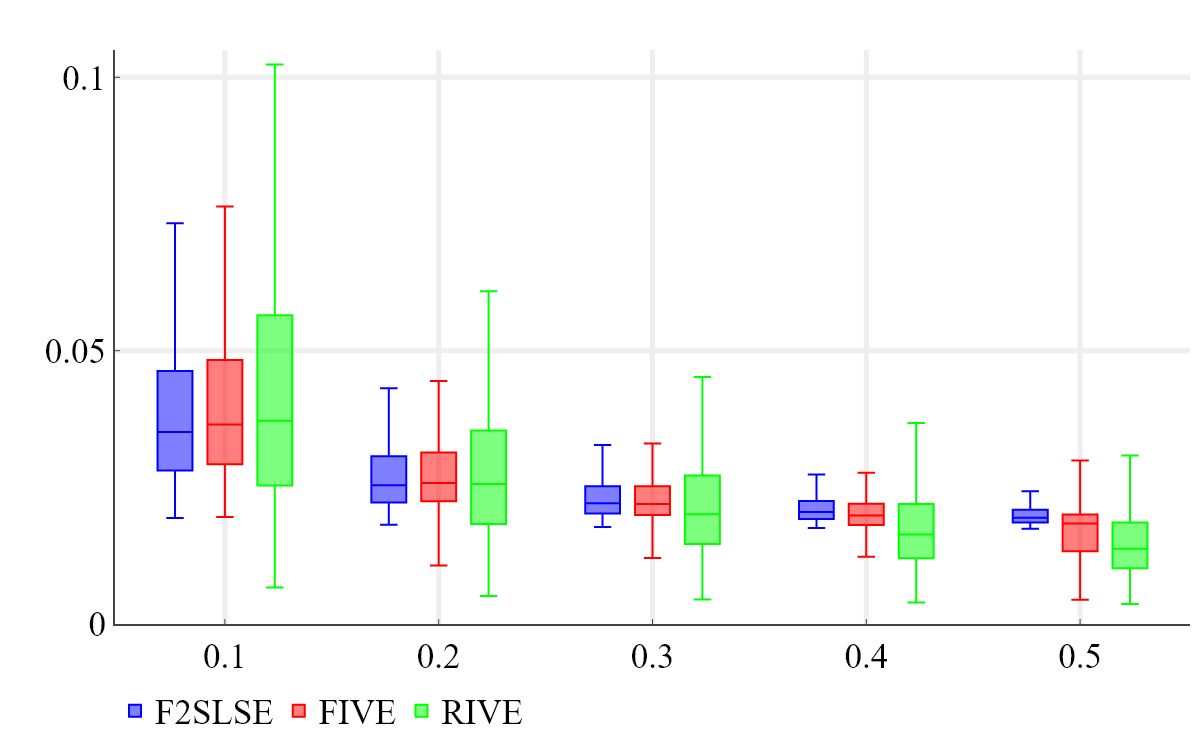}
	\end{subfigure}
\begin{subfigure}{.4\textwidth}\subcaption{$(n_{a}, n_b ) =(3,1.5)$ }
	\includegraphics[width=\textwidth ]{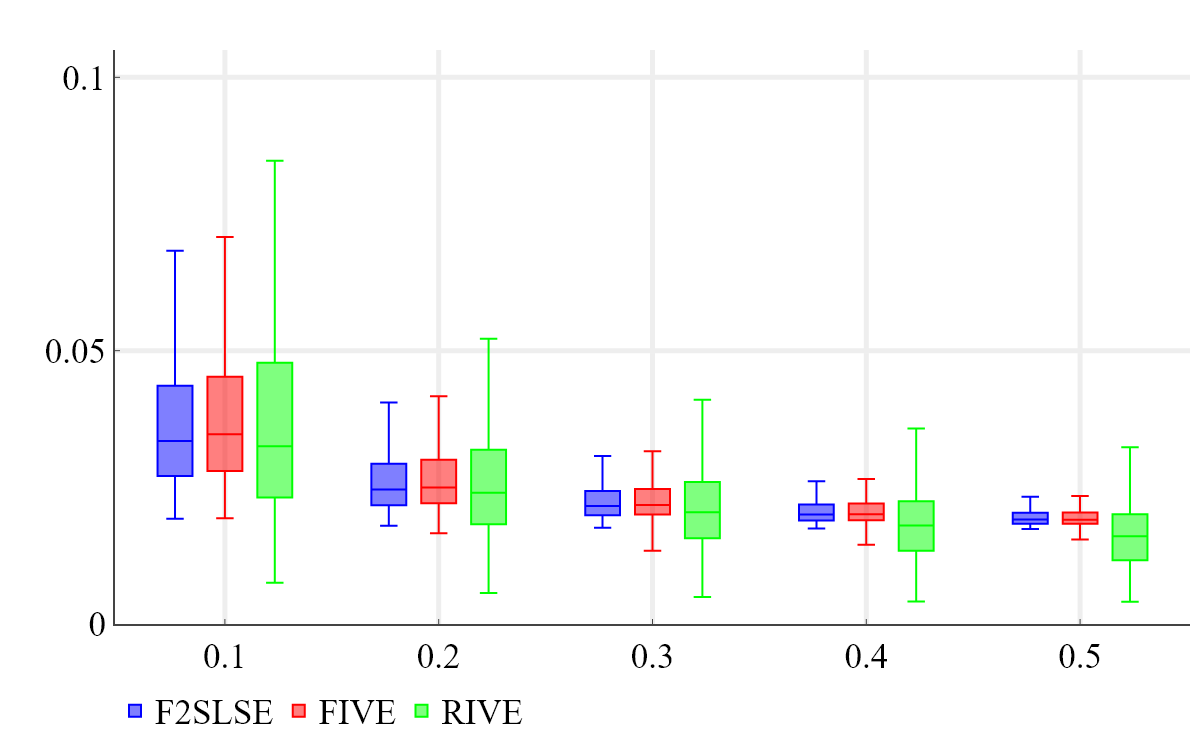}
\end{subfigure} \\
	\begin{subfigure}{.4\textwidth}\subcaption{$(n_{a}, n_b ) =(5,3/4)$ }
	\includegraphics[width=\textwidth ]{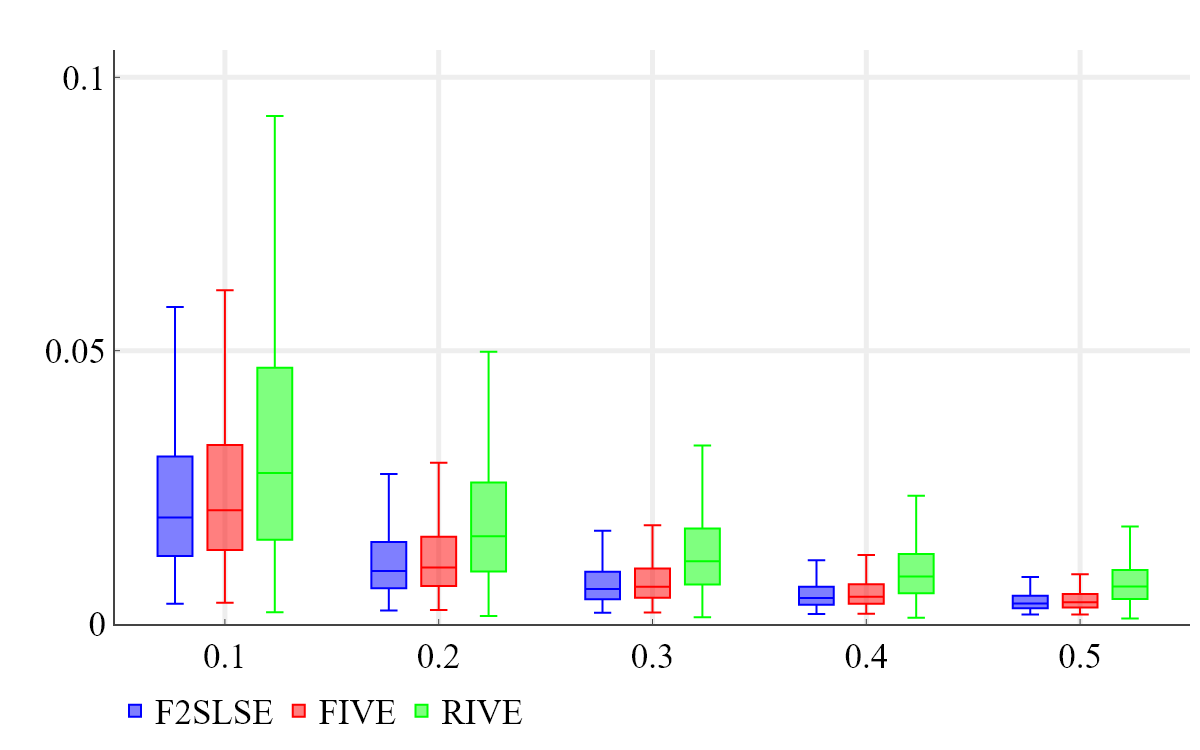}
	\end{subfigure}
\begin{subfigure}{.4\textwidth}\subcaption{$(n_{a}, n_b ) =(5,1.5)$ }
	\includegraphics[width=\textwidth ]{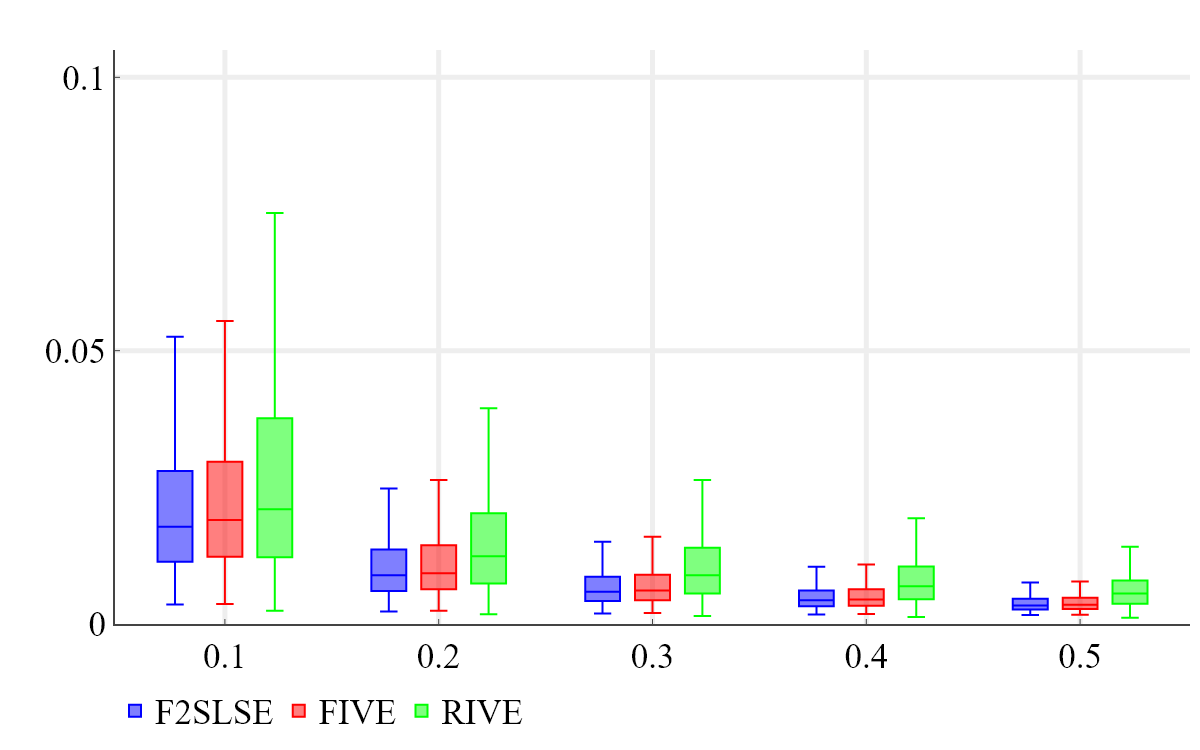}
\end{subfigure} \\ 
\footnotesize{Notes: Boxplots of the empirical MSEs of the FIVE (red), the F2SLSE (blue) and the RIVE (green) are reported for each value of the first-stage functional coefficient of determinations $\mathtt{r}^2 \in \{0.1,0.2,0.3,0.4,0.5\}$. }
\end{figure}

 In the figure, as the value of $\mathtt{r}^2$ decreases, the considered estimators tend to exhibit larger MSEs. Similar observations can be found in the standard IV literature, in which the so-called concentration parameter is used to measure the strength of IVs. Given that the coefficient of determination $\mathtt{r}^2$ is closely related to the concentration parameter in the IV literature, such a larger MSE may be understood as the distortion related to weak instruments.

In the subsequent sections, we will consider a more general setting to investigate the robustness of our estimators when the assumptions are unlikely to hold. Unlike the DGP considered in this section, verifying if all the required conditions are satisfied for the DGPs under consideration becomes nearly impossible. Given the practical challenges of confirming these conditions, practitioners may find it valuable to observe the performance of our estimators in the presence of potential violations of the required conditions.
\subsection{\label{subsub: exp1.2} Experiment 2: a modification of \citepos{Benatia2017} simulation DGP}
In this section, we consider a simulation DGP similar to that in \citepos{Benatia2017}.
Specifically, we let $\mathcal B$ in \eqref{eqsimdgp} be the identity operator $\mathcal I$ and let $\mathcal A$ be the integral operator with kernel $\kappa_{\mathcal A} (s_1, s_2) =1-|s_1-s_2|^2$ for $s_1,s_2 \in [0,1]$. In this setup, the first-stage signal is solely determined by the constant $\vartheta$. The IV $z_t$ is given as follows:  
\begin{equation}\label{eqsimu2}
z_t(s) = \tilde z_t (s; a_t, b_t) + \eta_t (s) \quad\text{and}\quad\tilde z_t (s; a_t, b_t)=   \frac{\Gamma (a_t) \Gamma(b_t)}{\Gamma (a_t +b_t)} s ^{a_t-1} (1-s) ^{b_t-1}  \text{ for } s \in [0,1],
\end{equation}
where $a_t$ and $b_t$ are randomly drawn from the uniform distribution $\text{U}[2,5]$ for each $t$. That is, $z_t$ is obtained by adding an additive noise $\eta_t$ to the beta density function with parameters $a_t$ and $b_t$. The IV in \eqref{eqsimu2} is analogous to that used for the simulation experiments in \cite{Benatia2017}, in which the additive noise   
 $\eta_t(s)$ is given by $q_t$ for all $s\in [0,1]$ with $q_t$ being randomly drawn from $\text{N}(0,1)$. In this section, we allow a more general form of $\eta_t$ by letting $\eta_t = \sum_{j=1} ^{n_J}  \sigma_{j} q_{t,j} \xi_j $ where $n_J=31$, $\{\xi_j\}_{j\geq1}$ is the  Fourier basis functions with the constant basis function $\xi_1$, and $q_{t,j}\sim_{\text{iid}}\text{N}(0,1)$ across $t$ and $j$.  \citepos{Benatia2017} IV can be understood to the case $n_J = 1$ under our notation. Then, we consider three different designs of $\{\sigma_j\}_{j\geq1}$.
Firstly, we consider the case $\sigma_{j} = c_1 \sigma_\eta$ for $j \leq 2$ and $\sigma_{j} = c_1 \sigma_\eta (0.1)^{j-2}$ for $j > 2$; this is called the sparse design. Secondly, we set $\sigma_{j} = \sigma_\eta (0.9)^{j-1}$ and call this setting the exponential design.  In the last design, which we call the geometric design, we let $\sigma_j =c_2 \sigma_\eta j^{-1}$. The parameter $\sigma_\eta$ is set to $0.5$ and $0.9$, and the constants $c_1$ and $c_2$ are chosen in such a way as to have the same Hilbert-Schmidt norm of $\mathbb E[\eta_t \otimes \eta_t]$ in all three designs. Lastly, the parameter $\vartheta$ is chosen as in Section~\ref{subsub: exp1.1} with $\mathtt{r}^2$ being set to 0.5; this can be done by using that $\mathbb{E}[\|v_t\|^2] = 1/6$, $\mathbb{E}[\|\eta_t\|^2] = \sum_{j=1}^{31}\sigma_j^2$, and the value of $\mathbb E [\Vert \tilde z_t(s_i; a_t, b_t)\Vert ^2]$ can be approximated from a large number of simulations.

\begin{table}[b!] 
\caption{Simulation results for Experiment 2: empirical MSEs and coverage probabilities}\label{tab1}
\vspace{-.5em}	\rule{1\textwidth}{.5pt}\vspace{.5em} 
\begin{tabular*}{1\textwidth}{@{\extracolsep{\fill}}rcccccccccccc} 
&\multicolumn{4}{c}{Sparse Design}&\multicolumn{4}{c}{Exponential Design}&\multicolumn{4}{c}{Geometric Design}\\\cmidrule{2-5}\cmidrule{6-9}\cmidrule{10-13}
$\sigma_\eta$	&\multicolumn{2}{c}{$ 0.5$}&\multicolumn{2}{c}{$ 0.9$}	&\multicolumn{2}{c}{$  0.5$}&\multicolumn{2}{c}{$  0.9$}	&\multicolumn{2}{c}{$ 0.5$}&\multicolumn{2}{c}{$  0.9$}\\\cmidrule{2-3}\cmidrule{4-5}\cmidrule{6-7}\cmidrule{8-9}\cmidrule{10-11}\cmidrule{12-13}
$T$ & 200 & 500 & 200 & 500 & 200 & 500 & 200 & 500 & 200 & 500 & 200 & 500 \\ \midrule
\multicolumn{13}{l}{Empirical MSE}\\
FIVE & ${0.043}$ & ${0.030}$ & ${0.042}$ & ${0.030}$ & ${0.111}$ & ${0.057}$ & ${0.108}$ & ${0.055}$ & ${0.046}$ & ${0.033}$ & ${0.046}$ & ${0.034}$ \\ 
F2SLSE & ${0.043}$ & ${0.030}$ & ${0.042}$ & ${0.030}$ & ${0.168}$ & ${0.081}$ & ${0.142}$ & ${0.076}$ & ${0.045}$ & ${0.033}$ & ${0.045}$ & ${0.033}$ \\ 
RIVE & ${0.041}$ & ${0.030}$ & ${0.040}$ & ${0.029}$ & ${0.141}$ & ${0.082}$ & ${0.134}$ & ${0.079}$ & ${0.045}$ & ${0.031}$ & ${0.043}$ & ${0.031}$ \\ 
\midrule
\multicolumn{13}{l}{Coverage probability for $  \langle  \mathcal A \widehat{\Pi}_{\K} \zeta ,\psi\rangle  $ or $  \langle  \mathcal A \widetilde{\Pi}_{\K_2} \zeta ,\psi\rangle  $}\\
FIVE & ${0.947}$ & ${0.950}$ & ${0.943}$ & ${0.949}$ & ${0.923}$ & ${0.938}$ & ${0.929}$ & ${0.936}$ & ${0.948}$ & ${0.951}$ & ${0.951}$ & ${0.952}$ \\ 
F2SLSE & ${0.947}$ & ${0.950}$ & ${0.943}$ & ${0.949}$ & ${0.907}$ & ${0.930}$ & ${0.905}$ & ${0.927}$ & ${0.950}$ & ${0.954}$ & ${0.949}$ & ${0.956}$ \\  \midrule
\multicolumn{13}{l}{Coverage probability for $ \langle   \mathcal A   \zeta , \psi\rangle  $}\\
FIVE & ${0.944}$ & ${0.940}$ & ${0.943}$ & ${0.941}$ & ${0.932}$ & ${0.950}$ & ${0.937}$ & ${0.945}$ & ${0.942}$ & ${0.946}$ & ${0.938}$ & ${0.944}$ \\ 
F2SLSE & ${0.944}$ & ${0.940}$ & ${0.943}$ & ${0.941}$ & ${0.855}$ & ${0.923}$ & ${0.882}$ & ${0.930}$ & ${0.941}$ & ${0.948}$ & ${0.937}$ & ${0.946}$ \\ 
\end{tabular*}
\rule{1\textwidth}{1pt} 
{\footnotesize Notes: Based on 1,000  replications. In the top, each cell reports the empirical mean squared error (MSE) of   three estimators: FIVE, F2SLSE, and \citepos{Benatia2017} RIVE. The last four rows report the coverage probabilities of the designated quantities; the  nominal level is 95\%.   } 
\end{table}


Table~\ref{tab1} summarizes simulation results. Overall, the MSEs  of our estimators  and those of the RIVE are similar to each other. However, in the exponential design, our estimators tend to have smaller MSEs compared to the RIVE. 
	
	 We note that  our estimators and related asymptotic results can be used to discuss the coverage probability of the interval \eqref{eq:conf} that is computed from the FIVE. 	The interval is expected to contain the random quantity $\langle \mathcal A \widehat{\Pi}_{\K} \zeta, \psi \rangle$ with ($100-\varpi$)\% of probability; moreover, if certain conditions are satisfied (see Remark \ref{remnormality}) the interval \eqref{eq:conf} can be understood as the ($100-\varpi$)\% confidence interval for $\langle \mathcal A \zeta, \psi \rangle$ which is nonrandom. Based on  Theorem \ref{thm2wa} (and also Remark \ref{remnormality2}), we may construct a similar interval with the F2SLSE and the interval is expected to include  $\langle \mathcal A \widetilde{\Pi}_{\K_2} \zeta, \psi \rangle$ (and also $\langle \mathcal A \zeta, \psi \rangle$ under certain conditions) with ($100-\varpi$)\% of probability; the coverage of this confidence interval will also be examined in this experiment. 
In order to compute the coverage probabilities, we let  $\psi=\ell_1$ and let $\zeta$ be randomly generated by $\zeta = \sum_{j=1} ^{11} \ddot q_{1,j} \ell_j$ for each realization of the DGP, where $\{ \ell_j \}_{j\geq1}$ is the polynomial basis with the constant basis function $\ell_1$ and $\ddot q_{j}\sim_{\text{iid}}\text{N}(0,j^{-4})$ across $j$. 
The simulation results are reported at the bottom of Table~\ref{tab1}. We note that, in all the considered cases, the coverage probabilities for $\langle \mathcal A \widehat{\Pi}_{\K} \zeta, \psi \rangle$  or  $\langle \mathcal A \widetilde{\Pi}_{\K_2} \zeta, \psi \rangle$ are close to the nominal level, which supports our findings in Theorems~\ref{thm2w2} and~\ref{thm2wa}. Moreover, even if  the reported coverage probabilities for $\langle \mathcal A  \zeta, \psi \rangle$ tend to be worse than those for $\langle \mathcal A \widehat{\Pi}_{\K} \zeta, \psi \rangle$ or  $\langle \mathcal A \widetilde{\Pi}_{\K_2} \zeta, \psi \rangle$,  they are still reasonably close to the nominal level of 95\%. This is what can be expected from Remarks~\ref{remnormality} and \ref{remnormality2}.  
In unreported simulations, we further experimented with different choices of $\zeta$ and $\psi$, but found no significant difference.

\subsection{Experiment 3: AR(1) model of probability density functions \label{sec:sim2}}
In this section, we examine the performance of the proposed estimators in the AR(1) model of probability density functions. What is mainly different from the earlier experiments given in Sections~\ref{subsub: exp1.1} and \ref{subsub: exp1.2} is that endogeneity is not explicitly imposed, but implicitly introduced by estimation errors.

We  let $\{p_t ^\circ\}_{t \geq 1}$ be a sequence of probability densities supported on $[0,1]$, and consider the linear prediction model of $p_t ^\circ$ given $p_{t-1}^\circ$.  Each density may be treated as a random variable taking values in $\mathcal H$, but the collection of probability densities in $\mathcal H$ is not a linear subspace. 
As a result, a direct application of the statistical methods developed in a Hilbert space setting may not be recommended; see e.g., \cite{delicado2011dimensionality}, \cite{petersen2016}, \cite*{Hron2016330}, \cite{kokoszka2019forecasting}, and \cite*{zhang2020wasserstein}. 
As a way to circumvent such issues, we consider the centered-log-ratio (clr) transformation $ y_t^\circ(s) =  \log p_t^\circ(s) - \int \log p_t^\circ(s)  ds$, $s \in [0,1]$ (see e.g., \citealp{Egozcue2006}). 
Then, $\{y_t^\circ\}_{t\geq 1}$ turns out to be a sequence in $\mathcal H_c$, the collection of all $\zeta \in \mathcal H$ satisfying $\int_{0}^1 \zeta(s)ds = 0$, and $\mathcal H_c$ is obviously a Hilbert space. Any element in $\mathcal H_c$ may be understood as a probability density via the inverse transformation $y_t^{\circ}(s) \mapsto \exp(y_t^{\circ}(s))/\int_{0}^{1} \exp(y_t^{\circ}(s))ds$. Thus, the linear prediction model of $p_t^\circ$ given $p_{t-1}^\circ$ may be recast into that of $y_t ^\circ$ given $y_{t-1} ^\circ$ in $\mathcal H_c$. 
We thus consider the following  prediction model: 
\begin{equation*}
y_t^{\circ} = c_y + \mathcal A (y_{t-1}^{\circ}-c_y) + \varepsilon_t, \label{eqtrans2}\end{equation*}
where $y_{t-1}^\circ$ and $\varepsilon_t$ are uncorrelated. 
To mimic situations commonly encountered in practice, we assume that $p_t^\circ$ (and thus $y_t^\circ$) is not observed, but only random samples  $\{s_{i,t}\}_{i=1}^{n_t}$ drawn from $p_t^\circ$ are available. If so, by replacing the density $p_t^\circ$ or the log-density $\log p_t^\circ$ with its proper nonparametric estimate, we may obtain an estimate $y_t$ of $y_t^\circ$, and then, as shown in Example \ref{exam2}, $\{y_t\}_{t\geq 1}$ satisfies $y_t = c_y + \mathcal A (y_{t-1}-c_y) + u_t$, but now $y_{t-1}$ and $u_t$ are generally correlated due to errors arising from the nonparametric estimation. We will compute the FIVE and F2SLSE by assuming that $y_{t-2}$ is a proper IV, as in Example \ref{exam2a}. Of course, this assumption may not be true depending on how estimation errors are generated. Even with this possibility, it may be of interest to practitioners, who are very often have no choice but to replace $p_t^\circ$ or $\log p_t^\circ$ with a standard nonparametric estimate, to see if a naive use of  our estimators  can make any actual improvements in estimating $\mathcal A$. This is the purpose of simulation experiments in this section. 

Specifically, we first estimate $y_t^\circ$ from $n$ random samples that are generated from $p_t^\circ$ by (i) the local likelihood density estimation method proposed by \cite{loader1996} (see Appendix \ref{sec:simul} for more details) and (ii) the standard kernel density estimation method with the Gaussian kernel and Silverman's rule-of-thumb bandwidth \citep{silverman2018density}. Even if the former is more suitable for estimation of  $y_t^\circ$ \citep[Section 4.2]{SEO2019}, the latter is considered as well because of its popularity in empirical studies. Once $y_t$ is computed, it is represented by the first 30 nonconstant Fourier basis functions for implementation of the FPCA in $\mathcal H_c$. We let $c_y$ be the clr transformation of the normal density function with mean $0.5$ and variance $0.25^2$ that is truncated on $[0,1]$. In addition, $\varepsilon_t = \sum_{j=1} ^{\infty}  \sigma_{j} q_{t,j} \xi_j^c$, where $\{\xi^c_j\}_{j\geq1}$ is the Fourier basis functions except for the constant basis function, and $q_{t,j}\sim_{\text{iid}}\text{N}(0,1)$ across $t$ and $j$.\footnote{In actual computation, $\varepsilon_t$ can be approximated by $ \sum_{j=1}^{L}  \sigma_{j} q_{t,j} \xi_j^c$ for some large $L$. We set $L$ to 50 in this example and found no significant difference even from big changes in $L$ as long as $L \geq 50$.} Below we consider two different specifications of $\sigma_j$, which are respectively called the exponential design and the sparse design; in the exponential design, $\sigma_{j} = 0.1 (0.9)^{j-1}$, and in the sparse design, $\sigma_{j} = c_{\sigma}$ for $j \leq 2$ and $\sigma_{j} =c_{\sigma}(0.1)^{j-2}$ for $j > 2$, where $c_{\sigma}$ is chosen so that the Hilbert-Schmidt norms of $\mathbb{E}[\varepsilon_t \otimes \varepsilon_t]$ in both designs are equal. These two designs are respectively obtained by setting $\sigma_{\eta}$ to $0.1$ in the sparse and exponential designs considered for $\eta_t$ in Section \ref{subsub: exp1.2}, and the reason why we choose a relatively smaller scale of $\sigma_j$ in this experiment is only to avoid as much as possible that the simulated densities have shapes that are rarely observed in practice (e.g., densities that are U-shaped or highly multimodal). We let $\mathcal A$ be defined by $\sum_{j=1}^\infty a_j \xi_j ^c \otimes \xi_j ^c$,\footnote{$\mathcal A$ is approximated by $\sum_{j=1}^{50} a_j \xi_j ^c \otimes \xi_j ^c$  in actual computation as in the case of $\varepsilon_t$} and, for each realization of the DGP, the coefficients $\{a_j\}_{j\geq 1}$ are independently determined across $j$ as follows,   
\begin{equation*}
a_1\sim \text{U}[0.4,0.9], \quad a_2 \sim \text{U}[0.4,0.9], \quad a_j = a_{u,j} (0.5)^{j-2}\quad\text{and}\quad a_{u,j} \sim_{\text{iid}}\text{U}[0,0.9] \quad\text{for}\quad j \geq 3.
\end{equation*}
Note here that we let the first two coefficients $a_{1}$ and $a_{2}$ be bounded below by $0.4$, which is to ensure that the operator norm of the cross-covariance operator of $y_{t-1}^{\circ}$ and $y_{t-2}^\circ$ is bounded away from zero. If this quantity is close to zero, then the employed IV {may }become `weak' and this case is not considered in the present paper. 
\begin{table}[t!]
\caption{Simulation results for Experiment 3: empirical MSEs ($a_1,a_2\geq 0.4$)}\label{tab3}
\vspace{-.5em}	\rule{1\textwidth}{.5pt}\vspace{.5em} 
\begin{tabular*}{1\textwidth}{@{\extracolsep{\fill}}llcccccccc} 
&&\multicolumn{4}{c}{Sparse Design}&\multicolumn{4}{c}{Exponential Design}\\\cmidrule{3-6}\cmidrule{7-10}
&$n$	&\multicolumn{2}{c}{$100$}&\multicolumn{2}{c}{$ 150$}	&\multicolumn{2}{c}{$  100$}&\multicolumn{2}{c}{$150$}\\\cmidrule{3-4}\cmidrule{5-6}\cmidrule{7-8}\cmidrule{9-10}
&T & 200 & 500 & 200 & 500 & 200 & 500 & 200 & 500 \\ \midrule
\multirow{4}{*}{Loader's}  & FIVE & ${0.206}$ & ${0.158}$ & ${0.187}$ & ${0.152}$ & ${0.400}$ & ${0.227}$ & ${0.315}$ & ${0.194}$ \\ 
& F2SLSE & ${0.204}$ & ${0.157}$ & ${0.186}$ & ${0.151}$ & ${0.392}$ & ${0.219}$ & ${0.310}$ & ${0.189}$ \\ 
& FLSE & ${0.255}$ & ${0.228}$ & ${0.208}$ & ${0.187}$ & ${0.427}$ & ${0.354}$ & ${0.333}$ & ${0.267}$ \\  \midrule
\multirow{4}{*}{Silverman's}
& FIVE & ${0.272}$ & ${0.218}$ & ${0.223}$ & ${0.188}$ & ${0.395}$ & ${0.257}$ & ${0.314}$ & ${0.215}$ \\ 
& F2SLSE & ${0.272}$ & ${0.215}$ & ${0.223}$ & ${0.186}$ & ${0.396}$ & ${0.247}$ & ${0.315}$ & ${0.209}$ \\ 
& FLSE & ${0.326}$ & ${0.291}$ & ${0.251}$ & ${0.226}$ & ${0.418}$ & ${0.351}$ & ${0.322}$ & ${0.258}$ \\ 
\end{tabular*}
\rule{1\textwidth}{1pt} 
{\footnotesize Notes: Based on 1,000  replications. Each cell reports the empirical mean squared error (MSE) of the three considered estimators: FIVE, F2SLSE, and \citepos{Park2012} FLSE. \commRV{The RIVE considered in Sections \ref{subsub: exp1.1} and \ref{subsub: exp1.2} is excluded in this experiment since the estimator is developed for iid functional data.}}
\end{table} 

Table \ref{tab3}  reports the empirical MSEs of the proposed estimators and \citepos{Park2012} FLSE  
when $n=100$ and $150$ (recall that $n$ is the number of random samples drawn from the distribution $p_t ^\circ$ to estimate $\log p_t^\circ$ or $p_t^\circ$). The IV estimators tend to exhibit smaller empirical MSEs than the FLSE. 
The superior comparative performance of the IV estimators is more noticeable  when $n$ is small and $T$ is large. This is what can be conjectured from our earlier discussion; as $n$ gets smaller, $y_t$ becomes a less accurate estimate of $y_t^\circ$, and hence the estimators that address the possible endogeneity caused by estimation errors will work better. The FIVE or the F2SLSE exhibits the smallest MSE in most of the cases (see also Table \ref{tab2aa} in Appendix reporting the simulation results for the case where the lower bound of $a_{1}$ and $a_{2}$ increases to $0.6$).  
However, it is hard to conclude the relative performance between the IV estimators; this may depend on various factors such as the DGP and the method of density estimation. Thus, it would be advisable to use those IV estimators complementarily in practice.

\subsection{Empirical application: effect of immigration on native wages \label{sec:emp}}  

In this section, we use our estimation methods to investigate the effect of immigrant inflows on the labor market outcomes of workers with heterogeneous skills, which has received due attention from both researchers and policymakers, see, e.g., \cite{Card2009}, \cite*{borjas2011}, \cite{Ottaviano2011}, and \cite{Glitz2012}. To begin with, we use national level data and generalize a widely used empirical model by viewing the variables of interest as functions depending on a  measure of relative communication skill provision. Our measure of relative communication skill provision is similar to \citepos{Peri2009} measure of occupation-specific relative provision of communication versus manual skills.  The number of distinct skill levels, denoted $s_j$, is 223, and, by construction, each occupation is uniquely identified by the skill score $s_j\in[0,1]$. Its formal definition is provided in the Supplementary Material.

We merge the percentile scores of relative communication skill provision to individuals in the monthly CPS data running from January 1996 to December 2019. The CPS data, which can be downloaded from the Integrated Public Use Microdata Series (IPUMS)\nocite{Flood2020}, provide information on various characteristics of individuals: hourly wage, citizenship status, age, employment status, and occupation. We focus on individuals who (i) are aged between 18 and 64 years, (ii) are not self-employed, and (iii) have positive income. Immigrants are defined by those who are not a citizen or are a naturalized citizen. The skill-dependent labor supply of immigrants ($\ell^\circ_{it}(s_j)$) and that of natives ($\ell^\circ_{nt}(s_j)$) are computed by the total hours of work per week (weighted by the variable WTFINL) provided by the foreign- and native-born workers for each $s_j$. The skill-dependent native wage is computed by weighted averaging weekly wages of native workers\footnote{The weekly wage of a native worker is computed as $(\text{hourly wage})\times (\text{usual hours of work})$, and the variables required to compute this quantity are also available in the CPS. We use the variable EARNWT as a weight.} in the occupation corresponding to $s_j$, and its logged value ($w_t^\circ(s_j)$) is used for the analysis.

The empirical models used in the labor economics literature (e.g., \citealp{Dustmann2012}; \citealp{Sharpe2020101902}) can be written as follows:  $
\Delta w_{t} ^\circ (s_j) =    \beta_j^\circ \Delta h_t ^\circ(s_j)  + u_t^\circ (s_j),
$ where $\Delta w_t ^\circ (s_j) = w_t^\circ(s_j)-w_{t-1}^\circ(s_j)$, $u_t^\circ (s_j)$ denotes the disturbance term, $\beta_j^\circ$ is the parameter of interest,  the explanatory variable $\Delta h_t ^\circ(s_j)$ is the first difference of $h_t ^\circ (s_j)$, and $h_t^\circ (s_j)= \ell_{it}^\circ (s_j) /( \ell_{nt}^\circ (s_j)+\ell_{it}^\circ (s_j))$.  
In the above model, an inflow of immigrants in the occupation with $s_j$ is assumed to affect only the  wages of natives in the occupation requiring the same skill level, which  seems to be restrictive. To resolve this issue, one may instead allow spillover effects across occupations, but this requires researchers to estimate too many parameters; for example, if we allow a spillover effect from the occupation corresponding to $s_i$ to another occupation corresponding to $s_j$ for any arbitrary $i,j\in \{1,\ldots,223\}$, then there are $223^2$ elements to be estimated. 
As an alternative, we view observations $w_t ^\circ(s_j)$ and $h_t^\circ (s_j)$ for each $t$ as imperfect realizations of curves $w_t$ and $h_t$, and use our methodology  developed in the previous sections.
To this end, we first estimate each of those curves with the standard Nadaraya-Watson estimator employing the second-order Gaussian kernel and the bandwidth minimizing the least square cross validation criterion. The smoothed curves are represented by 15 cubic B-Spline functions and are denoted by $  w_t  $ and $  h_t  $, respectively. Then, we estimate the following model:\begin{equation}
\Delta	{w}_t  = \mathcal A \Delta {h}_t  +u_t  , \label{eq:emp1}
\end{equation} 
where  $\Delta w_t=w_t-w_{t-1}$, $\Delta h_t=h_t-h_{t-1}$, and $\Delta h_t$ is likely to be correlated with $u_t$ due to, e.g., the  self-selection bias  pointed out by \cite{borjas1987} and \cite{llull2018}. Thus, we use the changes in the imputed share of immigrants as an IV, which has been employed in various contexts, including \cite{Card2009}, \cite{Peri2009}, and \cite{david2013growth}. Specifically, the imputed share of immigrants in the occupation corresponding to $s_j$, denoted $z_t ^\circ (s_j)$, is defined as follows: \begin{equation*}
z_t ^\circ (s_j) = \frac{\tilde{\ell}_{it}^\circ (s_j)}{ \ell_{nt} (s_j) +  \tilde{\ell}_{it}^\circ (s_j) }\quad\text{ and }\quad	\tilde{\ell}_{it} ^\circ (s_j) = \frac{1}{12}\sum_{b=1} ^B\sum_{t=1} ^{12} \frac{\ell_{it1994,b} ^\circ (s_j)}{\ell_{it1994,b} ^\circ } \ell_{it,b} ^\circ,
\end{equation*}
where $b$ denotes the country of birth of immigrants, $\ell ^\circ _{it1994,b} (s_j)$ is the labor supply of immigrants in the occupation corresponding to $s_j$ from the country $b$ in the month $t$ of the year 1994, and $\ell ^\circ_{it1994,b}$ is its aggregation over $s_j$. 
 The curve of imputed shares of immigrants, denoted $  z_t$, is  obtained by smoothing $   z_t ^\circ (s_j)$, and the instrument, denoted $\Delta z_t$, is the first difference of $  z_t$.

\begin{figure}[h!]
\caption{Functional data (grey) and their mean functions (black)\label{fig1}}\begin{subfigure}{.31\textwidth}\subcaption{average log native wages\label{fig1a}}
\includegraphics[width=\textwidth]{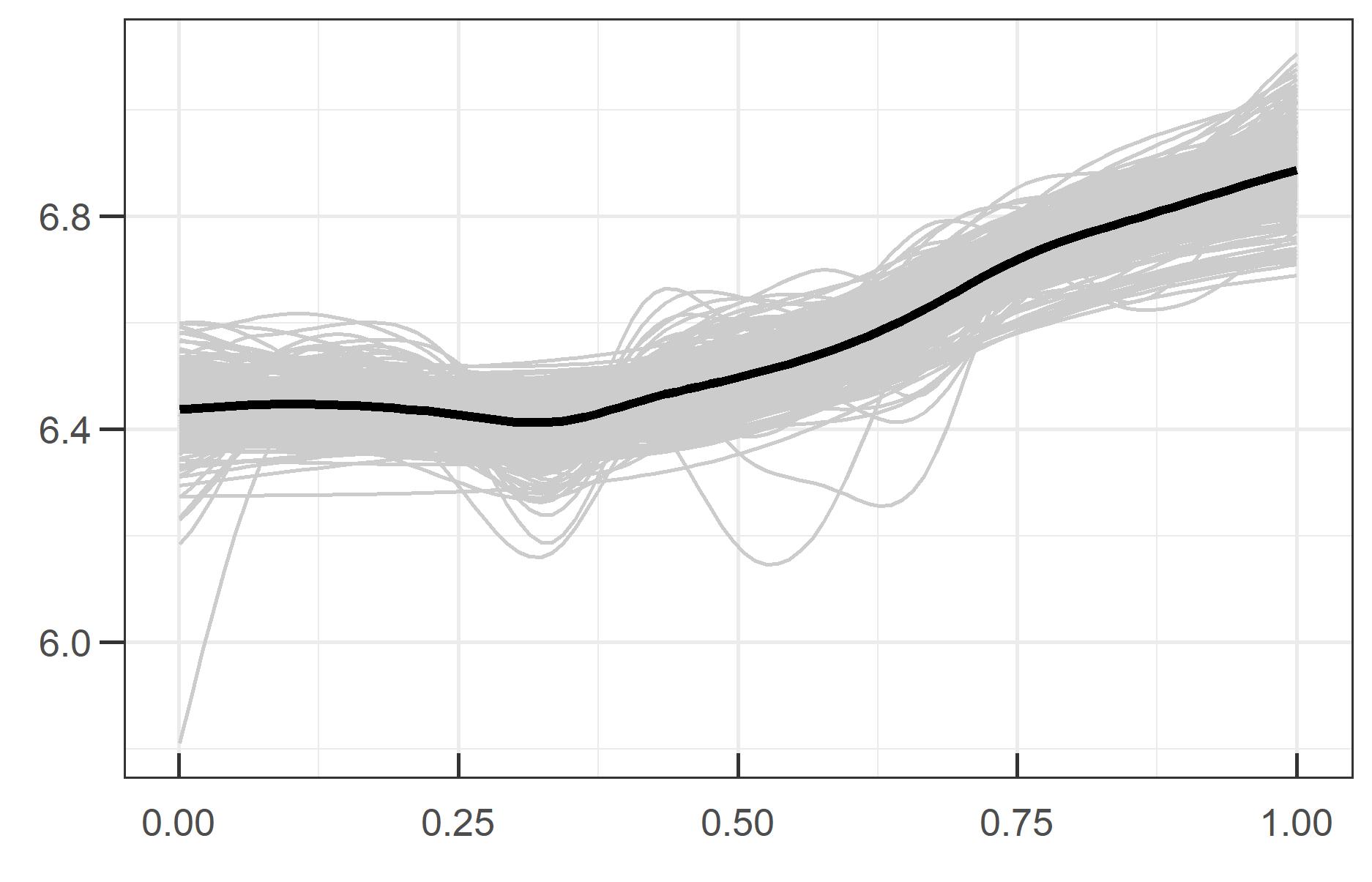}\vspace{0em}
\end{subfigure}\begin{subfigure}{.31\textwidth}\subcaption{relative share of immigrants\label{fig1b}}
\includegraphics[width=\textwidth]{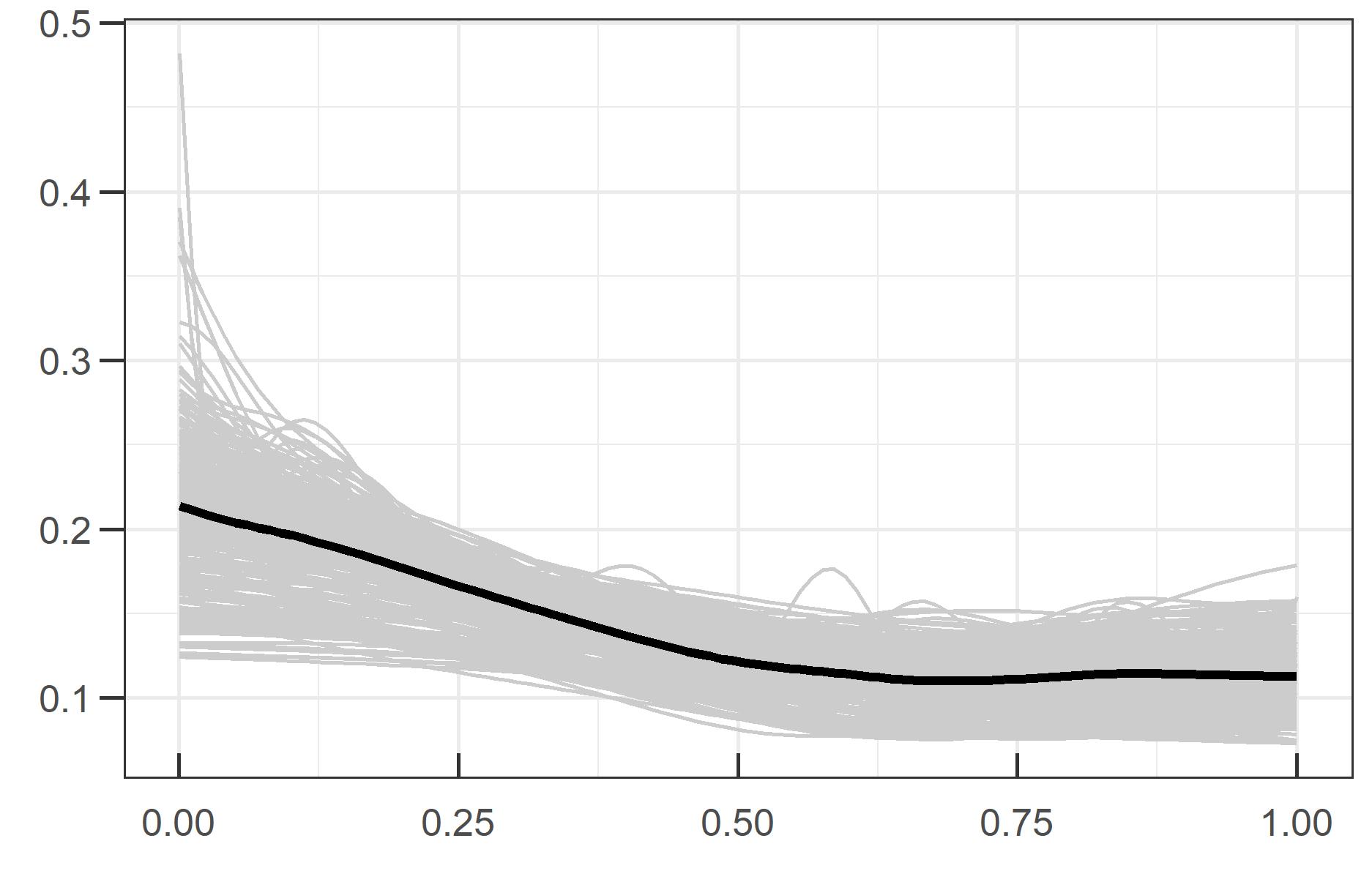}\vspace{0em}
\end{subfigure}	 \begin{subfigure}{.31\textwidth}\subcaption{imputed share of immigrants\label{fig1c}}
\includegraphics[width=\textwidth]{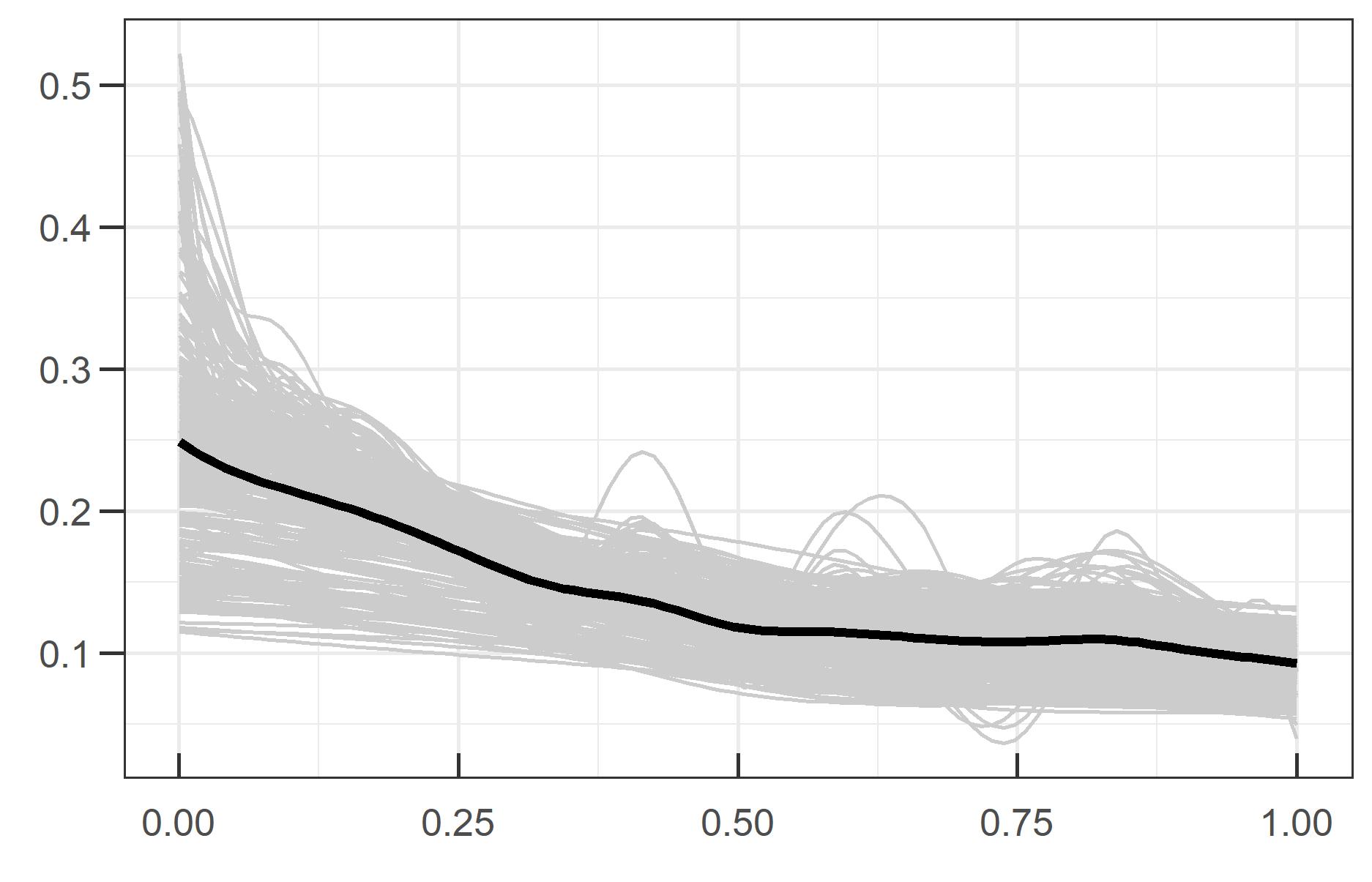}\vspace{0em}
\end{subfigure}	 
\end{figure}

The smoothed curves are reported in Figure~\ref{fig1}.  The solid lines in Figure~\ref{fig1} indicate the mean functions of $  w_t$, $  h_t$, and $  z_t$. Figure~\ref{fig1a} shows that native workers tend to be better paid if they are in occupations needing relatively higher communication skills. On the other hand, the share of immigrants tends to decrease in such occupations, and so does the imputed share of immigrants; this may be because natives have a comparative advantage in communication intensive tasks.

\begin{figure}[b!]\begin{center}
\caption{Estimated effects of immigration computed from the FIVE   \label{est.curves}}\vspace{-.5em}
\begin{subfigure}{.31\textwidth}\subcaption{$\zeta = 1\{ 0 \leq s < 1/3  \}$\label{eff.fig1}}
\includegraphics[width=\textwidth]{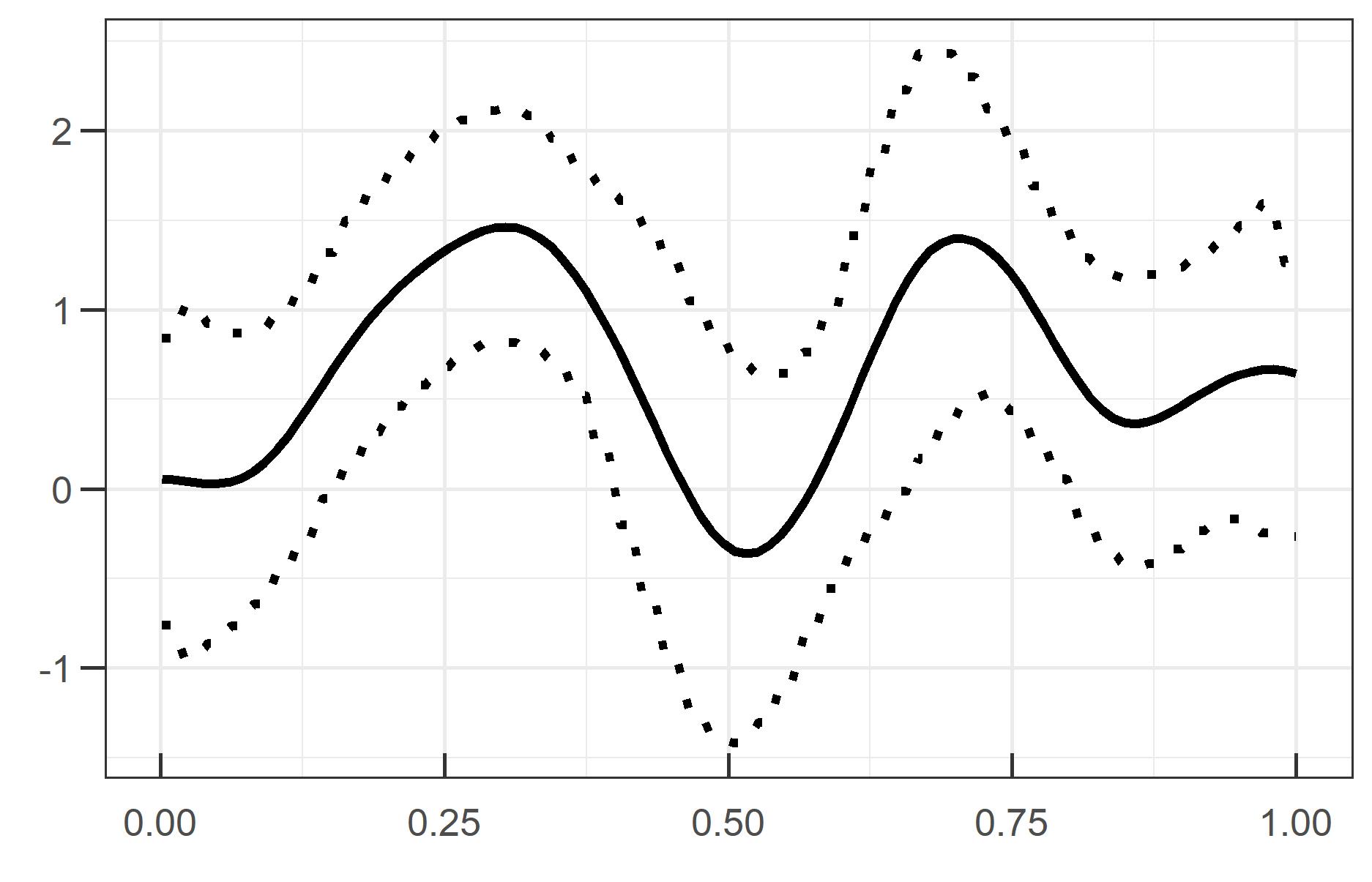}
\end{subfigure}\begin{subfigure}{.31\textwidth}\subcaption{$\zeta = 1\{ 1/3\leq s < 2/3  \}$\label{eff.fig2}}
\includegraphics[width=\textwidth]{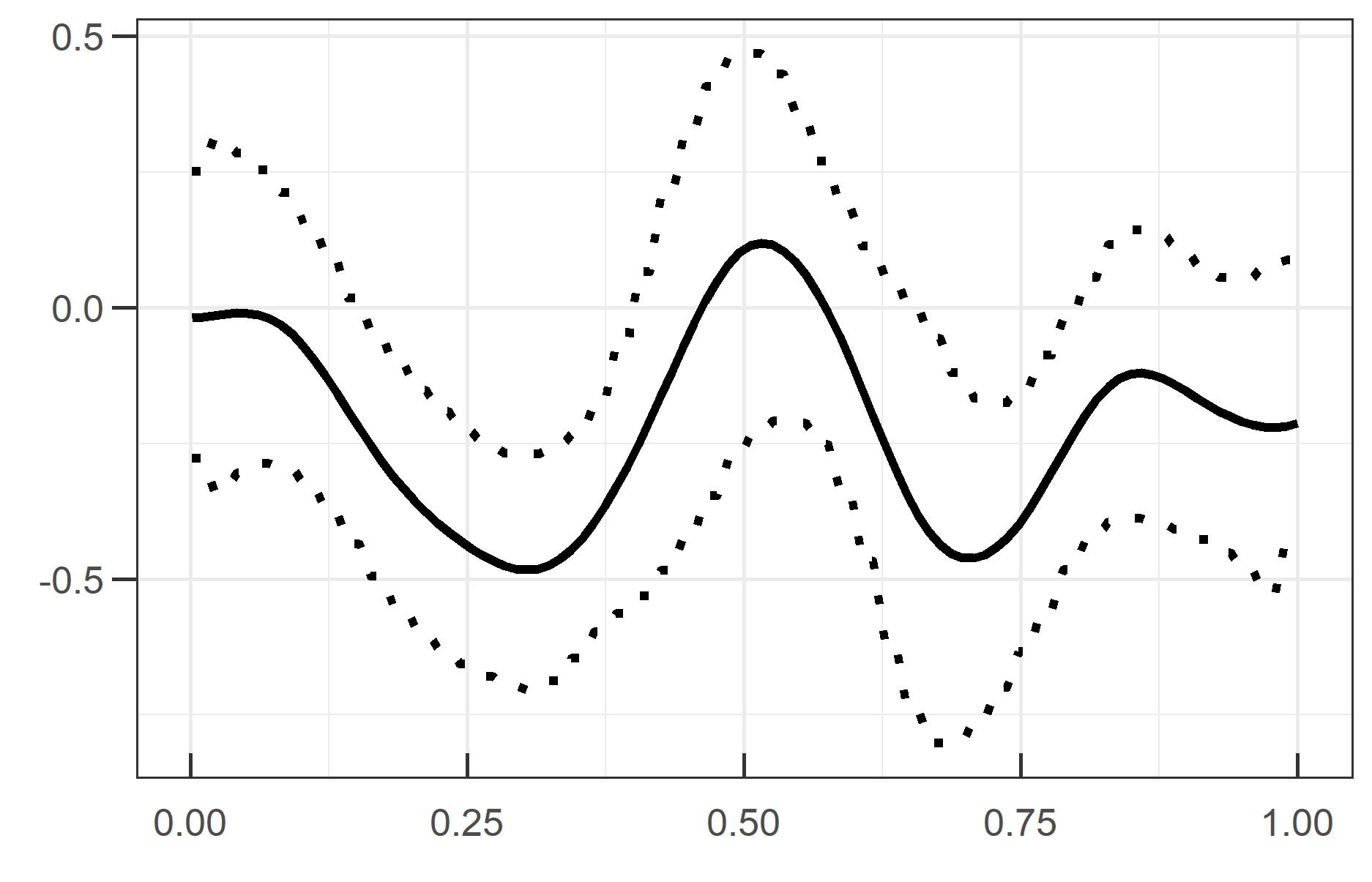}
\end{subfigure}\begin{subfigure}{.31\textwidth}\subcaption{$\zeta =1\{ 2/3 \leq s < 1  \}$\label{eff.fig3}}
\includegraphics[width=\textwidth]{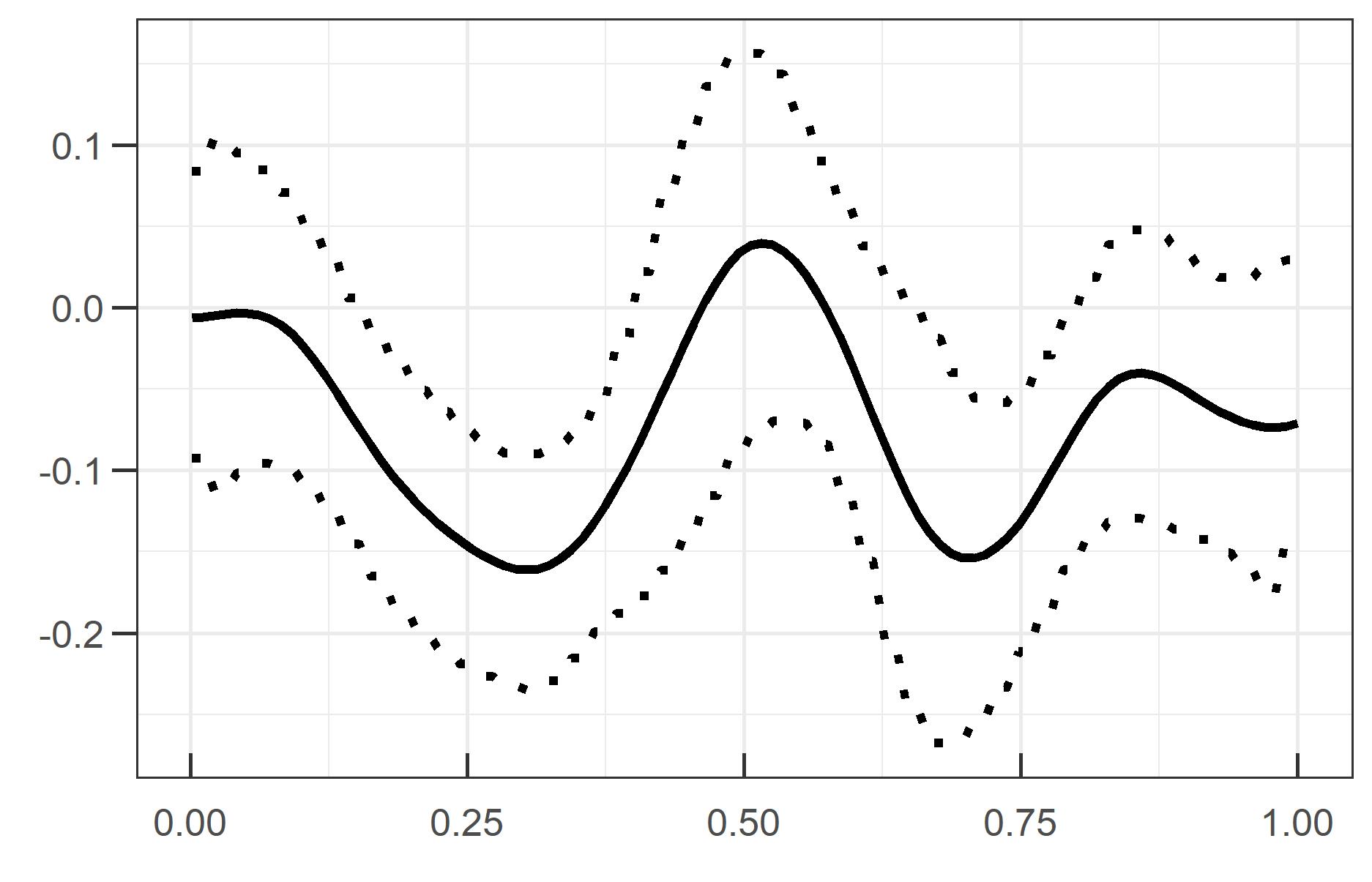}
\end{subfigure}  
\end{center} 
{\footnotesize Note: The solid line reports $ \widehat {\mathcal A} \zeta$. The dashed lines represent the collection of  confidence intervals for the local averages of $\mathcal A\widehat{\Pi}_{\K} \zeta$  over finely defined interval $[(m-1)/M, m/M]$, where $M=50$.  For each $m$, the interval is constructed as in \eqref{eq:conf} with $95\%$ significance level by noting that the local average is given by $\langle \mathcal A \widehat{\Pi}_{\K} \zeta ,\zeta_{m(M)} \rangle$, and this interval, of course, may be viewed as the confidence interval for the local average of $\mathcal A \zeta$  under certain conditions (see Remark \ref{remnormality})}.
\end{figure}

Then, we apply our estimation method to study if an inflow of immigrants has a heterogenous impact to native workers depending on the value of $s$. For the ease of interpretation, we focus on the case where immigrants are fully concentrated in a group of occupations with low, medium or high communication skill intensity, but they are evenly distributed within the group; that is, we set $\zeta$ in Theorem~\ref{thm2w2} to $1\{ 0 \leq s < 1/3  \}$, $1\{ 1/3\leq s < 2/3  \}$, and $1\{ 2/3 \leq s < 1  \}$. The regulariation parameter is chosen as in Section~\ref{subsub: exp1.1}. The results computed from the FIVE are summarized in Figure~\ref{est.curves}. 
The estimation results from the F2SLSE are similar and thus omitted. 
Overall,   our findings in Figure~\ref{est.curves} reconfirm the existing evidence that an inflow of immigrants heterogeneously affects the labor market outcomes of native workers according to workers' skills.    For example, in Figure~\ref{eff.fig1}, if the share of immigrants increases in occupations with a low value of $s_j$, then native wages are overall positively affected, although the size of this effect depends on the value of native workers' $s_j$ as well. In particular,  Figure~\ref{eff.fig1} suggests that native workers in the occupations with $s \in[0.1,0.4]$ will experience the most significant positive wage effects. This is somewhat consistent with \citepos{Peri2009} finding that native workers in occupations intensive in manual skills take advantage of having better-paid jobs when similarly skilled immigrants enter into the market. On the other hand, in Figure~\ref{eff.fig2}, it seems that the natives in occupations with $s \in [0.1,0.4]$ are negatively affected if the share of immigrants increases in occupations of which $s\in[1/3,2/3]$.

Before concluding this section, recall that the earlier literature mostly relies on the strategy of reducing dimensionality of the model by classifying workers into a few groups according to a measure of their skill. We will make a comparison of our estimation results with those based on this strategy. To this end, we let $\zeta_{j(J)} = J \times 1\{  {j -1}< J s \leq {j }  \}$. Then, the inner product $\langle \Delta w_t, \zeta_{j(J)}\rangle$ computes the average of changes in the log wages of natives in the occupations of which $s$ is between $({j}-1)/J$ and ${j}/J$, at time $t$. We then estimate the following using the standard 2SLSE,\begin{equation}
{\Delta   w_{ t(J) }} = \beta {\Delta   h_{ t(J) }} + u_{t(J)} ,  \label{emp:eq2}\end{equation}
where ${\Delta   w_{ t(J) }} = ( \langle \Delta w_t, \zeta_{1(J)}\rangle , \ldots, \langle \Delta w_t ,  \zeta_{J(J)}\rangle )' $ and ${\Delta   h_{ t(J) }} = ( \langle \Delta h_t, \zeta_{1(J)}\rangle, \ldots, \langle \Delta h_t ,  \zeta_{J(J)}\rangle )'$. The IV that is used to compute the 2SLSE is $ {\Delta   z_{ t(J) } } = ( \langle \Delta z_t, \zeta_{1(J)}\rangle , \ldots, \langle \Delta z_t ,  \zeta_{J(J)}\rangle )'$. For example, if $J=3$, we have $\langle \Delta w_t , \zeta_{1(3)}\rangle $, $\langle \Delta w_t , \zeta_{2(3)}\rangle $, and $\langle \Delta w_t, \zeta_{3(3)} \rangle$, which are plotted in  Figure~\ref{fig2}. In the following, we consider three different values for $J$: $3,$ $7$, and $11$.

\begin{figure}[t!]
\caption{Group characteristics (first difference of log native wages, $\langle \Delta w_t , \zeta_{j(3)}\rangle$)\label{fig2}}
\begin{subfigure}{.31\textwidth}\subcaption{$\zeta = \zeta_{1(3)}$\label{fig2a}}
\includegraphics[width=\textwidth]{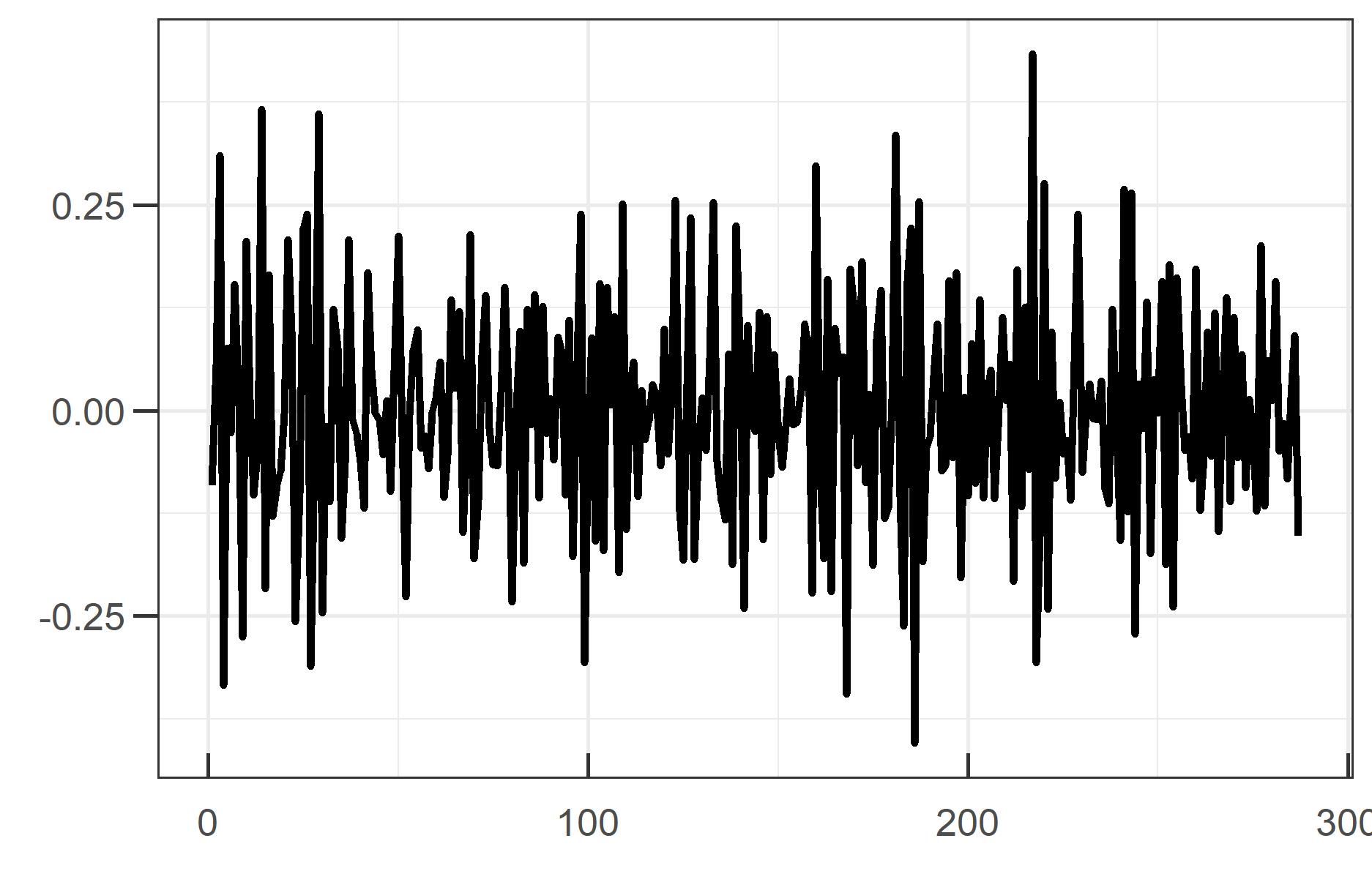}\vspace{0em}
\end{subfigure}	\begin{subfigure}{.31\textwidth}\subcaption{$\zeta = \zeta_{2(3)}$\label{fig2b}}
\includegraphics[width=\textwidth]{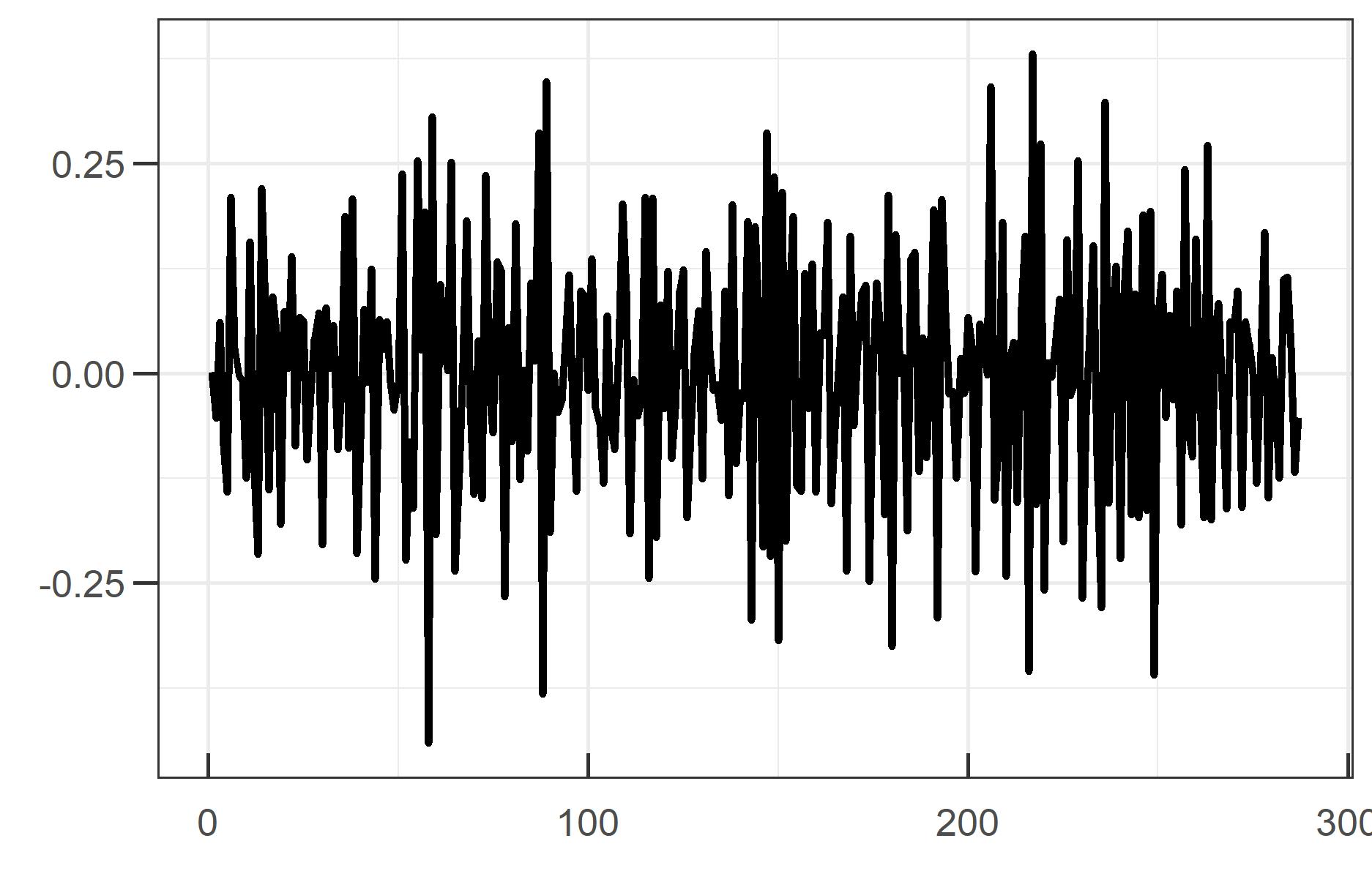}\vspace{0em}
\end{subfigure}	\begin{subfigure}{.31\textwidth}\subcaption{$\zeta = \zeta_{3(3)}$\label{fig2c}}
\includegraphics[width=\textwidth]{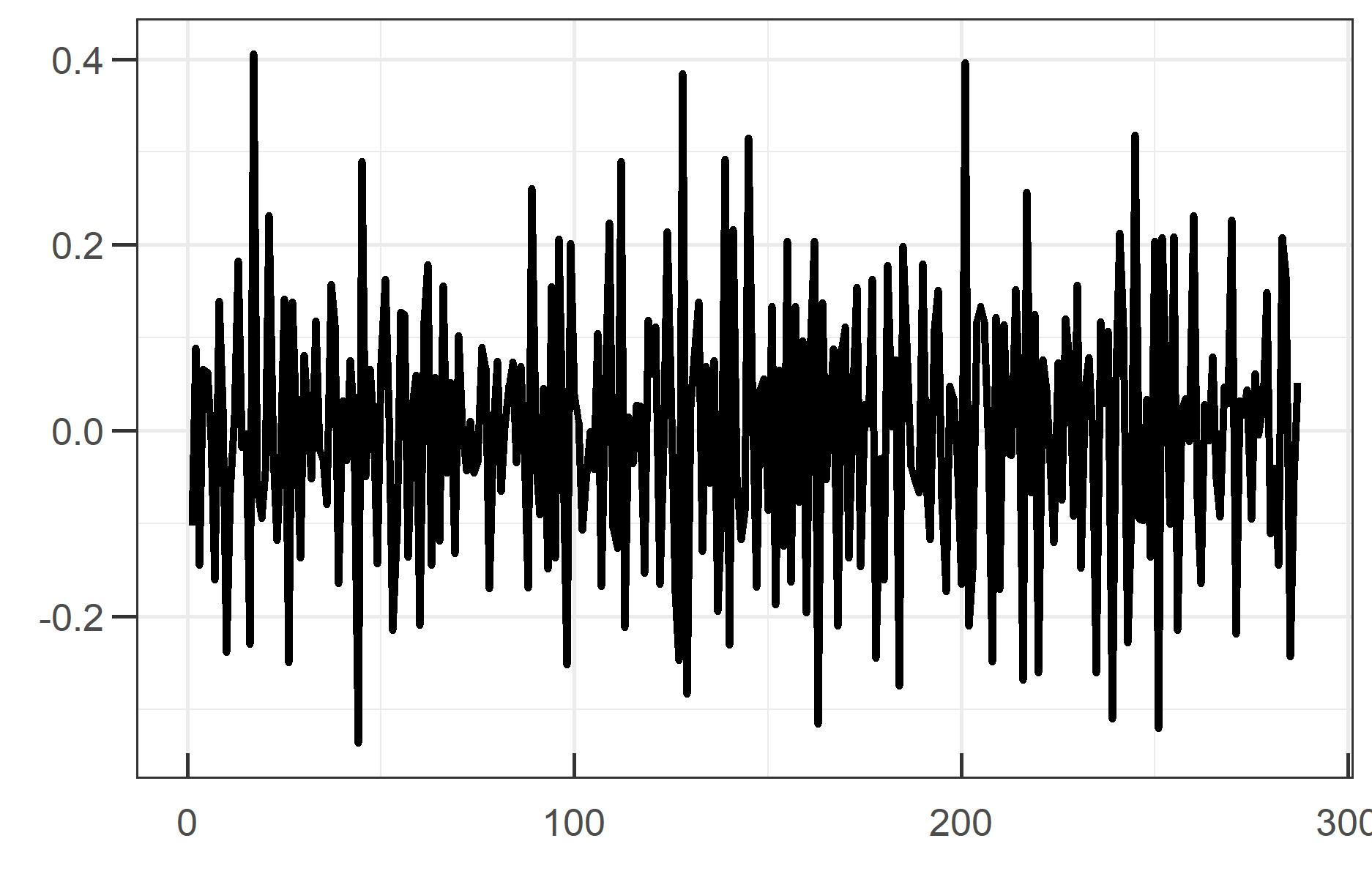}\vspace{0em}
\end{subfigure}	
\end{figure}


We compare the performance of estimators by using the root mean squared prediction error (RMSPE), which is computed with a rolling window for three test sets, with setting their starting points respectively as 2013/01, 2015/01, and 2017/01. Specifically, if we let \(\bar{u}_h\) be the forecasting error computed from the FIVE, F2SLSE or RIVE, then the RMSPE is given by \((H^{-1}\sum_{h=1}^H \int \bar{u}_h(s)^2 ds)^{1/2}\)  and the regularization parameters are chosen in such a way as to minimize the RMSPE. 
 Let $\hat{u}_{h,j(J)}$ denote the $j$th element of the forecasting error $\hat{u}_{h(J)}$ computed from the 2SLSE for each $J$. The standard RMSPE of the 2SLSE is given by $(H^{-1}\sum_{h=1} ^H \sum_{j=1} ^J\hat{u}_{h,j(J)}^2)^{1/2}$. {Because this measure is nondecreasing in $J$ and thus does not provide a fair comparison between RMSPEs,} we instead consider the normalized RMSPE, given by $((JH)^{-1} \sum_{h=1} ^H \sum_{j=1} ^J\hat{u}_{h,j(J)}^2)^{1/2}$, which can be reasonably compared to the RMSPEs of our estimators since, for each $h\geq 1$, (i) both $\hat{u}_{h,j(J)}$ and $\langle \bar{u}_h, \zeta_{j(J)} \rangle$ are estimates of the local average of $u_h$ over the interval $[(j-1)/J,j/J]$ and (ii) $\int \bar{u}_h(s)^2 ds$ may be approximated by $J^{-1} \sum_{j=1}^J \langle \bar{u}_h, \zeta_{j(J)} \rangle^2$.
\begin{table}[t!]
\caption{Root mean squared prediction errors }\label{tab.emp1}
\vspace{-.5em} 	\rule{1\textwidth}{.5pt}\vspace{.5em} 
\begin{tabular*}{1\textwidth}{@{\extracolsep{\fill}}lcccccc}  
\multirow{2}{*}{Test period} & \multirow{2}{*}{FIVE} & \multirow{2}{*}{F2SLSE}& \multirow{2}{*}{RIVE} &\multicolumn{3}{c}{2SLSE}\\\cmidrule{5-7}
&&&&$J=3$&$J=7$&$J=11$\\\midrule
2013/01$\sim$ &  0.1886 & 0.1889 & 0.1883 & 0.1514 & 0.2025 & 2.5877 \\ 
2015/01$\sim$ &0.1828 & 0.1829 & 0.1823& 0.1432 & 0.1884 & 3.4587 \\ 
2017/01$\sim$ &0.1679 & 0.1678 & 0.1676& 0.1236 & 0.1734 & 1.0399 \\ 
\end{tabular*}
\rule{1\textwidth}{1pt} 
{\footnotesize Notes: Each cell reports the estimated (normalized) RMSPE which is computed using three test sets.} 
\end{table}

Estimation results are reported in Table~\ref{tab.emp1}. We first note that the results from our estimators and those from the RIVE are very similar to each other. 
 These estimators report smaller RMSPEs than those of the 2SLSE except for the case $J=3$. Even if the 2SLSE reports the smallest RMSPEs when $J=3$, we should note that in this case,  223 different levels of skill are aggregated into only three groups, resulting in a lot of information loss. Moreover, the normalized RMSPE of the 2SLSE rapidly increases as we consider more finely defined skill groups.  This may be because, as $J$ gets larger, the number of parameters to be estimated rapidly increases. 
In addition, for a large $J$, the sample (cross-)covariance matrices for computing the 2SLSE tend to be singular,  
and thus the 2SLSE is expected to perform poorly. {This result also suggests that the pre-classification strategy can have a significant effect on the estimation results and our interpretation of them.} Therefore, the results given by Table~\ref{tab.emp1} imply that our functional IV methodology can be an appealing alternative to practitioners. 


\section{{Conclusion} \label{sec:con}}

 \phantomsection\label{rvlabel02}\commRV{This paper extends the existing results on the endogenous functional linear model to a more general setup, allowing for weakly dependent errors and without assuming a specific type of endogeneity. Additionally, it suitably extends the asymptotic approach of \cite{Hall2007} to encompass cases with endogeneity, a common occurrence in practical applications. 
Consequently, this paper proposes two novel estimators and provides their detailed asymptotic properties under this broader setting. Notably, even in the case where there is no endogeneity and hence the FIVE reduces to the FLSE estimator of \cite{Park2012} (see Section \ref{sec:estimators1}), most of the asymptotic results that we obtain in the present paper have not been explored in the literature, to the best of the authors' knowledge.  
Given the potential prevalence of endogeneity in the functional linear model (see Section \ref{model1}), we believe that the theoretical results presented in this paper hold value beyond the existing findings.}
 \phantomsection\label{rvlabel04}\commRV{From a practical perspective, our methodology can be applied to study relationships between economic functional variables of which availability has been being increasing; potential examples include density-on-density regression model (see e.g., \citealp{Park2012}). }

\newpage

\makeatletter 
\renewcommand{\thepage}{S\arabic{page}} 
\renewcommand{\thesection}{S\arabic{section}}  
\renewcommand{\thetable}{S\arabic{table}}  
\renewcommand{\thefigure}{S\arabic{figure}} 
\renewcommand{\theassumption}{S\arabic{assumption}} 
\renewcommand{\thelemma}{S\arabic{lemma}} 
\renewcommand{\thermk}{S\arabic{rmk}} 
\renewcommand{\hat}{\widehat}
\setcounter{page}{1}
\setcounter{section}{0}
\setcounter{rmk}{0}
\renewcommand{\abar}{\upalpha}
\renewcommand{\alpha}{\mathbbmss{a}}
\makeatother

\vspace{3em}
\noindent \LARGE{\textbf{Supplementary Material}} \normalsize 
\vspace{1em}

	\section{Preliminaries} \label{appprelim}
		Let $(\mathbb S,\mathbb F, \mathbb P)$ denote the underlying probability space and let $\mathcal H$ be a separable Hilbert space equipped with the inner product $\langle \cdot,\cdot\rangle$ and the usual Borel $\sigma$-field.

	\subsection{Random elements of Hilbert spaces}	\label{appintro}
A $\mathcal H$-valued random variable $X$ is defined by a measurable map from  $\mathbb S$ to $\mathcal H$. We say that such a random variable $X$ is integrable (resp.\ square-integrable) if $\mathbb{E} [\|X\|] < \infty$ (resp.\  $\mathbb{E}[ \|X\|^2] < \infty$), where $\|\cdot\|$ is the norm induced by the inner product.   If $X$ is integrable, there exists a unique element $\mathbb{E}[X] \in \mathcal H$ satisfying $\mathbb{E}[\langle X,\zeta\rangle] = \langle \mathbb{E}[X],\zeta\rangle$ for every $\zeta \in \mathcal H$. The element $\mathbb{E}[X]$ is called the expectation of $X$. 
	
	Let $Y$ be another $\mathcal H$-valued random variable. We let $\otimes$ denote the tensor product defined as follows: for all  $\zeta_1, \zeta_2\in \mathcal H$,
	\begin{equation}\label{eqtensor}
		\zeta_1 \otimes \zeta_2 (\cdot) = \langle \zeta_1,\cdot \rangle \zeta_2.
	\end{equation}
	Note that  $\zeta_1 \otimes \zeta_2$  is a linear map from $\mathcal H$ to $\mathcal H$. If  $\mathbb{E}[\|X\|\|Y\|] < \infty$,  we may  well define a linear map $\mathcal C_{XY}$ from $\mathcal H$ to $\mathcal H$  as follows: $\mathcal C_{XY} =   \mathbb{E}[(X-\mathbb{E}[X]) \otimes (Y-\mathbb{E}[Y])]$. $\mathcal C_{XY}$ is called the cross-covariance operator of $X$ and $Y$. If $X=Y$ and $X$ is square-integrable, we then may define $\mathcal C_{XX}$ similarly, and this is called the covariance  operator of $X$. If the cross-covariance  operator of two random variables $X$ and $Y$ are a nonzero operator, $X$ is said to be correlated with $Y$.  
		\subsection{Bounded linear operators on $\mathcal H$}
	Let $\mathcal{L}_{\mathcal H}$ denote the space of bounded  linear operators acting on $\mathcal H$, equipped with the  operator norm $\|\mathcal T\|_{\op} = \sup_{ \|\zeta\|\leq 1} \|\mathcal T\zeta\|$. For any $\mathcal T \in \mathcal L_{\mathcal H}$,   the adjoint of $\mathcal T$ (denoted  $\mathcal T^\ast$) is the unique element of $\mathcal L_{\mathcal H}$ satisfying that $\langle \mathcal T \zeta_1 ,\zeta_2 \rangle=\langle \zeta_1 ,  \mathcal T^\ast \zeta_2 \rangle$ for all $\zeta_1,\zeta_2 \in \mathcal H$. {The range (denoted $\ran \mathcal T$) and kernel (denoted $\ker \mathcal T$) of $\mathcal T\in \mathcal{L}_{\mathcal H}$ are respectively given by $\ran \mathcal T = \{\mathcal T\zeta : \zeta \in \mathcal H\}$ and $\ker \mathcal T = \{\zeta \in \mathcal H : \mathcal T\zeta = 0\}$.}  $\mathcal T$ is  said to be nonnegative if  $\langle \mathcal T \zeta,\zeta \rangle \geq 0$ for all $\zeta\in \mathcal H$, and positive if the inequality is strict.  An element $\mathcal T \in \mathcal L_{\mathcal H}$ is called compact if there exist two orthonormal bases $\{\zeta_{1j}\}_{j \geq 1}$ and $\{\zeta_{2j}\}_{j \geq 1}$ of $\mathcal H$, and a sequence of real numbers $\{a_j\}_{j \geq 1}$ tending to zero, such that  $\mathcal T = \sum_{j=1}^\infty a_j\zeta_{1j} \otimes \zeta_{2j} $. In this expression, it may be assumed that  $\zeta_{1j} = \zeta_{2j}$ and  $a_1 \geq a_2 \geq \ldots \geq 0$ if $\mathcal T$ is self-adjoint (i.e., $\mathcal T = \mathcal T^\ast$) and nonnegative \citep[p.35]{Bosq2000}. In this case, $a_j$ becomes an eigenvalue of $\mathcal T$ and hence $\zeta_{1j}$ is the corresponding eigenfunction, and moreover, we may define $\mathcal T^{1/2}$ by replacing $a_j$ with $\sqrt{a_j}$.  It is well known that the covariance of a $\mathcal H$-valued random variable is self-adjoint, nonnegative and compact if exists.  A linear operator $\mathcal T$ is  called a Hilbert-Schmidt operator if its Hilbert-Schmidt norm $\|\mathcal T\|_{\HS} = (\sum_{j=1} ^\infty \Vert  \mathcal T \zeta_j \Vert ^2)^{1/2}$ is finite, where $\{  \zeta_j \}_{j\geq 1}$ is an arbitrary orthonormal basis of $\mathcal H$. It is well known that $\|\mathcal T\|_{\op}\leq \|\mathcal T\|_{\HS}$ holds and the collection of Hilbert-Schmidt operators consists of a strict subclass of $\mathcal L_{\mathcal H}$; see Section 1.5  of \cite{Bosq2000}.

	\section{Appendix to Section \ref{sec:estimators} on ``Functional IV estimator" \label{sec:pf}}
We will hereafter let \(\alpha = \abar^{-1}\) and use both interchangeably for convenience. 
	\subsection{Proofs of the results in Section \ref{sec:asym1}}
	\subsubsection*{Proof of Theorem \ref{thm2w1}}
	Note that
	\begin{equation}
		\widehat{\mathcal A}  = \widehat{\mathcal C}_{yz}^\ast \widehat{\mathcal C}_{xz}	(\widehat{\mathcal C}_{xz}^\ast  \widehat{\mathcal C}_{xz})_{\K}^{-1}	=   \mathcal A \widehat{\Pi}_{\K} +   \widehat{\mathcal C}_{uz}^\ast \widehat{\mathcal C}_{xz}(\widehat{\mathcal C}_{xz}^\ast  \widehat{\mathcal C}_{xz})_{\K}^{-1}, \label{eq:b0}
	\end{equation}
	where $\widehat \Pi_{\K}= \sum_{j=1} ^{\K} \widehat f_j \otimes \widehat f_j$. 
	Since $\|\widehat{\mathcal C}_{xz}(\widehat{\mathcal C}_{xz}^\ast  \widehat{\mathcal C}_{xz})_{\K}^{-1}\|_{\op} \leq \alpha^{-1/2}$ and $\|\widehat{\mathcal C}_{uz}\|_{\HS}=O_p(T^{-1/2})$, we find that $	\|\widehat{\mathcal A}- \mathcal A \widehat{\Pi}_{\K}\|_{\HS} \leq O_p(\alpha^{-1/2} T^{-1/2})$. Thus the proof becomes complete if $\|\mathcal A\widehat{\Pi}_{\K} - \mathcal A\|_{\HS}\pto 0$ is shown. From  
	nearly identical arguments used to derive (8.63) of \cite{Bosq2000}, it can be shown that 
	\begin{equation}\label{eqpf0006add2wadd}
		\|\mathcal A\widehat{\Pi}_{\K} - \mathcal A\|^2 _{\HS}\leq \sum_{j=\K+1}^\infty \|\mathcal A\widehat{f}_j\|^2 \leq 	\sum_{j=\K+1}^\infty \|\mathcal A{f}^s_j\|^2 + 2 \|\mathcal A\|_{\op}^2 \|\widehat{\mathcal C}_{xz}^\ast  \widehat{\mathcal C}_{xz}-\mathcal C_{xz}^{\ast}\mathcal C_{xz}\|_{\op} \sum_{j=1}^{\K} \tau_j,
	\end{equation}
	where ${f}^s_j = \sgn\{\langle \widehat{f}_j, f_j \rangle \} f_j$.
	Since $\mathcal A$ is Hilbert-Schmidt, $\sum_{j=\K+1}^\infty \|\mathcal A{f}^s_j\|^2$ converges in probability to zero as $T$ gets larger (note that $\K$ diverges almost surely  as $T \to \infty$). In addition, 
	\begin{equation}
		\|\widehat{\mathcal C}_{xz}^\ast  \widehat{\mathcal C}_{xz}-\mathcal C_{xz}^{\ast}\mathcal C_{xz}\|_{\op} \leq \|\widehat{\mathcal C}_{xz}^\ast\|_{\op}\|\widehat{\mathcal C}_{xz}-\mathcal C_{xz}\|_{\op} + \|\widehat{\mathcal C}_{xz}^\ast-\mathcal C_{xz}^\ast\|_{\op}\|\mathcal C_{xz}\|_{\op} = O_p(T^{-1/2}), \label{addeq1}
	\end{equation} which in turn implies that the second term of the right hand side of \eqref{eqpf0006add2wadd} is $o_p(1)$.   Combining all these results, we find that  $\|\mathcal A\widehat{\Pi}_{\K} - \mathcal A\|_{\HS}$ is $o_p(1)$, which implies the desired result.\qed
	
	\subsubsection*{Proof of Theorem \ref{thm2w2}}
	To show (i), we note that $\mathcal C_{xz}\mathcal C_{xz}^\ast = \sum_{j=1}^\infty {\lambda}_j^2 {\xi}_j \otimes {\xi}_j$ and there exists an orthonormal basis $\{\widehat{\xi}_j\}_{j\geq 1}$ such that 
	$\widehat{\mathcal C}_{xz}\widehat{\mathcal C}_{xz}^{\ast}= \sum_{j=1}^\infty \widehat{\lambda}_j^2 \widehat{\xi}_j \otimes \widehat{\xi}_j$. 
 Moreover, the following can be shown: $\widehat {\mathcal C}_{xz} ^\ast \widehat \xi_j = \widehat{\lambda}_j \widehat f_j$, 	$\widehat {\mathcal C}_{xz}  \widehat f_j = \widehat{\lambda}_j \widehat \xi_j$, $  {\mathcal C}_{xz} ^\ast   \xi_j =  {\lambda}_j   f_j$, and 	$  {\mathcal C}_{xz}    f_j =  {\lambda}_j   \xi_j$. We first note that 	${\mathcal C}_{xz} ({\mathcal C}_{xz}^\ast{\mathcal C}_{xz})_{\K}^{-1} = \sum_{j=1}^{\K} (\lambda_j^s)^{-1} {f}_j^s\otimes   {\xi}_j^s$, where   ${f}^s_j = \sgn\{\langle \widehat{f}_j, f_j \rangle \} f_j$, ${\xi}^s_j = \sgn\{\langle \widehat{\xi}_j, \xi_j \rangle \} \xi_j$ and $\lambda_j^s = \sgn\{\langle \widehat{f}_j, f_j \rangle \} \cdot\sgn\{\langle \widehat{\xi}_j, \xi_j \rangle \}\cdot \lambda_j$.   Observe that
	\begin{equation} 
		\|\widehat{ \mathcal C}_{xz}(\widehat{ \mathcal C}_{xz}^\ast\widehat{\mathcal C}_{xz})_{\K}^{-1} -  { \mathcal C}_{xz} ({\mathcal C}_{xz}^\ast{\mathcal C}_{xz})_{\K}^{-1} \|_{\op} \leq \|\sum_{j=1}^{\K} ({({\lambda}_j^s)}^{-1}-\widehat{\lambda}_j^{-1}) {f}_j^s\otimes   {\xi}_j^s \|_{\op} +   \|\sum_{j=1}^{\K} \widehat{\lambda}_j^{-1} (\widehat{f}_j\otimes \widehat{\xi}_j-{f}^s_j\otimes{\xi}^s_j )\|_{\op}.\label{eqpf0001a0}
	\end{equation}
	The first term of \eqref{eqpf0001a0} is bounded above by $ \sup_{1 \leq j \leq \K} | \widehat{\lambda}_j ^{-1} -  {({\lambda}_j^s)}^{-1}|$, where $|\widehat{\lambda}_j ^{-1} -  {({\lambda}_j^s)} ^{-1}| = | ({\lambda}_j^s - \widehat \lambda_j){\widehat {\lambda}_j^{-1} \lambda_j^{-1}} | = | ({\lambda}_j^2 - \widehat{\lambda}_j^2){\widehat{\lambda}_j^{-1}\lambda_j^{-1} ({\lambda}_j^s + \widehat{\lambda}_j)^{-1}}|$. Since $ \sup_{1 \leq j \leq\K} |\widehat{\lambda}_{j}^2-\lambda_j^2| \leq \|\widehat{\mathcal C}_{xz}\widehat{\mathcal C}_{xz}^\ast -  {\mathcal C}_{xz}{\mathcal C}_{xz}^\ast\|_{\op}$ \citep[Lemma 4.2]{Bosq2000}  and $|\lambda_j ^s|| \widehat{\lambda}_j {\lambda}_j^s| \leq |\lambda_j ^s|| \widehat{\lambda}_j {\lambda}_j^s + \widehat{\lambda}_j^2|$,  we find that  
	\begin{align}
		\|\sum_{j=1}^{\K} ({({\lambda}_j^s)}^{-1}-\widehat{\lambda}_j^{-1}) {f}_j^s\otimes   {\xi}_j^s \|_{\op}  \leq \sup_{1 \leq j \leq\K}  \frac{|\widehat{\lambda}_j^2 - \lambda_j^2|}{|\lambda_j ^s|| \hat{\lambda}_j {\lambda}_j^s + \widehat{\lambda}_j^2|}\leq \sup_{1 \leq j \leq\K}  \frac{|\hat{\lambda}_j^2 - \lambda_j^2|}{|\lambda_j ^s|| \hat{\lambda}_j {\lambda}_j^s|} 
		\leq  \frac{\|\widehat{\mathcal C}_{xz}\widehat{\mathcal C}_{xz}^\ast -  {\mathcal C}_{xz}{\mathcal C}_{xz}^\ast\|_{\op}}{\sqrt{\alpha}\lambda_{\K}^2}. \label{eqpf001a}
	\end{align}
	As $\|\hat{f}_j - f^s_j\| \leq \tau_j \|\widehat{\mathcal C}_{xz}^\ast\widehat{\mathcal C}_{xz} -  {\mathcal C}_{xz}^\ast{\mathcal C}_{xz}\|_{\op}$ and $\|\hat{\xi}_j - \xi^s_j\| \leq \tau_j \|\widehat{\mathcal C}_{xz}\widehat{\mathcal C}_{xz}^\ast -  {\mathcal C}_{xz}{\mathcal C}_{xz}^\ast\|_{\op}$ \citep[Lemma 4.3]{Bosq2000}, 
	\begin{align}
		\|\sum_{j=1}^{\K} \hat{\lambda}_j^{-1} (\hat{f}_j\otimes \hat{\xi}_j-{f}^s_j\otimes{\xi}^s_j ) \|_{\op} &\leq  \alpha^{-1/2}  \sum_{j=1}^{\K} 	(\|\hat{f}_j -  f^s_j\| + 	\| \hat{\xi}_j - \xi^s_j\|)  \notag\\ &\leq \alpha^{-1/2}  \sum_{j=1}^{\K} 	\tau_j (\|\widehat{\mathcal C}_{xz}^\ast\widehat{\mathcal C}_{xz} -  {\mathcal C}_{xz}^\ast{\mathcal C}_{xz}\|_{\op} + \|\widehat{\mathcal C}_{xz}\widehat{\mathcal C}_{xz}^\ast -  {\mathcal C}_{xz}{\mathcal C}_{xz}^\ast\|_{\op}). \label{eqpf001b}
	\end{align}
	From \eqref{addeq1}, we know that $\|\widehat{\mathcal C}_{xz}^\ast\widehat{\mathcal C}_{xz} -  {\mathcal C}_{xz}^\ast{\mathcal C}_{xz}\|_{\op} = O_p(T^{-1/2})$ and $\|\widehat{\mathcal C}_{xz}\widehat{\mathcal C}_{xz}^\ast -  {\mathcal C}_{xz}{\mathcal C}_{xz}^\ast\|_{\op} = O_p(T^{-1/2})$. Moreover, it can be shown that $\lambda_{\K}^{-2} \leq (\lambda_{\K}^2 - \lambda_{\K+1}^2)^{-1} \leq \tau_{\K} \leq \sum_{j=1} ^{\K}\tau_j$, so the terms given in the right hand sides of \eqref{eqpf001a} and \eqref{eqpf001b} are  $o_p(1)$ under our assumptions. We thus deduce from \eqref{eqpf0001a0} that \begin{equation*}
		\sqrt{\frac{T}{\theta_{\K}(\zeta)}}(\widehat{\mathcal A}- \mathcal A \widehat{\Pi}_{\K}) \zeta  = \left(\frac{1}{\sqrt{T\theta_{\K}(\zeta)}} \sum_{t=1}^Tz_{t} \otimes u_{t}\right) {\mathcal C}_{xz}({\mathcal C}_{xz}^\ast{\mathcal C}_{xz})_{\K}^{-1}\zeta+  o_p(1).
	\end{equation*}
	Let $\zeta_t = (\theta_{\K}(\zeta))^{-1/2} [z_{t} \otimes u_t]  {\mathcal C}_{xz}({\mathcal C}_{xz}^\ast{\mathcal C}_{xz})_{\K}^{-1}\zeta = (\theta_{\K}(\zeta))^{-1/2} \langle z_t , {\mathcal C}_{xz}({\mathcal C}_{xz}^\ast{\mathcal C}_{xz})_{\K}^{-1} \zeta\rangle u_t .$ Then, we have 
	\begin{align}
		\mathbb{E}[\zeta_t\otimes \zeta_t] 
		&= (\theta_{\K} (\zeta))^{-1}  \langle {\mathcal C}_{xz}({\mathcal C}_{xz}^\ast{\mathcal C}_{xz})_{\K}^{-1}\zeta, \mathcal C_{zz}{\mathcal C}_{xz}({\mathcal C}_{xz}^\ast{\mathcal C}_{xz})_{\K}^{-1}\zeta \rangle \mathcal C_{uu}  =\mathcal C_{uu}, \label{eqlater1}
	\end{align}
	by Assumption \ref{assum1}.\ref{assum1.3}. Thus, under Assumption \ref{assum1}, $\{\langle \zeta_t, \psi\rangle\}_{t\geq 1}$ is a real-valued martingale difference sequence  for any $\psi \in \mathcal H$. 
	By the standard central limit theorem for such a sequence, we have 
	\begin{equation}
		\frac{1}{\sqrt{{T} }}\sum_{t=1}^T \langle \zeta_t,\psi\rangle  \dto N(0, \langle \mathcal C_{uu}\psi,\psi\rangle). \label{eqpf0008add}
	\end{equation}
	Let $	\ddot\zeta_T = {T}^{-1/2}\sum_{t=1}^T  \zeta_t$. 
	If we show that there exists an orthonormal basis $\{\ell_j\}_{j\geq 1}$ satisfying 
	\begin{equation}\label{eqpf0009add}
		\limsup_{n\to \infty}	\limsup_{T} \mathbb{P}\left( \sum_{j=n+1}^\infty  \langle \ddot\zeta_T, \ell_j  \rangle^2  > m \right) = 0  
	\end{equation} 
	for every $m >0$, then \eqref{eqpf0008add} implies that $\ddot\zeta_T\dto N(0,\mathcal C_{uu})$ \citep[Theorem 1.8.4]{van1996weak}. To show this, let  $\{\ell_j\}_{j\geq 1}$ be the eigenfunctions of $\mathcal C_{uu}$ and define $\mathcal L_n = \sum_{j=n+1}^{\infty} \ell_j\otimes \ell_j$. Then,
	\begin{equation}
		\mathbb{E}\left[\sum_{j=n+1}^\infty  \langle \ddot \zeta_T, \ell_j \rangle^2\right]  
		\leq  \frac{1}{T\theta_{\K}(\zeta)} \sum_{t=1}^T \mathbb{E}  \left[\langle z_t, \mathcal C_{xz} (\mathcal C_{xz} ^\ast \mathcal C_{xz})_{\K} ^{-1}\zeta \rangle^2 \|\mathcal L_nu_t\|^2 \right]
		=  \sum_{j=n+1}^\infty \langle \mathcal C_{uu} \ell_j,\ell_j\rangle^2, \label{eq0010add}
	\end{equation}
	where the equality follows from that $\{u_t\}_{t\geq 1}$ is a martingale difference sequence (with respect to $\mathfrak F_{t-1}$). Since $\mathcal C_{uu}$ is Hilbert-Schmidt, the right hand side of \eqref{eq0010add} converges to zero as $n \to \infty$. Combining this result with Markov's inequality, we find that for any $m>0$, $\mathbb{P} ( \sum_{j=n+1}^\infty  \langle \ddot\zeta_T, \ell_j  \rangle^2  > m ) \leq  m^{-1}\sum_{j=n+1}^\infty \langle \mathcal C_{uu} \ell_j,\ell_j\rangle^2$,
	from which \eqref{eqpf0009add} immediately follows. Thus, the desired result is obtained. 
	
	(ii)  follows from that  $\|\widehat{\mathcal C}_{zz}-\mathcal C_{zz}\|_{\op}$ and $	\| \widehat{ \mathcal C}_{xz}(\widehat{ \mathcal C}_{xz}^\ast\widehat{\mathcal C}_{xz})_{\K}^{-1} -  {\mathcal C}_{xz}({\mathcal C}_{xz}^\ast{\mathcal C}_{xz})_{\K}^{-1} \|_{\op}$ are all $o_p(1)$.\qed
	
	\subsection{Proofs of the results in Section \ref{sec:asym2}}	\label{sec:asym2proof}
	We first state a useful lemma and then provide our proofs of the main results given in Section \ref{sec:asym2}.	
	\begin{lemma}\label{lem:convrate}
		Suppose that Assumptions \ref{assumconvrate}.\ref{assumconvrate.1} and \ref{assumconvrate}.\ref{assumconvrate.2} hold and $\delta>1$. Then the following hold: \begin{enumerate*}[(i)]
			\item\label{lemconvrate8} 	$\sum_{\ell \neq j} \frac{ \lambda_\ell ^2 \ell ^{-\delta }}{(\lambda_j ^2  -\lambda _\ell ^2 )^{2} } \leq O(j^{\rho+2-\delta}) $ and
			\item\label{lemconvrate9} 	$\sum_{\ell \neq j} \frac{ \lambda_j ^2 \ell ^{-\delta }}{(\lambda_j ^2  -\lambda _\ell ^2 )^{2} } \leq O(j^{\rho+2-\delta})$.
		\end{enumerate*}  
	\end{lemma}
	\begin{proof}[\normalfont\textbf{Proof of  Lemma \ref{lem:convrate}}] We only show \ref{lemconvrate8}, because the remaining result can be obtained in a similar manner. As in \cite{imaizumi2018}, we can choose $j_0 \geq 1$ and $C>1 $ large enough so that $\lambda_j ^2 / \lambda_{\lfloor j/C\rfloor} ^2  \leq 1/2$ and $\lambda_{\lfloor jC\rfloor+1}  ^2 /\lambda_j^2 \leq 1/2$ for all $j \geq j_0$, where $\lfloor\cdot\rfloor$ denotes the floor function. In addition, because of Assumption \ref{assumconvrate}.\ref{assumconvrate.2} we have $(\lambda_j^2 - \lambda_\ell^2)^{2} \geq O(1) j^{-2\rho-2} (j-\ell) ^2 $ for $\ell \neq j$ and $\lfloor j/C\rfloor < \ell < \lfloor Cj \rfloor$ (\citealt[p.\ 29]{imaizumi2018}). Using these, for $\delta >1$, we find that
		\begin{equation*}
			\sum_{\ell=1} ^{\lfloor j/C \rfloor} \frac{\lambda_\ell ^2 \ell ^{-\delta}}{(\lambda_j ^2 -\lambda_\ell ^2 )^2} \leq 	\sum_{\ell=1} ^{\lfloor j/C \rfloor}  \frac{\lambda_\ell ^2 \ell^{-\delta}}{\lambda_\ell ^4 (1-\lambda_j ^2 /\lambda_{\lfloor j/C \rfloor} ^2)} \leq 4 \sum_{\ell=1} ^{\lfloor j/C \rfloor} \frac{\ell^{-\delta}}{\lambda_j ^2} \leq O(j^{\rho}),
		\end{equation*}and
		\begin{equation*}
			\sum_{\ell=\lfloor jC \rfloor+1 } ^{\infty} \frac{\lambda_\ell ^2 \ell ^{-\delta}}{(\lambda_j ^2 -\lambda_\ell ^2 )^2} \leq 4 \lambda_j ^{-2} \sum_{\ell=\lfloor jC\rfloor +1 } ^{\infty} \ell^{-\delta} \leq O(j^{\rho}). 
		\end{equation*}  Moreover, by using the inequality $\lambda_\ell ^2 \leq |\lambda_\ell^2 - \lambda_j ^2 | +\lambda_j^2$ and the property stated at the beginning of the proof, we have
		\begin{align*}
			\sum_{\ell=\lfloor j/C\rfloor+1 , \neq j} ^{\lfloor jC  \rfloor} \frac{\lambda_\ell ^2 \ell ^{-\delta}}{(\lambda_j ^2 -\lambda_\ell ^2 )^2}&\leq  \sum_{\ell=\lfloor j/C\rfloor+1 , \neq j}  ^{\lfloor jC  \rfloor}  \frac{\ell^{-\delta}}{|\lambda_\ell ^2 -\lambda_j^2|} +   \sum_{\ell=\lfloor j/C\rfloor+1 , \neq j}  ^{\lfloor jC  \rfloor} \frac{\lambda_j^2 \ell^{-\delta}}{(\lambda_j^2 -\lambda_\ell ^2)^2}\nonumber\\
			& \leq  j^{1+\rho} \sum_{\ell=\lfloor j/C\rfloor+1 , \neq j}  ^{\lfloor jC  \rfloor}  \frac{\ell ^{-\delta}}{|\ell-j|} + j^{2+\rho}\sum_{\ell=\lfloor j/C\rfloor+1 , \neq j}  ^{\lfloor jC  \rfloor}  \frac{\ell ^{-\delta}}{|\ell-j|^2}\nonumber\\
			& \leq  j^{1+\rho} \sum_{\ell=\lfloor j/C\rfloor+1 , \neq j}  ^{\lfloor jC  \rfloor}  \frac{\ell ^{-\delta+1}}{|\ell-j|^2} + j^{2+\rho}\sum_{\ell=\lfloor j/C\rfloor+1 , \neq j}  ^{\lfloor jC  \rfloor}  \frac{\ell ^{-\delta}}{|\ell-j|^2} \leq    O(j^{2+\rho-\delta}) ,
		\end{align*} 
		where the last two inequalities are obtained by using the fact that $  |\ell-j|<\ell$ and $\ell^{-\delta+1} \leq ( j/C)^{-\delta+1} \leq C^{\delta-1} j^{-\delta+1}$ for all $\ell > \lfloor j/C\rfloor+1 $. 
	\end{proof} 

	\subsubsection*{Proof of Theorem \ref{thm:convrate}}
	We first note that $\alpha \leq \widehat{\lambda}_{\K} ^2  = \widehat{\lambda}_{\K}^2  - \lambda_{\K}^2 +\lambda_{\K} ^2 \leq \Vert \widehat{\mathcal C}_{xz} ^\ast \widehat{\mathcal C}_{xz} - \mathcal C_{xz} ^\ast \mathcal C_{xz} \Vert_{\op} + c_\circ\K^{-\rho} \leq O_p (T^{-1/2}) + c_{\circ} \K^{-\rho}$ and  $\alpha^{-1} T^{-1/2}   = o(1)$. These imply that  
	\begin{equation}
		\K \leq (1+o_p(1))\alpha^{-1/\rho}.
		\label{eq:kalpha}
	\end{equation} 
	Using the fact that $\lambda_j ^2 \geq \sum_{\ell=j} ^{\infty} (\lambda_{\ell} ^2 -\lambda_{\ell+1}^2 )\geq \rho^{-1}c_\circ^{-1} j^{-\rho}$ and  $\alpha^{-1} T^{-1/2} = o(1)$, we also find that \begin{equation}
		(c_\circ \rho)^{-1}{(\K+1)^{-\rho }} \leq {\lambda}_{\K+1} ^2  = \lambda_{\K+1} ^2 - \widehat{\lambda}_{\K+1} ^2  +\widehat \lambda_{\K+1} ^2 \leq   O_p(T^{-1/2})  +\alpha 
		\leq (o_p(1) + 1) \alpha  \label{eq:kalpha2} . 
	\end{equation} 
	
	We will now obtain stochastic orders of $\Vert \widehat{\mathcal A} - \mathcal A \widehat{\Pi}_{\K} \Vert_{\HS}$,  $\Vert \mathcal A \widehat{\Pi}_{\K} - \mathcal A \Pi_{\K}\Vert_{\HS}$, and  $\Vert \mathcal A (\mathcal I - \Pi_{\K})\Vert_{\HS}$.
	From \eqref{eq:b0}, we know that $\Vert \widehat{\mathcal A} - \mathcal A \widehat{\Pi}_{\K} \Vert _{\HS} \leq \Vert \widehat{\mathcal C}_{uz} \Vert _{\HS} \Vert \widehat{\mathcal C}_{xz} (\widehat{\mathcal C}_{xz} ^{\ast} \widehat{\mathcal C}_{xz})_{\K} ^{-1} \Vert _{\op}\leq O_p (\alpha ^{-1/2} T^{-1/2})$. Moreover,  we have  
	\begin{equation}
		\Vert \mathcal A (\mathcal I  - \Pi_{\K}) \Vert_{\HS} ^2 = \sum_{\ell = \K+1} ^{\infty } \sum_{j=1} ^\infty \langle \mathcal A f_\ell,  {\xi_j} \rangle ^2 \leq c_\circ\sum_{\ell = \K+1} ^{\infty } \sum_{j=1} ^\infty  \ell ^{-2\varsigma} j ^{-2\gamma} \leq  O((\K+1)^{-2\varsigma+1 }) \leq O_p(\alpha^{(2\varsigma -1)/\rho} ),\label{eq:b07}
	\end{equation}
	where the first  inequality immediately follows from Assumption \ref{assumconvrate}.\ref{assum1.3}, and the second and third inequalities follow from \eqref{eq:kalpha2}, Assumption~\ref{assumconvrate}, and the Euler-Maclaurin summation formula for the Riemann zeta-function (see e.g., (5.6) of \citealp{ibukiyama2014euler}).  We then focus on the remaining term $\Vert \mathcal A \widehat{\Pi}_{\K} - \mathcal A\Pi_{\K}\Vert_{\HS} $ and find that 
	\begin{equation}
		\Vert \mathcal A \widehat{\Pi}_{\K} - \mathcal A\Pi_{\K}\Vert_{\HS} ^2 \leq 2\Vert  \sum_{j=1} ^{\K} \widehat{f}_j  \otimes \mathcal A (\widehat{f}_j - f_j ^s )\Vert_{\HS} ^2 + 2 \Vert \sum_{j=1} ^{\K} (\widehat{f}_j - f_j ^s) \otimes \mathcal A f_j ^s \Vert_{\HS} ^2 . \label{eq:b08} 
	\end{equation}	
We	observe that  \vspace{-0.75em}
	\begin{align}
		\Vert \widehat{f}_j - f_j^s \Vert^2 &= O_p(j^{2} T^{-1}), \label{eq000} \\
		\|\mathcal A (\widehat{f}_j - f_j ^s)\|^2 &= O_p(T^{-1}) (j^{2-2\varsigma} + j^{\rho+2-2\varsigma}).\label{eq001}
	\end{align}
	These will be proved below after discussing the main result of interest. Specifically, from \eqref{eq001}, we can show the second term in \eqref{eq:b08} satisfies that 
	\begin{align}
		\Vert \sum_{j=1} ^{\K} (\widehat{f}_j - f_j ^s) \otimes \mathcal A f_j ^s \Vert_{\HS} ^2 &= 
		\sum_{\ell=1}^{\infty}\Vert \sum_{j=1} ^{\K} \langle \mathcal A f_j,f_{\ell} \rangle (\widehat{f}_j - f_j ^s) \Vert ^2   \leq 	\sum_{\ell=1}^{\infty} \left(\sum_{j=1}^{\K} |\langle \mathcal A f_j,f_{\ell} \rangle| \|\widehat{f}_j - f_j ^s\| \right)^2 \notag \\
		&\leq \sum_{\ell=1}^{\infty} \ell^{-2\gamma} \left(\sum_{j=1}^{\K}  j^{-\varsigma}  \|\widehat{f}_j - f_j ^s\| \right)^2  =  O_p(T^{-1})  \left(\sum_{j=1}^{\K}  j^{1-\varsigma} \right)^2 \notag \\
		&= \begin{cases}O_p(T^{-1}) & \text{if }\varsigma>2,
			\\
			O_p(T^{-1}\max\{\log^2 \alpha^{-1}, \alpha^{(2\varsigma-4)/\rho} \}) &\text{if }\varsigma \leq 2, \end{cases} 
		\label{eq:b09}		\end{align}
	where the first equality follows from the properties of the Hilbert-Schmidt norm and the remaining relationships follow from \eqref{eq000}, Assumption \ref{assumconvrate}, and the fact that $\sum_{j=1}^{\K}  j^{1-\varsigma}=O_p(1)$ if $\varsigma>2$ and  $\sum_{j=1}^{\K}  j^{1-\varsigma}=O_p(\max\{\log \alpha^{-1}, \alpha^{(\varsigma-2)/\rho} \})$ otherwise.  Similarly, the first term in \eqref{eq:b08} satisfies that
	\begin{align}
		\Vert \sum_{j=1} ^{\K}   \widehat{f}_j   \otimes \mathcal A(\widehat{f}_j - f_j)\Vert_{\HS} ^2 &= 		\sum_{j=1} ^{\K}  \Vert \mathcal A(\widehat{f}_j - f_j)\Vert ^2 = O_p(T^{-1})  \sum_{j=1} ^{\K} (j^{2-2\varsigma} + j^{\rho+2-2\varsigma}) \notag \\
		&= \begin{cases}O_p(T^{-1}) & \text{if } \varsigma> \rho/2+3/2,
			\\
			O_p(T^{-1}\max\{\log \alpha^{-1}, \alpha^{(2\varsigma-\rho-3)/\rho} \}) &\text{if } \varsigma \leq \rho/2+3/2,\end{cases}   \label{eq:b11}
	\end{align}
	where \eqref{eq001} is used to establish the second equality.
	Since $2\varsigma-\rho-3 < 2\varsigma-4$, $\alpha \log \alpha^{-1} = o(1)$, and $\alpha \log^2 \alpha^{-1} = o(1)$, 
	\eqref{eqthmconvrate} may be deduced  from \eqref{eq:b07}, \eqref{eq:b09} and \eqref{eq:b11}. \\[9pt]
	\noindent \textbf{Proofs of \eqref{eq000} and \eqref{eq001}}:
	We first show \eqref{eq000}. Note that for each $j$, 
	\begin{equation}
		\widehat{f}_j - f_j ^s = \sum_{\ell \neq j }(\widehat{\lambda}_j ^2 -\lambda_\ell ^2 )^{-1} \langle (\widehat{\mathcal C}_{xz} ^\ast \widehat{\mathcal C}_{xz}  - \mathcal C_{xz} ^\ast \mathcal C_{xz})  \widehat{f}_j  , f_\ell^s \rangle f_\ell^s + \langle \widehat{f}_j - f_j ^s , f_j^s \rangle f_j ^s. \label{eq:b1}
	\end{equation}
	Then, using the arguments used to derive (4.48) of \cite{Bosq2000} and the expansion of $\langle \widehat{f}_j  - f_j ^s ,f_\ell^s \rangle $ that was used to derive \eqref{eq:b1}, it can be shown that \begin{equation}
		\Vert \widehat{f}_j - f_j^s \Vert^2 \leq 4\sum_{\ell \neq j}  (\widehat{\lambda}_j ^{2} - \lambda_\ell ^2 )^{-2} \langle  (\widehat{\mathcal C}_{xz}^\ast \widehat{\mathcal C}_{xz} - \mathcal C_{xz} ^\ast \mathcal C_{xz}   ) \widehat{f}_j , f_\ell\rangle ^2  .  \label{eq:b2}
	\end{equation}  
	Since $\widehat{\mathcal C}_{xz}^\ast \widehat{\mathcal C}_{xz} - \mathcal C_{xz} ^\ast \mathcal C_{xz}   = (\widehat{\mathcal C}_{xz} ^\ast - \mathcal C_{xz} ^\ast )\widehat{\mathcal C}_{xz} + \mathcal C_{xz}^\ast (\widehat{\mathcal C}_{xz}  - \mathcal C_{xz}  )$, the sum given in \eqref{eq:b2} satisfies that
	\begin{align}
		&\sum_{\ell \neq j}  (\widehat{\lambda}_j ^{2} - \lambda_\ell ^2 )^{-2} \langle  (\widehat{\mathcal C}_{xz}^\ast \widehat{\mathcal C}_{xz} - \mathcal C_{xz} ^\ast \mathcal C_{xz}   ) \widehat{f}_j , f_\ell\rangle ^2  \nonumber\\
		&\leq 2 \sum_{\ell \neq j}  (\widehat{\lambda}_j ^{2} - \lambda_\ell ^2 )^{-2} \langle (\widehat{\mathcal C}_{xz} ^\ast - \mathcal C_{xz} ^\ast ) \widehat{\lambda}_j \widehat{\xi}_j,f_\ell\rangle ^2 + 2 \sum_{\ell \neq j}  (\widehat{\lambda}_j ^{2} - \lambda_\ell ^2 )^{-2}  \langle(\widehat{\mathcal C}_{xz}  - \mathcal C_{xz}  ) \widehat{f}_j , \lambda_\ell \xi_\ell \rangle ^2.  \label{eqpf01}
	\end{align}
	The second term of \eqref{eqpf01} satisfies that  
	\begin{align}
		&\sum_{\ell \neq j}  (\widehat{\lambda}_j ^{2} - \lambda_\ell ^2 )^{-2}  \langle(\widehat{\mathcal C}_{xz}  - \mathcal C_{xz}  ) \widehat{f}_j , \lambda_\ell \xi_\ell \rangle ^2 \nonumber\\
		&\leq 2 \sum_{\ell \neq j}  (\widehat{\lambda}_j ^{2} - \lambda_\ell ^2 )^{-2} \lambda_\ell ^2 \langle (\widehat{\mathcal C}_{xz}  - \mathcal C_{xz}  ) (\widehat{f}_j - f_j ^s) ,\xi_\ell \rangle ^2 + 2 \sum_{\ell \neq j}  (\widehat{\lambda}_j ^{2} - \lambda_\ell ^2 )^{-2} \lambda_\ell ^2 \langle (\widehat{\mathcal C}_{xz}  - \mathcal C_{xz}  )   f_j ^s ,\xi_\ell \rangle ^2 \nonumber\\
		&\leq 2 \Delta_{1j} \Vert \widehat{f}_j - f_j ^s \Vert ^2 + 2 \sum_{\ell \neq j}  (\widehat{\lambda}_j ^{2} - \lambda_\ell ^2 )^{-2} \lambda_\ell ^2 \langle (\widehat{\mathcal C}_{xz}  - \mathcal C_{xz}  )   f_j ^s ,\xi_\ell \rangle ^2,  \label{eqpf03}
	\end{align}
	where $\Delta_{1j} = \sum_{\ell \neq j}  (\widehat{\lambda}_j ^2 -\lambda_\ell ^2 ) ^{-2} \lambda_\ell ^2 \Vert \widehat{\mathcal C}_{xz} - \mathcal C_{xz} \Vert_{\op} ^2 $.  Similarly, for the first term of \eqref{eqpf01}, we have 
	\begin{align}
		&\sum_{\ell \neq j}  (\widehat{\lambda}_j ^{2} - \lambda_\ell ^2 )^{-2} \widehat{\lambda}_j ^2\langle (\widehat{\mathcal C}_{xz} ^\ast - \mathcal C_{xz} ^\ast )  \widehat{\xi}_j,f_\ell\rangle ^2  \nonumber\\
		&\leq  \sum_{\ell \neq j} \left( (\widehat{\lambda}_j ^{2} - \lambda_\ell ^2 )^{-2}\lambda_\ell^2\langle (\widehat{\mathcal C}_{xz} ^\ast - \mathcal C_{xz} ^\ast )  \widehat{\xi}_j,f_\ell\rangle ^2+(\widehat{\lambda}_j ^{2} - \lambda_\ell ^2 )^{-1}\langle (\widehat{\mathcal C}_{xz} ^\ast - \mathcal C_{xz} ^\ast )  \widehat{\xi}_j,f_\ell\rangle ^2 \right)   \nonumber\\
		&\leq   2\sum_{\ell \neq j}  (\widehat{\lambda}_j ^{2} - \lambda_\ell ^2 )^{-2} \lambda_\ell ^2\langle  (\widehat{\mathcal C}_{xz} ^\ast - \mathcal C_{xz} ^\ast) (\widehat{\xi}_j - {\xi_j^s} ), f_\ell \rangle ^2 + 2 \sum_{\ell \neq j}  (\widehat{\lambda}_j ^{2} - \lambda_\ell ^2 )^{-2} \lambda_\ell ^2\langle  (\widehat{\mathcal C}_{xz} ^\ast - \mathcal C_{xz} ^\ast)  {\xi_j^s} , f_\ell \rangle ^2 \nonumber\\
		&\quad + 2 \sum_{\ell \neq j } (\widehat{\lambda}_j ^2 - \lambda_\ell ^2 )^{-1} \langle (\widehat{\mathcal C}_{xz} ^\ast - \mathcal C_{xz} ^\ast ) (\widehat{\xi}_j - \xi_j ^s) ,f_\ell \rangle ^2 + 2\sum_{\ell \neq j } (\widehat{\lambda}_j ^2 - \lambda_\ell ^2 )^{-1} \langle (\widehat{\mathcal C}_{xz} ^\ast - \mathcal C_{xz} ^\ast ) \xi_j ^s ,f_\ell \rangle ^2\nonumber\\
		&\leq 2(\Delta_{1j} + \Delta_{2j}) \Vert \widehat{\xi}_j - \xi_j^s \Vert ^2 +  2 \sum_{\ell \neq j}  (\widehat{\lambda}_j ^{2} - \lambda_\ell ^2 )^{-2}  {\widehat{\lambda}_j ^2}\langle  (\widehat{\mathcal C}_{xz} ^\ast - \mathcal C_{xz} ^\ast)  {\xi}_j^s  , f_\ell \rangle ^2, \label{eqpf01a}
	\end{align}
	where  $\Delta_{2j} = \max_{\ell \neq j, 1\leq j \leq \K} (\widehat{\lambda}_j ^2 -\lambda_\ell ^2)^{-1} \Vert \widehat{\mathcal C}_{xz} - \mathcal C_{xz}\Vert_{\op}^2$  and the second inequality simply follows from  the decomposition $\widehat{\xi}_j = (\widehat{\xi}_j - \xi_j^s) + \xi_{j}^s$ and the fact that $(a+b)^2 \leq  2a^2+ 2b^2$ for $a,b \in \mathbb{R}$. Let  \begin{equation}
		\Delta_{3j} =  \Delta_{3j,1} + \Delta_{3j,2}, \label{delta3}
	\end{equation}
	where \begin{equation}
		\Delta_{3j,1} = 	\sum_{\ell \neq j} (\widehat{\lambda}_j ^2 -\lambda_\ell ^2 ) ^{-2} \lambda_\ell ^2   \langle (\widehat{\mathcal C}_{xz}  -\mathcal C_{xz}  )  {{f}_j^s}, \xi_\ell \rangle ^2,\quad\Delta_{3j,2}  = \sum_{\ell \neq j} (\widehat{\lambda}_j ^2 -\lambda_\ell ^2 ) ^{-2}  {\widehat{\lambda}_j ^2}\langle (\widehat{\mathcal C}_{xz}  ^\ast- \mathcal C_{xz} ^\ast )  {{\xi}_j^s}, f_\ell \rangle^2 .  \label{delta32}
	\end{equation} 
	We then deduce from \eqref{eq:b2}-\eqref{delta32} that  
	\begin{equation} \label{eqpf01aa}
		\Vert \widehat{f}_j  - f_j ^s \Vert ^2 \leq 16 \Delta_{1j} \Vert \widehat{f}_j - f_j ^s \Vert ^2 +16 (\Delta_{1j} +\Delta_{2j}) \Vert \widehat{\xi}_j - \xi_j^s \Vert^2 +16 \Delta_{3j}.
	\end{equation}
	A similar bound of $\Vert \widehat{\xi}_j - \xi_j ^s \Vert ^2 $ can be obtained from nearly identical arguments to derive \eqref{eqpf01aa}, from which the following can be deduced with a little algebra: 
	\begin{equation*}
		\Vert \widehat{f}_j - f_j^s \Vert ^2 \leq \frac{16(1+16\Delta_{2j})}{(1-16\Delta_{1j})^2 - 16^2(\Delta_{1j}+\Delta_{2j})^2} \Delta_{3j} . 
	\end{equation*} 	
	From \eqref{eq:kalpha2}, the condition  $\abar = \alpha^{-1} = o(T^{\rho/(2\rho+2)})$, and  similar arguments used in the proof of Theorem 1 of \cite{imaizumi2018}, we find that  
	\begin{equation}\label{eqpf02}
		\mathbb P \{  | \widehat{\lambda}_j ^{2} - \lambda_\ell ^2   | \geq |\lambda_j ^2- \lambda_{\ell }^2|/\sqrt{2},\text{for }j=1,\ldots, \K \text{ and }\ell \neq j\} \to 1.
	\end{equation} 
	Because of \eqref{eq:kalpha}, \eqref{eqpf02}, and the condition $\abar = \alpha^{-1} = o(T^{\rho/(2\rho+2)})$, we first find that $\Delta_{1j} \leq O_p(T^{-1}) \max_{\ell \neq j, 1\leq j \leq \K}  |\widehat{\lambda}_j ^2 -\lambda_\ell ^2 |^{-2} \leq O_p(T^{-1}\K^{2\rho+2}) \leq O_p(T^{-1}\alpha^{-(2\rho+2)/\rho})= o_p(1)$.  Similarly, from \eqref{eq:kalpha}, \eqref{eqpf02} and Assumption \ref{assumconvrate}.\ref{assumconvrate.2}, we also deduce that $\Delta_{2j} \leq   O_p(  \K^{ \rho+1 } \Vert \widehat{\mathcal C}_{xz} ^\ast - \mathcal C_{xz } \Vert_{\op} ^2)  \leq  O_p(\alpha ^{-(\rho+1)/\rho} T^{-1}) =o_p(1)$.  We thus find that 
	\begin{equation}\label{eqfhatbound}
		\Vert \widehat{f}_j - f_j ^s\Vert ^2 \leq 16(1+o_p(1))\Delta_{3j}.
	\end{equation}
	We now focus on $\Delta_{3j}$. \bt{First note that} 
	\begin{equation*} 
		|\Delta_{3j,1}| \leq O_p(1) \sum_{\ell \neq j}  (\lambda_j ^2 -\lambda_\ell ^2 ) ^{-2} \lambda_\ell ^2    \langle (\widehat{\mathcal C}_{xz}  -\mathcal C_{xz}  ) {{f}_j^s}, \xi_\ell \rangle ^2\leq  O_p(j^{2 }T^{-1});
	\end{equation*}
	this may be deduced from \eqref{eqpf02}, Assumption~\ref{assumconvrate}, Lemma \ref{lem:convrate}\ref{lemconvrate8} and the fact that
	\begin{align*} 
		&	T\mathbb{E}[\langle (\widehat{\mathcal C}_{xz}  -\mathcal C_{xz}  ) {{f}_j^s}, \xi_\ell \rangle ^2]  \leq \sum_{s=0}^T  		\mathbb{E}[\upsilon_{t}(j,\ell)\upsilon_{t-s}(j,\ell)]  \leq O(1) \mathbb E[ \langle x_t, f_j\rangle ^2 \langle z_t, \xi_\ell \rangle ^2 ]\nonumber\\
		&\leq O(1) \mathbb E [\Vert \langle x_t, f_j \rangle z_t \Vert ^2 ]^{1/2} \mathbb E [\Vert \langle z_t, \xi_\ell \rangle x_t \Vert ^2 ]^{1/2} \leq O(1) |\lambda_j \lambda_\ell |,
	\end{align*}
	where the second inequality follows from 
	Assumption~\ref{assumconvrate}.\ref{assumconvrate.4}. To obtain the third inequality, note that  $\mathbb E[\langle x_t, f_j \rangle^2 \langle z_t, \xi_\ell \rangle^2 ] = \mathbb E [\langle \langle x_t, f_j \rangle z_t, \xi_\ell \rangle \langle \langle z_t, \xi_\ell \rangle x_t, f_j \rangle  ] \leq \mathbb E [\Vert \langle x_t, f_j \rangle z_t \Vert \Vert \langle z_t, \xi_\ell \rangle x_t \Vert  ]$ and then apply the Cauchy-Schwarz inequality and Assumption~\ref{assumconvrate}.\ref{assumconvrate.4}.
	In a similar manner, we also find that $|\Delta_{3j,2}| \leq O_p(j^{2}T^{-1})$, and thus $|\Delta_{3j}| \leq O_p(j^{2}T^{-1})$. Combining this result with \eqref{eqfhatbound}, we find that the desired result \eqref{eq000} holds. 
	
	We next show \eqref{eq001}. Note that
	\begin{equation}
		\mathcal A (\widehat{f}_j - f_j ^s) = \sum_{\ell \neq j }(\widehat{\lambda}_j ^2 -\lambda_\ell ^2 )^{-1} \langle (\widehat{\mathcal C}_{xz} ^\ast \widehat{\mathcal C}_{xz}  - \mathcal C_{xz} ^\ast \mathcal C_{xz})  \widehat{f}_j  , f_\ell^s \rangle \mathcal A f_\ell^s + \langle \widehat{f}_j - f_j ^s , f_j ^s\rangle \mathcal A f_j^s, \label{eqafhat0}
	\end{equation}
	where $\Vert \langle \widehat{f}_j - f_j ^s , f_j ^s\rangle \mathcal A f_j ^s \Vert ^2 \leq O_p(T^{-1})j^{2-2\varsigma} $. For each $j=1,\ldots, \K$, the first term in \eqref{eqafhat0} is bounded above as follows:
	\begin{align}
		&( \sum_{\ell \neq j} (\widehat{\lambda}_j ^2 - \lambda_\ell^2) ^{-1} \langle  (\widehat{\mathcal C}_{xz} ^\ast \widehat{ \mathcal C}_{xz}- \mathcal C_{xz} ^\ast \mathcal C_{xz}) \widehat{f}_j ,  f_\ell^s \rangle \mathcal A f_\ell^s   )^2
		\leq O(1)(\sum_{\ell \neq j}|\widehat{\lambda}_j ^2 - \lambda_\ell^2| ^{-1} \ell^{-\varsigma} | \langle  (\widehat{\mathcal C}_{xz} ^\ast \widehat{ \mathcal C}_{xz}- \mathcal C_{xz} ^\ast \mathcal C_{xz}) \widehat{f}_j ,  f_\ell \rangle| )^2\nonumber\\
		&\leq O(1) ( \sum_{\ell \neq j} |\widehat{\lambda}_j ^2 - \lambda_\ell^2| ^{-1} \ell^{-\varsigma} \lambda_\ell | \langle ( \widehat{\mathcal C}_{xz} - \mathcal C_{xz})\widehat{f}_j,\xi_\ell \rangle|     )^2+ O(1)  ( \sum_{\ell \neq j}  |\widehat{\lambda}_j ^2 - \lambda_\ell^2|^{-1} \ell^{-\varsigma}\widehat{\lambda}_j  |\langle ( \widehat{\mathcal C}_{xz} ^\ast  - \mathcal C_{xz} ^\ast )\widehat{\xi}_j,f_\ell  \rangle |    )^2 \nonumber\\
		&\leq O(1)\Vert \widehat{\mathcal C}_{xz} - \mathcal C_{xz}\Vert_{\op} ^2 \left( \sum_{\ell \neq j} ( \widehat{\lambda}_j ^2  -\lambda _\ell ^2 )^{-2}   \lambda_\ell ^2 \ell^{-2\varsigma} +  \sum_{\ell \neq j} ( \widehat{\lambda}_j ^2  -\lambda _\ell ^2 )^{-2}   \widehat\lambda_j ^2 \ell^{-2\varsigma} \right) \nonumber\\
		&  {\leq O_p(T^{-1}) \left( \sum_{\ell \neq j} ( \widehat{\lambda}_j ^2  -\lambda _\ell ^2 )^{-2}   \lambda_\ell ^2 \ell ^{-2\varsigma} +  \sum_{\ell \neq j} ( \widehat{\lambda}_j ^2  -\lambda _\ell ^2 )^{-2}   (\widehat\lambda_j ^2 - \lambda_j ^2 ) \ell ^{-2\varsigma}+  \sum_{\ell \neq j} ( \widehat{\lambda}_j ^2  -\lambda _\ell ^2 )^{-2}     \lambda_j ^2   \ell ^{-2\varsigma} \right) }\nonumber \\
		& \leq  O_p(T^{-1})\left(  j^{\rho-2\varsigma+2}+  (O_p(T^{-1/2}) \lambda_j ^{-2}+1)\sum_{\ell \neq j} ( \widehat{\lambda}_j ^2  -\lambda _\ell ^2 )^{-2}   \lambda_j ^2  \ell ^{-2\varsigma} \right)\nonumber\\
		&\leq (1+O_p(T^{-1/2}\alpha^{-1}))O_p(T^{-1} j^{\rho-2\varsigma+2})   \leq (1+o_p(1))O_p(T^{-1} j^{\rho-2\varsigma+2}), \label{eqafhat1} 
	\end{align}
	where the second inequality follows from Assumption~\ref{assumconvrate}, and the remaining inequalities follow from H\"older's inequality, Lemma \ref{lem:convrate}\ref{lemconvrate8}-\ref{lemconvrate9}, and the fact that $\Vert \widehat{\mathcal C}_{xz} - \mathcal C_{xz}\Vert_{\op}  =O_p(T^{-1/2})$. From \eqref{eqafhat0} and \eqref{eqafhat1}, we find that \eqref{eq001} holds. 
	\qed 
	
	\subsubsection*{Proof of Theorem \ref{thm:convrate2}} 

In this proof, we consider the scenario where $\rho/2+2\geq\varsigma+\delta_{\zeta}$, thus incorporating the complementary result presented in Section \ref{sec_thm_conti1}. The whole proof is divided into two parts.\\
	\noindent	{\textbf{1. Proof of the convergence results}}: 	
	For the subsequent discussion, we first need to obtain an upper bound of  $ \langle \widehat{f}_j - f_j ^s,\zeta\rangle$. From the expansion given in \eqref{eq:b1}, we have
	\begin{equation}
		\langle \widehat{f}_j - f_j ^s,\zeta\rangle  =\sum_{\ell \neq j} (\widehat{\lambda}_j ^2 - \lambda_\ell^2) ^{-1} \langle  (\widehat{\mathcal C}_{xz} ^\ast  \widehat{\mathcal C}_{xz}- \mathcal C_{xz} ^\ast \mathcal C_{xz}) \widehat{f}_j ,  f_\ell \rangle \langle f_\ell ,\zeta\rangle  + \langle \widehat{f}_j - f_j ^s , f_j ^s \rangle \langle f_j ^s ,\zeta\rangle  . \label{eqbb6}
	\end{equation}
	Note that the second term in \eqref{eqbb6} satisfies that  $
	( \langle \widehat{f}_j - f_j , f_j^s \rangle  \langle f_j ^s ,\zeta\rangle) ^2  \leq O_p (T^{-1}  j^{-2\delta_{\zeta}+2 }) ,  
	$  for all $j=1,\ldots,\K$. 
	Moreover, the first term in \eqref{eqbb6} satisfies the following:
{\allowdisplaybreaks\begin{align}
		&( \sum_{\ell \neq j} (\widehat{\lambda}_j ^2 - \lambda_\ell^2) ^{-1} \langle  (\widehat{\mathcal C}_{xz} ^\ast \widehat{ \mathcal C}_{xz}- \mathcal C_{xz} ^\ast \mathcal C_{xz}) \widehat{f}_j ,  f_\ell \rangle \langle f_\ell ,\zeta\rangle  )^2  \leq O(1) (\sum_{\ell \neq j}|\widehat{\lambda}_j ^2 - \lambda_\ell^2| ^{-1} \ell^{-\delta_{\zeta}} | \langle  (\widehat{\mathcal C}_{xz} ^\ast \widehat{ \mathcal C}_{xz}- \mathcal C_{xz} ^\ast \mathcal C_{xz}) \widehat{f}_j ,  f_\ell \rangle| )^2\nonumber\\
		&\leq O(1)  ( \sum_{\ell \neq j} |\widehat{\lambda}_j ^2 - \lambda_\ell^2| ^{-1} \ell^{-\delta_{\zeta}} \lambda_\ell | \langle ( \widehat{\mathcal C}_{xz} - \mathcal C_{xz})\widehat{f}_j,\xi_\ell \rangle|     )^2+ O(1)   ( \sum_{\ell \neq j}  |\widehat{\lambda}_j ^2 - \lambda_\ell^2|^{-1} \ell^{-\delta_{\zeta}}\widehat{\lambda}_j  |\langle ( \widehat{\mathcal C}_{xz} ^\ast  - \mathcal C_{xz} ^\ast )\widehat{\xi}_j,f_\ell  \rangle |    )^2 \nonumber\\
		&\leq O(1)\Vert \widehat{\mathcal C}_{xz} - \mathcal C_{xz}\Vert_{\op} ^2 \left( \sum_{\ell \neq j} ( \widehat{\lambda}_j ^2  -\lambda _\ell ^2 )^{-2}   \lambda_\ell ^2 \ell ^{-2\delta_{\zeta}} +  \sum_{\ell \neq j} ( \widehat{\lambda}_j ^2  -\lambda _\ell ^2 )^{-2}   \widehat\lambda_j ^2 \ell ^{-2\delta_{\zeta}} \right) \nonumber\\
		&  {\leq O_p(T^{-1}) \left( \sum_{\ell \neq j} ( \widehat{\lambda}_j ^2  -\lambda _\ell ^2 )^{-2}   \lambda_\ell ^2 \ell ^{-2\delta_{\zeta}} +  \sum_{\ell \neq j} ( \widehat{\lambda}_j ^2  -\lambda _\ell ^2 )^{-2}   (\widehat\lambda_j ^2 - \lambda_j ^2 ) \ell ^{-2\delta_{\zeta}}+  \sum_{\ell \neq j} ( \widehat{\lambda}_j ^2  -\lambda _\ell ^2 )^{-2}     \lambda_j ^2   \ell ^{-2\delta_{\zeta}} \right) }\nonumber \\
		& \leq  O_p(T^{-1})\left(  j^{\rho-2\delta_{\zeta}+2}+  (O_p(T^{-1/2}) \lambda_j ^{-2}+1)\sum_{\ell \neq j} ( \widehat{\lambda}_j ^2  -\lambda _\ell ^2 )^{-2}   \lambda_j ^2  \ell ^{-2\delta_{\zeta}} \right)\nonumber\\
		&\leq O_p(T^{-1} )j^{\rho-2\delta_{\zeta}+2} (1+O_p(T^{-1/2})j^{\rho}), \label{eqb14} 
	\end{align}}where the first inequality follows from the assumption on $\langle f_\ell, \zeta\rangle$ and the remaining inequalities are deduced from similar arguments that are used to obtain \eqref{eqafhat1}.   Because $T^{-1/2}j^{\rho} \leq O_p(T^{-1/2} \alpha ^{-1}) = o_p(1)$ uniformly in $j=1,\ldots, \K$, we conclude that, for each $j=1,\ldots, \K$, \begin{equation}
		\langle \hat{f}_j-f_j^s,\zeta\rangle^2 =
		O_p(T^{-1})j^{-2\delta_{\zeta}+2} +  O_p(T^{-1})j^{-2\delta_{\zeta}+2+\rho}(1+o_p(1)).\label{eqfhatdelta}
	\end{equation} 
	
	We next show the following:
	\begin{align}\label{eqdesire0}
		\Vert \widehat{\mathcal C}_{xz} (\widehat{\mathcal C}_{xz} ^\ast \widehat{\mathcal C}_{xz})_{\K} ^{-1}\zeta - \mathcal C_{xz} (\mathcal C_{xz} ^\ast \mathcal C_{xz})_{\K} ^{-1}\zeta\Vert_{\op} = o_p(1).  
	\end{align}
	Given that $\|\widehat{\mathcal C}_{zz}-\mathcal C_{zz}\|_{\op}=o_p(1)$, the asymptotic results given in Theorem \ref{thm2w2}.\ref{thm2w2a} and \ref{thm2w2}.\ref{thm2w2b} are deduced without difficulty from \eqref{eqdesire0} and similar arguments used in our proofs given in Section \ref{sec:pf}. From the decomposition given in \eqref{eqpf0001a0}, it can be deduced that the desired result \eqref{eqdesire0} is established if the following terms are all $o_p(1)$: $\|\sum_{j=1}^{\K} ({({\lambda}_j^s)}^{-1}-\hat{\lambda}_j^{-1})  \langle {f}_j^s , \zeta \rangle  \widehat{\xi}_j \|$, $\|\sum_{j=1}^{\K} ({({\lambda}_j^s)}^{-1}-\hat{\lambda}_j^{-1})  \langle \widehat{f}_j - f_j ^s , \zeta \rangle  \widehat  {\xi}_j \|$, $\Vert \sum_{j=1} ^{\K} ( {\lambda}_j ^s) ^{-1} \langle   f_j ^s, \zeta \rangle (\widehat{\xi}_j  - \xi_j ^s )\Vert$, and $	\Vert \sum_{j=1} ^{\K} ( {\lambda}_j ^s) ^{-1} \langle \widehat{f}_j - f_j ^s, \zeta \rangle \widehat{\xi}_j\Vert$.	
	
	First, note that
	\begin{align}
		&\|\sum_{j=1}^{\K} ({({\lambda}_j^s)}^{-1}-\hat{\lambda}_j^{-1})  \langle {f}_j^s , \zeta \rangle  \widehat{\xi}_j \| ^2  = \sum_{j=1} ^{\K} ({({\lambda}_j^s)}^{-1}-\hat{\lambda}_j^{-1})^2 \langle f_j ^s , \zeta\rangle ^2 	\leq \sum_{j=1} ^{\K} \frac{(\lambda_j ^2 -\widehat{\lambda}_j ^2 )^2}{   {\lambda}_{j} ^2 (  \widehat{\lambda}_j ^2 + \lambda_j ^s \widehat {\lambda}_j)^2 } c_\zeta j^{-2\delta_{\zeta}} \nonumber\\
		&\leq \sum_{j=1} ^{\K} \frac{(\lambda_j ^2 -\widehat{\lambda}_j ^2 )^2 }{\lambda_{j} ^4 \widehat{\lambda}_{j} ^2 }c_{\zeta} j^{-2\delta_{\zeta} } 
		\leq    O_p(T^{-1}\alpha^{-1}\sum_{j=1} ^{\K}     j^{  2\rho-2\delta_{\zeta}} ) 
		\leq   O_p(T^{-1}\max\{\alpha^{-(3\rho-2\delta_{\zeta} +1)/\rho}, \alpha ^{-(1+\rho)/\rho}\}  ),\label{eqbb1}
	\end{align}
	where the last inequality  follows from \eqref{eq:kalpha}.  In addition, using \eqref{eqfhatdelta} and the arguments used to derive \eqref{eqbb1}, it can be shown that
	\begin{align}
		&\|\sum_{j=1}^{\K} ({({\lambda}_j^s)}^{-1}-\hat{\lambda}_j^{-1})  \langle \widehat{f}_j - f_j ^s , \zeta \rangle  \widehat  {\xi}_j \| ^2   =\sum_{j=1} ^{\K} ({({\lambda}_j^s)}^{-1}-\hat{\lambda}_j^{-1})^2 \langle \widehat f_j - f_j ^s , \zeta\rangle^2\nonumber\\
		&\leq O_p(T^{-2}\alpha^{-1})\sum_{j=1} ^{\K} j^{2\rho-2\delta_{\zeta}+2} +O_p(T^{-2}\alpha^{-1})\sum_{j=1} ^{\K} j^{3\rho-2\delta_{\zeta}+2} \nonumber\\
		&\leq O_p(T^{-2}) \max\{\alpha^{-(4\rho-2\delta_{\zeta}+3)/\rho} , \alpha^{-(1+\rho)/\rho}    \} .\label{eqbb2}
	\end{align}		
	Note that $T^{-1}\alpha^{-(1+\rho)/\rho} = o(1)$ and $T^{-2} \alpha^{-(4\rho-2\delta_{\zeta}+3)/\rho} = T^{-1} \alpha ^{-(3\rho-2\delta_{\zeta}+1)/\rho} T^{-1} \alpha ^{-(\rho+2)/\rho} = o(1)$ by \eqref{eqrbalpha1}. These imply that  the terms given in \eqref{eqbb1} and \eqref{eqbb2} are all $o_p(1)$. 
We also find that 
\begin{align}
	&\Vert \sum_{j=1} ^{\K} ( {\lambda}_j ^s) ^{-1} \langle   f_j ^s, \zeta \rangle (\widehat{\xi}_j  - \xi_j ^s )\Vert \leq \sum_{j=1} ^{\K}  \Vert  \lambda_j ^{-1}   \langle  f_j ^s, \zeta \rangle     (\widehat{\xi}_j  - \xi_j ^s)\Vert  \nonumber \\
	&\leq O_p(T^{-1/2})\sum_{j=1} ^{\K} j^{-\delta_{\zeta}+\rho/2+1 } \leq O_p(T^{-1/2}\max\{\alpha^{-1/\rho},  \alpha^{(\delta_{\zeta}-\rho/2-2)/\rho}\}) =o_p(1),     \label{eqbb4}
\end{align}  
where the first inequality follows from the triangular inequality and the second inequality is deduced from Assumption~\ref{assumconvrate}.\ref{assumconvrate.2} and the fact that $\Vert \widehat{\xi}_j - \xi_j^s\Vert^2 \leq O_p(T^{-1}j^{2})$ for $j=1,\ldots, \K$ (this can obtained from nearly identical arguments used to derive \eqref{eq000}). The remaining inequalities are deduced since {$2\delta_{\zeta} >1$} and \eqref{eq:kalpha} holds.  It only remains to show that $	\Vert \sum_{j=1} ^{\K} ( {\lambda}_j ^s) ^{-1} \langle \widehat{f}_j - f_j ^s, \zeta \rangle \widehat{\xi}_j\Vert ^2 = o_p(1)$. This can be obtained from Assumption \ref{assumconvrate}.\ref{assumconvrate.1} and \eqref{eqfhatdelta}; specifically, we observe that \begin{equation}
	\Vert \sum_{j=1} ^{\K} ( {\lambda}_j ^s) ^{-1} \langle \widehat{f}_j - f_j ^s, \zeta \rangle \widehat{\xi}_j\Vert ^2 \leq  O_p(T^{-1}) \max\{\alpha^{-(2\rho-2\delta_{\zeta}+3)/\rho}, \alpha^{-1/\rho}\}   = o_p(1).   \label{eqbb5}
\end{equation}		 
Hence, from the results given in \eqref{eqbb1}-\eqref{eqbb5},  \eqref{eqdesire0} is established.
\\[10pt]
\noindent	{\textbf{2. Analysis on the regularization bias}}: 
Next, we focus on the regularization bias term, $\|\mathcal A(\widehat{\Pi}_{\K} -\Pi_{\K})\zeta\| $.  For convenience, we let 
\begin{align*}
	\mathcal{A}(\widehat\Pi_{\K}-\mathcal I)\zeta = F_1 + F_2 + F_3 +F_4,	
\end{align*}
where $F_4 = \mathcal A(\Pi_{\K} - \mathcal I)\zeta$, 
\begin{align*}
	F_1 = \sum_{j=1}^{\K} \langle \hat{f}_j-f_j^s,\zeta\rangle \mathcal A(\hat{f}_j-f_j^s),
	\quad	F_2 = \sum_{j=1}^{\K} \langle f_j^s,\zeta\rangle \mathcal A(\hat{f}_j-f_j^s), 
	\quad	F_3 =  \sum_{j=1}^{\K} \langle \hat{f}_j-{f}_j^s,\zeta\rangle \mathcal A{f}_j^s,\quad  
\end{align*}
and thus $F_1+F_2+F_3 = \mathcal A(\widehat{\Pi}_{\K} - \Pi_{\K})\zeta$. Then, by using \eqref{eq001} and \eqref{eqfhatdelta}, we find that	\begin{equation*}
	\Vert F_1 \Vert   \leq  \sum_{j=1} ^{\K} |\langle \widehat{f}_j - f_j ^s,\zeta\rangle| \Vert\mathcal A(\widehat{f}_j - f_j ^s) \Vert  \leq   O_p(T^{-1})   \sum_{j=1} ^{\K}j^{\rho-\varsigma-\delta_{\zeta}+2 } \leq o_p(T^{-1/2}) \sum_{j=1} ^{\K} j^{\rho/2-\varsigma-\delta_{\zeta}+1 } ,
\end{equation*}
where the last bound is obtained because $\abar = o(T^{\rho/(2\rho+2)})$. In a similar manner, it can be shown that \begin{equation*}
	\|F_2\|  \leq  \sum_{j=1}^{\K} |\langle f_j^s,\zeta\rangle|  \| \mathcal A(\hat{f}_j-f_j^s)\|  \leq     O_p(T^{-1/2}) \sum_{j=1} ^{\K} j^{\rho/2-\varsigma-\delta_{\zeta}+1}  ,
\end{equation*}	
and		
\begin{equation*}
	\Vert F_3 \Vert  \leq \sum_{j=1} ^{\K} |\langle \widehat{f}_j - f_j ^s,\zeta\rangle| \Vert \mathcal A f_j ^s\Vert \leq O_p(T^{-1/2})  \sum_{j=1} ^{\K}j^{\rho/2-\varsigma-\delta_{\zeta}+1 }    .
\end{equation*}
Therefore, $\Vert F_1\Vert$, $\Vert F_2\Vert$ and $\Vert F_3\Vert  $ are bounded by the following quantity:\begin{equation*}
	O_p(T^{-1/2})  \sum_{j=1} ^{\K}j^{\rho/2-\varsigma-\delta_{\zeta}+1 } 	\leq  \begin{cases}O_p(T^{-1/2}    )  &\text{if }\rho/2+2<\varsigma+\delta_{\zeta},
		\\ 
		O_p(T^{-1/2}\max\{\log\alpha^{-1}, \alpha^{-(\rho/2-\varsigma - \delta_{\zeta}+2)/\rho} \}) &\text{if }\rho/2+2\geq\varsigma+\delta_{\zeta}.\end{cases}  
\end{equation*}
Lastly, the following can be shown:
\begin{equation*}
	\|F_4\|^2  \leq \sum_{j=\K+1} ^\infty \|  \langle f_j,\zeta \rangle \mathcal Af_j\|^2 \leq  \sum_{j=\K+1} ^\infty  j^{-2\delta_{\zeta}}\| \mathcal Af_j\|^2 =   O_p(\sum_{j=\K+1} ^\infty j^{-2\delta_{\zeta} -2\varsigma})  \leq O_p(\alpha^{(2\varsigma +2\delta_{\zeta} -1)/\rho}).   
\end{equation*}
This concludes the  proof. \qed

\subsection{Supplementary results}
	\subsubsection{Strong consistency of the FIVE}\label{sec:strongfive}
We first review some essential mathematics to establish the strong consistency of our estimators. The space of Hilbert-Schmidt operators, denoted $\mathcal S_{\mathcal H}$, is a separable Hilbert space with respect to the inner product given by $\langle \mathcal T_1,\mathcal T_2 \rangle_{\mathcal S_{\mathcal H}} = \sum_{j,k \geq 1} \langle \mathcal T_1\zeta_{1j},\zeta_{2k} \rangle\langle \mathcal T_2\zeta_{1j},\zeta_{2k} \rangle$ for two arbitrary orthonormal bases $\{\zeta_{1j}\}_{j \geq 1}$ and $\{\zeta_{2j}\}_{j \geq 1}$ of $\mathcal H$; this inner product does not depend on the choice of orthonormal bases \citep[Chapter 1]{Bosq2000}. We then note that $\{x_t \otimes z_t - \mathcal C_{xz}\}_{t\geq 1}$ is a zero-mean stationary and geometrically strongly mixing sequence in $\mathcal S_{\mathcal H}$, and, in the sequel, employ the following additional assumption: below, $\{ \Lambda_j  \}_{j \geq 1}$ is the sequence of eigenvalues of the covariance operator of $d_t=x_t \otimes z_t - \mathcal C_{xz}$.
\begin{assumption} \label{assum2ad}\begin{enumerate*}[(a)]
		\item\label{ass2ad-1} $\sup_{t\geq 1} \|x_t\| \leq m_x$, $\sup_{t\geq 1} \|z_t\| \leq m_z$, and $\sup_{t\geq 1} \|u_t\| \leq m_u$ a.s., \item\label{ass2ad-2}$\Lambda_j \leq  a b^j$ for some $a>0$ and $0<b<1$.
	\end{enumerate*}
\end{assumption}
As shown by Corollaries 2.4 and 4.2 of \cite{Bosq2000}, Assumption \ref{assum2ad} combined with Assumption \ref{assum1}.\ref{assum1.2} helps us obtain a stochastic bound of $\|\widehat{\mathcal C}_{xz}-{\mathcal C}_{xz}\|_{\op}$, which is given as follows:
\begin{lemma} \label{lem1ad} Under Assumptions \ref{assum1}.\ref{assum1.2} and \ref{assum2ad}, the following holds almost surely:
	\begin{equation*}
		\| \widehat{\mathcal C}_{xz}-\mathcal C_{xz}\| _{\op}=  O(T^{-1/2} \log^{3/2}T) .  
	\end{equation*}
\end{lemma}
We omit the proof of Lemma \ref{lem1ad} since it is a direct consequence of Theorem 2.12 and Corollary 2.4 of \cite{Bosq2000}, and the fact that $\sup_{t\geq 1}\|d_t\|\leq m_d$ holds for some $m_d > 0$ under the employed assumptions. Based upon this result, we can establish the strong consistency of the FIVE as follows:
\setcounter{theorem}{0}
\begin{theorem}[continued] \label{thm2.1}
	If Assumption \ref{assum2ad} is additionally satisfied, $\tau(\abar)  = o(T^{1/2}\log^{-3/2} T)$ a.s., and $\abar T^{-1} \log T \to 0$, then   
	$\|\widehat{\mathcal A} - \mathcal A\|_{\op} \to 0$ a.s.
\end{theorem}

\begin{proof}
	It can be easily shown from our proof of Theorem \ref{thm2w1} that $\|\widehat{\mathcal A}-\mathcal A \widehat{\Pi}_{\K}\|_{\op}
	\leq  \alpha^{-1/2}  \|T^{-1}\sum_{t=1}^Tz_{t}  \otimes  u_t \|_{\op}$ holds a.s.
	Under Assumptions \ref{assum1} and \ref{assum2ad}, the sequence of $z_{t}  \otimes  u_t$ is a martingale difference, and $\|z_t\otimes u_t\|_{\HS}$ and $\mathbb{E}\|z_t\otimes u_t\|_{\HS}^2$ are uniformly bounded, and we thus know from Theorem 2.14 of \cite{Bosq2000} that $\|T^{-1}\sum_{t=1}^Tz_{t}  \otimes  u_t\|_{\op} = O(T^{-1/2}\log^{1/2}T)$, a.s. This implies that $\|\widehat{\mathcal A}-  \mathcal  A \widehat{\Pi}_{\K}\|_{\op}  = O(\alpha^{-1/2}T^{-1/2} \log^{1/2}T)$ a.s. 
	Moreover, we note that $\|\mathcal A\widehat{\Pi}_{\K} - \mathcal A\|^2 _{\op}$ is bounded above by the term given in the right hand side of \eqref{eqpf0006add2wadd}, and deduce from Lemma \ref{lem1ad} that $\Vert \widehat{\mathcal C}_{xz} ^\ast \widehat{\mathcal C}_{xz} -\mathcal C_{xz} ^\ast \mathcal C_{xz}\Vert_{\op} = O(T^{-1/2} \log ^{3/2} T)$ a.s. These results imply that $\|\mathcal A\widehat{\Pi}_{\K} - \mathcal A\|^2 _{\op} = o(1)$ a.s.
\end{proof}

\subsubsection{Refinements of the general asymptotic results for the FIVE}\label{sec_thm_conti1}
We now provide a complementaty result to Theorem \ref{thm:convrate2} for the case where \(\rho/2+2\geq\varsigma+\delta_{\zeta}\). Specifically, the following can be shown:
\setcounter{theorem}{3}
\begin{theorem}[Continued]\label{thm:convrate2add}
Let everything be as in Theorem \ref{thm:convrate2} but with $\rho/2+2\geq\varsigma+\delta_{\zeta}$.
Then,  Theorem \ref{thm2w2} holds and 
\begin{align*}
&\Vert \mathcal A(\widehat \Pi_{\K}-\Pi_{\K}) \zeta \Vert   =	O_p( T^{-1/2}\max\{\log\abar, \abar^{(\rho/2-\varsigma-\delta_{\zeta}+2)/\rho} \}), \\
&	\|\mathcal A(\Pi_{\K} - \mathcal I)\zeta\| = O_p(\abar^{(1/2-\varsigma- \delta_{\zeta})/\rho}).  \label{eqthmrb4}
\end{align*}
\end{theorem}
Our proof of the above result is already contained in the proof of Theorem \ref{thm:convrate2} given in Section \ref{sec:asym2proof}, and hence omitted.


\section{Appendix to Section \ref{sec:f2sls00} on ``Functional two-stage least square estimator" \label{sec:pfe}}	
As in Section \ref{sec:pf}, we will hereafter let \(\alpha_1 = \abar_1^{-1}\) and \(\alpha_2 = \abar_2^{-1}\),  and use them interchangeably. We first provide a useful lemma that is related to our discussion on the F2SLSE in Section \ref{sec:f2sls00}.
	\begin{lemma}\label{lem:ftsls} There exist unique bounded linear operators $\mathcal R_{xz}$ and $\mathcal R_{yz}^\ast$ satisfying the following: 
		\begin{align*}
			&\mathcal C_{zz}^{1/2} \mathcal R_{xz} \mathcal C_{xx}^{1/2} = \mathcal C_{xz}, \quad \mathcal R_{xz}[\ran \mathcal C_{xx}^{1/2}]^\perp = \{0\}, \quad  \mathcal R_{xz}^\ast[ \ran \mathcal C_{zz}^{1/2}]^\perp = \{0\}, \\
			&\mathcal C_{yy}^{1/2} \mathcal R_{yz}^\ast \mathcal C_{zz}^{1/2} = \mathcal C_{yz}^\ast, \quad \mathcal R_{yz}^\ast [ \ran \mathcal C_{zz}^{1/2}]^\perp = \{0\}, \quad  \mathcal R_{yz}[ \ran \mathcal C_{yy}^{1/2}]^\perp = \{0\},
		\end{align*}
		where $V^\perp$ denotes the  orthogonal complement of $V \subset \mathcal H$.
	\end{lemma}

	Lemma \ref{lem:ftsls} directly follows from Theorem 1 of \cite{baker1973joint}.	From the properties of $\mathcal R_{xz}$ (resp.\ $\mathcal R_{yz}$) given above, it can be understood as the cross-correlation operator of $x_t$ and $z_t$ (resp.\ $y_t$ and $z_t$). Let $\mathcal C_{zz}^{-1/2}$ be defined by $\sum_{j=1}^\infty \mu_j^{-1/2} g_j \otimes g_j$, which is not a bounded linear operator since $\mu_j \to 0$ as $j \to \infty$. However, even with this property, we know as a direct consequence of Lemma~\ref{lem:ftsls} that $\mathcal C_{zz}^{-1/2} \mathcal C_{xz}$ and $\mathcal C_{yz}^\ast \mathcal C_{zz}^{-1/2}$ are well defined bounded linear operators and they are respectively given by 
	\begin{equation*}
		\mathcal C_{zz}^{-1/2} \mathcal C_{xz} =   \mathcal R_{xz} 	\mathcal C_{xx}^{1/2} \quad\text{and}\quad \mathcal C_{yz}^{\ast}  \mathcal C_{zz}^{-1/2}=   \mathcal C_{yy}^{1/2} \mathcal R_{yz}^\ast.
	\end{equation*}
	We thus find that\vspace{-0.7em} 
	\begin{equation*}\begin{aligned}
			&\mathcal C_{yz} ^\ast \mathcal C_{zz}^{-1}\mathcal C_{xz}= \mathcal C_{yz} ^\ast \mathcal C_{zz}^{-1/2}\mathcal C_{zz}^{-1/2}\mathcal C_{xz} =  \mathcal C_{yy}^{1/2}\mathcal R_{yz}^\ast\mathcal R_{yz}  	 \mathcal C_{zz}^{1/2} \eqqcolon \mathcal P,  \\
			&\mathcal C_{xz} ^\ast \mathcal C_{zz}^{-1}\mathcal C_{xz}= \mathcal C_{xz} ^\ast \mathcal C_{zz}^{-1/2}\mathcal C_{zz}^{-1/2}\mathcal C_{xz} =  \mathcal C_{xx}^{1/2}\mathcal R_{xz}^\ast\mathcal R_{xz}  	 \mathcal C_{xx}^{1/2} \eqqcolon \mathcal Q.\end{aligned}
	\end{equation*}
	As desired, $\mathcal P$ and $\mathcal Q$ are uniquely defined elements of $\mathcal L_{\mathcal H}$, and moreover, they are compact since $\mathcal C_{xx}^{1/2}$ and $\mathcal C_{yy}^{1/2}$ are compact.
	\subsection{Proofs of the results in Section \ref{sec:asym1a}} \label{sec:f2sls0proof}	
	\subsubsection*{Proof of Theorem \ref{thm2w}}
	Since $\|\widehat{\mathcal C}_{uz}\|_{\HS}=O_p(T^{-1/2})$, we find that
	\begin{equation}
		\|\widetilde{\mathcal A}-\mathcal A \widetilde{\Pi}_{\K_2}\|_{\HS} 
		\leq \|\widehat{\mathcal C}_{uz }\|_{\HS}\|(\widehat{\mathcal C}_{zz })_{\K_1}^{-1/2}\|_{\op}\|(\widehat{\mathcal C}_{zz })_{\K_1}^{-1/2}\widehat{\mathcal C}_{xz}\widehat{\mathcal Q}_{\K_2}^{-1} \|_{\op} \leq O_p(\alpha_1^{-1/4}\alpha_2^{-1/4} T^{-1/2}).
		\label{eqpf0001vv}
	\end{equation}
	Thus, $\|\widetilde{\mathcal A}-\mathcal A \widetilde{\Pi}_{\K_2}\|_{\HS} =  o_p(1)$, and hence it suffices to show that $	\|\mathcal A \widetilde{\Pi}_{\K_2} -\mathcal A \|_{\HS}^2 = o_p(1)$.   Note  that 
	\begin{equation}
		\|\mathcal A \widetilde{\Pi}_{\K_2} -\mathcal A \|_{\HS}^2  \leq  \sum_{j= \K_2+1}^{\infty} \|\mathcal A {h}^s_j\|^2 + |\mathcal R|, \label{eq001a}
	\end{equation}
	where ${h}^s_j = \sgn\{\langle \hat{h}_j, {h}_j \rangle \} {h}_j$ and $\mathcal R = 	\sum_{j= \K_2+1}^{\infty} (\|\mathcal A \hat{h}_j\|^2 -  \|\mathcal A {h}^s_j\|^2)$.  Since $\mathcal A$ is Hilbert-Schmidt, the first term of \eqref{eq001a} is $o_p(1)$. It thus only remains to verify that $|\mathcal R| = o_p(1)$. To show this, we first deduce the following inequality from similar arguments used to derive the equation between (8.62) and (8.63) of \cite{Bosq2000}: 
	\begin{equation} \label{eq01a}
		|\mathcal R| \leq 2 \|\mathcal A\|_{\op}^2 \sum_{j=1}^{\K_2}\|\hat{h}_j-{h}^s_j\|.
	\end{equation}
	We find that $\mathcal Q \hat{h}_j - {\nu}_j \hat{h}_j =  (\mathcal Q-\widehat{\mathcal Q}) \hat{h}_j + (\hat{\nu}_j- {\nu}_j) \hat{h}_j$. Hence, by Lemma 4.2 of  \cite{Bosq2000}, 
	\begin{equation}
		\|\mathcal Q \hat{h}_j - {\nu}_j \hat{h}_j \|\leq 2 \|\widehat{\mathcal Q}- {\mathcal Q}\|_{\op}.
		\label{eq001b}	\end{equation} 
	Moreover, it can be shown from similar arguments used in the proof of Lemma 4.3 of  \cite{Bosq2000} that $\|\hat{h}_j - h^s_j\| \leq \frac{\tau_{2,j}}{2} \|\mathcal Q \hat{h}_j - \nu_j \hat{h}_j\|$,	which, combined with \eqref{eq001b},  implies that   
	\begin{equation} \label{eqeigenvector}
		\|\hat{h}_j - h^s_j\| \leq \tau_{2,j} \|\widehat{\mathcal Q}-\mathcal Q\|_{\op}.
	\end{equation}
	We then deduce from \eqref{eq01a} and \eqref{eqeigenvector} that
	\begin{equation}
		|\mathcal R| \leq  2 \|\mathcal A\|_{\op}^2 \left( \sum_{j=1}^{\K_2} \tau_{2,j}\right)  \|\widehat{\mathcal Q}-\mathcal Q\|_{\op} . \label{eqpfpf01}
	\end{equation}
	Then the following can be shown:
	\begin{equation}
		\|\widehat{\mathcal Q} -  {\mathcal Q}\|_{\op} \leq \|\mathcal B\widehat{\Pi}_{\K_1}\widehat{\mathcal C}_{zz}\widehat{\Pi}_{\K_1}\mathcal B^\ast - \mathcal B {\mathcal C}_{zz}\mathcal B^\ast \|_{\op} + \|\mathcal S\|_{\op},\label{eqReq00}
	\end{equation}			
	where $\mathcal S = \widehat{\mathcal C}_{vz}^\ast\widehat{\Pi}_{\K_1}\mathcal B^\ast + \mathcal B \widehat{\Pi}_{\K_1}\widehat{\mathcal C}_{vz} + \widehat{\mathcal C}_{vz}^\ast(\widehat{\mathcal C}_{zz})_{\K_1}^{-1} \widehat{\mathcal C}_{vz}$.
	Let  $\Pi_{\K_1} = \sum_{j=1}^{\K_1} g_j \otimes g_j$ and let 
	\begin{equation*}
		\mathcal T = \mathcal B(\mathcal I-\Pi_{\K_1}){\mathcal C}_{zz}(\mathcal I-\Pi_{\K_1})\mathcal B^\ast. 
	\end{equation*}
	We further find that{\allowdisplaybreaks\begin{align}
		&	\|\mathcal B\widehat{\Pi}_{\K_1}\widehat{\mathcal C}_{zz}\widehat{\Pi}_{\K_1}\mathcal B^\ast - \mathcal B{\mathcal C}_{zz}\mathcal B^\ast \|_{\op} \leq \|\mathcal B\widehat{\Pi}_{\K_1}\widehat{\mathcal C}_{zz}\widehat{\Pi}_{\K_1}\mathcal B^\ast - \mathcal B\Pi_{\K_1}{\mathcal C}_{zz}\Pi_{\K_1}\mathcal B^\ast \|_{\op}  +  \|\mathcal T\|_{\op} \notag \\ &\leq  \|\mathcal B\|_{\op}^2\|\sum_{j=1}^{\K_1} (\hat{\mu}_j-\mu_j)\hat{g}_j \otimes \hat{g}_j\|_{\op} + 2 \|\mathcal B\|_{\op}^2 \sum_{j=1}^{\K_1} \mu_j \|\hat{g}_j- g^s_j \| +  \|\mathcal T\|_{\op}  \notag\\ 
		&\leq  \left(\|\mathcal B\|_{\op}^2 +2\|\mathcal B\|_{\op}^2 \sum_{j=1}^{\K_1} \mu_j \tau_{1,j}  \right) \|\widehat{\mathcal C}_{zz}-\mathcal C_{zz}\|_{\op}+  \|\mathcal T\|_{\op}  \leq   O_p\left(\sum_{j=1}^{\K_1} \mu_j \tau_{1,j} \right) \|\widehat{\mathcal C}_{zz}-\mathcal C_{zz}\|_{\op}  +  \|\mathcal T\|_{\op}, \label{eqnormpf01}	
	\end{align}}
	where ${g}^s_j = \sgn\{\langle \hat{g}_j, g_j \rangle \} g_j$.
	From \eqref{eqpfpf01}, \eqref{eqReq00} and  \eqref{eqnormpf01}, the following is established:
	\begin{equation}
		|\mathcal R|	\leq  \left\{ O_p\left(\sum_{j=1}^{\K_1} \mu_j \tau_{1,j}\right) \|\widehat{\mathcal C}_{zz}-\mathcal C_{zz}\|_{\op} + \|\mathcal S\|_{\op} +  \|\mathcal T\|_{\op} \right\} O_p\left(\sum_{j=1}^{\K_2} \tau_{2,j} \right). \label{eqnormpf02adad}
	\end{equation}
	Since $\|\widehat{\mathcal C}_{vz}\|_{\HS}= O_p(T^{-1/2})$, $\|\widehat{\Pi}_{\K_1}\|_{\op} \leq 1$ and $\|(\widehat{\mathcal C}_{zz})_{\K_1}^{-1}\|_{\op} \leq \alpha_1^{-1/2}$, we have
	\begin{equation} \label{eqReq}
		\|\mathcal S\|_{\op} \leq O_p(T^{-1/2}) + O_p(\alpha_1^{-1/2} T^{-1}). 
	\end{equation}			
	Note that $\|\widehat{\mathcal C}_{zz}-\mathcal C_{zz}\|_{\op} = O_p(T^{-1/2})$ and $\|\mathcal T\|_{\op}\sum_{j=1}^{\K_2} \tau_{2,j}= o_p(1)$ (which follows from the fact that $\|\mathcal T\|_{\op} \leq \|\mathcal B\|^2_{\op} \sum_{j=\K_1+1}^{\infty} \mu_j$).  Combining these results with \eqref{eqnormpf02adad} and \eqref{eqReq}, we find that $|\mathcal R|	\leq  O_p(T^{-1/2}(\sum_{j=1}^{\K_1} \mu_j \tau_{1,j})(\sum_{j=1}^{\K_2} \tau_{2,j})) + O_p((T^{-1/2} + \alpha_1^{-1/2}T^{-1}) \sum_{j=1}^{\K_2} \tau_{2,j} ) $. Given that $\alpha_1^{-1} T^{-1} \to 0$  and  $\sum_{j=1}^{\K_2} \tau_{2,j}\leq O_p((\sum_{j=1}^{\K_1} \mu_j  \tau_{1,j})(\sum_{j=1}^{\K_2} \tau_{2,j}))=o_p(T^{1/2})$ (the  inequality follows from that $\mu_1^{-1} \tau_{1,1}^{-1}(\sum_{j=1}^{\K_1} \mu_j \tau_{1,j}) \geq 1$), it may be easily deduced that $|\mathcal R| =  o_p(1)$ 	as desired.  \qed
	\subsubsection*{Proof of Theorem \ref{thm2wa}}
	To show (i), we will first verify that 
	\begin{align}
		\|(\widehat{\mathcal C}_{zz})_{\K_1}^{-1}\widehat{\mathcal C}_{xz}\widehat{\mathcal Q}_{\K_2}^{-1}-({\mathcal C}_{zz})_{\K_1}^{-1}{\mathcal C}_{xz}{\mathcal Q}_{\K_2}^{-1}\|_{\op} &\leq  E_1+E_2+E_3 = o_p(1), \label{eqdesired}
	\end{align} 
	where $E_1$, $E_2$ and $E_3$ are defined as follows:
	\begin{align}
		&E_1=	\|((\widehat{\mathcal C}_{zz})_{\K_1}^{-1}-(\mathcal C_{zz})_{\K_1}^{-1})\widehat{\mathcal C}_{xz}\widehat{\mathcal Q}_{\K_2}^{-1}\|_{\op}, \quad\quad E_2= \|(\mathcal C_{zz})_{\K_1}^{-1}(\widehat{\mathcal C}_{xz}-\mathcal C_{xz})\widehat{\mathcal Q}_{\K_2}^{-1}\|_{\op}, \nonumber\\
		& E_3=\|(\mathcal C_{zz})_{\K_1}^{-1}\mathcal C_{xz}(\widehat{\mathcal Q}_{\K_2}^{-1}-{\mathcal Q}_{\K_2}^{-1})\|_{\op}. \nonumber
	\end{align}
	Note first that $\|\widehat{\mathcal C}_{xz}\widehat{\mathcal Q}_{\K_2}^{-1}\|_{\op}=O_p(\alpha_2^{-1/2})$, and thus 
	\begin{equation*} 
		E_1  \leq  O_p(\alpha_2^{-1/2}) \left(\|\sum_{j=1}^{\K_1} ( {\mu}_j ^{-1}-\hat{\mu}_j^{-1}) {g}_j^s\otimes{g}_j^s \|_{\op} +   \|\sum_{j=1}^{\K_1} \hat{\mu}_j^{-1} (\hat{g}_j\otimes \hat{g}_j-{g}_j^s\otimes{g}_j^s )\|_{\op}\right),
	\end{equation*}
	where ${g}^s_j = \sgn\{\langle \hat{g}_j, g_j \rangle \} g_j$.	We then find that
	\begin{equation}
		\|\sum_{j=1}^{\K_1} ({\mu}_j^{-1}-\hat{\mu}_j^{-1}) {g}_j^s\otimes{g}_j^s \|_{\op} \leq \sup_{1 \leq j \leq \K_1}|\hat{\mu}_j^{-1}-\mu_j^{-1}| = \sup_{1 \leq j \leq \K_1}\left|\frac{\hat{\mu}_j-\mu_j}{{\mu}_j\hat{\mu}_j}\right|
		\leq \alpha_1^{-1/2}\mu_{\K_1}^{-1}\|\widehat{\mathcal C}_{zz} -  {\mathcal C}_{zz}\|_{\op}   \label{eqnormpf02ad}
	\end{equation}
	and
	\begin{equation}
		\|\sum_{j=1}^{\K_1} \hat{\mu}_j^{-1} (\hat{g}_j\otimes \hat{g}_j-{g}_j^s\otimes{g}_j ^s)\|_{\op} \leq 	2\alpha_1^{-1/2} \sum_{j=1}^{\K_1}  \| \hat{g}_j-g^s_j\|\leq 2 \alpha_1^{-1/2}\|\widehat{\mathcal C}_{zz} -  {\mathcal C}_{zz}\|_{\op} \sum_{j=1}^{\K_1}\tau_{1,j}.   \label{eqnormpf02ad2}
	\end{equation}	
	Since	$ \mu_{\K_1}^{-1} \leq\sum_{j=1}^{\K_1} \tau_{1,j} = o_p(\alpha_{1}^{1/2}T^{1/2})$ and $\|\widehat{\mathcal C}_{zz}-\mathcal C_{zz}\|_{\op}=O_p(T^{-1/2})$, the right hand sides of \eqref{eqnormpf02ad} and  \eqref{eqnormpf02ad2} are $o_p(1)$, and hence $E_1 = o_p(1)$. Since $\|(\mathcal C_{zz})_{\K_1}^{-1}\|_{\op} \leq \mu_{\K_1}^{-1}$, $\|\widehat{\mathcal Q}_{\K_2}^{-1}\|_{\op}\leq \alpha_2^{-1/2}$, and $\|\widehat{\mathcal C}_{xz}-\mathcal C_{xz}\|_{\op} = O_p(T^{-1/2})$, we also find that 
	\begin{align*}
		E_2 \leq \mu_{\K_1}^{-1} O_p(\alpha_2^{-1/2} T^{-1/2} ) \leq  O_p( \alpha_2^{-1/2} T^{-1/2}) \sum_{j=1}^{\K_1} \tau_{1,j} = o_p(1).  
	\end{align*}
	Given that $\|(\mathcal C_{zz})_{\K_1}^{-1}\mathcal C_{xz}\|_{\op} \leq  \Vert\mathcal B^\ast\Vert_{\op} = O_p(1)$, {it remains to show that $\|\widehat{\mathcal Q}_{\K_2}^{-1} - {\mathcal Q}_{\K_2}^{-1}\|_{\op}=o_p(1)$ since this implies $E_3 = o_p(1)$ (and thus the desired result \eqref{eqdesired} is obtained).} To show this, we first note that 
	\begin{equation*} 
		\|\widehat{\mathcal Q}_{\K_2}^{-1} - {\mathcal Q}_{\K_2}^{-1}\|_{\op}\leq  \|\sum_{j=1}^{\K_2} ({\nu}_j^{-1}-\hat{\nu}_j^{-1}) {h}_j\otimes{h}_j \|_{\op} +   \|\sum_{j=1}^{\K_2} \hat{\nu}_j^{-1} (\hat{h}_j\otimes \hat{h}_j-{h}_j\otimes{h}_j )\|_{\op}.  
	\end{equation*}
	Let $\mathcal S = \widehat{\mathcal C}_{vz}^\ast\widehat{\Pi}_{\K_1} \mathcal B^\ast + \mathcal B \widehat{\Pi}_{\K_1}\widehat{\mathcal C}_{vz} + \widehat{\mathcal C}_{vz}^\ast(\widehat{\mathcal C}_{zz})_{\K_1}^{-1} \widehat{\mathcal C}_{vz}$. We then deduce from our proof of Theorem \ref{thm2w} that 
	\begin{equation}
		\|\widehat{\mathcal Q} -  {\mathcal Q}\|_{\op} \leq 	 O_p\left(\sum_{j=1}^{\K_1} \mu_j \tau_{1,j}\right) \|\widehat{\mathcal C}_{zz}-\mathcal C_{zz}\|_{\op} + \|\mathcal S\|_{\op} + \|\mathcal T\|_{\op}.  \label{eqpf0002w2}
	\end{equation}
	As in \eqref{eqnormpf02ad}, it can be shown that $\|\sum_{j=1}^{\K_2} ({\nu}_j^{-1}-\hat{\nu}_j^{-1}) {h}_j\otimes{h}_j \|_{\op}  \leq \alpha_2^{-1/2}\nu_{\K_2}^{-1}	\|\widehat{\mathcal Q}-  {\mathcal Q}\|_{\op}$, hence
	\begin{equation}
		\|\sum_{j=1}^{\K_2} ({\nu}_j^{-1}-\hat{\nu}_j^{-1}) {h}_j\otimes{h}_j \|_{\op} \leq  \alpha_2^{-1/2}\nu_{\K_2}^{-1}  \left(O_p\left(\sum_{j=1}^{\K_1} \mu_j \tau_{1,j}\right)   \|\widehat{\mathcal C}_{zz}-\mathcal C_{zz}\|_{\op} + \|\mathcal S\|_{\op} + \|\mathcal T\|_{\op} \right).   \label{eqnormpf02}
	\end{equation}
	We also deduce the following from \eqref{eqeigenvector}, \eqref{eqnormpf02ad2} and \eqref{eqpf0002w2}:
	\begin{equation}
		\|\sum_{j=1}^{\K_2} \hat{\nu}_j^{-1} (\hat{h}_j\otimes \hat{h}_j-{h}_j\otimes{h}_j )\|_{\op} \leq  2 \alpha_2^{-1/2} \sum_{j=1}^{\K_2}\tau_{2,j} \left( O_p\left(\sum_{j=1}^{\K_1} \mu_j \tau_{1,j}\right)  \|\widehat{\mathcal C}_{zz}-\mathcal C_{zz}\|_{\op} + \| \mathcal S\|_\op + \|\mathcal T\|_{\op}\right).
		\label{eqnormpf03}
	\end{equation}
	We find that $\nu_{\K_2}^{-1} \leq \sum_{j=1}^{\K_2}\tau_{2,j} \leq O_p((\sum_{j=1}^{\K_1} \mu_j \tau_{1,j})(\sum_{j=1}^{\K_2}\tau_{2,j}))= o_p(\alpha_2^{1/2}T^{1/2})$, which follows from that $\mu_1^{-1} \tau_{1,1}^{-1}(\sum_{j=1}^{\K_1} \mu_j \tau_{1,j}) \geq 1$. Moreover, note that
	\begin{equation*} 
		\sum_{j=1}^{\K_1} \mu_j \tau_{1,j} \leq \nu_1 \nu_{\K_2} ^{-1} \sum_{j=1}^{\K_1} \mu_j\tau_{1,j} \leq \nu_1 \left(\sum_{j=1}^{\K_1} \mu_j \tau_{1,j}\right)\left(\sum_{j=1}^{\K_2}\tau_{2,j}\right)  = o_p(\alpha_2^{1/2}T^{1/2}) ,\end{equation*} and
	\begin{equation*}
		\nu_{\K_2} ^{-1} \|\mathcal T\|_{\op} \leq  \sum_{j=1}^{\K_2}\tau_{2,j}  \|\mathcal T\|_{\op} \leq  \|\mathcal B\|_{\op}^2\left(\sum_{j=\K_1+1}^\infty \mu_j \right) \left(\sum_{j=1}^{\K_2}\tau_{2,j}\right)  =o_p(\alpha_2^{1/2}). 
	\end{equation*}
	From  these results, \eqref{eqReq}  and the fact that $\|\widehat{\mathcal C}_{zz}-\mathcal C_{zz}\|_{\op} = O_p(T^{-1/2})$, we may deduce that  the right hand sides of \eqref{eqnormpf02} and \eqref{eqnormpf03} are all $o_p(1)$, and thus $E_3 =o_p(1)$ and \eqref{eqdesired} holds. 
	
	We thus know that 
	\begin{equation*}
		\sqrt{T}(\widetilde{\mathcal A}- \mathcal A \widetilde{\Pi}_{\K_2}) 
		\zeta = \left(\frac{1}{\sqrt{T}} \sum_{t=1}^Tz_{t} \otimes u_{t}\right) ({\mathcal C}_{zz})_{\K_1}^{-1} {\mathcal C}_{xz}\mathcal Q_{\K_2}^{-1}\zeta +  o_p(1).
	\end{equation*}
	Define $\zeta_t =  (\phi_{\K_2}(\zeta))^{-1/2}  [z_{t} \otimes u_t]   ({\mathcal C}_{zz})_{\K_1}^{-1} {\mathcal C}_{xz}\mathcal Q_{\K_2}^{-1}\zeta$ and let $\ddot\zeta_T = T^{-1/2}\sum_{t=1}^T  \zeta_t$. Then from nearly identical arguments used to derive \eqref{eqlater1}, \eqref{eqpf0008add} and \eqref{eqpf0009add}, we find that, for any $\psi \in \mathcal H$ and $m>0$, $T^{-1/2}\sum_{t=1}^T \langle \zeta_t,\zeta_1\rangle  \dto N(0, \langle \mathcal C_{uu}\psi,\psi\rangle)$ and  		$\limsup_{n\to \infty}	\limsup_{T} \mathbb{P}( \sum_{j=n+1}^\infty  \langle\ddot \zeta_T, \ell_j  \rangle^2  > m) = 0$,
	where $\{\ell_j\}_{j \geq 1}$ denote the eigenfunctions of $\mathcal C_{uu}$.  Hence (i) is established. 
	
	Given that $\|\widehat{\mathcal Q}_{\K_2}^{-1} -  {\mathcal Q}_{\K_2}^{-1}\|_{\op}=o_p(1)$, (ii) is immediately deduced.  \qed

	\subsection{Proofs of the results in Section \ref{sec:asym2a}} \label{sec:f2sls0proof2}	
We hereafter define
	\begin{equation*}
		\mathcal Q_{\K_1} = \mathcal C_{xz}^\ast (\mathcal C_{zz})_{\K_1}^{-1}\mathcal C_{xz},
	\end{equation*}
	which is repeatedly used in the subsequent proofs.

	\subsubsection*{Proof of Theorem \ref{thm3:convrate}}
	We will show the following: 
	\begin{align}
		\K_2  &\leq (1+o_p(1))\alpha_2^{-1/\rho_{\nu}},\label{eq:k2alpha1} \\
		(c_\circ \rho)^{-1}{(\K_2+1)^{-\rho_{\nu} }} &\leq (1+o_p(1)) \alpha_2, \label{eq:k2alpha2} \\ 
		\Vert \widehat{h}_j - h_j^s \Vert ^2 &\leq O_p(\alpha_1 )\bt{j^{\rho_{\nu}-4\gamma_{\mu}+2}}, \label{eqadd01a} \\
		\Vert \mathcal A(\widehat{h}_j - h_j^s)\Vert ^2 &\leq O_p(\alpha_1) \bt{j^{\rho_{\nu}-2\varsigma_\nu-4\gamma_{\mu}+2}} + O_p(d_T)j^{\rho_\nu-2\varsigma_{\nu}+2}, \label{eqahh}
	\end{align}
	where $h_j^s$ is defined as in our proof of Theorem \ref{thm2w} and $d_T$ is defined by 
	\begin{equation} \label{eqdt}
		d_T =\alpha_1^{(4\varsigma_{\mu}+\rho_{\mu}-2)/\rho_{\mu}}  + T^{-1}\max\{  \alpha_1 ^{-1/\rho_{\mu}},  \alpha_1^{-(\rho_\mu-2\varsigma_{\mu}+3)/\rho_{\mu}} \}.
	\end{equation}	
	Note that
	\begin{equation*}
		\Vert \widetilde{\mathcal A} - {\mathcal A} \Vert_{\HS} \leq \Vert \widetilde{\mathcal A} - \mathcal A \widetilde{\Pi}_{\K_2} \Vert_{\HS} + \Vert \mathcal A \widetilde{\Pi}_{\K_2} - \mathcal A \Pi_{\K_2} \Vert _{\HS} + \Vert (\mathcal I - \Pi_{\K_2}) \mathcal A \Vert_{\HS},
	\end{equation*}
	where $\Vert \widetilde{\mathcal A} - \mathcal A \widetilde {\Pi}_{\K_2} \Vert_{\HS} = O_p(\alpha_1 ^{-1/4}\alpha_2 ^{-1/4} T^{-1/2})$ as shown in \eqref{eqpf0001vv}. 
	Using \eqref{eq:k2alpha2}, we also find that \begin{equation}
		\Vert (\mathcal I - \Pi_{\K_2}) \mathcal A \Vert_{\HS} ^2 =\sum_{\ell = \K_2+1} ^\infty \Vert \mathcal A h_\ell  \Vert ^2 \leq O(1)\sum_{\ell = \K_2+1} ^\infty \sum_{j=1} ^\infty \ell ^{-2\varsigma_{\nu}} \bt{j^{-2\gamma_{\nu}}} 
		\leq O_p(\alpha_2 ^{(2\varsigma_{\nu}-1)/\rho_{\nu}}).\label{eqqhatt00}
	\end{equation}
	We next focus on the remaining term $\Vert \mathcal A \widetilde{\Pi}_{\K_2} - \mathcal A \Pi_{\K_2} \Vert _{\HS}$. Note that \begin{equation}
		2^{-1}\Vert \mathcal A \widetilde{\Pi}_{\K_2} - \mathcal A \Pi_{\K_2} \Vert _{\HS} ^2 \leq \Vert \sum_{j = 1}^{ \K_2} \widehat{h}_j \otimes \mathcal A(\widehat{h}_j - h_j^s )\Vert _{\HS} ^2 + \Vert \sum_{j = 1}^{ \K_2} (\widehat{h}_j - h_j^s ) \otimes \mathcal Ah_j\Vert _{\HS} ^2 . \label{eqqhatt0}
	\end{equation}
	We know from \eqref{eqadd01a} and \eqref{eqqhatt0} that
	{\allowdisplaybreaks\begin{align}
		\Vert& \sum_{j=1} ^{\K_2} (\widehat{h}_j - h_j ^s) \otimes \mathcal A h_j ^s \Vert_{\HS} ^2 = 
		\sum_{\ell=1}^{\infty}\Vert \sum_{j=1} ^{\K_2} \langle \mathcal A h_j,h_{\ell} \rangle (\widehat{h}_j - h_j ^s) \Vert ^2   \leq 	\sum_{\ell=1}^{\infty} \left(\sum_{j=1}^{\K_2} |\langle \mathcal A h_j,h_{\ell} \rangle| \|\widehat{h}_j - h_j ^s\| \right)^2 \nonumber\\
		&\leq   O_p(\alpha_1)  \left(\sum_{j=1}^{\K_2}  \bt{j^{\rho_{\nu}/2-2\gamma_{\mu}-\varsigma_{\nu}+1}} \right)^2   \nonumber\\ &= \begin{cases}O_p(\alpha_1) & \text{if }\varsigma_{\nu}>2 + \bt{\rho_{\nu}/2 -2\gamma_{\mu}},
			\\
			O_p(\alpha_1\max\{\log^2 \alpha_2^{-1}, \bt{\alpha_2^{(2\varsigma_{\nu}-\rho_{\nu}+4\gamma_{\mu}-4)/\rho_{\nu}}} \}) &\text{if }\varsigma_{\nu} \leq 2 + \bt{\rho_{\nu}/2 -2\gamma_{\mu}} .\end{cases} 
		\label{eqadd10}		\end{align}}
	Moreover, from \eqref{eqahh} and the fact that  $d_T \alpha_1 ^{-1} = o(1)$, the following may be deduced:
		\begin{align}
			&\Vert \sum_{j=1} ^{\K_2} \widehat{h}_j \otimes \mathcal A(\widehat{h}_j - h_j^s)\Vert ^2_{\HS} \leq  \sum_{j=1} ^{\K_2} 
			\Vert\mathcal A(\widehat{h}_j - h_j^s)\Vert ^2 \leq O_p(\alpha_1) \sum_{j=1}^{\K_2} \bt{j^{\rho_{\nu}-2\varsigma_\nu-4\gamma_{\mu}+2}} + O_p(d_T)\sum_{j=1}^{\K_2}j^{\rho_\nu-2\varsigma_{\nu}+2} \nonumber \\ 
			&\leq \begin{cases}
				O_p(\alpha_1 )	&\text{if }\rho_\nu/2+3/2 <\varsigma_\nu,\\
				O_p(\alpha_1 \max\{ \log \alpha_2 ^{-1}, \alpha_2 ^{\bt{(2\varsigma_\nu-\rho_\nu-3)}/\rho_\nu }   \})    &\text{if }\rho_\nu/2+3/2 \geq \varsigma_\nu. \label{eqadd11}
			\end{cases}
		\end{align}
		Since \bt{$4\gamma_{\mu}-4 >-3$}, $\alpha_2 \log \alpha_2^{-1} = o(1)$, and $\alpha_2 \log^2 \alpha_2^{-1} = o(1)$,  
		\eqref{eqthmconvrate2} may be deduced  from \eqref{eqqhatt00}, \eqref{eqadd10} and \eqref{eqadd11}.  \\[9pt]
		\noindent \textbf{Proofs of \eqref{eq:k2alpha1}-\eqref{eqahh}}: To obtain the desired results, we first need to discuss on  $\Vert (\widehat{\mathcal Q} - \mathcal Q)h_\ell \Vert$ and  $\Vert \widehat{\mathcal Q} - \mathcal Q\Vert_{\HS}$. Note that, for any $\ell$,  \begin{equation}
			2^{-1}\Vert (\widehat{\mathcal Q} - \mathcal Q)h_\ell \Vert  ^2 \leq \Vert (\widehat{\mathcal Q} - \mathcal Q_{\K_1})h_\ell \Vert ^2 + \Vert  ({\mathcal Q}_{\K_1} - \mathcal Q )h_\ell \Vert  ^2 . \label{eqqhat0}
		\end{equation}
		The second term in \eqref{eqqhat0} is bounded above as follows,  
		\begin{align}
			\Vert ({\mathcal Q}_{\K_1} - \mathcal Q)h_\ell \Vert ^2  & = \Vert \sum_{j = \K_1+1} ^\infty \mu_j  \langle \mathcal B g_j, h_\ell \rangle  \mathcal B g_j  \Vert ^2 \leq  O(1)  \sum_{j = \K_1+1} ^\infty \mu_j ^2 \Vert \mathcal B g_j \Vert ^2  \sum_{j = \K_1+1} ^\infty  j^{-2\varsigma_{\mu}} \bt{\ell ^{-2\gamma_{\mu}}}\nonumber\\
			&\leq   O(1)\mu_{\K_1+1} ^2(\K_1+1)^{-4\varsigma_{\mu}+2} \bt{\ell ^{-2\gamma_{\mu}}} \leq O_p(1)  \alpha_1^{(4\varsigma_{\mu}+\rho_{\mu}-2)/\rho_{\mu}} \bt{\ell ^{-2\gamma_{\mu}}},  \label{eqadd1}
		\end{align} 
		where the first inequality follows from the H\"older's inequality and the second is obtained because $\bt{2\gamma_{\mu}}>1$, $\sum_{j=\K_1+1} ^\infty  j^{-2\varsigma_{\mu}} \leq(\K_1+1) ^{-2\varsigma_{\mu}+1}$, and $\mu_j ^2 \leq \mu_{\K_1+1} ^2$ for $j > \K_1$.   The last inequality is  obtained using the arguments that are used to derive \eqref{eq:kalpha2}. We now focus on the first term in \eqref{eqqhat0}. Note that \begin{equation}\begin{aligned}
				4^{-1}\Vert (\widehat{\mathcal Q} - \mathcal Q_{\K_1} )h_\ell \Vert ^2 &\leq  \Vert \widehat{\mathcal C}_{vz}^\ast (\widehat{\mathcal C}_{zz})_{\K_1} ^{-1} \widehat{\mathcal C}_{vz} h_\ell   \Vert ^2 +  {\Vert \widehat{\mathcal C}_{vz} ^\ast  \widehat{\Pi}_{\K_1} \mathcal B^\ast h_\ell \Vert ^2} +  {\Vert \mathcal B\widehat{\Pi}_{\K_1}  \widehat{\mathcal C}_{vz}  h_\ell \Vert ^2} \nonumber\\ &\quad+ \Vert (\mathcal B(\widehat{\Pi}_{\K_1} \widehat{\mathcal C}_{zz} \widehat{\Pi}_{\K_1}  - \Pi_{\K_1} \mathcal C_{zz} \Pi_{\K_1} ) \mathcal B^\ast )h_\ell \Vert^2, \label{eqqhat1}
		\end{aligned}\end{equation}
		where $\Vert \widehat{\mathcal C}_{zv} \widehat{\Pi}_{\K_1} \mathcal B^\ast h_\ell  \Vert  ^2 \leq \bt{\ell ^{-2\gamma_{\mu}}}O_p(T^{-1})$.  Moreover, we have ${\Vert \mathcal B\widehat{\Pi}_{\K_1}  \widehat{\mathcal C}_{vz}  h_\ell \Vert ^2} 
		\leq  \bt{\ell ^{-2\gamma_{\mu}}} O_p(T^{-1})$ and $\Vert \widehat{\mathcal C}_{vz} ^\ast  (\widehat{\mathcal C}_{zz})_{\K_1} ^{-1} \widehat{\mathcal C}_{vz}  h_\ell \Vert ^2 \leq O_p(\alpha_1^{-1}T^{-1}) \Vert \widehat{\mathcal C}_{vz} h_\ell \Vert ^2 \leq O_p(T^{-2}\alpha_1 ^{-1}) \ell^{-\rho_\nu/2}$ since 
		\begin{equation}
			T\mathbb E[ \Vert \widehat{\mathcal C}_{vz}    h_j \Vert ^2  ] = T^{-1} \mathbb E[ \Vert \sum_{t=1} ^{T} \langle v_t , h_j \rangle z_t \Vert ^2 ] \leq  O(1) \mathbb E [ \Vert  \langle v_t,h_j \rangle z_t\Vert ^2  ]   \leq {O(1)\nu_j}\bt{\leq O(j^{-\rho_{\nu}/2}) }, \label{eqcvhat}
		\end{equation}
		where the inequalities are obtained from Assumption~\ref{assconvratetsls}; specifically, under the assumption, \bt{we have that $ \mathbb E[\Vert\langle v_t, h_j \rangle z_t\Vert ^2 ]\leq\mathbb E[\Vert \langle x_t, h_j \rangle z_t \Vert ^2] \leq  c_\circ\Vert \mathcal C_{xz} h_j\Vert ^2 \leq \Vert \Pi_{\K_1} \mathcal C_{zz} \Vert_{\op} \Vert (\mathcal C_{zz})_{\K_1} ^{-1/2} \mathcal C_{xz} h_j\Vert ^2 \leq O(1)\nu_j $}. Lastly, using the arguments used to obtain \eqref{eq001}, we can show that \bt{$\Vert  \widehat{g}_j - g_j ^s\Vert ^2 \leq O_p(T^{-1}) j^{2}$} and $\Vert \mathcal B(\widehat{g}_j - g_j ^s)\Vert ^2 \leq O_p(T^{-1})(j^{2-2\varsigma_{\mu}}+   j^{\rho_\mu+2-2\varsigma_{\mu}})$. Using this bound, we find that 
		\begin{align}
			&4^{-1}	\Vert (\mathcal B(\widehat{\Pi}_{\K_1} \widehat{\mathcal C}_{zz} \widehat{\Pi}_{\K_1}  - \Pi_{\K_1} \mathcal C_{zz} \Pi_{\K_1} ) \mathcal B^\ast )h_\ell \Vert^2\nonumber\\
			&\leq \Vert \sum_{j=1}^{\K_1} (\widehat{\mu}_j - \mu_j ) \langle \widehat{g}_j , \mathcal B^\ast h_\ell \rangle \mathcal B\widehat{g}_j \Vert ^2 + \Vert  \sum_{j=1}^{\K_1} \mu_j \langle \widehat{g}_j - g_j^s, \mathcal B^\ast h_\ell \rangle \mathcal B g_j \Vert ^2 \nonumber\\
			&\quad+\Vert  \sum_{j=1}^{\K_1} \mu_j \langle     g_j ^s, \mathcal B^\ast h_\ell \rangle \mathcal B (\widehat{g}_j - g_j^s) \Vert ^2+ \Vert  \sum_{j=1}^{\K_1} \mu_j \langle \widehat{g}_j - g_j^s, \mathcal B^\ast h_\ell \rangle \mathcal B(\widehat{g}_j - g_j^s ) \Vert ^2\nonumber\\
			&\leq \K_1 \Vert \mathcal B\Vert_{\op} ^2 \Vert \widehat{\mathcal C}_{zz} - \mathcal C_{zz} \Vert_{\op} ^2 \Vert  \mathcal B ^\ast h_\ell \Vert ^2 +  \Vert \mathcal C_{zz}\Vert_{\HS} ^2  \sum_{j=1}^{\K_1} \Vert \widehat{g}_j - g_j^s \Vert ^2 \Vert \mathcal B^\ast h_\ell \Vert ^2\Vert \mathcal B g_j\Vert ^2 \nonumber\\
			&\quad+ \Vert \mathcal C_{zz}\Vert_{\HS}^2\sum_{j=1} ^{\K} \Vert \mathcal B^\ast h_\ell\Vert ^2 \Vert \mathcal B(\widehat{g}_j - g_j^s)\Vert^2 + \Vert \mathcal C_{zz} \Vert_{\HS} ^2 \Vert \mathcal B^\ast h_\ell \Vert^2  \sum_{j=1}^{\K_1}\Vert \widehat{g}_ j - g_j^s \Vert ^4 \nonumber\\
			&\leq\bt{\ell ^{-2\gamma_{\mu}}} \alpha^{-1/\rho_\mu}O_p(T^{-1}) + \bt{\ell ^{-2\gamma_{\mu}}} O_p(T^{-1}) \sum_{j=1} ^{\K_1} \bt{j^{-2\varsigma_{\mu}+2}}\nonumber\\
			&\quad+\bt{\ell ^{-2\gamma_{\mu}}} O_p(T^{-1}) \sum_{j=1} ^{\K_1} (  j^{-2\varsigma_{\mu}+2}+  j^{\rho_\mu-2\varsigma_{\mu}+2}) + \bt{\ell ^{-2\gamma_{\mu}}} O_p(T^{-2})\sum_{j=1} ^{\K_1} j^4 \nonumber\\
			&\leq \bt{\ell ^{-2\gamma_{\mu}}} O_p(T^{-1}\max\{  \alpha_1 ^{-1/\rho_{\mu}}, \alpha_1^{-(\rho_\mu-2\varsigma_{\mu}+3)/\rho_{\mu}} \} ),\label{eqqhat2}
		\end{align} 
		where the second inequality is obtained by using Lemma 4.2 in \cite{Bosq2000} and noting that $ \Vert \sum_{j =1}^{ \K_1} \mu_j \langle \widehat{g}_j - g_j, \mathcal B^\ast h_\ell \rangle \mathcal B g_j \Vert ^2 \leq (\sum_{j =1}^{ \K_1} \mu_j ^2 )(\sum_{j =1}^{ \K_1} \Vert \mathcal B g_j \Vert ^2 \langle \widehat{g_j} - g_j ,\mathcal B^\ast h_\ell \rangle ^2)$ holds by the H\"older's inequality. The last two inequalities are deduced from  Assumption~\ref{assum1}, Assumption~\ref{assconvratetsls}, and the the fact that $\Vert \mathcal B^\ast h_\ell \Vert ^2 \leq O(1) \bt{\ell ^{-2\gamma_{\mu}}}$ and $\alpha_1 ^{-1} = \abar_1= o(T^{\rho_\mu/(2\rho_\mu+2)})$.	Then, from the results given in \eqref{eqqhat0} to \eqref{eqqhat2} and the definition of $d_T$ given in \eqref{eqdt}, we conclude the following: for any $\ell\leq\K_2$, 	 
		\begin{equation}
			\Vert (\widehat{\mathcal Q}- \mathcal Q)h_\ell \Vert  ^2 \leq     \bt{\ell ^{-2\gamma_{\mu}}} O_p(d_T), \label{eqqhat3}
		\end{equation}
		from which we also find that
		\begin{equation}
			\Vert \widehat{\mathcal Q} - \mathcal Q \Vert_{\HS}  ^2 = O_p(d_T).
			\label{eqqhat3a}
		\end{equation}
		
		We now verify \eqref{eq:k2alpha1} and \eqref{eq:k2alpha2}.	It can be shown without difficulty that $d_T = O(\alpha_1)$. Given that $\alpha_1 ^{-1} = o(T^{1/2})$ and $\alpha_2 ^{-1}\alpha_1^{1/2} =o(1)$, we find that $\alpha_2 ^{-1}d_T^{1/2} =o(1)$, from which the following is deduced:
		\begin{equation*}
			\alpha_2 = \widehat{\nu}_{\K_2}^2  - \nu_{\K_2}^2 +\nu_{\K_2} ^2 \leq \Vert \widehat{\mathcal Q} - \mathcal Q \Vert_{\op} + c_\circ\K_2^{-\rho} \leq o(1) \alpha_2 + c_\circ\K_2 ^{-\rho_\nu}.
		\end{equation*}
		Then \eqref{eq:k2alpha1} follows from the above. Similarly as in \eqref{eq:kalpha2}, it can be shown that \begin{equation*}
			(c_\circ \rho)^{-1}{(\K_2+1)^{-\rho_{\nu} }} \leq {\nu}_{\K_2+1} ^2  = \nu_{\K_2+1} ^2 - \widehat{\nu}_{\K_2+1} ^2  +\widehat \nu_{\K_2+1}   \leq (1 + o_p(1)) \alpha_2,  
		\end{equation*}
		and thus  we find that \eqref{eq:k2alpha2} holds. 
		
		We then show \eqref{eqadd01a}. To this end, it should first be noted that, under the employed assumptions,  \eqref{eqpf02} holds if $\widehat{\lambda}_{j}$ (resp.\ $\lambda_{\ell}$) is replaced by $\widehat{\nu}_j$ (resp.\ $\nu_\ell$) . Moreover, note that the eigenfunctions of $\mathcal Q ^2$ and $\widehat{\mathcal Q} ^2$ are, respectively, equivalent to those of $\mathcal Q$ and $\widehat{\mathcal Q}$. Therefore, by applying the arguments that are used in our proof of Theorem~\ref{thm:convrate}, we can show that \begin{equation}
			8^{-1}\Vert \widehat{h}_j - h_j^s \Vert ^2 
			\leq\sum_{\ell \neq j} (\widehat{\nu}_j ^2 - \nu_\ell ^2 )^{-2}  {\nu}_\ell ^2  \langle (\widehat{\mathcal Q}  -\mathcal Q ) \widehat{h}_j , h_\ell \rangle ^2  +\sum_{\ell \neq j} (\widehat{\nu}_j ^2 - \nu_\ell ^2 )^{-2} \widehat{\nu}_j^2  \langle (\widehat{\mathcal Q}  -\mathcal Q ) \widehat{h}_j , h_\ell \rangle ^2    .\label{eqhhat1}
		\end{equation}
		From similar arguments used to derive \eqref{eqpf03}, the first term in \eqref{eqhhat1} is bounded above as follows:
		\begin{equation}
			\sum_{\ell \neq j} (\widehat{\nu}_j ^2 - \nu_\ell ^2 )^{-2}  {\nu}_\ell ^2  \langle (\widehat{\mathcal Q}  -\mathcal Q ) \widehat{h}_j , h_\ell \rangle ^2  
			\leq 2 \widetilde{\Delta}_{1j} \Vert \widehat{h}_j - h_j^s\Vert ^2 + 2 \sum_{\ell \neq j} (\widehat{\nu}_j ^2 - \nu_\ell ^2 )^{-2} \nu_\ell ^2 \langle  h_j , (\widehat{\mathcal Q} - \mathcal Q) h_\ell \rangle ^2,
		\end{equation}
		where $ \widetilde{\Delta}_{1j}  = \sum_{\ell \neq j}(\widehat{\nu}_j ^2 - \nu_\ell ^2 )^{-2} \nu_\ell ^2 \Vert (\widehat{\mathcal Q} - \mathcal Q) h_\ell \Vert ^2$. 
		Moreover, by using similar arguments that are used to obtain \eqref{eqpf01a}, we can show that
		\begin{equation}
			\sum_{\ell \neq j} (\widehat{\nu}_j ^2 - \nu_\ell ^2 )^{-2} \widehat{\nu}_j^2  \langle (\widehat{\mathcal Q}  -\mathcal Q ) \widehat{h}_j , h_\ell \rangle ^2  
			\leq 2\sum_{\ell \neq j} (\widehat{\nu}_j ^2 - \nu_\ell ^2 )^{-2} \widehat{\nu}_j^2  \langle (\widehat{\mathcal Q}  -\mathcal Q )  {h}_j  , h_\ell \rangle ^2  +  {2} (\widetilde{\Delta}_{1j} + \widetilde{\Delta}_{2j}) \Vert \widehat{h}_j - h_j^s \Vert ^2, 
		\end{equation}
		where $\widetilde{\Delta}_{2j} =  \sum_{\ell \neq j}(\widehat{\nu}_j ^2 - \nu_\ell ^2 )^{-1}   \Vert (\widehat{\mathcal Q} - \mathcal Q) h_\ell \Vert ^2 $. 	We then use the results given in \eqref{eqqhat3} and \eqref{eqqhat3a} to obtain the following bounds of $\widetilde{\Delta}_{1j}$ and $\widetilde{\Delta}_{2j}$: for $j=1,\ldots, \K_2$, \begin{equation}
			\widetilde{\Delta}_{1j} \leq  O_p(d_T)\sum_{\ell \neq j} (\nu_j ^2 -\nu_\ell ^2)^{-2} \nu_\ell^2 \bt{\ell ^{-2\gamma_{\mu}}} \leq \bt{j^{\rho_{\nu} + 2 - 2\gamma_{\mu}}}O_p(d_T)  \leq   \bt{O_p(\alpha_2 ^{-(2\gamma_{\mu}-2-\rho_{\nu})/\rho_\nu}d_T)} = o_p(1),   
		\end{equation} 
		where the second inequality follows from Lemma~\ref{lem:convrate}\ref{lemconvrate8} and \bt{and the last equality follows from that  $\gamma_{\mu} \leq 1+\frac{3}{2}\rho_{\nu}$ and $\alpha_2^{-1}d_T^{1/2} = o_p(1)$}. Similarly, for $j=1,\ldots, \K_2$,   
		\begin{equation} \label{eqhhat1end}
			\widetilde{\Delta}_{2j}   
			\leq   {\max_{1\leq \ell \leq \K_2}} (\nu_j ^2 - \nu_\ell ^2)^{-1}     \sum_{\ell \neq j} \Vert (\widehat{\mathcal Q} - \mathcal Q)h_\ell \Vert ^2   \leq j^{1+\rho_{\nu}} O_p(d_T) \leq O(\alpha_2^{-(1+\rho_{\nu})/\rho_{\nu}}d_T) = o_p(1). 
		\end{equation} 
		From \eqref{eqhhat1}-\eqref{eqhhat1end}, we have
		\begin{equation}
			\Vert \widehat{h}_j - h_j^s \Vert ^2 \leq O(1) (1+o_p(1)) \widetilde{\Delta}_{3j} \label{eqadd00}
		\end{equation}
		for $j=1,\ldots,\K_2,$		where \begin{equation}
			\widetilde	\Delta_{3j} = \sum_{\ell \neq j} (\widehat{\nu}_j ^2 - \nu_\ell ^2 )^{-2} \widehat{\nu}_j^2  \langle (\widehat{\mathcal Q}  -\mathcal Q )  {h}_j , h_\ell \rangle ^2 +\sum_{\ell \neq j} (\widehat{\nu}_j ^2 - \nu_\ell ^2 )^{-2}  {\nu}_\ell^2  \langle (\widehat{\mathcal Q}  -\mathcal Q )  {h}_j , h_\ell \rangle ^2.   \label{eqadd01}
		\end{equation}
		We will analyze the above term using the decomposition $\widehat{\mathcal Q}  -\mathcal Q  = \widehat{\mathcal Q} - \mathcal Q_{\K_1} + \mathcal Q_{\K_1} - \mathcal Q$. Note that
		\begin{equation}
			\langle (\mathcal Q_{\K_1} - \mathcal Q) h_j , h_\ell \rangle ^2  \leq \left(\sum_{i = \K_1+1} ^\infty \mu_i ^2  \langle g_i , \mathcal B^\ast h_j \rangle ^2\right) \left(\sum_{i = \K_1+1} ^\infty \langle g_i , \mathcal B^\ast h_\ell \rangle ^2\right)   
			\leq O_p(\alpha_1)\bt{\ell^{-2\gamma_{\mu}} j^{-2\gamma_{\mu}}}, \label{eqqhat00}
		\end{equation}
		where the first inequality follows from the H\"older's inequality, and the second is deduced from Assumption~\ref{assconvratetsls} and similar arguments used to derive \eqref{eq:kalpha2}. 
		From \eqref{eqqhat00}, Lemma \ref{lem:convrate}\ref{lemconvrate8} \bt{and a result similar to \eqref{eqpf02}}, we find that  \begin{equation}
			\sum_{\ell \neq j} (\widehat{\nu}_j ^2 - \nu_\ell ^2 )^{-2}  {\nu}_\ell^2 \langle (\mathcal Q_{\K_1} - \mathcal Q) h_j , h_\ell \rangle ^2 \leq O_p(\alpha_1)\bt{j^{\rho_{\nu}-4\gamma_{\mu}+2}}. \label{eqqhat01}
		\end{equation}
		Similarly, for $j=1,\ldots, \K_2$, we have\begin{align} 
			&\sum_{\ell \neq j} (\widehat{\nu}_j ^2 - \nu_\ell ^2 )^{-2}  \widehat{\nu}_j^2 \langle (\mathcal Q_{\K_1} - \mathcal Q) h_j , h_\ell \rangle ^2 \nonumber\\
			& \leq |\widehat{\nu}_j ^2 -\nu_j ^2| \sum_{\ell \neq j} (\widehat{\nu}_j ^2 - \nu_\ell ^2)^{-2} \nu_j ^2 \nu_j ^{-2} \langle (\mathcal Q_{\K_1} - \mathcal Q) h_j , h_\ell \rangle ^2 + \sum_{\ell \neq j} (\widehat{\nu}_j ^2 - \nu_\ell ^2)^{-2} \nu_j ^2   \langle (\mathcal Q_{\K_1} - \mathcal Q) h_j , h_\ell \rangle ^2 \nonumber\\
			&\leq    (O_p(d_T^{1/2})j^{\rho_\nu} +1)  O_p(\alpha_1) \bt{j^{\rho_{\nu}-4\gamma_{\mu}+2}}  ,
			\label{eqqhat02}
		\end{align}
		where the second inequality follows from that $|\widehat{\nu}_j ^2 - \nu_j^2 | \leq \Vert \widehat{\mathcal Q}^2 - \mathcal Q ^2 \Vert_{\op} \leq O_p(1)\Vert  \widehat{\mathcal Q} - \mathcal Q\Vert_{\op}$.  Given that   $  \alpha_2^{-1}d_T^{1/2} =o(1)$, \eqref{eqqhat02} is bounded above by $O_p(\alpha_1)\bt{j^{\rho_{\nu}-4\gamma_{\mu}+2}}$. Next, we will obtain an upper bound of $\langle (\widehat{\mathcal Q} - \mathcal Q_{\K_1 })h_j, h_\ell \rangle ^2$. Note that \begin{align} 
			4^{-1}\langle  (\widehat{\mathcal Q} - \mathcal Q_{\K_1} )h_\ell, h_j \rangle  ^2 	&\leq  \langle  \widehat{\mathcal C}_{vz}^\ast (\widehat{\mathcal C}_{zz})_{\K_1} ^{-1} \widehat{\mathcal C}_{vz} h_\ell, h_j \rangle ^2 +  \langle  \widehat{\mathcal C}_{vz} ^\ast  \widehat{\Pi}_{\K_1} \mathcal B^\ast h_\ell , h_j \rangle  ^2 +  \langle  \mathcal B\widehat{\Pi}_{\K_1}  \widehat{\mathcal C}_{vz}  h_\ell ,h_j\rangle ^2 \nonumber\\ &\quad+ \langle  (\mathcal B(\widehat{\Pi}_{\K_1} \widehat{\mathcal C}_{zz} \widehat{\Pi}_{\K_1}  - \Pi_{\K_1} \mathcal C_{zz} \Pi_{\K_1} ) \mathcal B^\ast )h_\ell , h_j \rangle ^2 \nonumber\\
			&  \leq \alpha_1 ^{-1} \nu_j   \nu_\ell   O_p(T^{-2})  + \nu_j   \bt{\ell ^{-2\gamma_{\mu}}}O_p(T^{-1}) + \nu_\ell   \bt{j^{-2\gamma_{\mu}}} O_p(T^{-1}) \nonumber\\
			&\quad + \langle  (\mathcal B(\widehat{\Pi}_{\K_1} \widehat{\mathcal C}_{zz} \widehat{\Pi}_{\K_1}  - \Pi_{\K_1} \mathcal C_{zz} \Pi_{\K_1} ) \mathcal B^\ast )h_\ell , h_j \rangle ^2,  \label{eqqhat11}
		\end{align}
		where  the last inequality is obtained by using \eqref{eqcvhat}. The last term in \eqref{eqqhat11} satisfies the following: \begin{align}
			&4^{-1} \langle  (\mathcal B(\widehat{\Pi}_{\K_1} \widehat{\mathcal C}_{zz} \widehat{\Pi}_{\K_1}  - \Pi_{\K_1} \mathcal C_{zz} \Pi_{\K_1} ) \mathcal B^\ast )h_\ell , h_j \rangle ^2 	\nonumber\\
			&\leq ( \sum_{i=1} ^{\K_1} (\widehat{\mu}_i - \mu_i ) \langle \widehat{g}_i , \mathcal B^\ast h_\ell \rangle \langle \mathcal B\widehat{g}_i, h_j \rangle  )^2 + (  \sum_{i=1} ^{\K_1} \mu_i \langle \widehat{g}_i - g_i, \mathcal B^\ast h_\ell \rangle \langle \mathcal B \widehat g_i , h_j \rangle ) ^2 \nonumber\\
			&\quad +(  \sum_{i=1} ^{\K_1} \mu_i \langle     g_i, \mathcal B^\ast h_\ell \rangle \langle \mathcal B (\widehat{g}_i - g_i) , h_j \rangle ) ^2  \nonumber\\
			&\leq  \max_{1\leq i \leq \K_1} |\widehat{\mu}_i - \mu_i|^2 \Vert \widehat{\Pi}_{\K_1} \mathcal B^\ast h_\ell \Vert ^2 \Vert \widehat{\Pi}_{\K_1} \mathcal B^\ast h_j\Vert ^2 +  \sum_{i=1} ^{\K_1} \mu_i ^2 \Vert \widehat{g}_i - g_i \Vert ^2 \Vert \mathcal B^\ast h_\ell \Vert ^2 \Vert \widehat{\Pi}_{\K_1} \mathcal B^\ast h_j \Vert ^2  \nonumber\\
			&\quad\quad+  \sum_{i=1} ^{\K_1}\mu_i ^2 \Vert \widehat{g}_i - g_i \Vert ^2 \Vert \mathcal B^\ast h_\ell \Vert ^2 \Vert {\Pi}_{\K_1} \mathcal B^\ast h_j \Vert ^2 \nonumber\\
			&\leq \bt{\ell ^{-2\gamma_{\mu}} j^{-2\gamma_{\mu}}} O_p (T^{-1}\alpha_1 ^{-1/\rho_\mu}),\label{eqqhat12}
		\end{align}
		where 
		the last inequality follows from Assumptions \ref{assum1a} and \ref{assconvratetsls}. Combining the results given in \eqref{eqqhat11} and \eqref{eqqhat12}, we find that\begin{align*}
			&4^{-1}\langle  (\widehat{\mathcal Q} - \mathcal Q_{\K_1} )h_\ell, h_j \rangle  ^2 \notag \\ &\leq O_p(T^{-1}) (\alpha_1 ^{-1}T^{-1} j^{-\rho_{\nu}/2} \ell ^{-\rho_{\nu}/2} + \bt{j^{-\rho_{\nu}/2} \ell^{-2\gamma_{\mu}} + j^{-2\gamma_{\mu}} \ell^{-\rho_{\nu}/2}} +  \alpha_1 ^{-1/\rho_\mu}\bt{\ell ^{-2\gamma_{\mu}} j^{-2\gamma_{\mu}}})  \nonumber\\
			&\leq O_p(T^{-1})  \alpha_1 ^{-1/\rho_\mu}\bt{\ell ^{-2\gamma_{\mu}} j^{-2\gamma_{\mu}}} + \bt{O_p(T^{-1}) (j^{-\rho_{\nu}/2} \ell^{-2\gamma_{\mu}} + j^{-2\gamma_{\mu}} \ell^{-\rho_{\nu}/2})} . 
		\end{align*} 
		Together with Lemma~\ref{lem:convrate}\ref{lemconvrate8}, this implies that \begin{align}
			\sum_{\ell \neq j} (\widehat{\nu}_j ^2 - \nu_\ell ^2 )^{-2}  {\nu}_\ell^2\langle  (\widehat{\mathcal Q} - \mathcal Q_{\K_1} )h_\ell, h_j \rangle  ^2 
			&\leq O_p(T^{-1}\alpha_1 ^{-1/\rho_\mu}) \bt{j^{\rho_{\nu} -4\gamma_{\mu}+2} + O_p(T^{-1})j^{\rho_{\nu}/2-2\gamma_{\mu}+2}} \notag \\
& \leq O_p(T^{-1}\alpha_1 ^{-1/\rho_\mu}) j^{\rho_{\nu}-4\gamma_{\mu}+2},  \label{eqadd02}
		\end{align}
		where the inequalities follow from Lemma \ref{lem:convrate}\ref{lemconvrate8} and \bt{the fact that the first term is dominant under the condition $\gamma_{\mu} \leq \rho_{\nu}/4+1/2$}. In addition, from the arguments that are used to derive \eqref{eqqhat02} and the fact that $\alpha_{2}^{-1}d_T^{1/2} = o(1)$, the following may be deduced: for $j=1,\ldots, \K_2$, \begin{align}
			\sum_{\ell \neq j} (\widehat{\nu}_j ^2 - \nu_\ell ^2 )^{-2}  \widehat{\nu}_j^2\langle  (\widehat{\mathcal Q} - \mathcal Q_{\K_1} )h_\ell, h_j \rangle  ^2  &\leq (O_p(\alpha_{2}^{-1}d_T^{1/2})+1) O_p(T^{-1}\alpha_1 ^{-1/\rho_{\mu}} )\bt{j^{\rho_{\nu}-4\gamma_{\mu}+2}} \notag \\ &=   O_p(T^{-1}\alpha_1 ^{-1/\rho_\mu}) \bt{j^{\rho_{\nu}-4\gamma_{\mu}+2}}.  \label{eqadd03}
		\end{align} 
		From  \eqref{eqadd00}, \eqref{eqadd01}, \eqref{eqqhat01}, \eqref{eqqhat02}, \eqref{eqadd02}, \eqref{eqadd03}, and the decomposition $\widehat{\mathcal Q}  -\mathcal Q  = \widehat{\mathcal Q} - \mathcal Q_{\K_1} + \mathcal Q_{\K_1} - \mathcal Q$, we conclude that 
		\begin{equation*}
			\Vert \widehat{h}_j - h_j^s \Vert ^2 \leq (1+o_p(1))\widetilde{\Delta}_{3j}\leq O_p(\alpha_1) \bt{j^{\rho_{\nu}-4\gamma_{\mu}+2}} + O_p(T^{-1}\alpha_1 ^{-1/\rho_\mu}) \bt{j^{\rho_{\nu}-4\gamma_{\mu}+2}} \leq O_p(\alpha_1 )\bt{j^{\rho_{\nu}-4\gamma_{\mu}+2}},
		\end{equation*} 
		where the last inequality follows from $\alpha_1 ^{-1} = o(T^{\rho_\mu/(2\rho_\mu+2)})$. This completes our proof of \eqref{eqadd01a}. 
		
		Lastly, to show \eqref{eqahh}, we note that	
		\begin{equation}
			\mathcal A (\widehat{h}_j - h_j ^s) = \sum_{\ell \neq j }(\widehat{\nu}_j^2  -\nu_\ell^2  )^{-1} \langle (\widehat{\mathcal Q}^2  - \mathcal Q^2)  \widehat{h}_j  , h_\ell \rangle \mathcal A h_\ell + \langle \widehat{h}_j - h_j ^s , h_j \rangle \mathcal A h_j, \label{eqahhat0}
		\end{equation}
		where $\Vert \langle \widehat{h}_j - h_j ^s ,h_j \rangle \mathcal A h_j  \Vert ^2 \leq O_p(\alpha_1)\bt{j^{\rho_{\nu}-2\varsigma_\nu-4\gamma_{\mu}+2}} $. The first term in \eqref{eqahhat0} is bounded above as follows,
		{\allowdisplaybreaks\begin{align}
			&( \sum_{\ell \neq j} (\widehat{\nu}_j^2  -\nu_\ell  ^2)^{-1} \langle (\widehat{\mathcal Q}^2  - \mathcal Q^2)  \widehat{h}_j  , h_\ell \rangle \mathcal A h_\ell)^2 \leq ( \sum_{\ell \neq j} |\widehat{\nu}_j ^2 - \nu_\ell^2| ^{-1} |\langle  (\widehat{\mathcal Q}^2- \mathcal Q^2) \widehat{h}_j ,  h_\ell \rangle|\Vert  \mathcal A h_\ell \Vert  )^2   \nonumber\\ 
			&\leq O(1)\Vert \widehat{\mathcal Q} - \mathcal Q\Vert_{\op} ^2 \left( \sum_{\ell \neq j} ( \widehat{\nu}_j ^2  -\nu _\ell ^2 )^{-2}   \nu_\ell ^2 \ell^{-2\varsigma_\nu} +  \sum_{\ell \neq j} ( \widehat{\nu}_j ^2  -\nu _\ell ^2 )^{-2}   \widehat\nu_j ^2 \ell^{-2\varsigma_\nu} \right) \nonumber\\
			&  {\leq O_p(d_T) \left( \sum_{\ell \neq j} ( \widehat{\nu}_j ^2  -\nu _\ell ^2 )^{-2}   \nu_\ell ^2 \ell ^{-2\varsigma_\nu} +  \sum_{\ell \neq j} ( \widehat{\nu}_j ^2  -\nu _\ell ^2 )^{-2}   (\widehat\nu_j ^2 - \nu_j ^2 ) \ell ^{-2\varsigma_\nu}+  \sum_{\ell \neq j} ( \widehat{\nu}_j ^2  -\nu _\ell ^2 )^{-2}     \nu_j ^2   \ell ^{-2\varsigma_\nu} \right) }\nonumber \\
			& \leq  O_p(d_T)\left(  j^{\rho_\nu-2\varsigma_\nu+2}+  (O_p(d_T^{1/2}) \nu_j ^{-2}+1)\sum_{\ell \neq j} ( \widehat{\nu}_j ^2  -\nu _\ell ^2 )^{-2}   \nu_j ^2  \ell ^{-2\varsigma_\nu} \right)\nonumber\\
			& \leq (1+o_p(1))O_p(d_T) j^{\rho_\nu-2\varsigma_\nu+2}.\label{eqahhat01}
		\end{align}}
		Combining \eqref{eqahhat0} and \eqref{eqahhat01}, we obtain \eqref{eqahh} as desired. 
		\qed 

				
				\subsubsection*{Proof of Theorem \ref{thm4:convrate}}
In this proof, we allow the case where \(\rho_{\nu}/2+2\geq\varsigma_{\nu}+\delta_{\zeta}\), thereby encompassing the complementary result provided in Section \ref{sec_thm_conti2} as a special case. 
				The whole proof is divided into two parts.\\ 		
				\noindent	{\textbf{1. Proof of the convergence results}}: 			
				We need an upper bound of $\langle \widehat{h}_j - h_j ^s,\zeta\rangle $, which is importantly used in the following discussion. Using the expansion in \eqref{eq:b1}, we find that  \begin{equation*}
					\langle \widehat{h}_j - h_j ^s,\zeta\rangle  =\sum_{\ell \neq j} (\widehat{\nu}_j ^2 - \nu_\ell^2) ^{-1} \langle  (\widehat{\mathcal Q}- \mathcal Q) \widehat{h}_j ,  h_\ell \rangle \langle h_\ell ,\zeta\rangle  + \langle \widehat{h}_j - h_j ^s , h_j ^s \rangle \langle h_j ^s ,\zeta\rangle. 
				\end{equation*}
				Note that $	( \langle \widehat{h}_j - h_j^s , h_j^s \rangle  \langle h_j ^s ,\zeta\rangle) ^2 \leq O_p(\alpha_1) \bt{j^{\rho_{\nu}-4\gamma_{\mu}-2\delta_{\zeta}+2}}$ for $j=1,\ldots,\K_2$, because of \eqref{eqadd01a}. Moreover, using similar arguments that are used to derive \eqref{eqb14} and \eqref{eqahhat01}, we find that   \begin{align}
					&( \sum_{\ell \neq j} (\widehat{\nu}_j ^2 - \nu_\ell^2) ^{-1} \langle  (\widehat{\mathcal Q}- \mathcal Q) \widehat{h}_j ,  h_\ell \rangle \langle h_\ell ,\zeta\rangle  )^2  
					\leq (1+o_p(1))O_p(d_T )j^{\rho_{\nu}-2\delta_{\zeta}+2}. \label{eqb14hh} 
				\end{align} 
				Hence, we conclude that, for $j=1,\ldots, \K_2$, \begin{equation}
					\langle \widehat{h}_j - h_j ^s,\zeta\rangle^2 \leq O_p(\alpha_1) \bt{j^{\rho_{\nu}-4\gamma_{\mu}-2\delta_{\zeta}+2}} +  O_p(d_T )j^{\rho_{\nu}+2-2\delta_{\zeta}} \leq O_p(\alpha_1) j^{\rho_\nu-2\delta_{\zeta}+2}, \label{eqbhh}
				\end{equation}
				where the last inequality follows from the fact that $d_T \alpha_1 ^{-1} = o(1)$ and \bt{$j^{-\gamma_{\mu}} \leq 1$}.
				
				Using the result given in \eqref{eqbhh}, we will show that 
				\begin{align}\label{eqdesire1}
					\Vert (\widehat{\mathcal C}_{zz})_{\K_1} ^{-1} \widehat{\mathcal C}_{xz} \widehat{\mathcal Q}_{\K_2} ^{-1}\zeta - (\mathcal C_{zz})_{\K_1} ^{-1} \mathcal C_{xz} \mathcal Q _{\K_2} ^{-1}\zeta  \Vert &= o_p(1), \\
					\Vert \widehat{\mathcal Q}_{\K_2} ^{-1}\zeta - \mathcal Q_{\K_2} ^{-1}\zeta   \Vert &= o_p(1). \label{eqdesire2}
				\end{align}
				To show \eqref{eqdesire1}, note that 
				\begin{equation}
					\Vert (\widehat{\mathcal C}_{zz} )_{\K_1}^{-1} \widehat{\mathcal C}_{xz} \widehat{\mathcal Q}_{\K_2} ^{-1}\zeta -(\mathcal C_{zz})_{\K_1} ^{-1} \mathcal C_{xz} \mathcal Q_{\K_2} ^{-1} \zeta\Vert \leq   	\Vert (\widehat{\mathcal C}_{zz} )_{\K_1}^{-1} \widehat{\mathcal C}_{xz} (\widehat{\mathcal Q}_{\K_2} ^{-1}  - \mathcal Q_{\K_2} ^{-1})\zeta\Vert + \Vert ( (\widehat{\mathcal C}_{zz} )_{\K_1}^{-1} \widehat{\mathcal C}_{xz} - (\mathcal C_{zz})_{\K_1} ^{-1} \mathcal C_{xz})\mathcal Q_{\K_2} ^{-1} \zeta\Vert.  \label{tmpeq1}
				\end{equation}
				Because $ (\widehat{\mathcal C}_{zz} )_{\K_1}^{-1} \widehat{\mathcal C}_{xz} = \widehat{\Pi}_{\K_1}\mathcal B^\ast + (\widehat{\mathcal C}_{zz} )_{\K_1}^{-1} \widehat{\mathcal C}_{vz}$ and  $\| (\widehat{\mathcal C}_{zz})_{\K_1} ^{-1} \widehat{\mathcal C}_{vz} \|_{\op}=O_p(\alpha_1^{-1/2}T^{-1/2})$,  the first term in \eqref{tmpeq1} satisfies that\begin{equation} \label{tmpeq2}
					\Vert (\widehat{\mathcal C}_{zz} )_{\K_1}^{-1} \widehat{\mathcal C}_{xz} (\widehat{\mathcal Q}_{\K_2} ^{-1}  - \mathcal Q_{\K_2} ^{-1})\zeta\Vert \leq  O_p(1)\Vert (\widehat{\mathcal Q}_{\K_2} ^{-1}  - \mathcal Q_{\K_2} ^{-1}) \zeta\Vert .
				\end{equation}
				Moreover, under the employed assumptions, the following holds:
				\begin{align}
					\Vert \mathcal Q_{\K_2} ^{-1} \zeta\Vert ^2 &= \Vert \sum_{j=1} ^{\K_2} \nu_j^{-1} \langle h_j , \zeta\rangle h_j\Vert ^2  
					\leq O(1) \sum_{j=1} ^{\K_2} j^{\rho_\nu -2\delta_\zeta} \leq O (\max\{  \alpha_2 ^{-1/\rho_\nu}, \alpha_2 ^{-(\rho_\nu-2\delta_\zeta+1)/\rho_\nu}\}).
				\end{align}
				Given that $\alpha_1^{1/2} \alpha_2 ^{-1} = o(1)$, we have $\alpha_1 \alpha_2 ^{-1/\rho_\nu} \leq \alpha_1 \alpha_2 ^{-1} \alpha_2 ^{(\rho_\nu-1)/\rho_\nu}  = o(1)$ and $\alpha_1 \alpha_2 ^{-(\rho_\nu-2\delta_\zeta+1)/\rho_\nu} = \alpha_1 \alpha_2 ^{-1} \alpha_2 ^{(2\delta_\zeta-1)/\rho_\nu} = o(1)$. This implies that $\Vert \mathcal Q_{\K_2} ^{-1} \zeta\Vert ^2  \leq o_p(\alpha_1 ^{-1})$. Note also that 
				\begin{align}
					((\widehat{\mathcal C}_{zz} )_{\K_1}^{-1} \widehat{\mathcal C}_{xz} - (\mathcal C_{zz})_{\K_1} ^{-1} \mathcal C_{xz})\mathcal Q_{\K_2} ^{-1} \zeta = (\widehat{\Pi}_{\K_1} - \Pi_{\K_1})\mathcal B^\ast \mathcal Q_{\K_2} ^{-1} \zeta  + (\widehat{\mathcal C}_{zz})_{\K_1} ^{-1} \widehat{\mathcal C}_{vz}\mathcal Q_{\K_2} ^{-1} \zeta,
				\end{align}
				where the second term on the right hand side is $o_p(1)$ since $\| (\widehat{\mathcal C}_{zz})_{\K_1} ^{-1} \widehat{\mathcal C}_{vz} \|_{\op}=O_p(\alpha_1^{-1/2}T^{-1/2})$, $\Vert \mathcal Q_{\K_2} ^{-1} \zeta\Vert  \leq o_p(\alpha_1 ^{-1/2})$, and $\alpha_1^{-1}T^{-1/2} = o(1)$.
				Furthermore, since $\|\widehat{\Pi}_{\K_1} - \Pi_{\K_1}\|_{\op} \leq O(1)\sum_{j=1}^{\K_1} \|\widehat{g}_j - g_j ^s\| \leq O_p(T^{-1/2}\K_1^2) = O_p(T^{-1/2}\alpha_1^{-2/\rho_{\mu}})$ and $\alpha_1^{-1}=o(T^{\rho_{\mu}/(2\rho_{\mu}+2)})$, we find that 
				\begin{align}\label{thmeq3}
					\|(\widehat{\Pi}_{\K_1} - \Pi_{\K_1})\mathcal B^\ast \mathcal Q_{\K_2} ^{-1}\zeta\| &\leq \|\widehat{\Pi}_{\K_1} - \Pi_{\K_1}\|_{\op}\|\mathcal B\|_{\op} \|\mathcal Q_{\K_2} ^{-1}\zeta\| = O_p(T^{-1/2}\alpha_1^{-2/\rho_{\mu}-1/2})\nonumber \\
					&= O_p(T^{-1/2}\alpha_1^{-(\rho_{\mu}+1)/\rho_{\mu}})\alpha_1^{(\rho_{\mu}/2-1)/\rho_{\mu}} =o_p(1).
				\end{align}
				From \eqref{tmpeq2}-\eqref{thmeq3}, it is deduced that \eqref{eqdesire2} implies \eqref{eqdesire1} and thus we only need to show \eqref{eqdesire2} for the desired results. From a similar decomposition to that given in \eqref{eqpf0001a0}, it can be shown that \eqref{eqdesire2} holds if the following terms are all $o_p(1)$: $\|\sum_{j=1}^{\K_2} ({\nu}_j^{-1}-\hat{\nu}_j^{-1})  \langle {h}_j , \zeta \rangle  \widehat{h}_j \| $, $\|\sum_{j=1}^{\K_2} ({\nu}_j^{-1}-\hat{\nu}_j^{-1})  \langle \widehat{h}_j - h_j ^s , \zeta \rangle  \widehat  {h}_j \|$, $	\Vert \sum_{j=1} ^{\K}  {\nu}_j^{-1} \langle   h_j  , \zeta \rangle (\widehat{h}_j  - h_j ^s )\Vert$, and 	$\Vert \sum_{j=1} ^{\K}  {\nu}_j ^{-1} \langle \widehat{h}_j - h_j ^s, \zeta \rangle \widehat{h}_j\Vert$.	
				
				As in \eqref{eqbb1} and \eqref{eqbb2}, we obtain the following: 
				{\allowdisplaybreaks\begin{align}
					&\|\sum_{j=1}^{\K_2} ({\nu}_j^{-1}-\hat{\nu}_j^{-1})  \langle {h}_j , \zeta \rangle  \widehat{h}_j \| ^2  = \sum_{j=1} ^{\K_2} ({\nu}_j^{-1}-\hat{\nu}_j^{-1})^2 \langle h_j, \zeta\rangle ^2 	\leq \sum_{j=1} ^{\K_2} \frac{(\nu_j ^2 -\widehat{\nu}_j ^2 )^2}{  {\nu}_{j} ^2 ( \widehat {\nu}_j ^2 + \nu_j \widehat {\nu}_j)^2 } c_\zeta j^{-2\delta_{\zeta}} \nonumber\\
					&\leq \sum_{j=1} ^{\K_2} \frac{(\nu_j ^2 -\widehat{\nu}_j ^2 )^2 }{\nu_{j} ^4 \widehat{\nu}_{j} ^2 }c_{\zeta} j^{-2\delta_{\zeta} } 
					\leq   O_p(d_T \max \{ \alpha_2^{-(3\rho_{\nu}-2\delta_{\zeta} +1)/\rho_{\nu}} ,  \alpha_2 ^{-(1+\rho_\nu)/\rho_\nu}\} ) 
					,\label{eqbb1hh} \\
					&\|\sum_{j=1}^{\K_2} ({\nu}_j^{-1}-\hat{\nu}_j^{-1})  \langle \widehat{h}_j - h_j ^s , \zeta \rangle  \widehat  {h}_j \| ^2   =O_p(d_T \alpha_2 ^{-1}) \sum_{j=1} ^{\K_2} j^{2\rho_\nu} \langle \widehat{h}_j - h_j ^s ,\zeta\rangle ^2 \leq O_p(d_T \alpha_1 \alpha_2 ^{-1}) \sum_{j=1} ^{\K_2} j^{3\rho_\nu-2\delta_{\zeta}+2}\nonumber\\
					&\leq O_p(d_T \alpha_1  ) \max\{ \alpha_2 ^{-(4\rho_\nu -2\delta_{\zeta} +3)/\rho_\nu} \alpha_2 ^{-(1+\rho_\nu)/\rho_\nu }   \}. \label{eqbb2hh}
				\end{align}}
				Under the conditions on $\alpha_1$ and $\alpha_2$ given in {Theorem \ref{thm4:convrate}},  we have  $d_T = O(\alpha_1)$, and $\alpha_1^{1/2}\alpha_2^{-1} = o(1)$. These imply that the right hand sides of \eqref{eqbb1hh} and \eqref{eqbb2hh} are $o_p(1)$. 
				We then use the arguments that are used to show \eqref{eqbb4} and find that
				\begin{equation}
					\Vert \sum_{j=1} ^{\K_2}  {\nu}_j^{-1} \langle   h_j  , \zeta \rangle (\widehat{h}_j  - h_j ^s )\Vert   \leq  O_p(\alpha_1^{1/2})\sum_{j=1} ^{\K_2} \bt{j^{-\delta_{\zeta}+\rho_{\nu}-2\gamma_{\mu} +1}}  \leq O_p(\alpha_1^{1/2} \max\{  \alpha_2 ^{-1/\rho_\nu} , \alpha_2 ^{\bt{(\delta_\zeta -\rho_\nu+2\gamma_{\mu}-2)}/\rho_\nu}  \}  ).
					\label{eqbb4hh}
				\end{equation}  
				Due to $\alpha_1 ^{1/2} \alpha_2 ^{-1} = o(1)$ and \eqref{eq001thm}, we can show that the above term is $o_p(1)$. We then note that
				\begin{equation}
					\Vert \sum_{j=1} ^{\K_2}  {\nu}_j ^{-1} \langle \widehat{h}_j - h_j ^s, \zeta \rangle \widehat{h}_j\Vert ^2 
					\leq O_p(\alpha_1)\sum_{j=1} ^{\K_2} j^{2\rho_{\nu}-2\delta_{\zeta}+2} \leq O_p(\alpha_1  \max\{  \alpha_2 ^{-1/\rho_\nu} , \alpha_2 ^{(2\delta_\zeta -2\rho_\nu-3)/\rho_\nu}  \}  )   ,  \label{eqbb5hh}
				\end{equation}		 
				where the inequalities can be deduced from \eqref{eqbhh}. Because of \eqref{eq001thm}, the right hand side of \eqref{eqbb5hh} is also $o_p(1)$. Thus, by combining the results given in \eqref{eqbb1hh}-\eqref{eqbb5hh}, we conclude that \eqref{eqdesire2} holds. Then the results given in Theorem \ref{thm2wa}.\ref{thm2waa} and \ref{thm2wa}.\ref{thm2wab} immediately follow from \eqref{eqdesire1} and \eqref{eqdesire2}, and hence the details are omitted. \\[10pt]
				\noindent	{\textbf{2. Analysis on the regularization bias}}: 			
				We next focus on the regularization bias term, $\|\mathcal A(\widetilde{\Pi}_{\K_2} -\Pi_{\K_2})\zeta\| $. Note that $	\|\mathcal A(\widetilde{\Pi}_{\K_2} -\Pi_{\K_2})\zeta\| \leq G_1+G_2+G_3 + 	\|\mathcal A(\Pi_{\K_2}-I)\zeta\|^2$, where 
				\begin{equation*}
					G_1 = \Vert  \sum_{j=1}^{\K_2} \langle \hat{h}_j-h_j^s,\zeta\rangle \mathcal A(\hat{h}_j-h_j^s)\Vert, \quad G_2= \Vert    \sum_{j=1}^{\K_2} \langle h_j^s,\zeta\rangle \mathcal A(\hat{h}_j-h_j^s)\Vert, \quad G_3=\Vert  \sum_{j=1}^{\K_2} \langle \hat{h}_j-{h}_j^s,\zeta\rangle \mathcal A{h}_j^s\Vert.
				\end{equation*}		
				Then, by using \eqref{eqahh}, \eqref{eqbhh} \bt{and the fact that $d_T \alpha_1 ^{-1} = o(1)$}, we find that 
				\begin{equation*}
					 G_1    \leq   \sum_{j=1} ^{\K_2} |\langle \widehat{h}_j - h_j ^s, \zeta\rangle | \Vert \mathcal A (\widehat{h}_j - h_j^s) \Vert    
					\leq O_p(\alpha_1) \sum_{j=1}^{\K_2} j^{\rho_\nu-\delta_\zeta-\varsigma_{\nu}+2} \leq o_p(\alpha_1 ^{1/2})  \sum_{j=1} ^{\K_2} j^{\rho_\nu/2-\varsigma_{\nu}-\delta_{\zeta}+1},
				\end{equation*} 
				\begin{equation*}
					G_2   \leq  \sum_{j=1}^{\K_2}  |\langle h_j^s,\zeta\rangle |  \| \mathcal A(\hat{h}_j-h_j^s)\|  \leq   O_p(\alpha_1^{1/2})\sum_{j=1} ^{\K_2} j^{\rho_{\nu}/2-\varsigma_\nu-\delta_{\zeta}+1 },
				\end{equation*}  
				and	
				\begin{equation*}
					 G_3  \leq \sum_{j=1} ^{\K_2} |\langle \widehat{h}_j - h_j ^s,\zeta\rangle| \Vert \mathcal A h_j ^s\Vert \leq  O_p(\alpha_1^{1/2}) \sum_{j=1} ^{\K_2} j^{\rho_\nu/2-\varsigma_{\nu}-\delta_{\zeta}+1}.
					\label{eqf3hh}
				\end{equation*}	
				Hence, $G_1, $ $G_2$ and $G_3$ are bounded by the following.\begin{equation*}
					O_p(\alpha_1^{1/2})\sum_{j=1} ^{\K_2} j^{\rho_{\nu}/2-\varsigma_\nu-\delta_{\zeta}+1 } 	\leq \begin{cases}
						O_p(\alpha_1^{1/2}) &\text{if }\rho_\nu/2+2 < \varsigma_\nu +\delta_{\zeta},\\
						O_p(\alpha_1^{1/2}) \max\{  \log \alpha_2 ^{-1},\alpha_2 ^{-(\rho_\nu/2-\varsigma_{\nu}-\delta_{\zeta}+2)/\rho_{\nu}}\} &\text{if }\rho_\nu/2+2 \geq  \varsigma_\nu +\delta_{\zeta}.
					\end{cases}
				\end{equation*}
				Lastly, we have  
				\begin{equation*}
					\|\mathcal A(\Pi_{\K_2}-I)\zeta\|^2  \leq \sum_{j=\K_2+1} ^\infty \|  \langle h_j,\zeta \rangle \mathcal Ah_j\|^2 \leq 
					O(\sum_{j=\K_2+1} ^\infty j^{-2\delta_{\zeta} -2\varsigma_{\nu}})
					\leq O_p(\alpha_2^{(2\varsigma_{\nu}+2\delta_{\zeta}-1)/\rho_{\nu}}), 
				\end{equation*}
				from which the desired result follows. \qed 

	\subsection{Supplementary results}
\subsubsection{Strong consistency of the F2SLSE} \label{sec:strong2sls}
As in the case of the FIVE, we need some additional assumptions: below, as we did for the sequence of $d_t$ in Section \ref{sec:strongfive}, we let $\{ M_j  \}_{j \geq 1}$ be the sequence of eigenvalues of the covariance of $z_t \otimes z_t - \mathcal C_{zz}$.
\begin{assumption} \label{assum2ad2}\begin{enumerate*}[(a)] \item $\sup_{t\geq 1} \|x_t\| \leq m_x$, $\sup_{t\geq 1} \|z_t\| \leq m_z$, and $\sup_{t\geq 1} \|u_t\| \leq m_u$ a.s., \item  ${M}_j^\dag \leq  a b^j$ for some $a>0$ and $0<b<1$,  \item   \label{assum2ad2.3} the sequence of $v_t$ is a martingale difference with respect to $\mathcal G_t = \sigma(\{z_s\}_{s \leq t+1}, \{v_s\}_{s \leq t})$.
	\end{enumerate*}
\end{assumption}
We may establish the following preliminary results: 
\begin{lemma} \label{lem1ad2} Under  Assumptions \ref{assum1}.\ref{assum1.2}, \ref{assum1a}\ref{assum1a2} and \ref{assum2ad2},  the following hold almost surely: \vspace{-0.7em}
	\begin{equation*}
		\|\widehat{\mathcal C}_{zz}-\mathcal C_{zz}\|_{\op}  =  O(T^{-1/2} \log^{3/2}T) \quad \text{and}\quad
		\|\widehat{\mathcal C}_{vz}\|_{\op} = O(T^{-1/2}\log^{1/2} T).  
	\end{equation*}
\end{lemma}
\begin{proof}
	The first result follows from Theorem 2.12 and Corollary 2.4 of \cite{Bosq2000}. 	We then note that under Assumptions \ref{assum1}.\ref{assum1.2} and \ref{assum1a}\ref{assum1a2}, $\sup_{t\geq 1} \|v_t\| \leq \sup_{t\geq 1}\|x_t\| + \|\mathcal B\|_{\op}\sup_{t\geq 1}\|z_t\| < \infty$, and apply Theorem 2.14 of \cite{Bosq2000} to find that $\|\widehat{\mathcal C}_{zv}\|_{\op} = O(T^{-1/2}\log^{1/2}T)$ a.s. 
\end{proof}
In our proof of the strong consistency of the F2SLSE, what we want to have by employing Assumption \ref{assum2ad2}.\ref{assum2ad2.3} is the asymptotic order of $\|\widehat{\mathcal C}_{vz}\|_{\op}$ given in Lemma~\ref{lem1ad2}; in fact, our proof does not require any change once the following weaker condition holds:
\begin{equation} \label{alteq}
	\|\widehat{\mathcal C}_{vz}\|_{\op}  = O(T^{-1/2}\log T), \quad a.s.
\end{equation} 
Hence, in the sequel,  \eqref{alteq} may replace  Assumption \ref{assum2ad2}.\ref{assum2ad2.3}.  We now establish the strong consistency.
\setcounter{theorem}{4}
\begin{theorem}[continued] \label{thm2s} If Assumption \ref{assum2ad2} is additionally satisfied,  $(\sum_{j=1}^{\K_1} \mu_j \tau_{1,j})(\sum_{j=1}^{\K_2}\tau_{2,j}) = o(T^{1/2}\log^{-3/2}T)$ a.s., $(\sum_{j=\K_1+1}^{\infty}\mu_j)(\sum_{j=1}^{\K_2}\tau_{2,j} )  = o(1)$ a.s.,  $\abar_1^{-1}\abar_2\to0$, and  $\abar_1 T^{-1} \log T \to 0$,  then	$\|\widetilde{\mathcal A} - \mathcal A\|_{\op} \to 0$ a.s.
\end{theorem}

\begin{proof}
	From \eqref{eqpf0001vv}, we know that $	\|	\widetilde{\mathcal A}-\mathcal A \widetilde{\Pi}_{\K_2}\|_{\op} 	\leq  \alpha_1^{-1/4} \alpha_2^{-1/4}\|T^{-1}\sum_{t=1}^Tz_{t}  \otimes  u_t\|_{\op}$.	Moreover, $\{z_{t}  \otimes  u_t\}_{t\geq 1}$ is a martingale difference sequence satisfying  that $\sup_t\|z_t\otimes u_t\|_{\HS}<\infty$ a.s and $\sup_t\mathbb{E}\|z_t\otimes u_t\|_{\HS}^2 < \infty$. We therefore deduce from Theorem 2.14 of \cite{Bosq2000} that $	\|	\widetilde{\mathcal A}-\mathcal A \widetilde{\Pi}_{\K_2}\|_{\op} = O(\alpha_1^{-1/4}\alpha_2^{-1/4}T^{-1/2}\log^{1/2}T)$ a.s., and hence, $	\|\widetilde{\mathcal A}-\mathcal A \widetilde{\Pi}_{\K_2}\|_{\op} =o(1)$ a.s.  
	Note also that $	\|\mathcal A \widetilde{\Pi}_{\K_2} -\mathcal A \|_{\op}^2 \leq  \sum_{j= \K_2+1}^{\infty} \|\mathcal A {h}^s_j\|^2 + |\mathcal R|$,
	where $h^s_j$ and $\mathcal R$ are defined as in our proof of Theorem \ref{thm2w}. Since $\mathcal A$ is a Hilbert-Schmidt operator, we have $\sum_{j= \K_2+1}^{\infty} \|\mathcal A {h}^s_j\|^2 = o(1)$ a.s. It thus only remains to show that $|\mathcal R|=o(1)$ a.s. We know  from Lemma \ref{lem1ad2}  that that $\|\widehat{\mathcal C}_{zz}-\mathcal C_{zz}\|_{\op} =  O(T^{-1/2} \log^{3/2}T)$ a.s., and hence 
	\begin{equation}
		O\left(\sum_{j=1}^{\K_1} \mu_j \tau_{1,j}\right)   \sum_{j=1}^{\K_2} \tau_{2,j} \|\widehat{\mathcal C}_{zz}-\mathcal C_{zz}\|_{\op} = o(1), \quad a.s. \label{eq001a3}
	\end{equation}
	As shown in Lemma \ref{lem1ad2}, we have $\|\widehat{\mathcal C}_{zv}\|_{\op} = O(T^{-1/2}\log^{1/2}T)$, a.s., and also find that
	\begin{align}
		\|\mathcal S\|_{\op} \sum_{j=1}^{\K_2}\tau_{2,j}& \leq (O(T^{-1/2} \log^{1/2}T) + O( \alpha_1^{-1/2}T^{-1} \log T))  \sum_{j=1}^{\K_2}\tau_{2,j}  \notag \\ &= o(\log^{-1}T) + o(\alpha_1^{-1/2}T^{-1/2}\log^{-1/2}T) = o(1), \quad a.s. \label{eq001a4a}
	\end{align}
	by the definition of $\mathcal S$. Moreover, 
	\begin{equation}
		\|\mathcal T\|_{\op}\sum_{j=1}^{\K_2} \tau_{2,j}\leq O\left(\sum_{j=\K_1 +1}^\infty \mu_j\right)\sum_{j=1}^{\K_2} \tau_{2,j} = o(1), \quad a.s. \label{eq001as}
	\end{equation}	
	From \eqref{eqnormpf02adad},  \eqref{eq001a3},  \eqref{eq001a4a}, and \eqref{eq001as}, it immediately follows that  $|\mathcal R| = o(1)$ a.s. 
\end{proof}		

\subsubsection{Refinements of the general asymptotic results for the F2SLSE}\label{sec_thm_conti2}
\setcounter{theorem}{7}
We in this section provide a complementaty result to Theorem \ref{thm4:convrate} for the case where \(\rho_\nu/2+2\geq\varsigma_\nu+\delta_{\zeta}\). Specifically, the following can be shown:
\begin{theorem}[Continued]\label{thm4:convrateadd}
Let everything as in Theorem \ref{thm4:convrate} but with $\rho_\nu/2+2\geq\varsigma_\nu+\delta_{\zeta}$.
Then, Theorem \ref{thm2wa} holds and
\begin{align*}
&\Vert \mathcal A(\widetilde \Pi_{\K_2}-\Pi_{\K_2}) \zeta \Vert   =
O_p(\abar_1^{-1/2} \max\{  \log \abar_2, \abar_2 ^{(\rho_\nu/2-\varsigma_{\nu}-\delta_{\zeta}+2)/\rho_{\nu}} \}),\nonumber\\
&	\|\mathcal A(\Pi_{\K_2} - \mathcal I)\zeta\| =O_p(\abar_2^{(1/2-\varsigma_{\nu}-\delta_{\zeta})/\rho_{\nu}}).  
\end{align*}
\end{theorem}
Our proof of the above result is contained in the proof of Theorem \ref{thm4:convrate} given in Section \ref{sec:f2sls0proof}, and hence omitted.

\section{Appendix to Section \ref{sec:sim} on ``Numerical studies"}\label{sec_app_num0}
\subsection{Appendix to Section \ref{subsub: exp1.1}}\label{sec_app_num1}
We show that the simulation DGP considered in Section \ref{subsub: exp1.1} satisfies the employed assumptions.  From our construction of the variables ($\{y_t,x_t,z_t,u_t,v_t\}_{t\geq 1}$) and  operators ($\mathcal A$ and $\mathcal B$) in Section \ref{subsub: exp1.1} and from Theorem 2.7 of \cite{Bosq2000}, it is not difficult to show that Assumptions~\ref{assum1} and \ref{assum1a} are satisfied. Moreover, this setup simplifies the representation of the eigenvalues and eigenfunctions associated with the FIVE and F2SLSE.  Specifically, we have $\lambda_j ^2  = b_j^2 \mu_j ^2$, $\nu_j^2 = b_j^4\mu_j^2$, $f_j  = \xi_j =h_j = g_j$,   $\langle z_t,g_j \rangle \sim N(0,\mu_j^2)$  and   $\langle x_t,g_j \rangle = \langle \vartheta \mathcal Bz_t,g_j \rangle + \langle v_t,g_j \rangle \sim N(0,(\vartheta^2 b_j^2+1) \mu_j^2)$. 
Since $\mathbb{E}[\Vert x_t \Vert ^4 ] < \infty$ and $ \mathbb{E}[\Vert z_t \Vert ^4 ] < \infty$,  we find the following: for some $c_{\circ}>0$, 
\begin{align}
&\mathbb E[ \Vert \langle z_t , \xi_j\rangle x_t \Vert ^2 ] 
\leq \mathbb E[ |\langle z_t , \xi_j\rangle|^4]^{1/2} \mathbb{E}[\Vert x_t \Vert ^4 ]^{1/2} \leq  c_{\circ} \lambda_j^2, \label{eqaddpf01}\\
&\mathbb E[ \Vert \langle x_t , \xi_j\rangle z_t \Vert ^2 ] 
\leq \mathbb E[ |\langle x_t , \xi_j\rangle|^4]^{1/2} \mathbb{E}[\Vert z_t \Vert ^4 ]^{1/2} \leq  c_{\circ} \lambda_j^2. \label{eqaddpf02}
\end{align} Moreover, we have $$|\langle \mathcal Af_j,\xi_{\ell} \rangle| = 2 j^{-n_a} {1}\{j=\ell\} \leq 2 j^{-n_{a}/2}\ell^{-n_{a}/2}.$$
From these results and the fact that $\{x_t,z_t\}_{t\geq 1}$ is an iid sequence,  Assumptions~\ref{assum1eigen}, \ref{assum1eigen2}  and  \ref{assumconvrate} are also satisfied. It is obvious that Assumption \ref{assconvratetsls}.\ref{assconvratetsls.1}-\ref{assconvratetsls.4} from our construction of the variables and operators, and  Assumption \ref{assconvratetsls}.\ref{assconvratetsls.70} can be shown to hold from the fact that $h_j = g_j$. We next note that 
$$|\langle   h_j ,\mathcal Bg_{\ell} \rangle|  \leq b_j 1\{j=\ell\} \leq j^{-n_b/2} \ell^{-n_b/2},$$
and, in the considered setup, $\rho_{\nu} = 4n_b  + 4$. and thus $n_b/2 \leq \rho_{\nu}/4 + 1/2$.  Thus, Assumption \ref{assconvratetsls}.\ref{assconvratetsls.7} is satisfied. Lastly, it can also be shown that Assumption \ref{assconvratetsls}.\ref{assconvratetsls.5} holds from similar arguments used to show \eqref{eqaddpf01} and \eqref{eqaddpf02} and the facts that $\{z_t\}_{t\geq 1}$ is an iid sequence and $\mathcal C_{xz}h_j = \mathcal C_{xz}g_j = \lambda_j$.

				\subsection{Appendix to Section \ref{sec:sim2}}\label{sec:simul}
				In Section~\ref{sec:sim2}, we compute the local likelihood estimate of $\log p_t^\circ$ using random samples $\{s_{i,t}\}_{i=1}^{n}$ drawn from the distribution $p_t ^\circ$. Consider the  following log-likelihood:
				\begin{equation}
					l(\{s_{i,t}\}_{i=1}^{n}) = \sum_{i=1}^n \log p_t(s_{i,t}) - n \left(\int p_t(v) dv - 1 \right). \label{locallike}
				\end{equation} 
				Under some local smoothness assumptions \citep{loader1996}, we can obtain a localized version of \eqref{locallike} by approximating   $\log p_t(s)$ using polynomial functions, as follows.
				\begin{equation} \label{locallf}
					l(\{s_{i,t}\}_{i=1}^{n})(s) = \sum_{i=1}^n  \, w\left(\frac{s_{i,t}-s}{h_s}\right) H(s_{i,t}-s ; \beta_t) - n \int  w \left(\frac{v-s}{h_s}\right) \exp(H(v-s ; \beta_t)) dv,
				\end{equation} 
				where  $w(\cdot)$ is a suitable weight function,  $h_s$ is a bandwidth, and $H(v ; \beta_t)$ is polynomial in $v$ with coefficients $\beta_t$, i.e., $H(v ; \beta_t)= \sum_{j=0} ^q \beta_{j,t} v^{j}$ for some nonnegative integer $q$.   For a fixed $s \in [0,1]$, let $\hat{\beta}_{t}$ be the maximizer of \eqref{locallf}, then the local likelihood log-density estimate is given by $\widehat{\log p}_t (s) = \hat{\beta}_{0,t}$. By repeating this procedure for a find grid of points and interpolating the results as described by \citet[Chapter 12]{Loader2006}, we can obtain $\widehat{\log p}_t$. In our simulation experiment in Section \ref{sec:sim2},	$w(\cdot)$ is set to the tricube kernel that is used in many examples given by \cite{Loader2006}, $q=1$, and $h_s$ is set to the nearest neighbor bandwidth covering 33.3\% of observations \citep[Section 2.2.1]{Loader2006}.  
				
				\subsection{Appendix to Section \ref{sec:emp}}\label{sec:app:emp}
We here define a measure of worker's skill similar to \citepos{Peri2009} measure of occupation-specific relative provision of communication versus manual skills. Specifically, we use the O*Net ability survey data,\footnote{Version 24, provided by the US Department of Labor} in which the importance of each of 52 distinct abilities required by each occupation is quantified. Using the data, we construct the communication skill measure $(c_j ^\circ)$ and the manual skill measure ($m_j ^\circ$) for each occupation $j$, where the definitions of communication and manual skills are equivalent to the extended definitions of those in Table~A.1 of \cite{Peri2009}. We merge the values of $c_j ^\circ$ and $m_j ^\circ$ to individuals in the 2000 census using the monthly US Current Population Survey (CPS) data. Then, as done similarly by \cite{Peri2009}, the measure of occupation-specific skill intensity ($s_j$) is obtained by converting the value of $c_{j}^\circ / m_j ^\circ$ to its percentile score ($s_{jt} ^\circ$) for each month in 2000 and averaging the monthly scores for each occupation $j$, i.e., $s_j = 12^{-1}\sum_{t=1} ^{12} s_{jt} ^\circ$.  The number of distinct skill levels, $s_j$, is 223, and, by construction, each occupation is uniquely identified by the skill score $s_j\in[0,1]$.\footnote{{A similar rescaling procedure is taken by \cite{Peri2009}. Specifically, they first converted the values of $c_j ^\circ$ and $m_j^\circ$ into their percentile scores in 2000 for each $j$, and then their measure of relative provision of communication versus manual skills is given by the ratio of the percentile scores for each $j$. In contrast, we first take the ratio of $c_j ^\circ$ and $m_j ^\circ$  and then convert the ratio into the percentile score for each $j$. This is to ensure that our measure of relative communication skill provision takes values in [0,1].}} In Table~\ref{tab.cmscore} in Appendix, we report occupations with the lowest and highest scores of relative communication skill provision.

In addition to estimation results reported in Section~\ref{sec:emp}, we examine the null hypothesis $H_0: \mathcal A^\ast \psi = 0$ using the significance test given in Section~\ref{sec:hypo}. Note that $\psi$ can be any arbitrary element of $\mathcal H$, but here we consider only a few cases, where $\psi\in \{\zeta_{j(3)}\}_{j=1} ^3$, for the purpose of illustration. Then, by testing $H_0$, we can examine whether  the average of changes in the wages of native workers in the  occupations of which $s \in [(j-1)/3,j/3]$ is affected by an inflow of immigrants. Similarly to the simulation experiments in Section~\ref{sec:sim}, we set $D$ to $\lceil T^{1/3} \rceil$ and compute the critical values based on 10,000 Monte Carlo simulations. The testing results are reported in Table~\ref{tab.emp2}. In the table, we found that an inflow of immigrants significantly affects the wages of native workers who are in occupations intensive in either manual or communication skills. 

\begin{table}[h!]
\caption{Significance testing results}\label{tab.emp2}
\vspace{-.5em}	\rule{1\textwidth}{.5pt}\vspace{.1em} 
\begin{tabular*}{1\textwidth}{@{\extracolsep{\fill}}lclll }  
&$\psi$&$\zeta_{1(3)}$&$\zeta_{2(3)}$&$\zeta_{3(3)}$   \\  \midrule	 
&Test statistic&0.00050$^{*}$ & 0.00038$^{}$ & 0.00084$^{**}$ \\  
\end{tabular*}
\rule{1\textwidth}{1pt} 
{\footnotesize Notes: 
{We use $^{\ast}$ and $^{\ast\ast}$ to denote rejection at 10\% and 5\% significance levels, respectively. } }
\end{table}

				
				\subsection{Additional tables}\label{sec:addtab}
			
\begin{figure}[H]
\centering	\caption{Simulation results for Experiment 1: boxplots of the empirical MSEs ($T=250$)}\label{fig:box} 
	\begin{subfigure}{.4\textwidth}\subcaption{$(n_{a}, n_b ) =(3,3/4)$ }
	\includegraphics[width=\textwidth ]{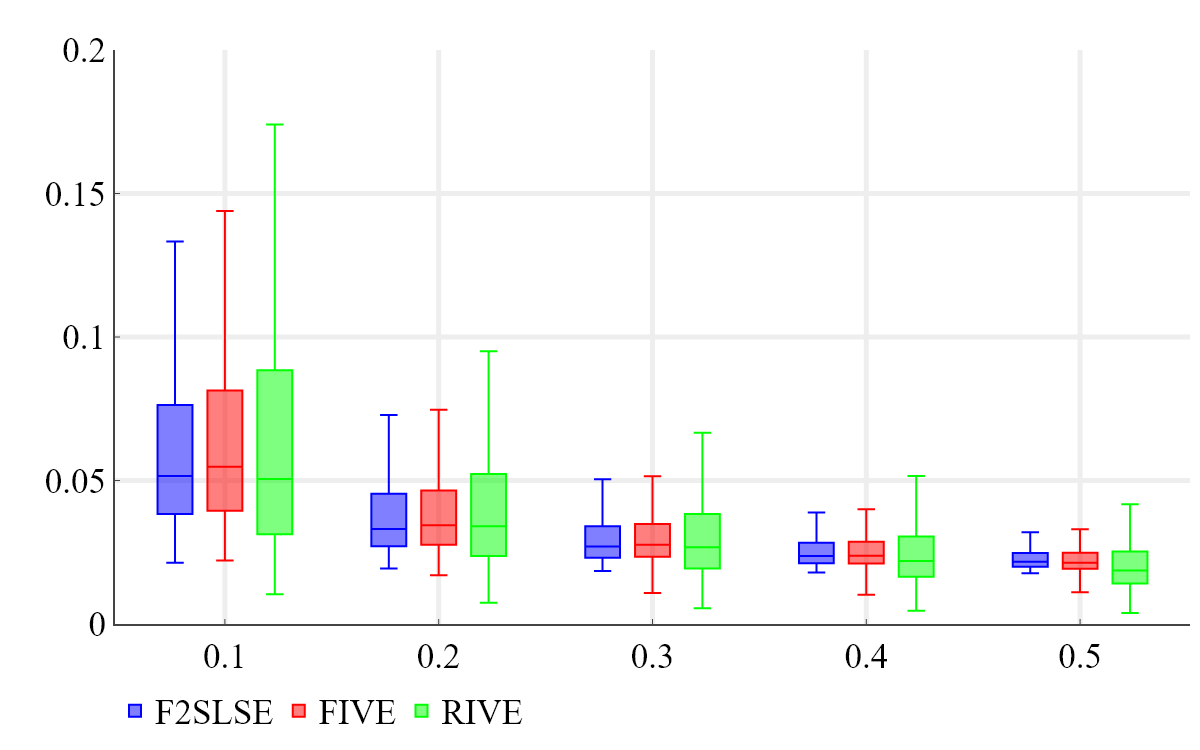}
	\end{subfigure}
\begin{subfigure}{.4\textwidth}\subcaption{$(n_{a}, n_b ) =(3,1.5)$ }
	\includegraphics[width=\textwidth ]{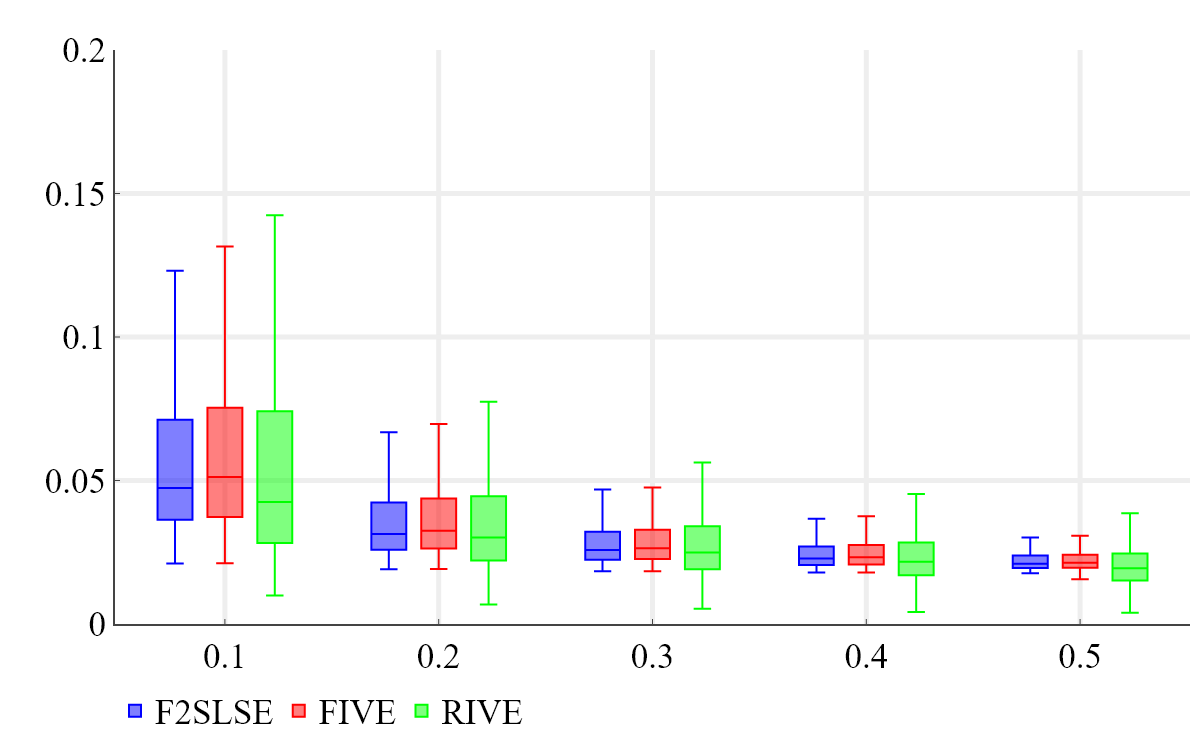}
\end{subfigure} \\
	\begin{subfigure}{.4\textwidth}\subcaption{$(n_{a}, n_b ) =(5,3/4)$ }
	\includegraphics[width=\textwidth ]{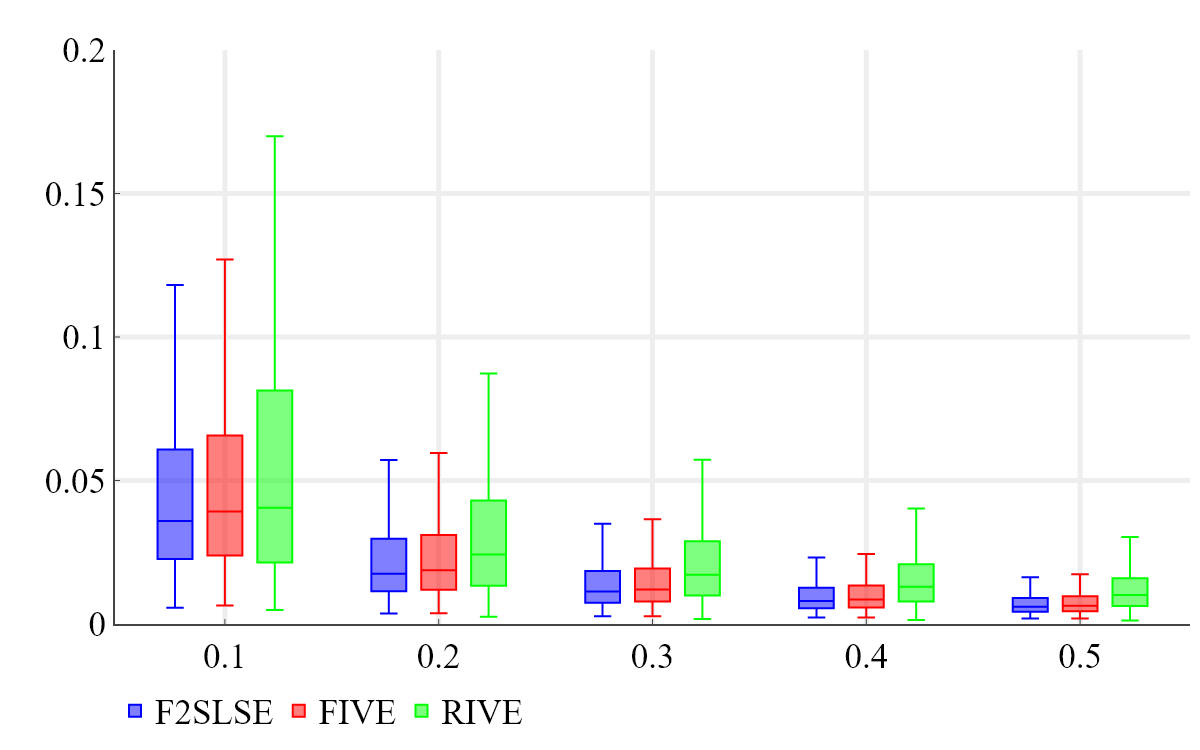}
	\end{subfigure}
\begin{subfigure}{.4\textwidth}\subcaption{$(n_{a}, n_b ) =(5,1.5)$ }
	\includegraphics[width=\textwidth ]{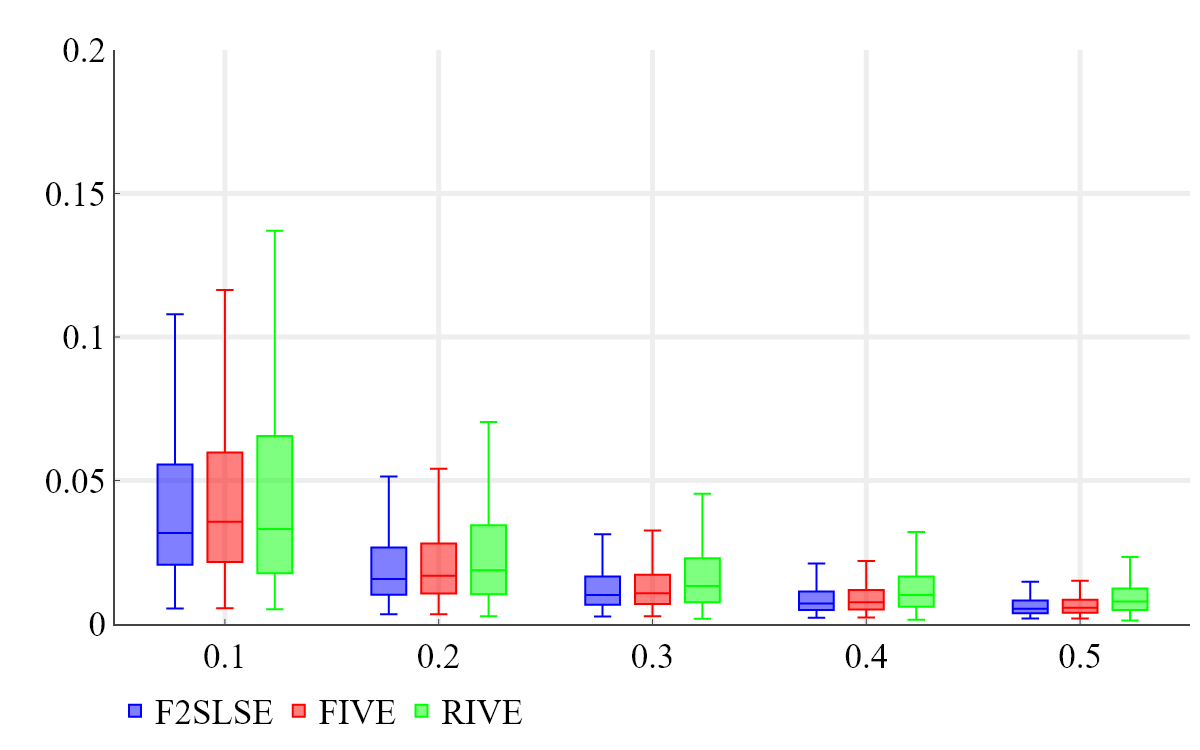}
\end{subfigure} \\ 
\footnotesize{Notes: Boxplots of the empirical MSEs of the FIVE (red), the F2SLSE (blue) and the RIVE (green) are reported for each value of the first-stage functional coefficient of determinations $\mathtt{r}^2 \in \{0.1,0.2,0.3,0.4,0.5\}$. }
\end{figure}

		\begin{table}[H]
		\caption{Simulation results for Experiment 3: empirical MSEs ($a_1,a_2\geq 0.6$)}\label{tab2aa}
		\vspace{-.5em}	\rule{1\textwidth}{.5pt}\vspace{.5em}
		\begin{tabular*}{1\textwidth}{@{\extracolsep{\fill}}llcccccccc}
			&&\multicolumn{4}{c}{Sparse Design}&\multicolumn{4}{c}{Exponential Design}\\\cmidrule{3-6}\cmidrule{7-10}
			&$n$    &\multicolumn{2}{c}{$100$}&\multicolumn{2}{c}{$ 150$}    &\multicolumn{2}{c}{$  100$}&\multicolumn{2}{c}{$150$}\\\cmidrule{3-4}\cmidrule{5-6}\cmidrule{7-8}\cmidrule{9-10}
			&T & 200 & 500 & 200 & 500 & 200 & 500 & 200 & 500 \\ \midrule
			\multirow{4}{*}{Loader's}
			& FIVE & ${0.185}$ & ${0.152}$ & ${0.174}$ & ${0.149}$ & ${0.314}$ & ${0.194}$ & ${0.255}$ & ${0.173}$ \\ 
			& F2SLSE & ${0.188}$ & ${0.151}$ & ${0.174}$ & ${0.148}$ & ${0.313}$ & ${0.188}$ & ${0.256}$ & ${0.170}$ \\ 
			& RIVE & ${0.194}$ & ${0.154}$ & ${0.179}$ & ${0.150}$ & ${0.342}$ & ${0.205}$ & ${0.276}$ & ${0.174}$ \\ 
			& FLSE & ${0.248}$ & ${0.221}$ & ${0.206}$ & ${0.183}$ & ${0.419}$ & ${0.347}$ & ${0.322}$ & ${0.260}$ \\ 
			\midrule
			\multirow{4}{*}{Silverman's}
			& FIVE & ${0.273}$ & ${0.231}$ & ${0.222}$ & ${0.194}$ & ${0.339}$ & ${0.242}$ & ${0.271}$ & ${0.207}$ \\ 
			& F2SLSE & ${0.275}$ & ${0.228}$ & ${0.225}$ & ${0.193}$ & ${0.352}$ & ${0.234}$ & ${0.277}$ & ${0.201}$ \\ 
			& RIVE & ${0.271}$ & ${0.226}$ & ${0.224}$ & ${0.192}$ & ${0.342}$ & ${0.237}$ & ${0.278}$ & ${0.199}$ \\ 
			& FLSE & ${0.351}$ & ${0.318}$ & ${0.267}$ & ${0.240}$ & ${0.422}$ & ${0.350}$ & ${0.319}$ & ${0.257}$ \\ 
		\end{tabular*}
		\rule{1\textwidth}{1pt}
		{\footnotesize Notes: Based on 1,000  replications. Each cell reports the empirical mean squared error (MSE) of the four considered estimators: FIVE, F2SLSE, \citepos{Benatia2017} RIVE and \citepos{Park2012} FLSE.}
	\end{table}	
\begin{table}[H]
\caption{Occupations with the lowest and highest communication skill intensity in 2000 (denoted $s$)}\label{tab.cmscore}
\vspace{-.5em}	\rule{1\textwidth}{.5pt} \vspace{.5em}
\begin{tabular*}{1\textwidth}{@{\extracolsep{\fill}}lc} 
	\multicolumn{1}{l}{\sl{Four occupations with the lowest $s$}}&$s$\\ \midrule   
	Pressing machine operators (clothing)&0.0010\\
	Construction Trades & 0.0035\\
	Machine operators & 0.0235 \\
	Garbage and recyclable material collectors& 0.0242\\
	\midrule 
	\multicolumn{1}{l}{\sl{Four occupations with the highest $s$}}&$s$\\ \midrule   
	Chief executives and public administrators&0.9926\\
	Operations and systems researchers and analysts& 0.9954\\
	Management Analysts&0.9985\\
	Economists, market researchers, and survey researchers& 1.0000	\\ \bottomrule		 
\end{tabular*}
\end{table} 

\section{Computation}\label{sec_comp}
We here only describe how to compute the FIVE $\widehat{\mathcal A}$ from observations $\{y_t,x_t,z_t\}_{t=1}^T$; in fact, computation of the F2SLSE $\widetilde{\mathcal A}$ can be done with only slight modifications and thus omitted.  Specifically, for each $T$, $\widehat{\mathcal A}$ is a finite rank operator acting on the Hilbert space of square-integrable functions defined on $[0,1]$, so it allows the following representation \citep*[Chapter 8]{Gohberg2013}: 
\begin{equation} \label{eqcompute}
	\widehat{\mathcal A} \upsilon(s_1) =  \int_{0}^1 \hat{\kappa}(s_1,s_2) \upsilon(s_2) ds_2, \quad s_1,s_2\in [0,1], 
\end{equation}
where $\upsilon$ is any arbitrary random or nonrandom element  $\upsilon$ taking values in $\mathcal H$ (for example, $\upsilon$ can be $x_t$ or any fixed element in $\mathcal H$). Therefore, computation of $\widehat{\mathcal A}$ reduces to obtaining an explicit formula for the integral kernel $\hat{\kappa}(s_1,s_2)$ for $s_1,s_2 \in [0,1]$. We here present a way to compute this integral kernel from  the eigenelements of $\widehat {\mathcal C}_{xz}^\ast \widehat {\mathcal C}_{xz}$ and $\widehat {\mathcal C}_{xz} \widehat {\mathcal C}_{xz}^\ast$, which can be obtained by the standard functional principal component method; see e.g., \citet[Chapter 8]{Ramsay2005} and \citet[Chapter 3]{HK2012}.

Let $\{ \hat \xi_j  \}_{j \geq 1}$ be the collection of the eigenfunctions of $\widehat {\mathcal C}_{xz} \widehat {\mathcal C}_{xz}^\ast$, and then note that $\widehat{\mathcal C}_{xz} \hat{f}_j = \hat{\lambda}_j \hat \xi_j$ and $\widehat{\mathcal C}_{xz} ^\ast \hat {\xi}_j = \hat{\lambda}_j \hat{f}_j$,  
where $\hat{\lambda}_j =\langle \widehat{\mathcal C}_{xz} \hat f_j , \hat{\xi}_j \rangle $ \citep[Section 4.3]{Bosq2000}. Then,  $\widehat{\mathcal A}$ is given by 
\begin{equation}
	\widehat{\mathcal A} = \widehat{\mathcal C}_{yz}^\ast  \widehat{\mathcal C}_{xz} (\widehat{\mathcal C}_{xz}^\ast  \widehat{\mathcal C}_{xz})_{\K}^{-1} = \frac{1}{T}\sum_{t=1}^T \sum_{j=1}^{\K} \hat{\lambda}_j^{-1}   \langle\hat {\xi}_j , z_t \rangle \hat{f}_j \otimes y_t. \label{eq:int:kernel0}
\end{equation}
It is  quite  obvious from  \eqref{eq:int:kernel0} that the integral kernel $\hat{\kappa}(s_1,s_2)$ for $s_1,s_2 \in [0,1]$ is given by $T^{-1}\sum_{t=1} ^T\sum_{j=1} ^{\K} \hat{\lambda}_j^{-1}   \langle\hat {\xi}_j , z_t \rangle \hat f_j (s_1) y_t(s_2)$, and this can be equivalently expressed as follows,\begin{equation}
	\hat{\kappa}(s_1,s_2)   =   T^{-1}\hat F_{\K} (s_1)'\text{diag}(\hat{\lambda}_j ^{-1}) \hat G_{\K} Y_T(s_2),\label{eq:int:kernel}
\end{equation}
where $Y_T(s) = (y_1 (s), \ldots , y_T (s))'$, $\hat F_{\K} (s) = (\hat f_1 (s) , \ldots, \hat f_{\K} (s))'$ for $s \in [0,1]$, and $\hat G_{\K}$ is the $\K \times T$ matrix whose $(i,t)$-th element is given by $ \langle\hat{\xi}_i , z_t\rangle$. Thus, for each choice of $s_1$ and $s_2$, $\hat{\kappa}(s_1,s_2)$ can be obtained by simple matrix multiplications.	In view of \eqref{eq:int:kernel0} and \eqref{eq:int:kernel},  $\widehat{\mathcal A} \upsilon$, for any arbitrary random or nonrandom element  $\upsilon$ taking values in $\mathcal H$, is computed as follows:
$	\widehat{\mathcal A}  \upsilon(s)  = T^{-1}\hat F_{\K}({\upsilon})' \text{diag}(\hat{\lambda}_j ^{-1}) \hat G_{\K} Y_T(s),$	where $s \in [0,1]$ and $ \hat F_{\K}({\upsilon})' = (\langle \hat{f}_1 , \upsilon \rangle , \ldots,\langle \hat{f}_{\K}, \upsilon\rangle  )'$.

Computing the FIVE requires choosing $\alpha = \abar^{-1}$ defined in  \eqref{eqdef2}. The eigenvalue $\hat{\lambda}_j^2$ of $\widehat{\mathcal C}_{xz}^\ast \widehat{\mathcal C}_{xz}$ depends on the scales of $x_t$ and $z_t$, and this needs to be considered in choosing $\alpha$. In practice, it thus may be of interest to have a scale-invariant choice of $\alpha$. This can be done by computing the contribution of each of the eigenvalues to a magnitude of the operator $\widehat{\mathcal C}_{xz}$  and then viewing $\alpha$ as the threshold parameter for such computed contributions. We illustrate an easy-to-implement way here. Define $\hat{r}_k = \widehat{\lambda}_k^2 \Big/ \sum_{j=1}^\infty \widehat{\lambda}_j^2$.
Since $\|\widehat{\mathcal C}_{xz}\|_{\HS}^2=\sum_{j=1}^T \widehat{\lambda}_j^2$, the ratio $\hat{r}_k$ computes the contribution of the $k$-th eigenvalue to the squared Hilbert-Schmidt norm of $\widehat{\mathcal C}_{xz}$. Of course, the above quantity does not depend on the scales of $x_t$ and $z_t$, and hence, a scale-invariant version of \eqref{eqdef2} may be written as 
\begin{equation}
	\K = \#\{j  :	\hat{r}_k > \alpha_T\},  \label{eqdef2a}
\end{equation}
where $\alpha_T \in (0,1)$ depends only on $T$, and shrinks to zero as $T$ increases.  An alternative way  is to directly choose $\K$, instead of $\alpha$, as the minimal number of the eigenvalues whose sum exceeds a pre-specified proportion of $\|\widehat{\mathcal C}_{xz}\|_{\HS}^2$. (Of course, even in this case, it is more natural to understand $\K$ as a random variable.) To be more specific, let
\begin{equation}
	\widehat{\mathrm{R}}_{k} = {\sum_{j=1}^k \widehat{\lambda}_j^2  \Big/
		\sum_{j=1}^\infty \widehat{\lambda}_j^2} \quad\text{and}\quad\K = \min_k \{k : \widehat{\mathrm{R}}_{k} >  (1-\alpha_T)\}, \label{eqdetermine1} 
\end{equation}
where $\alpha_T \in (0,1)$ is similarly defined. This choice is obviously scale-invariant as well. 

It is also possible to pursue a data-driven selection of $\alpha_T$ in \eqref{eqdef2a} and \eqref{eqdetermine1}, such as a cross-validation approach, proposed by \cite{Benatia2017} developed in an iid setting. Such a procedure may be adapted for dependent non iid data, but this will not be further studied in this paper.  

\section{FIVE with a general weighting operator} \label{sec_five_weight}
The theoretical results given for the FIVE in Section \ref{sec:estimators} can be extended to the case involving a general weighting operator as in Euclidean space setting. 
 Let $\mathcal W$ be a self-adjoint positive definite operator, and define $$\overline{z}_t = {\mathcal W}^{1/2} z_t.$$ If $\ker \mathcal C_{xz} = \{0\}$, due to the positive definiteness of $\mathcal W$, $\mathcal A$ is uniquely identified from the equation: $\mathcal C_{y \bar{z}}^\ast \mathcal C_{y \bar{z}} = \mathcal A  \mathcal C_{y \bar{z}}^\ast \mathcal C_{y \bar{z}}$. A natural extension of the FIVE can be defined as follows:
\begin{equation} \label{appendeqadd01}
	\widehat{\mathcal A}({\mathcal W}) = \widehat{\mathcal C}_{y\overline{z}}^\ast \widehat{\mathcal C}_{x\overline{z}} \left(\widehat{\mathcal C}_{x\overline{z}}^\ast\widehat{\mathcal C}_{x\overline{z}}\right)_{\K}^{-1}  = \widehat{\mathcal C}_{y{z}}^\ast {\mathcal W}\widehat{\mathcal C}_{x{z}} \left(\widehat{\mathcal C}_{x{z}}^\ast {\mathcal W} \widehat{\mathcal C}_{x{z}}\right)^{-1}_{\K}. 
\end{equation}
In this case, the theoretical results given in Section \ref{sec:estimators} can easily be extended by an obvious conversion of the assumptions given for $z_t$ into those for $\overline{z}_t$. This can further be extended to the case with a possibly random weighting operator $\widehat{\mathcal W}$ satisfying, which will be discussed in this section. 

It will be useful to define additional notation. We let 
\begin{equation}
\breve{z}_t = \widehat{\mathcal W}^{1/2} z_t,
\end{equation}  
and let 
$\widehat{\mathcal A}(\widehat{\mathcal W})$ (resp.\ $\K$) be defined by replacing $\mathcal W$ (resp.\ $\lambda_j$) with $\widehat{\mathcal W}$ in \eqref{appendeqadd01} (the $j$-th ordered eigenvalue $\breve{\lambda}_j$ of $\widehat{\mathcal C}_{x\breve{z}}^\ast \widehat{\mathcal C}_{x\breve{z}} =\widehat{\mathcal C}_{x{z}}^\ast \widehat{\mathcal W} \widehat{\mathcal C}_{x{z}}$ in \eqref{eqdef2}).  Let $\mathfrak L_{+}$ be the collection of self-adjoint positive definite operators in $\mathcal L_{\mathcal H}$, and write the spectral representations of $\mathcal C_{x\bar{z}}^\ast\mathcal C_{x\bar{z}}$ and $\mathcal C_{x\bar{z}}\mathcal C_{x\bar{z}}^\ast$ as  
\begin{equation}
\mathcal C_{x\bar{z}}^\ast\mathcal C_{x\bar{z}} = \sum_{j=1}^\infty \bar{\lambda}_j^2 \bar{f}_j \otimes \bar{f}_j \quad \text{and} \quad   \mathcal C_{x\bar{z}}\mathcal C_{x\bar{z}}^\ast = \sum_{j=1}^\infty \bar{\lambda}_j^2 \bar{\xi}_j \otimes \bar{\xi}_j,
\end{equation}
respectively. We also let $\bar{\Pi}_{\K}$ denote the projection on the span of the first $\K$ eigenvectors of $\widehat{\mathcal C}_{x\breve{z}}^\ast \widehat{\mathcal C}_{x\breve{z}}$. 
\begin{assumption}\label{assumadd}
 $\widehat{\mathcal W} \in \mathfrak L_{+}$ a.s.\ and $\|\widehat{\mathcal W} - \mathcal W\|_{\op} = O_p(T^{-1/2})$ for some fixed $\mathcal W \in \mathfrak L_{+}$. 
\end{assumption}
Note that the F2SLSE discussed in Section \ref{sec:f2sls00} does not satisfy Assumption \ref{assumadd} as \((\widehat{\mathcal C}_{zz})_{\K_1}^{-1}\) diverges in operator norm. 

\begin{theo}\label{thm2w1addone}
	Suppose that Assumption \ref{assumadd} holds. 
\begin{enumerate}[(i)] 
		\item\label{thm2w1addonea} 	If the assumptions in Theorem \ref{thm2w1} when $\lambda_j$ is replaced with $\bar{\lambda}_j$ are also satisfied,  
	\begin{equation*}
		\|\widehat{\mathcal A}(\widehat{\mathcal W}) - \mathcal A\bar{\Pi}_{\K}\|_{\op}^2 =O_p(T^{-1}\abar)\quad \text{and} \quad 	\|\mathcal A (\mathcal I-\bar{\Pi}_{\K})\|_{\op}^2 =o_p(1).	
	\end{equation*}  
\item\label{thm2w1addoneb} 	 If the assumptions in Theorem \ref{thm2w2} when $\lambda_j$ is replaced with $\bar{\lambda}_j$ are also satisfied, 
\begin{align}
\sqrt{{{T}}/{{\bar{\theta}_{\K}(\zeta)}}}(\widehat{\mathcal A}(\widehat{\mathcal W})- \mathcal A \bar{\Pi}_{\K})\zeta \dto  N(0, \mathcal  C_{uu}) \quad \text{and}\quad  |\breve{\theta }_{\K}(\zeta) -  \bar{\theta}_ {\K}(\zeta) | \pto 0,
\end{align} where 
\begin{equation}
	{\bar{\theta}_{\K}(\zeta)}\coloneqq \langle \zeta,  ({\mathcal C}_{x\bar{z}}^\ast{\mathcal C}_{x\bar{z}})_{\K}^{-1}{\mathcal C}_{x\bar{z}}^\ast \mathcal C_{\bar{z}\bar{z}}{\mathcal C}_{x\bar{z}}({\mathcal C}_{x\bar{z}}^\ast{\mathcal C}_{x\bar{z}})_{\K}^{-1}\zeta
\end{equation}
and $\breve{\theta }_{\K}(\zeta)$ is defined by replacing $\bar{z}_t$ with $\breve{z}_t$ in the above.
	\end{enumerate} 
\end{theo}
\begin{proof}
	We note that 
	\begin{equation}
	\widehat{\mathcal A}(\widehat{\mathcal W})  = \widehat{\mathcal C}_{y\breve{z}}^\ast \widehat{\mathcal C}_{x\breve{z}}	(\widehat{\mathcal C}_{x\breve{z}}^\ast  \widehat{\mathcal C}_{x\breve{z}})_{\K}^{-1}	=   \mathcal A \bar{\Pi}_{\K} +   \widehat{\mathcal C}_{u\breve{z}}^\ast \widehat{\mathcal C}_{x\breve{z}}(\widehat{\mathcal C}_{x\breve{z}}^\ast  \widehat{\mathcal C}_{x\breve{z}})_{\K}^{-1}. \label{eq:b0}
	\end{equation}
From Assumption \ref{assumadd} and our construction of $\K$, the following can be shown: (a) $\|\widehat{\mathcal C}_{u\breve{z}}\|_{\HS} = O_p(T^{-1/2})$, (b) $\|\widehat{\mathcal C}_{x\breve{z}}	(\widehat{\mathcal C}_{x\breve{z}}^\ast  \widehat{\mathcal C}_{x\breve{z}})_{\K}^{-1}\|_{\op} = \alpha^{-1/2} = \abar^{1/2}$, and (c) $\|\widehat{\mathcal C}_{x\breve{z}}^\ast  \widehat{\mathcal C}_{x\breve{z}}-{\mathcal C}_{x\bar{z}}^\ast  {\mathcal C}_{x\bar{z}}\|_{\op}=\|\widehat{\mathcal C}_{x{z}}^\ast \widehat{\mathcal W} \widehat{\mathcal C}_{x{z}}-{\mathcal C}_{x{z}}^\ast  \mathcal W {\mathcal C}_{x{z}}\|_{\op} = O_p(T^{-1/2})$. Combining these with similar arguments used in our proofs of Theorems \ref{thm2w1} and \ref{thm2w2}, the desired results are obtained. 
\end{proof}

\begin{theo}\label{thm2w1addone2}
	Suppose that Assumption \ref{assumadd} holds. 
\begin{enumerate}[(i)] 
		\item\label{thm2w1addone2a} 	If the assumptions in Theorem \ref{thm:convrate} when $\lambda_j$, $f_j$ and $\xi_j$ are, respectively, replaced with $\bar{\lambda}_j$, $\bar{f}_j$ and $\bar{\xi}_j$ are also satisfied, then $	\|\widehat{\mathcal A}(\widehat{\mathcal W}) - \mathcal A\bar{\Pi}_{\K}\|_{\op}^2=O_p(T^{-1}\abar)$ as in Theorem \ref{thm2w1addone} and  and
			\begin{equation} \label{eqthmconvrate}
				\|\mathcal A (\mathcal I-\bar{\Pi}_{\K})\|_{\op}^2 = O_p (T^{-1}\abar \max\{1, \abar^{(3-2\varsigma)/\rho}\} + \abar^{(1-2\varsigma)/\rho} ). 
			\end{equation} 	Thus, $\|\widehat{\mathcal A}(\widehat{\mathcal  W}) - \mathcal A\|_{\op} = o_p(1)$ for any $\rho > 2$ and $\varsigma > 1/2$. 
\item If the assumptions in Theorem \ref{thm:convrate2} when $\lambda_j$, $f_j$ and $\xi_j$ are, respectively, replaced with $\bar{\lambda}_j$, $\bar{f}_j$ and $\bar{\xi}_j$ are also satisfied, then  Theorem \ref{thm2w1addone}.\ref{thm2w1addoneb} holds and 
\begin{align*}
&\Vert \mathcal A(\bar \Pi_{\K}-\Pi_{\K}) \zeta \Vert   =\begin{cases}  O_p(T^{-1/2})  & \text{if }\rho/2+2<\varsigma+\delta_{\zeta},\\
	O_p( T^{-1/2}\max\{\log\abar, \abar^{(\rho/2-\varsigma-\delta_{\zeta}+2)/\rho} \} )  & \text{if }\rho/2+2\geq \varsigma+\delta_{\zeta},\\
\end{cases} \nonumber\\
&	\|\mathcal A(\Pi_{\K} - \mathcal I)\zeta\| = O_p(\abar^{(1/2-\varsigma -\delta_{\zeta})/\rho}).  
\end{align*}
	\end{enumerate} 
\end{theo}
\begin{proof}
Let $\breve{f}_j$ be the eigenvector corresponding to the $j$-th largest eigenvalue of $\widehat{\mathcal C}_{x\breve{z}}^\ast\widehat{\mathcal C}_{x\breve{z}}$ and let   $\bar{f}^s_j = \sgn\{\langle \breve{f}_j, \bar{f}_j \rangle \} \bar{f}_j$. Using the fact that $\|\widehat{\mathcal C}_{x\breve{z}}^\ast  \widehat{\mathcal C}_{x\breve{z}}-{\mathcal C}_{x\bar{z}}^\ast  {\mathcal C}_{x\bar{z}}\|_{\op}=\|\widehat{\mathcal C}_{x{z}}^\ast \widehat{\mathcal W} \widehat{\mathcal C}_{x{z}}-{\mathcal C}_{x{z}}^\ast  \mathcal W {\mathcal C}_{x{z}}\|_{\op} = O_p(T^{-1/2})$ and nearly identical arguments used in our proof of \eqref{eq000} and \eqref{eq001}, the following can be shown under the assumptions employed for (i):
	\begin{align}
		\Vert \breve{f}_j - \bar{f}_j^s \Vert^2 &= O_p(j^{2} T^{-1}), \label{eq000add} \\
		\|\mathcal A (\breve{f}_j - \bar{f}_j ^s)\|^2 &= O_p(T^{-1}) (j^{2-2\varsigma} + j^{\rho+2-2\varsigma}).\label{eq001add}
	\end{align}
Then, the desired result given in (i) can easily be obtained from \eqref{eq000add}, \eqref{eq001add} and similar arguments that we used to establish \eqref{eq:b09} and  \eqref{eq:b11}. 

Moreover, under the assumptions employed for (ii), the following can be shown by nearly identical arguments that are used to obtain \eqref{eqfhatdelta}:  \begin{equation}
		\langle \breve{f}_j-\bar{f}_j^s,\zeta\rangle^2 = 
		O_p(T^{-1})j^{-2\delta_{\zeta}+2} +  O_p(T^{-1})j^{-2\delta_{\zeta}+2+\rho}(1+o_p(1)). \label{eq002add}
	\end{equation} The remaining proof is almost identical to that of Theorem \ref{thm:convrate2};  by using \eqref{eq000add}-\eqref{eq002add}, it can be shown that $\Vert \widehat{\mathcal C}_{x\breve{z}} (\widehat{\mathcal C}_{x\breve{z}} ^\ast \widehat{\mathcal C}_{x\breve{z}})_{\K} ^{-1}\zeta - \mathcal C_{x\bar{z}} (\mathcal C_{x\bar{z}} ^\ast \mathcal C_{x\bar{z}})_{\K} ^{-1}\zeta\Vert = o_p(1)$. We can also analyze the term $\|\mathcal A(\bar{\Pi}_{\K}-\mathcal I)\zeta\|$ as $\|\mathcal A(\widehat{\Pi}_{\K}-\mathcal I)\zeta\|$ in our proof  of Theorem \ref{thm:convrate2}. The details are omitted.
\end{proof}


\section{Significance testing in functional endogenous linear model}\label{sec:hypo}
\phantomsection\label{rvlabel05}\commRV{Practitioners may often be interested in examining if various characteristics of $y_t$ depend on $x_t$.}  \commRV{For example, in our empirical application where $y_t(s)$ represents the wage of workers of skill level $s \in [0,1]$, practitioners might be interested in examining if the average wage ($\int_{0}^1 y_t(s)ds$) is affected by the considered explanatory variable $x_t$.} Likewise, because various characteristics of $y_t$ can be written as $\langle y_t,\psi \rangle$ for $\psi \in \mathcal H$, we in this section develop a significance test for examining if $\langle y_t,\psi \rangle$ is affected by $x_t$. Specifically, for any  $\psi \in \mathcal H$,  observe that 
\begin{equation*}
\langle y_t,\psi \rangle = \langle x_t,\mathcal A^\ast \psi \rangle + \langle u_t, \psi\rangle.
\end{equation*}
We then want to test the following null and alternative hypotheses:
\begin{equation} \label{eqhypo}
H_0: \mathcal A^\ast \psi = 0	\quad \text{v.s.} \quad H_1: \mathcal A^\ast \psi \neq 0.
\end{equation}
The null hypothesis means that the characteristic $\langle y_t,\psi \rangle $ of $y_t$ does not linearly depend on $x_t$. Note that $\widehat{\mathcal C}_{yz}\psi= \widehat{\mathcal C}_{xz} \mathcal A^\ast \psi+  \widehat{\mathcal C}_{uz}\psi$, and hence $\widehat{\mathcal C}_{yz}\psi$  reduces to $ \widehat{\mathcal C}_{uz}\psi$ if the null is true; moreover, in this case, $\sqrt{T}\widehat{\mathcal C}_{uz}\psi$ turns out to weakly converge to a $\mathcal H$-valued Gaussian random element under relevant assumptions.   Using this property, we develop a significance test, which is described by Theorem \ref{thmsigtest}.


\begin{theo}\label{thmsigtest}
Suppose that (i) $\mathcal C_{uu}$ is positive definite, (ii)  either Assumption \ref{assum1} or Assumption \ref{assum1a} holds, (iii) $\bar{\mathcal A}$ is the FIVE (if  Assumption \ref{assum1} holds) or the F2SLSE (if Assumption \ref{assum1a} holds) and  the other assumptions for $\|\bar{\mathcal A}-\mathcal A\|_{\op} \pto 0$ are satisfied (see Theorems \ref{thm2w1}, \ref{thm:convrate},  \ref{thm2w} and \ref{thm3:convrate}).   Let $\bar{u}_t = y_t- \bar{\mathcal A}x_t$ and let $\mathcal J$ denote $T \Vert  \hat{c}_{\psi}  ^{-1} \widehat{\mathcal C}_{yz} \psi\Vert^2$, where $\hat{c}_{\psi}^2 = \langle T^{-1}\sum_{t=1}^T \bar{u}_t\otimes \bar{u}_t (\psi), \psi \rangle$. Then, the following hold (below $\varkappa_j \sim_{\text{iid}} N (0,1)$). \begin{enumerate}[(i)]
\item \label{thm5.1}  $\mathcal J \dto \sum_{j=1} ^{\infty} \mu_j \varkappa_j^2$ under $H_0$ of \eqref{eqhypo} while $\mathcal J \pto \infty$ under $H_1$  of \eqref{eqhypo}.
\item \label{thm5.2} \commRV{Let $\hat{q}_{1-a_0}$ be the $100(1-a_0)\%$ quantile of $\sum_{j=1} ^{D} \hat{\mu}_j \varkappa_j^2$ for $a_0 \in (0,1)$, $D \to \infty$ and $D = o(T^{1/2})$. Then $\mathbb{P}\left\{ \mathcal J >\hat{q}_{1-a_0} \right\} \to a_0$ under $H_0$ of \eqref{eqhypo} while $\mathbb{P}\left\{ \mathcal J >\hat{q}_{1-a_0} \right\} \to 1$ under $H_1$ of \eqref{eqhypo}.} 
\end{enumerate}
\end{theo}
\begin{proof}
	For notational convenience,  let $c_\psi = \langle \mathcal C_{uu}\psi, \psi \rangle ^{1/2}$. To show \ref{thm5.1}, first note that \begin{equation}
					\sqrt{T} \widehat {\mathcal C}_{yz} \psi =   \sqrt{T}\widehat{\mathcal C}_{xz} \mathcal A^\ast \psi   +   \sqrt{T}\widehat{\mathcal C}_{uz} \psi .\label{tm5.eq1}
				\end{equation} 
				Under $H_0$, the first term in \eqref{tm5.eq1} is equal to zero, and thus $\sqrt{T}\widehat{\mathcal C}_{yz} \psi = T^{-1/2} \sum_{t=1} ^Tc_\psi \psi_t$, where $\psi_t =c_\psi^{-1} \langle u_t, \psi \rangle z_t$. Then, note that $\mathbb E [ \psi_t ] = 0$, $\mathbb E[\psi_t \otimes \psi_t] =   \mathcal C_{zz}$ and  $\{   \langle \psi_t , \zeta \rangle  \}_{t \geq1}$ is a real-valued martingale difference sequence for any $\zeta \in \mathcal H$. Thus, by applying  nearly identical arguments that are used to show \eqref{eqpf0008add} and \eqref{eqpf0009add}, it can be shown that, for any $\zeta \in \mathcal H$ and $m > 0$, $T^{-1/2}\sum_{t=1}^T \langle \psi_t , \zeta \rangle \dto N (0, \langle\mathcal C_{zz} \zeta, \zeta\rangle)$ and $	\limsup_{n\to \infty}	\limsup_{T} \mathbb{P}( \sum_{j=n+1}^\infty  \langle  T^{-1/2}\sum_{t=1} ^T \psi_t, g_j  \rangle^2  > m) = 0$ since $\mathcal C_{zz}$ is Hilbert-Schmidt. Therefore, $T^{-1/2} \sum_{t=1} ^T \psi_t \dto N(0,\mathcal C_{zz})$, and we note that $N(0,\mathcal C_{zz}) \stackrel{d}{=} \sum_{j=1} ^\infty \sqrt{\mu}_j \varkappa_j g_j$,  where $\varkappa_j \sim_{\text{iid}} N(0,1)$ across $j$. The rest of the proof follows from the consistency of $\hat{c}_\psi $ (Corollaries~\ref{cor1} and~\ref{cor2}), continuous mapping theorem, and orthonormality of $\{  g_j \}_{j \geq 1}$. 
				
				Under $H_1$, the first term in \eqref{tm5.eq1} is not equal to zero, and, by combining the results given above, we find that $\Vert \sqrt{T} \widehat{\mathcal C}_{yz} \psi\Vert ^2 = T\Vert  \mathcal C_{xz} \mathcal A^\ast \psi\Vert^2 + O_p(1)$ holds in this case. Since $\ker \mathcal C_{xz} = \{ 0 \}$ under either of Assumptions~\ref{assum1} or \ref{assum1a}, we find that $\Vert  \mathcal C_{xz} \mathcal A^\ast \psi\Vert^2 >0$ under $H_1$, and therefore $\mathcal J \pto \infty$.
				
				To show \ref{thm5.2}, note that $
				|\sum_{j=1} ^{\infty} \mu_j \varkappa_j ^2 - \sum_{j=1} ^{D} \hat\mu_j \varkappa_j ^2| \leq \sum_{j=1} ^{D} |\mu_j - \hat \mu_j| |\varkappa_j ^2| + \sum_{j=D+1} ^\infty \mu_j \varkappa_j^2.$ 	The first term in the right hand side is bounded above by $\sup_{j\geq 1} |\mu_j - \hat{\mu}_j| \sum_{j=1} ^{D} \varkappa_j ^2$, where $D^{-1}\sum_{j=1} ^{D} \varkappa_j ^2 = O_p(1)$ by Markov's inequality, and $\sup_{j\geq 1} |\mu_j -  \hat{\mu}_j| \leq \Vert \widehat{\mathcal C}_{zz} - \mathcal C_{zz} \Vert_{\op} \leq O_p(T^{-1/2})$ under Assumption~\ref{assum1}.\ref{assum1.7} and Lemma 4.2 in \cite{Bosq2000}. Since $\mathcal C_{zz}$ is nonnegative, we also have $\mathbb E[|\sum_{j=D+1} ^\infty \mu_j \varkappa_j^2|] \leq  \sum_{j=D+1} ^\infty \mu_j \to 0$ as $D \to \infty$, which implies that $\sum_{j=D+1} ^\infty \mu_j \varkappa_j^2$ is $o_p(1)$. Since $D\to \infty$ and $D/\sqrt{T} \to 0$ as $T \to \infty$, we find that $	|\sum_{j=1} ^{\infty} \mu_j \varkappa_j ^2 - \sum_{j=1} ^{D} \hat\mu_j \varkappa_j ^2| = o_p(1)$. \commRV{From this result, we find that $\hat{\eta}_{1-a_0}$ converges to $\eta_{1-a_0}$ (see e.g., Lemma 21.2 of \citealp{vaart_1998}), where $\eta_{1-a_0}$ is the  $100(1-a_0)\%$ quantile of the distribution function $G$ of $\sum_{j=1} ^{\infty} {\mu}_j \varkappa_j^2$. It is then obvious that $1-\mathbb{P}\left\{ \mathcal J >\hat{\eta}_{1-a_0} \right\}$ converges to $G(\eta_{1-a_0})$. By combining this with the previous result Theorem \ref{thmsigtest}.\ref{thm5.1}, the desired limiting behavior of $\mathbb{P}\{ \mathcal J >\hat{q}_{1-a_0}\}$ is established. }
\end{proof}
\commRV{Theorem \ref{thmsigtest}.\ref{thm5.1} shows that the asymptotic null distribution of the proposed statistic does not depend on  $\psi \in \mathcal H$, but does depend on all the eigenvalues of $\mathcal C_{zz}$; this means that there are infinitely many nuisance parameters. However, we can approximate the limiting distribution using the estimated eigenvalues of $\mathcal C_{zz}$ and thus implement a valid asymptotic test without a significant difficulty, as detailed in Theorem~\ref{thmsigtest}.\ref{thm5.2}; once the estimated eigenvalues $\hat{\mu}_j$ are obtained, then the approximate quantile $\hat{q}_{1-a_0}$ for any significance level $a_0 \in (0,1)$ can easily be computed by Monte Carlo simulations. Thus, implementation of the proposed test in practice is straightforward, which will be further illustrated in Section~\ref{sec:sim}.}

\begin{remark}\label{remtest} \normalfont
The test proposed in Theorem \ref{thmsigtest} can obviously be extended to examine the following hypotheses:  
\begin{equation} \label{eqremtest}
H_0: \mathcal A^\ast \psi = \psi_0	\quad \text{v.s.} \quad H_1: \mathcal A^\ast  \psi \neq \psi_0,
\end{equation}
for any  $\psi_0 \in \mathcal H$.  The extension only requires redefining  $\mathcal J$ as $T\Vert \hat c_\psi ^{-1} (\widehat{\mathcal C}_{yz}\psi - \widehat{\mathcal C}_{xz}\psi_0)\Vert^2 $, and this does not make any change in the convergence results given in Theorem \ref{thmsigtest}.  
\end{remark} 

\subsection*{Finite sample performance}\label{sec_app_num2}
In this section, we explore the finite sample performance of the test for examining the null and alternative hypotheses \eqref{eqremtest}, which is proposed in Remark \ref{remtest} of Section~\ref{sec:hypo}. To this end, we consider the DGP employed in Section~\ref{subsub: exp1.2}.
The test statistic $\mathcal J$ is computed with $\hat{c}_\psi$ obtained from the FIVE (see Theorem \ref{thmsigtest} and Corollary \ref{cor1}).   For each realization of the DGP, the critical value  at 5\% significance level is obtained from 500 Monte Carlo simulations of the distribution given in Theorem~\ref{thmsigtest}.\ref{thm5.2} with $D = \lceil T^{1/3} \rceil$.  The finite sample properties of the test are investigated by computing its rejection probabilities when $\mathcal A^\ast \zeta = \psi_0 + c_\zeta \tilde \psi$ holds, where $ \zeta$  is defined in Section \ref{subsub: exp1.2}, $c_\zeta  \in \{0,0.01,\ldots,0.5\}$, and  $\tilde \psi$ is a perturbation element with unit norm and randomly generated for each realization of the DGP. Specifically,  $\tilde \psi = \ddot{\psi}/\|\ddot{\psi}\|$ and $\ddot{\psi}=\sum_{j=1} ^{11} \ddot q_{3,j} \xi_j$, where $\ddot q_{3,j}\sim_{\text{iid}}\text{N} (0,0.5^{2(j-1)})$ across $j$ and $\mathbb{E}[\ddot q_{1,i}\ddot q_{3,j}]=0$ for all $i$ and $j$. This section only considers the case where $\sigma_{\eta} = 0.5$; in unreported simulations, we also investigated the performance of the test (i) when $\hat{c}_\psi$ is computed from the F2SLSE and (ii) when $\sigma_\eta$ is given by $0.9$, but found no significant difference.

\begin{figure}[h!] 
\caption{Simulation results for Experiment 2: rejection probability of $\mathcal J$ when $\mathcal A ^\ast \psi =\psi_0 +  c_\zeta  \tilde \psi $\label{fig.power} }\vspace{-.5em}
\begin{subfigure}{.32\textwidth}\subcaption{Sparse Design \label{fig.power.j1.d}}
\includegraphics[width=\textwidth]{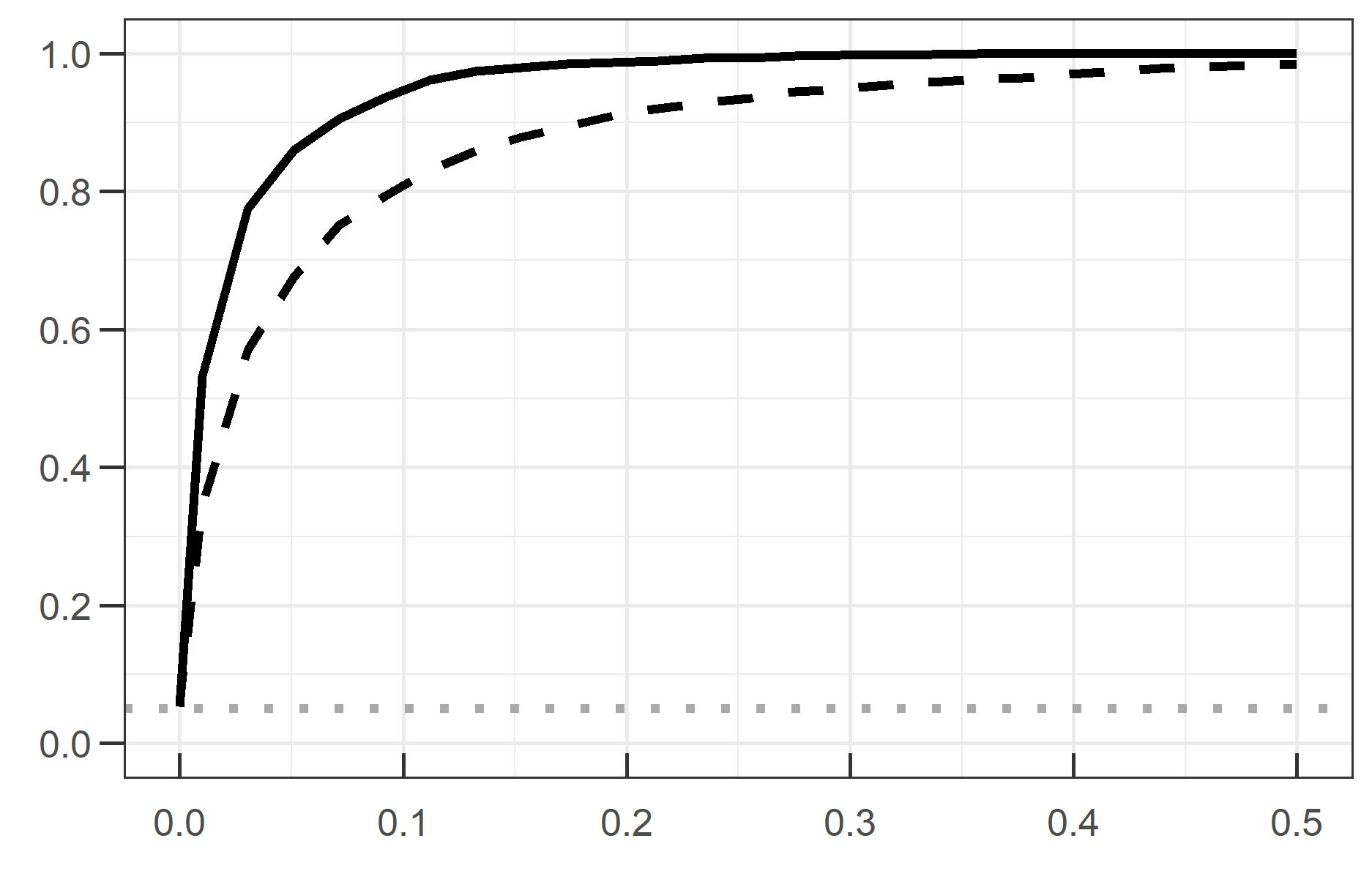} 
\end{subfigure}
\begin{subfigure}{.32\textwidth} \subcaption{Exponential Design  \label{fig.power.j1.e}}
\includegraphics[width=\textwidth]{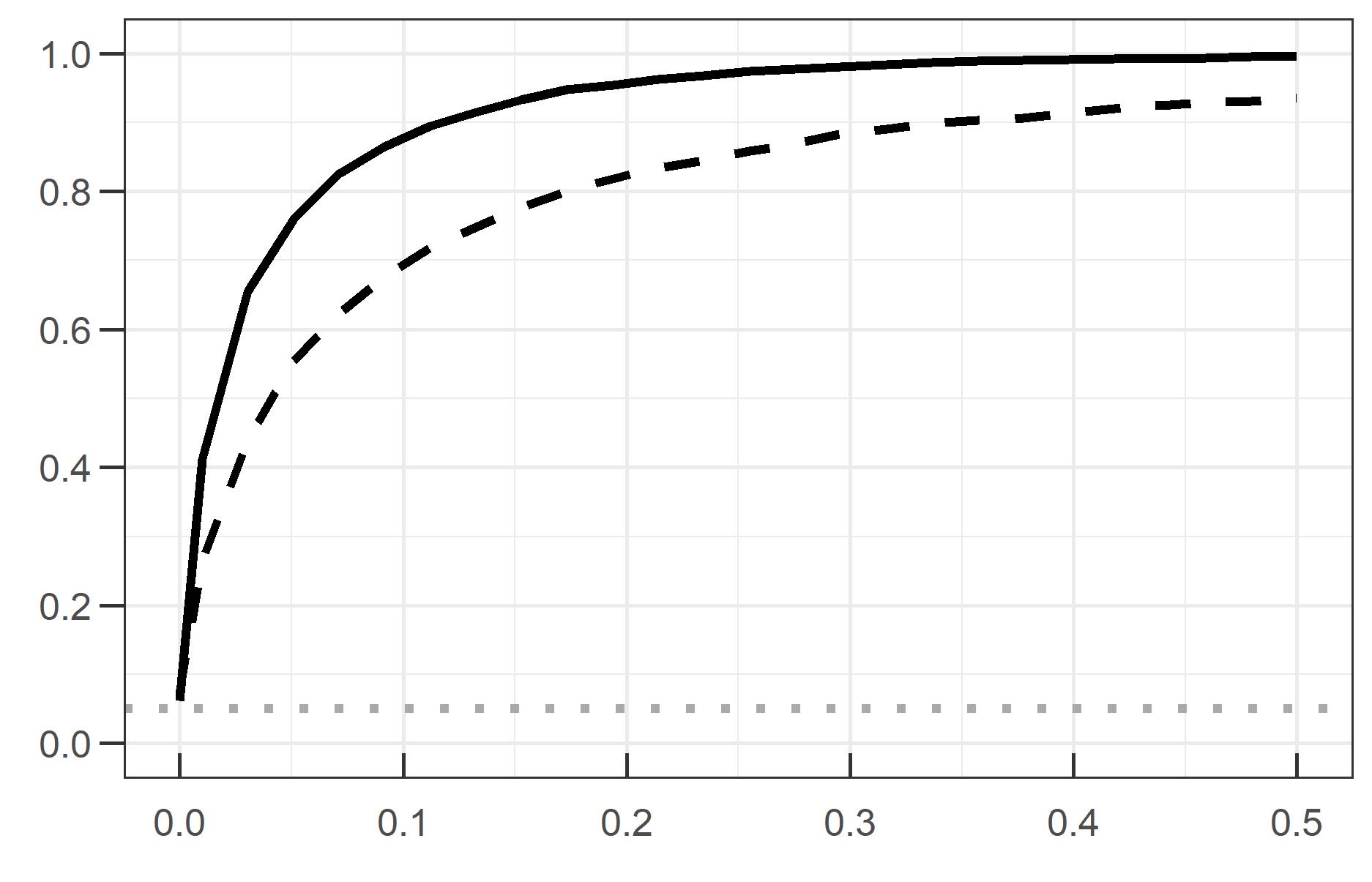} 
\end{subfigure}
\begin{subfigure}{.32\textwidth} \subcaption{Noisy Design  \label{fig.power.j1.f}}
\includegraphics[width=\textwidth]{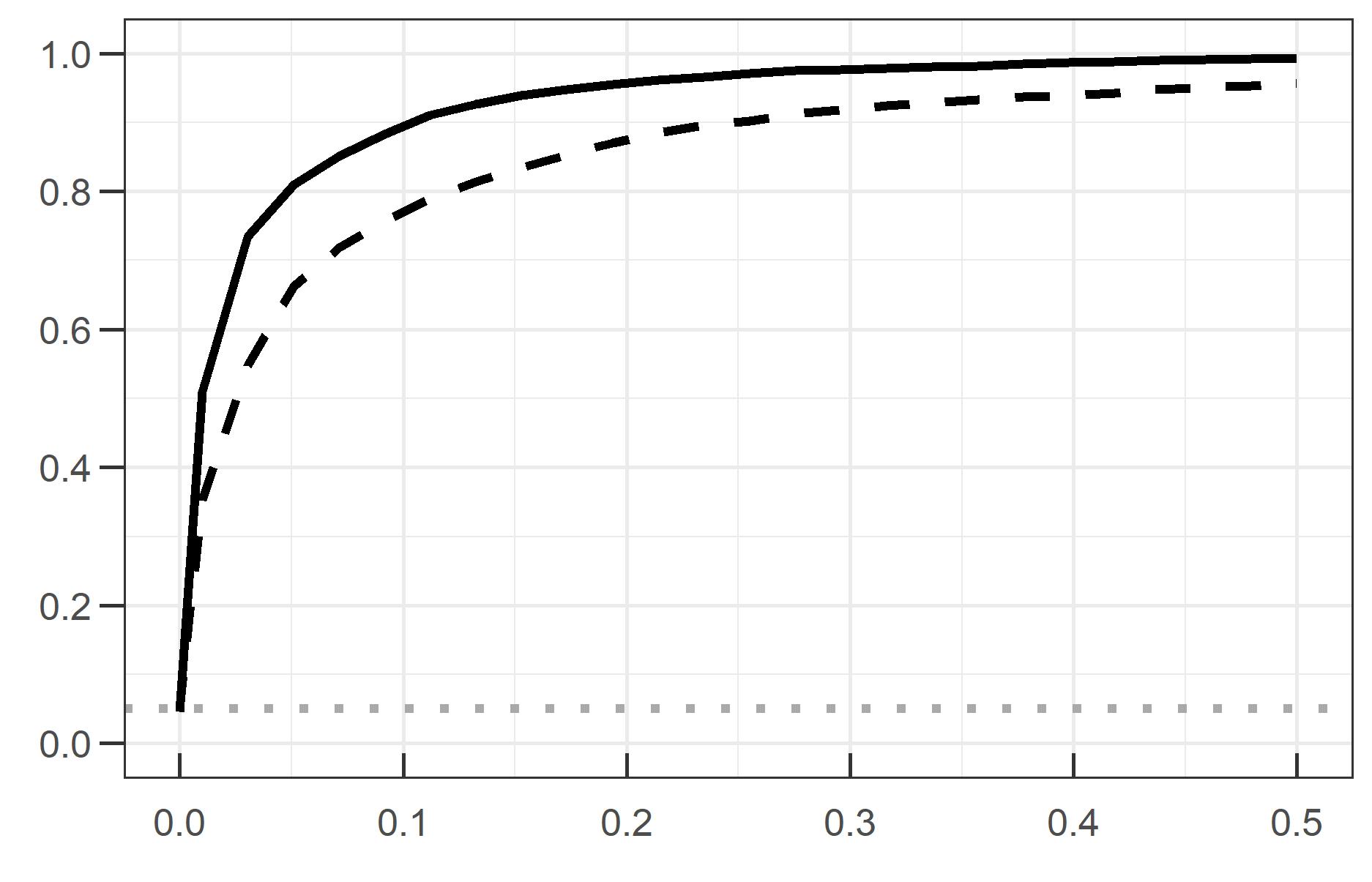} 
\end{subfigure}
{\footnotesize Notes: Based on 1,000 Monte Carlo replications. The rejection probability of $\mathcal J$ when $\hat c_\psi$ is computed from the FIVE is reported according to the value of $c_\zeta  \in \{0,0.01,\ldots, 0.5\}$. 
}
\end{figure}

Figure~\ref{fig.power} shows the rejection probability of the test depending on $c_\zeta $. The dashed and solid lines indicate the rejection probabilities when $T = 200$ and $500$, respectively, and the dotted horizontal lines indicate the nominal size of the test. As expected, the proposed test exhibits a higher power as $c_\zeta $ gets deviated from zero, and it seems that the power of the test more rapidly increases in the sparse design compared to those in the other designs. Moreover, in all the considered cases, the test has excellent size control, as it can be seen from the case where $c_\zeta  = 0$.

\bibliographystyle{apalike}

\bibliography{bib_fregiv}

\end{document}